\documentclass[a4paper,11pt]{article}

\usepackage[utf8]{inputenc}
\usepackage[british]{babel}
\usepackage[mono=false]{libertine}
\usepackage[T1]{fontenc}
\usepackage{amsthm}
\usepackage[top=2.5cm, bottom=2.5cm, left=2cm, right=2cm]{geometry}
\usepackage{xcolor}
\usepackage[font=small]{caption}
\usepackage{enumitem}
\usepackage{tikz-cd}

\usepackage{stackrel}
\usepackage[normalem]{ulem}
\usepackage{soul}

\usepackage{amsmath}
\usepackage{blindtext}
\usepackage{etoolbox}

\newtheorem{claim}{Claim}

\usepackage{titlesec}
\makeatletter
\newcommand{\setappendix}{Appendix~\thesection:~~}
\newcommand{\setsection}{\thesection~~}
\titleformat{\section}{\bfseries\LARGE}{%
	\ifnum\pdfstrcmp{\@currenvir}{appendices}=0
	\setappendix
	\else
	\setsection
\fi}{0em}{}
\makeatother
\usepackage[titletoc]{appendix}

\usepackage{cite}
\usepackage{array}
\usepackage{url}
\usepackage{mathtools}
\usepackage{amssymb,amsfonts,amsmath}
\usepackage{dsfont}
\usepackage{stmaryrd}
\usepackage{bm}
\usepackage[pdfpagemode=UseNone,bookmarksopen=false,colorlinks=true,urlcolor=blue,citecolor=magenta,citebordercolor=blue,linkcolor=blue]{hyperref}
\usepackage{footmisc}
\DefineFNsymbols{mySymbols}{{\ensuremath\dagger}{\ensuremath\Diamond}{\ensuremath{\ast}}{\ensuremath{\star}}}
\setfnsymbol{mySymbols}
\graphicspath{{./Figures}}

\DeclareUnicodeCharacter{00A0}{~}

\def \({\left(}
\def \){\right)}
\def \[{\left[}
\def \]{\right]}
\def \nn{\nonumber \\}

\newcommand{\defeq}{\vcentcolon=}

\newcommand{\bY}{{\textbf {Y}}}

\newcommand{\bU}{{\textbf {U}}}
\newcommand{\bV}{{\textbf {V}}}
\newcommand{\bW}{{\textbf {W}}}
\newcommand{\bZ}{{\textbf {Z}}}

\newcommand{\bw}{{\textbf {w}}}

\newcommand{\bg}{{\textbf {g}}}
\newcommand{\bA}{{\textbf {A}}}
\newcommand{\bB}{{\textbf {B}}}

\newcommand{\bX}{{\textbf {X}}}
\newcommand{\bx}{{\textbf {x}}}

\newcommand{\bz}{{\textbf {z}}}

\newcommand{\bu}{{\textbf {u}}}
\newcommand{\bS}{{\textbf {S}}}

\newcommand{\ba}{{\textbf {a}}}
\newcommand{\cC}{{\mathcal{C}}}

\newcommand{\sign}{\text{ sign}}

\newcommand{\be}{\begin{equation}}
\newcommand{\ee}{\end{equation}}

\newcommand\smallO{
  \mathchoice
    {{\scriptstyle\mathcal{O}}}
    {{\scriptstyle\mathcal{O}}}
    {{\scriptscriptstyle\mathcal{O}}}
    {\scalebox{.7}{$\scriptscriptstyle\mathcal{O}$}}
  }
\newcommand{\bea}{\begin{align}}
\newcommand{\eea}{\end{align}}

\newtheorem{theorem}{Theorem}
\newtheorem{lemma}{\textbf{Lemma}}
\newtheorem{thm}{\textbf{Theorem}}
\newtheorem{remark}{\textbf{Remark}}
\newtheorem{proposition}{\textbf{Proposition}}
\newtheorem{corollary}{\textbf{Corollary}}
\newtheorem{definition}{\textbf{Definition}}

\DeclareMathAlphabet{\varmathbb}{U}{bbold}{m}{n}

\newcommand{\EE}{\mathbb{E}}

\renewcommand{\P}{{\rm P}}

\usepackage{notations}
\usepackage{setspace}

\DeclareMathOperator*{\argmin}{arg\,min}
\newcommand*{\QEDA}{\hfill\ensuremath{\square}}
\makeatletter
\def\blfootnote{\xdef\@thefnmark{}\@footnotetext}
\makeatother

\newenvironment{talign}
 {\let\displaystyle\textstyle\align}
 {\endalign}
\newenvironment{talign*}
 {\let\displaystyle\textstyle\csname align*\endcsname}
 {\endalign}

\usepackage{chngcntr}
\counterwithin*{section}{part}

\begin{document}
\title{Optimal Errors and Phase Transitions \\in High-Dimensional Generalized Linear Models}
\author{Jean Barbier$^{\dagger\Diamond\star\otimes}$, Florent Krzakala$^{\star}$, Nicolas Macris$^{\dagger}$, L\'eo Miolane$^{*\otimes}$ and Lenka Zdeborov\'a$^\cup$}
\date{}
\maketitle
\blfootnote{
\!\!\!\!\!\!\!\!\!\!\!\!\!$\dagger$ Laboratoire de Th\'eorie des Communications, Facult\'e Informatique et Communications, Ecole Polytechnique F\'ed\'erale de Lausanne, Suisse. \\
$\Diamond$ International Center for Theoretical Physics, Trieste, Italy.\\
$\star$ Laboratoire de Physique Statistique, CNRS \& Universit\'e Pierre et Marie Curie \& \'Ecole Normale Sup\'erieure \& PSL Universit\'e, Paris, France.\\
$*$ D\'epartement d'Informatique de l'ENS, \'Ecole Normale Sup\'erieure \& CNRS \& PSL Research University \& Inria, Paris, France. \\
$\cup$ Institut de Physique Th\'eorique, CNRS \& CEA \& Universit\'e Paris-Saclay, Saclay, France.\\
$\otimes$ Corresponding authors: jean.barbier@lps.ens.fr, leo.miolane@gmail.com\\
}
\begin{abstract}
  Generalized linear models (GLMs) arise in high-dimensional
  machine learning, statistics, communications and signal processing.
  In this paper we analyze GLMs when the data matrix is random, as
  relevant in problems such as compressed sensing, error-correcting codes
  or benchmark models in neural networks.
  We evaluate the mutual information (or ``free entropy'') from which we deduce the Bayes-optimal estimation and generalization errors. Our analysis applies to the high-dimensional limit where both the
  number of samples and the dimension are large and their ratio is fixed.
  Non-rigorous predictions for the optimal errors existed for special cases of GLMs, e.g. for
  the perceptron, in the field of statistical physics based on
  the so-called replica method. Our present paper rigorously establishes those
  decades old conjectures and brings forward their
  algorithmic interpretation in terms of performance of the
  generalized approximate message-passing algorithm.
  Furthermore, we tightly characterize, for many learning problems,
  regions of parameters for which this algorithm achieves the optimal
  performance, and locate the associated sharp phase transitions
  separating learnable and non-learnable regions.
  We believe that this random version of GLMs can serve as a
  challenging benchmark for multi-purpose algorithms.

  This paper is divided in two parts that can be read independently: The first part (main part) presents the model and main results, discusses some applications and sketches the main ideas of the proof. The second part (supplementary informations) is much more detailed and provides more examples as well as all the proofs.
\end{abstract}
%

\setcounter{tocdepth}{2}
{%
	\singlespacing
	\hypersetup{linkcolor=black}
	\tableofcontents
}
%
\part{Main part}
\section{Introduction}
As datasets grow larger and more complex, modern data
analysis requires solving high-dimensional estimation problems
with very many parameters. Developing algorithms up to
the task and understanding their limitations has become a major
challenge in computer science, machine learning, statistics, signal processing, communications and related fields.

In the present contribution, we address this challenge in the case of generalized linear estimation models (GLMs)
\cite{nelder1972generalized,mccullagh1984generalized} where data are
generated as follows: Given a $n$-dimensional vector $\bX^*$, hidden
to the statistician, he/she observes instead a $m$-dimensional vector
$\bY$ where each component reads
\begin{align}\label{measurements_0}
	Y_\mu = \varphi\Big(\frac{1}{\sqrt{n}} [\boldsymbol{\Phi}	\bX^*]_{\mu},A_\mu \Big)\,, \qquad 1\leq \mu \leq m\,,
\end{align}
where $\boldsymbol{\Phi}$ is a $m \times n$ ``measurement'' or
``data'' matrix, the random variables $(A_\mu) \iid P_A$ account for noise (or randomness) of the model. The model is ``linear'' because the output $Y_\mu$
depends on a {\it linear} combination of the data
$z_\mu=\frac{1}{\sqrt{n}}[\boldsymbol{\Phi} \bX^*]_{\mu}=\frac{1}{\sqrt{n}}\sum_{i=1}^n\Phi_{\mu i}X^*_i$. The GLM generalizes the
ordinary linear regression by allowing the output function
$\varphi(z,A)$ to be non-linear and/or stochastic; in the case of a deterministic
model we simply write $\varphi(z)$. Explicit
examples will be given below.

GLMs belong to the realm of supervised learning and arise in a wide
variety of scientific fields. In signal processing one usually
observes $Y_\mu$ given as a linear combination of the signal-elements
$\bX^*$. In a range of applications these observations are 
obtained via a non-linear function $\varphi$. In optics or X-ray crystallography one
often measures only the amplitude of
$[\boldsymbol{\Phi} \bX^*]_{\mu} $, leading to the phase retrieval
problem \cite{fienup1982phase}. A real-valued analog is the problem of
sign-retrieval when we only observe $|[\boldsymbol{\Phi} \bX^*]_{\mu}|$ \cite{demanet2014stable,candes2013phaselift}. Observations are
sometimes quantized in order to reduce the storage, leading
for instance to the problem of 1-bit compressed sensing \cite{boufounos20081}. In statistics and machine learning,
classification is often described via a GLM where
the output function $\varphi$ is discrete and
corresponds to the labels that classify the data-points
$\boldsymbol{\Phi}_{\mu}$ \cite{nelder1972generalized,mccullagh1984generalized,buhlmann2011statistics}.
GLMs with non-linear output functions are also the basic
building blocks of each layer of neural networks
\cite{lecun2015deep}: $\varphi$ corresponds to the
activation, the rows of the matrix $\boldsymbol{\Phi}$ are
different data samples while $\bX^*$ are the set of synaptic
weights to be learned.

There are two main learning problems in GLMs: $i)$ The {\it estimation} task requires, knowing the measured vector $\bY$ and the matrix $\boldsymbol{\Phi}$, to infer the unknown vector
$\bX^*$; $ii)$ the {\it prediction} or {\it generalization} task instead requires, again knowing $\bY$ and $\boldsymbol{\Phi}$, to predict accurately new values $Y_{\rm new}$ when new rows (i.e.\ data-points) are added to the matrix $\boldsymbol{\Phi}$.

In the present paper we build a rigorous theory for both these tasks for {\it random instances} of the GLM. In this setting each element $\Phi_{\mu i}$ of the matrix is sampled
independently from a probability distribution of zero mean and unit
variance, and the unknown vector $\bX^*$ has been also created
randomly from a probability distribution $P_0$, with each of its
components $X^*_1, \dots, X^*_n \iid P_0$. Since our main aim is to
study the intrinsic information-theoretic and algorithmic limitations
caused by the lack of samples and/or the amplitude of the noise, we assume throughout this paper that $P_0$ and $\varphi$ are known to the statistician (if they are not
the task can only be harder). Our results are derived in the
challenging and interesting high-dimensional limit where
$m,n\to \infty$ while $m/n\to\alpha$ a constant. Random instances of GLMs are both practically and theoretically relevant in many different contexts:

a) In {\it signal processing}, GLM estimation with a random matrix
$\boldsymbol{\Phi}$ has been studied with considerable attention in
the context of compressed
sensing~\cite{donoho2005sparse,candes2006near,donoho2009message}
where a $n$-dimensional sparse signal is recovered from $m<n$ noisy
measurements. While standard compressed sensing focused on the linear
case --where $\varphi(z,A)=z+A$ with a Gaussian noise $A$--
the generalized case was also widely studied
\cite{GAMP,REVIEWFLOANDLENKA}, especially for quantized output
\cite{kamilov2011optimal} and 1-bit compressed
sensing \cite{boufounos20081,xu2014bayesian} where
$\varphi(z,A)={\rm sign}(z+A)$, as well as for compressive phase
retrieval when $\varphi(z,A)=|z+A|$ \cite{schniter2015compressive}.

b) In {\it statistical
    learning}, important activity is dedicated to
  understand the limitation of learning with data generated by GLMs, both in the linear case, e.g.\ in the context of ridge
  regression or LASSO \cite{lassoMontanariBayati}, or with non-linear
  probabilistic output, e.g.\ logistic
  regression. Random instances were studied in particular in
  the context of so-called M-estimators \cite{ElKaroui,Donoho2016,gribonval2013reconciling,advani2016equivalence}.

  c) In studies of {\it artificial neural networks} there has been a
  large amount of works using random instances of GLMs, with
  $\varphi$ playing the role of a non-linear activation function. In
  this context the random GLM was introduced as the teacher-student
  setting for the perceptron in the pioneering work of Gardner and Derrida
  \cite{gardner1989three}. Large volume of work followed and is
  reviewed, e.g., in
  \cite{ReviewTishby,RevModPhys,engel2001statistical}. While initial
  works concentrated on a simple activation functions
  $\varphi(z)={\rm sign}(z-K)$ ($K$ is the threshold constant), many other functions were
  considered, e.g.\ in
  \cite{EngelReimers_wedge,BexBroeck_wedge,hosaka2002statistical}. Recently,
  the study of random instances of neural networks have emerged as a
  key ingredient in understanding the performance of deep learning
  algorithms \cite{Baldassi29112016,StatGen}. Computing mutual informations in GLMs is also a critical issue in confirming the information bottleneck scenario of \cite{tishby99information,DBLP:journals/corr/Shwartz-ZivT17}

  d) In {\it communications}, error-correcting codes that use
  random constructions are particularly efficient, as discussed by Shannon in his seminal paper
  \cite{shannon2001mathematical}. Random instances of GLMs describe
  both the setting of code-division multiple access --a multi-user access method used in communication technologies
  \cite{tanaka2002statistical,guo2005randomly}-- as well as an error
  correction scheme called sparse superposition codes, that have been shown
  to achieve the Shannon capacity for {\it any type} of
  noisy channel
  \cite{barron2010sparse,barbier2015approximate,rush2017capacity,barbier2016threshold,DBLP:journals/corr/BarbierDM17}.

%

Interestingly there is an important gap in the above volume of
work. On the one hand there are studies that rely on the algorithmic performance
of the so-called generalized approximate message-passing algorithm (GAMP)
\cite{mezard1989space,donoho2009message,GAMP}. GAMP is remarkable in
that its asymptotic ($n,m\to \infty$, $m/n\to\alpha$) performance can be analyzed
rigorously using the so-called state evolution
\cite{Bolthausen2014,bayati2011dynamics,bayati2015universality,javanmard2013state}.
However, GAMP is not expected to be always information-theoretically optimal. On the other hand, other results are concerned with the linear case of the GLM with additive Gaussian noise
for which the information-theoretically optimal performance was
established in \cite{BarbierDMK16,BarbierMDK17,reeves2016replica} (the
methodology of these works unfortunately does not generalize
straightforwardly to the important non-linear case or
to other types of additive noise). All the other works, giving
information-theoretic results for the non-linear case, are based on
powerful and sophisticated but {\it non-rigorous}
techniques originating in statistical physics of disordered
systems, such as the cavity and replica methods
\cite{mezard1990spin}. Historically, the first of these non-rigorous, yet correct,
results on information-theoretic limitations of learning was for the perceptron with binary weights and was established using the replica
method in \cite{gardner1989three,gyorgyi1990first,sompolinsky1990learning}, including a discontinuous phase
transition to perfect learning that appears as the ratio between
number of samples and the dimension exceeds $\alpha\approx 1.249$. 

In the present paper we close the above gap between mathematically rigorous
work and conjectures (some of them several decades old) from statistical mechanics. In particular, we
prove that the results for GLMs stemming from the
replica method are indeed correct and imply the optimal
value of both the estimation and generalization error. These results
are summarized in section ``Main results''. 
The proof is based on the {\it adaptive interpolation method} recently developed in \cite{barbier_stoInt} and is of
independent interest as it is applicable to a range of other models, see section ``Methods and proofs''
and the supplemantary informations (SI). We compare our
information-theoretic results to the
performance of the GAMP algorithm and its state
evolution (as reviewed briefly in section ``Main results'').
We determine regions of parameters where this algorithm is or is not
information-theoretically optimal. Up to technical assumptions (as
specified below), our results apply to all activation functions $\varphi$ and priors $P_0$,
thus unifying a large volume of previous work where many particular
functions have been analyzed on a case by case basis. This generality allows us
to provide a unifying understanding of the types of phase transitions and
phase diagrams that we can encounter in GLMs, which is as well of
independent interest and we devote section ``Application to learning
and inference'' 
to its presentation.

\section{Main results}
\label{sec:results_0}
This section summarizes our main results. Their formal statement
together with all technical assumptions and full proofs are provided in
section ``Methods and proofs'' and in the SI.

For the random GLM problem as defined in the introduction, the optimal way to estimate the ground-truth
signal/weights $\bX^*$ relies on its posterior probability
distribution
\begin{align}\label{posterior_0}
      P(\bx | \bY, \boldsymbol{\Phi} ) = \frac{1}{{\cal Z}(\bY,
  \boldsymbol{\Phi})}   \prod_{i=1}^n  P_0(x_i)   \prod_{\mu=1}^m
  P_{\rm out}\Big(Y_\mu \Big| \frac{[\boldsymbol{\Phi} \bx]_{\mu}}{\sqrt{n}}\Big)
\end{align}
where we used the prior $P_0$ of
$\bX^*$, and introduced the likelihood $P_{\rm out}$ that an
output $Y_\mu$ is observed given $\frac{1}{\sqrt{n}}[\boldsymbol{\Phi} \bx]_{\mu}$. 
$P_{\rm out}( \cdot \, | \, z)$ is the probability density function of $\varphi(z,A)$ (where again the r.v. $A \sim P_A$ accounts for noise).
This paper is concerned with the so-called {\it Bayes-optimal} setting
where the prior $P_0$ and the likelihood $P_{\rm out}$ that appear in
the posterior \eqref{posterior_0} were also used to generate the
ground-truth signal $\bX^*$ and the labels $\bY$, using a known random
matrix $\boldsymbol{\Phi}$. 

A first quantity of interest is the {\it free
entropy} (which is the {\it free energy} up to a sign) defined as
$f_n(\bY,\boldsymbol{\Phi}) \equiv \frac1n\ln{ {\cal Z}(\bY,\boldsymbol{\Phi})}$. The expectation of the free entropy is equal to minus the conditional entropy density of the observation $-\frac1nH(\bY|\boldsymbol{\Phi})$, as well as (up to an additive
constant) to the mutual information density between the signal and the observations $\frac{1}{n}I(\bX^*;\bY|\boldsymbol{\Phi})$.

\subsection{The free entropy}
Our first result is the rigorous determination of the free
entropy, in the high-dimensional asymptotic regime $n,m\to \infty$, $m/n\to\alpha$. For a random matrix
$\boldsymbol{\Phi}$ with independent entries of zero mean and unit variance, for output $\bY$
that was generated using \eqref{measurements_0}, 
and under appropriate
technical assumptions stated precisely in section ``Methods and proofs'', 
the free entropy converges in probability to:
\begin{align}
f_n(\bY,\boldsymbol{\Phi}) \equiv \frac1n\ln{ {\cal Z}(\bY,\boldsymbol{\Phi})} \xrightarrow[n \to \infty]{\mathbb{P}}  {\adjustlimits \sup_{q \in [0,\rho]} \inf_{r \geq  0}}  f_{\rm RS}(q,r;\rho) \label{RSf_0}
\end{align}
where $\rho\equiv \E_{P_0}[(X^*)^2]$ and where the {\it potential} $f_{\rm RS}(q,r;\rho)$ is
\begin{talign}
	f_{\rm RS}(q,r;\rho) &\equiv \psi_{P_0}(r)
  + \alpha \Psi_{P_{\rm out}}(q;\rho) - rq/2 \,, \label{result-one_0}\\
\psi_{P_0}(r) &\equiv 
\stackrel[Z_0,X_0]{}{\E}\!\! \ln \int dP_0(x) \,e^{r x X_0 +
                \sqrt{r} x Z_0   - r  x^2/2}\, ,  \label{psi0_0} \\ 
\label{PsiPout_0} \Psi_{P_{\rm out}}(q;\rho) &\equiv
\stackrel[V,W,\tilde Y_0]{}{\E}\!\!
\ln \int {\cal D}w P_{\rm out}(\tilde{Y}_0 | \sqrt{q}\, V \!+\!
  \sqrt{\rho - q}\, w) \,,
\end{talign}
where ${\cal D}w=dw \exp(-w^2/2)/\sqrt{2\pi}$ is a standard Gaussian measure and the scalar r.v. are independently sampled from
$X_0 \sim P_0$, then $V, W, Z_0 \iid {\cal N}(0,1)$
and $\tilde Y_0\sim P_{\rm out}(\cdot| \sqrt{q}\, V + \sqrt{\rho - q}\,
W)$.
Only the special linear case with Gaussian $P_{\rm out}$ was
known rigorously so far \cite{BarbierDMK16,BarbierMDK17,reeves2016replica}. 
Convergence of the {\it averaged} free entropy 
is precisely stated in Theorem \ref{th:RS_1layer_0}; 
the one in probability follows from concentration results in the SI.

One can check by explicit comparison that for specific choices of
$P_0$ and $P_{\rm out}$ the expression \eqref{result-one_0} is the replica-symmetric free entropy derived in numerous statistical physics papers (thus the $\rm{RS}$ in $f_{\rm RS}$), and in particular in
\cite{gardner1989three,
  mezard1989space,gyorgyi1990first,sompolinsky1990learning} for $\varphi(z)={\rm sign}(z)$. The formula for general
$P_0$ and $P_{\rm out}$ was conjectured based on the statistical
physics derivation in \cite{REVIEWFLOANDLENKA}.
Establishing \eqref{RSf_0} closes these old
conjectures and yields an important step
towards vindication of the cavity and replica
methods for inference, alongside with e.g. \cite{coja2017information,bayati2011dynamics}. We now discuss the main consequences of this formula.

\subsection{Overlap and optimal estimation error}
Our second result concerns the overlap between a sample $\bx$ from the posterior \eqref{posterior_0} and the ground-truth. We obtain that as $n,m \to \infty$, $n/m\to \alpha$, 
\begin{align}
  \frac1n\,\big|\bx\cdot \bX^*\big|\xrightarrow[n \to \infty]{\mathbb{P}} q^*\label{asymptOverlap_0}
\end{align}
whenever $q^*=q^*(\alpha)$ the maximizer in formula \eqref{RSf_0} is unique. This is the case for almost every $\alpha$ (see the SI). 
%

It is a simple fact of Bayesian inference that, given the measurements $\bY$
and the measurement matrix $\boldsymbol{\Phi}$, the estimator $\hat \bX$ that minimizes the mean-square error with the ground-truth
$\bX^*$ is the mean
of the posterior distribution
\eqref{posterior_0}, i.e. $\hat \bX= \E_{P(\bx|\bY,\bbf{\Phi})}[\bx]$. The minimum mean-square error (MMSE) that is
achieved by such ``Bayes-optimal'' estimator is deduced, again in the limit $n\to\infty, m/n\to\alpha$, as follows:
\begin{align}\label{eq:MMSEresult_0}
{\rm MMSE}=\frac1n\E\Big[\big\|\bX^*-\hat\bX\big\|^2\Big] \to \rho-q^* \, . 
\end{align}
We refer to Theorem \ref{th:errors_0} in section ``Main theorems'' for rigorous statements. Again the value of the MMSE was known rigorously so far only for the linear case with Gaussian noise
\cite{BarbierDMK16,reeves2016replica,BarbierMDK17} (and
conjectured for the non-linear case e.g. in \cite{REVIEWFLOANDLENKA}).

\subsection{Optimal generalization error}
Our third result concerns the prediction error, also called generalization error. Consider again the statistical model
\eqref{measurements_0}. 
To define the {\it
  Bayes-optimal generalization error}, one is given a new row of the
matrix/data point, denoted
$\boldsymbol{\Phi}_{\rm new} \in {\mathbb R}^n$ (in addition to the data $\boldsymbol{\Phi}$ and associated outputs $\bY$ used for the learning), and is asked to
estimate the corresponding output value $Y_{\rm new}$. We seek for an estimator $\hat Y_{\rm new}=\hat Y_{\rm new}(\bY,\boldsymbol{\Phi},\boldsymbol{\Phi}_{\rm new})$ that achieves ${\cal E}_{\rm gen} \equiv \min_{\hat Y_{\rm new}}\EE[ (Y_{\rm new}-\hat Y_{\rm new})^2
]$, i.e. that minimizes the MSE with the true $Y_{\rm new}$ obtained using
the ground-truth weights $\bX^*$. Such estimator is again obtained from the
posterior: $\hat Y_{\rm new}  = \E_{P_A(a)}\E_{P(\bx|\bY,\bbf{\Phi})}\varphi(\frac{1}{\sqrt{n}}\boldsymbol{\Phi}_{\rm new}\cdot
\bx,a)$. Note that this is different than the {\it plug-in} estimator $\tilde Y_{\rm new}=\varphi(\frac{1}{\sqrt{n}}\boldsymbol{\Phi}_{\rm new}\cdot \hat
\bX)$, which leads to a
worse MSE than $\hat Y_{\rm new}$. Yet it is often used in practice for deterministic models
since most algorithms for generalized linear
regression do not provide 
the full posterior distribution.

Our result states that the optimal generalization error follows from
the I-MMSE theorem \cite{GuoShamaiVerdu_IMMSE} applied to the free
entropy \eqref{RSf_0} (see the SI for the details). 
The optimal
generalization error reads as $n\to\infty$, $m/n\to\alpha$ ($q^*$ is the maximizer in \eqref{RSf_0})
\begin{equation}
{\cal E}_{\rm gen}\!\to \! \mathop{\mathbb{E}}_{V,a}\!\!\big[\varphi(\sqrt{\rho}\,V, \!a)^2\big] 
 \!- \! \mathop{\mathbb{E}}_V \!\big[\!\mathop{\mathbb{E}}_{w,a}\!\!\big[\varphi(\sqrt{q^*}\,V \!+ \! \sqrt{\rho \!- \!q^*} w, \!a)\big]^2\big], \label{Egen_generic_0}
\end{equation}
where $V, w\iid{\cal N}(0,1)$ and $a\sim P_A$. See again Theorem \ref{th:errors_0} in the section ``Main theorems'' for the precise statement (and Theorems 3 and 4 in the SI).



Note that for labels $\bY$ belonging to a discrete set the
MSE might not be a suitable loss and we are
more often interested in maximizing the so-called overlap, i.e.\ the
probability of obtaining the correct label. In that case the
Bayes-optimal estimator is computed as the argmax of the posterior
marginals, rather than as its mean, i.e.\ for discrete labels $\bar Y_{\rm new}  = {\rm argmax}_y {\mathbb{P}}(y= \varphi(\frac{1}{\sqrt{n}}\boldsymbol{\Phi}_{\rm new} \cdot \bx,a))$ where again $\bx$ is distributed according to \eqref{posterior_0}, $a\sim P_A$.
%
The replica method has been used to compute the optimal generalization
error for the perceptron where $\varphi(x)=\sign(z)$ in
the pioneering works of
\cite{gyorgyi1990first,ReviewTishby,opper1991generalization}. 
We note that in this special case the plug-in
estimator $\tilde Y_{\rm new}$ is actually equal to the optimal one $\bar Y_{\rm new}$. 

A final note concerns the issue of overfitting. In optimization-based approaches
to learning overfitting may lead to a generalization error which is too large as compared to the training error. In the 
Bayes-optimal setting the estimators are constructed in order not to
overfit.
This is related to general properties of Bayes-optimal inference
and learning that are called 
``Nishimori conditions'' in the physics literature
\cite{REVIEWFLOANDLENKA} and that turn out to be crucial in our proofs.


%

\subsection{Optimality of approximate message-passing}
While the three results stated above are of an information-theoretic nature, our fourth
one concerns the performance of an algorithm to solve random
instances of GLMs called generalized approximate
message-passing (GAMP) \cite{donoho2009message,GAMP,REVIEWFLOANDLENKA}, which is closely related to the TAP equations developed in statistical physics \cite{thouless1977solution,mezard1989space,Kaba}.


The GAMP algorithm can be summarized as follows \cite{donoho2009message,GAMP,REVIEWFLOANDLENKA}: 
Given initial estimates $\widehat{\bx}^0,\mathbf{v}^0$ for the
marginal posterior means and variances of the unknown signal vector $\bX^*$ entries, GAMP iterates the following equations, with $g_{\mu}^{0} =0$:
\begin{align*}
	\left\{
	\begin{array}{llll}
		V^{t}  &=& \overline {\mathbf{v}^{t-1}} \\
		\boldsymbol{\omega}^{t}  &=& \boldsymbol{\Phi} \widehat{\bx}^{t-1}/\sqrt n - V^{t} \mathbf{g}^{t-1} &\label{omega} \\
		g_{\mu}^{t} &=& g_{P_{\rm out}}(Y_{\mu},\omega^{t}_{\mu},V^{t}) &\forall\ \mu=1,\ldots m\\
		\lambda^t &=&  \alpha \,\overline {  g^2_{P_{\rm out}}(\bY,\boldsymbol{\omega}^{t},V^{t})} &\\
		{\mathbf R}^t &=&  \widehat{\bx}^{t-1} +
                                  (\lambda^t)^{-1} \boldsymbol{\Phi}^\intercal   \bg^{t}/\sqrt{n} &\\
		\widehat{x}_i^{t} &=&  g_{P_0}(R^{t}_i,\lambda^t) & \forall \ i= 1,\ldots n \\
		\mathrm{v}_i^{t} &=&  (\lambda^t)^{-1} \,\partial_R g_{P_0}(R,\lambda^t)|_{R=R^{t}_i} & \forall\ i= 1,\ldots n
\end{array}
\right.
\end{align*}
(here we denote by $\overline \bu$ the average over all the components of a vector $\bu$).
The so-called thresholding function $g_{P_0}(R,\lambda)$ is defined as
the mean of the normalized distribution
$\propto P_0(x) \exp(-\lambda(R-x)^2/2)$ and the output function $g_{P_{\rm
    out}}(Y,\omega,V)$ is similarly the mean of the normalized distribution (of $x$) $\propto P_{\rm out}(Y|\omega + \sqrt{V} x)\exp(-x^2/2)$.
%

The heuristic derivation of GAMP in statistical physics \cite{REVIEWFLOANDLENKA} suggests
via the definition of the function $g_{P_{\rm out}}$ that 
$\boldsymbol{\omega}$ and $V$ are the estimates of the means and average variance of the
components of the variable $\bz=\boldsymbol{\Phi} \bx$. This, in turn,
suggests a GAMP prediction of labels of new data
points:
\begin{equation*}
       \hat Y^{{\scriptsize\rm GAMP},t}_{\rm new} = \int 
  y\,   P_{\rm out}(y|\omega^t_{\rm new}+z\sqrt{V^t}) \,dy  {\cal D}z
\end{equation*}
where $\omega^t_{\rm new} \equiv \frac{1}{\sqrt n}\boldsymbol{\Phi}_{\rm new}
\cdot\widehat{\mathbf{x}}^{t-1}$. Comparing it with the test-set labels,
this serves to compute GAMP's generalization error.

One of the strongest assets of GAMP is that its performance can be
tracked via a closed form procedure known as state evolution (SE), again in the asymptotic limit when $n,m \to
\infty$, $m/n\to\alpha$. For proofs of SE see \cite{bayati2011dynamics,bayati2015universality} for the linear case,
and \cite{javanmard2013state} for the generalized one. In our
notations, SE tracks the correlation (or ``overlap'') between
the true weights $\bX^*$ and their estimate $\widehat{\mathbf{x}}^t$ defined as $q^t\equiv\lim\limits_{n\to\infty}\frac{1}n\mathbf{ X}^*\cdot \widehat{\mathbf{x}}^t$ via:
%
%
%
\begin{talign} \label{eq:SE_0}
&q^{t} =  2 \psi'_{P_0}( r^t) \,, \qquad r^t = 2 \alpha  \Psi'_{P_{\rm out}}(q^{t-1};\rho)  \,.
\end{talign}
The derivatives are w.r.t.\ the first argument. Similarly for the evolution of GAMP's generalization error ${\cal E}^{\rm GAMP,t}_{\rm gen}$ (see SI) we get that it is asymptotically, and with high probability, given by the r.h.s.\ of formula \eqref{Egen_generic_0} but with $q^*$ replaced by $q^t$.
%

It is a simple algebraic fact that the fixed points of the SE equations \eqref{eq:SE_0}
correspond to the critical points of the potential~\eqref{result-one_0}. The question of GAMP achieving asymptotically optimal MMSE or
generalization error therefore reduces to the study of the extrema of
the two-scalar-variables potential \eqref{result-one_0}. 
If the SE \eqref{eq:SE_0} converges to
the same couple $(q, r)$ as the extremizer $(q^*, r^*)$ of
\eqref{RSf_0} then GAMP is optimal, and if it does not then GAMP
is sub-optimal. In the next section we  illustrate this result on
several examples, delimiting regions where GAMP reaches optimality. We note that optimality of AMP-based algorithms in terms of the MMSE
on the ground-truth vector $\bX^*$ was proven for several cases where
the extremizer $q^*$ in \eqref{RSf_0} is unique, see
e.g.~\cite{donoho2013information}, or in the linear case of GLM in
\cite{BarbierMDK17}. Our results allow to complete the
characterization of regions of parameters where the algorithm reaches
optimal performance in terms of the estimation and generalization errors. 
While the asymptotic value of the Bayes-optimal generalization
error was predicted for some cases of $P_{\rm out}$ and $P_0$
\cite{opper1991generalization}, and TAP-based
algorithms were argued to reach this performance in \cite{opper1996mean,opper2001tractable}, it was not known whether this error can be achieved provably nor for what exact regions of
parameters the algorithm is sub-optimal. Our present work settles
this question thanks to the state evolution of the GAMP algorithm. 
Interestingly, heuristic arguments based on the glassy nature of the
corresponding probability measure were used to argue that direct
sampling or optimization-based approaches will not be able to match
this performance \cite{sompolinsky1990learning}. Whether this statement is correct
goes beyond the scope of the present paper. 


%
\section{Application to learning and inference}\label{estimation_0}
\begin{figure*}[ht!]
\hspace{-1cm}
\includegraphics[width=.44\linewidth]{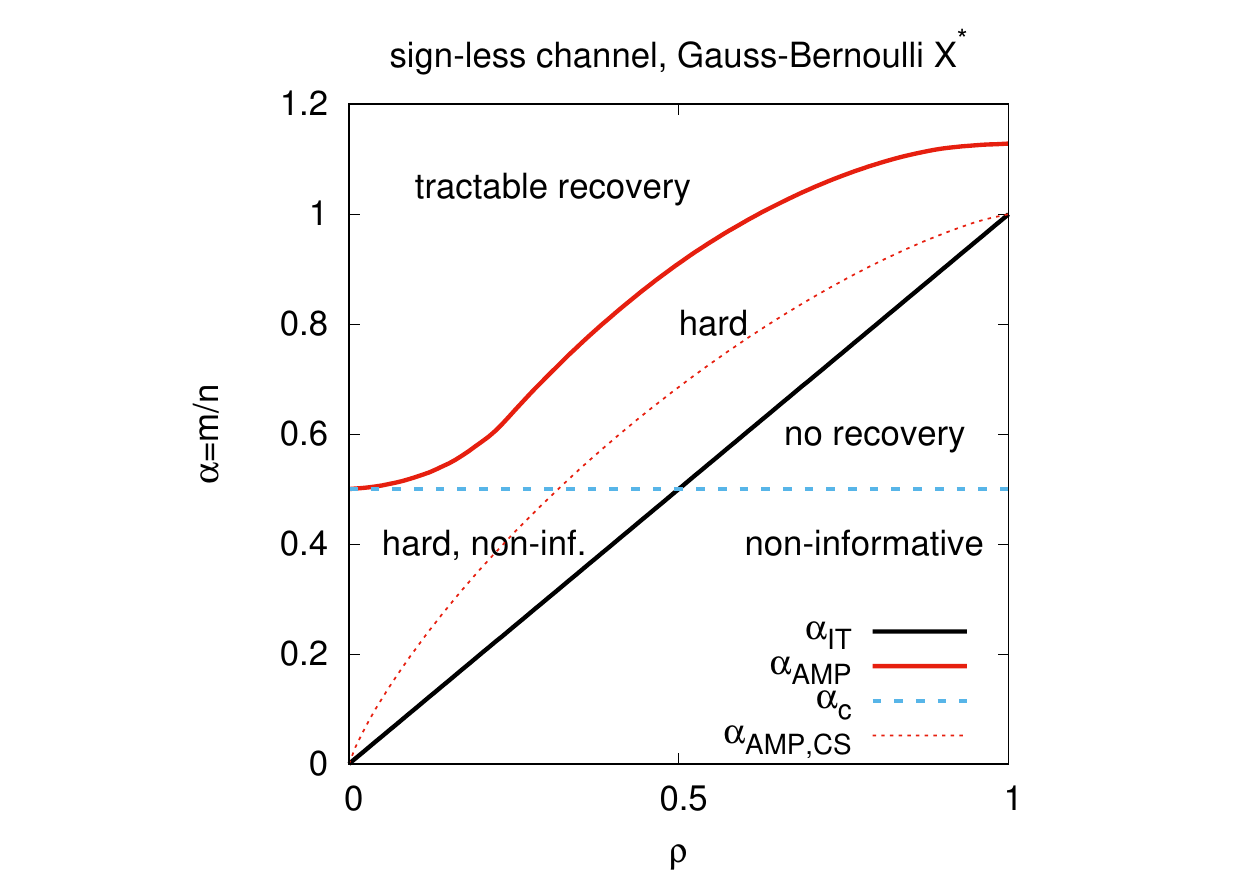}
\hspace{-2.5cm}
\includegraphics[width=.44\linewidth]{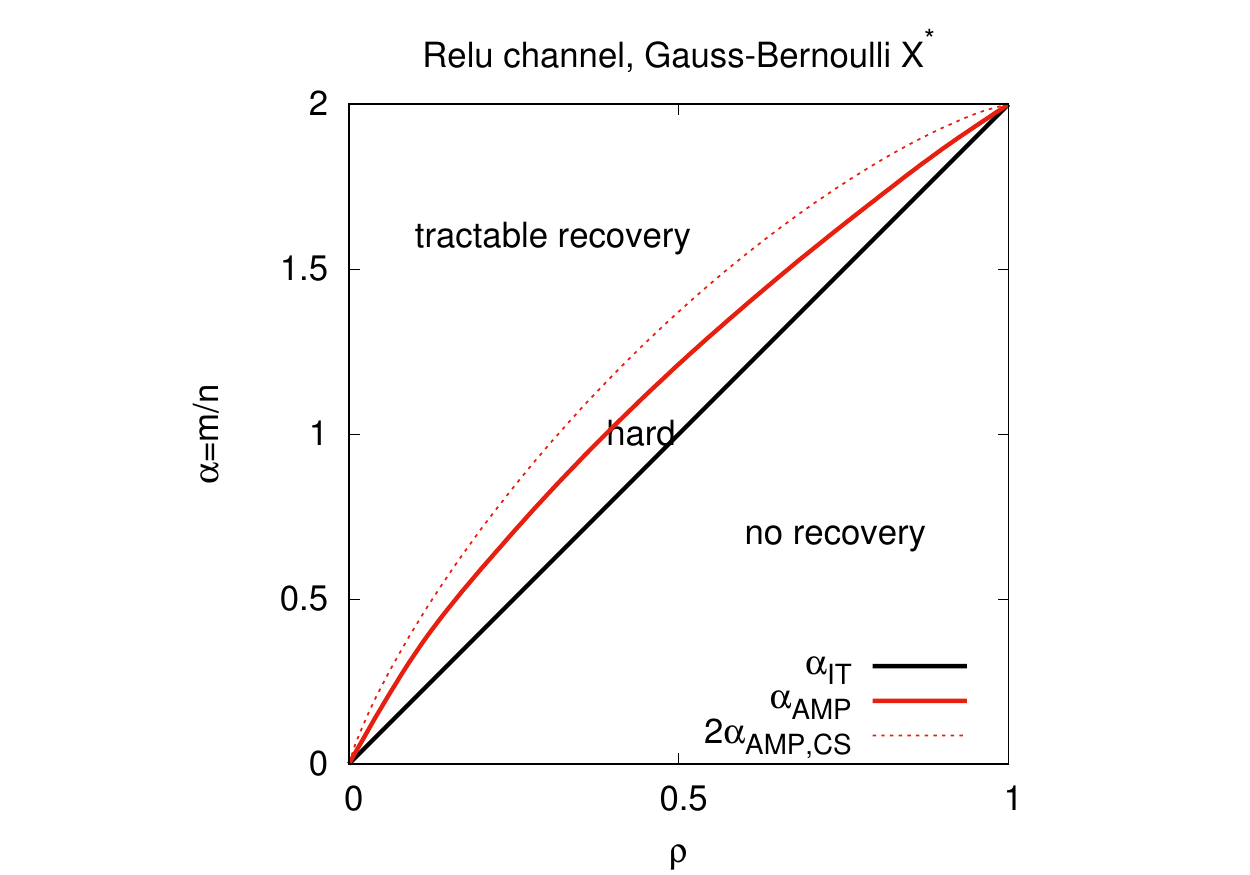}
\hspace{-2.5cm}
\includegraphics[width=.44\linewidth]{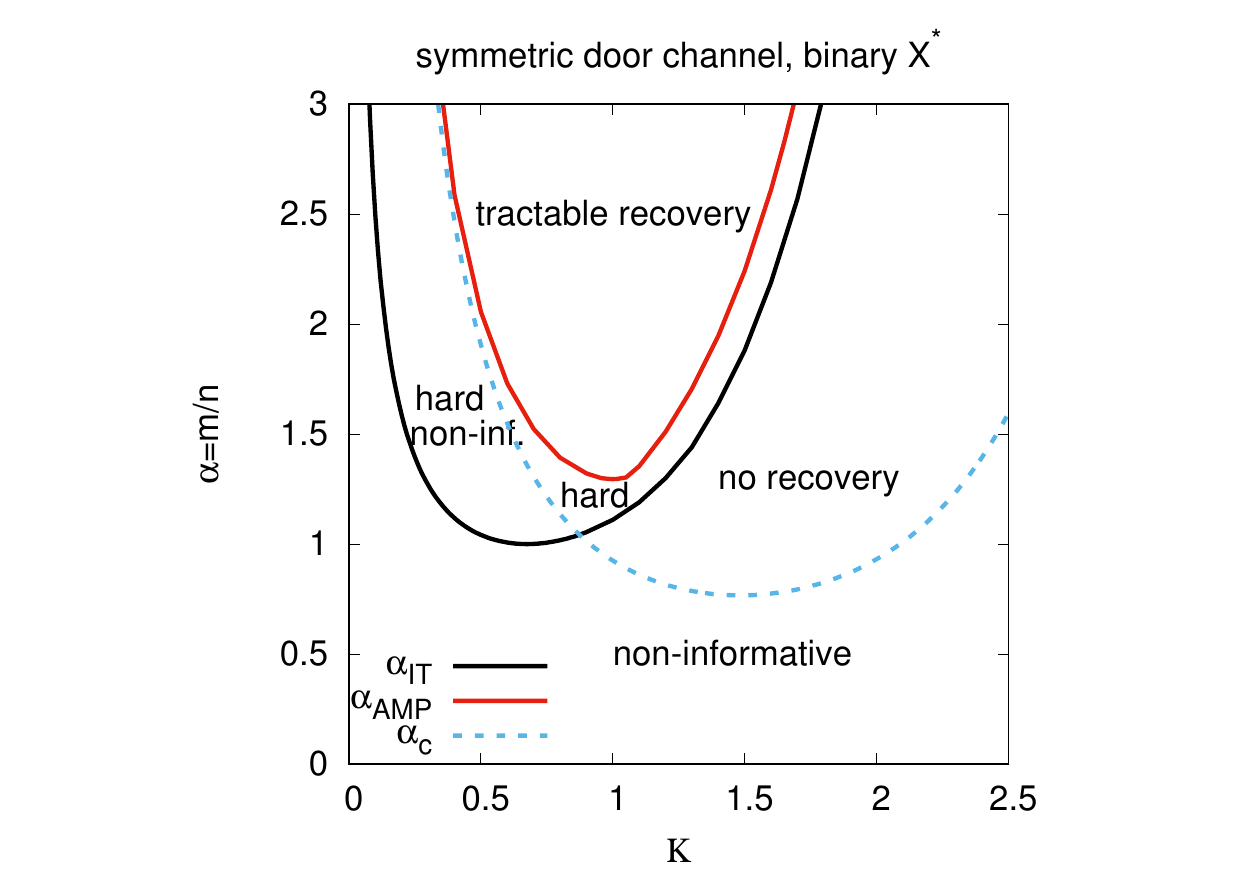}
\caption{\label{fig:phase_0} Phase diagrams showing boundaries
  of the region where almost exact recovery is possible (in absence of
  noise). \textbf{Left:} The case of sign-less sparse recovery, $\varphi (x)=|x|$ with a Gauss-Bernoulli signal, as a
  function of the ratio between number of samples/measurements and the
  dimension $\alpha=m/n$, and the fraction of non-zero components
  $\rho$. Evaluating \eqref{result-one_0} for this case, we
  find that a recovery of the signal is information-theoretically impossible for
  $\alpha<\alpha_{\rm IT}=\rho$. Recovery becomes possible starting from
  $\alpha>\rho$, just as in the canonical compressed
  sensing. Algorithmically the sign-less case is much
  harder. Evaluating \eqref{eq:stability_0} we conclude that GAMP is not able to perform better than a
  random guess as long as $\alpha<\alpha_c=1/2$, and the same is true for spectral algorithms, see \cite{mondelli2017fundamental}.
For larger values of $\alpha$, the
  inference using GAMP leads to better results than a purely random guess.
  GAMP can recover the signal and generalize perfectly only for values of $\alpha$
  larger than $\alpha_{\rm AMP}$ (full red line).
  The dotted red line shows for comparison the
  algorithmic phase transition of the canonical compressed sensing. 
  \textbf{Center:} Analogous to the left panel, for the ReLU output function,
  $\varphi (x)=\max(0,x)$. 
  Here it is always possible to perform better than random
  guessing using GAMP. The dotted red line shows the algorithmic phase transition
  when using information only about the non-zero observations. 
\textbf{Right:} Phase diagram for
  the symmetric door output function $\varphi(z)={\rm sign}(
  |z|-K)$ for a Rademacher signal, as
  a function of $\alpha$ and $K$. The stability line $\alpha_c$ is
  depicted in dashed blue, the information-theoretic phase transition
  to almost exact recovery $\alpha_{\rm IT}$ in black, and the algorithmic
  one $\alpha_{\rm AMP}$ in red.}
\end{figure*}

In this section, we report what our results imply for the
information-theoretically optimal errors, and those reached by the
GAMP algorithm for several interesting cases of output functions $\varphi$
and prior distributions $P_0$. We do not seek to be exhaustive in any way,
we simply aim to illustrate the kind of insights about the GLM that can be obtained from our results.  We focus on
determination of phase transitions in performance as we vary
parameters of the model, e.g.\ the number of samples or the sparsity of
the signal. We use careful numerical procedures to compute the
expectations required in the formula \eqref{result-one_0}, and check that
the reported results are stable towards the choice of various
precision-parameters. In this section we, however, do not seek rigor
in bounding formally the corresponding numerical errors. Many of the
codes used in this section are given online in a github repository \cite{githubrepo}.

\subsection{General observations about fixed points and terminology} \

$\bullet$ {\bf Non-informative fixed point and its stability:} It is
instrumental to analyze under what conditions $q^*=0$ is the optimizer
in \eqref{RSf_0}. Our result \eqref{eq:MMSEresult_0} about the MMSE implies that if $q^*=0$ then
the MMSE is as large as if we had no samples/measurements at our
disposition. A necessary condition for $q^*=0$ is that it is a fixed
point of the state evolution. In turn, a sufficient condition for the
state evolution~\eqref{eq:SE_0} to have such a fixed point is that $i)$
the output density $P_{\rm out}(y|z)$ is even in the
argument $z$, and $ii)$ that the prior $P_0$ has zero mean. A proof of this 
is given in the SI. In order for $q^*=0$ to be a fixed point to which the
state evolution \eqref{eq:SE_0}
converges, it needs to be stable. We detail in the SI that under properties $i)$ and $ii)$ this fixed point is
stable when 
\begin{align}\label{eq:stability_0}
           \alpha \int dy \frac{ \big(\int {\cal D}z
				   (z^2 - 1) 
		 P_{\rm out}(y|\sqrt{\rho}z) \big)^2
           }{ \int {\cal D}z  P_{\rm
               out}(y|\sqrt{\rho}z) }  < 1 \,.
\end{align}
In what follows we will denote $\alpha_c$ the
largest value of $\alpha$ for which the above condition
holds. Consequently the error reachable by the GAMP algorithm is as bad
as random guessing for both the estimation and generalization errors as
long at $\alpha<\alpha_c$. For $\alpha>\alpha_c$, starting with infinitesimal positive $q$ the state evolution will move towards larger $q$ as in \cite{fletcher2018iterative}. 
 Note that the condition
\eqref{eq:stability_0} also appears in a recent work
\cite{mondelli2017fundamental} as a barrier for performance of
spectral algorithms. 

Concerning the information-theoretically optimal error, we will call the phase where ${\rm MMSE}=\rho$, i.e.\ $q^*=0$ is the
extremizer of \eqref{result-one_0}, the {\it non-informative phase}. 
Existing literature sometimes refers to such behavior 
as {\it retarded
  learning} phase \cite{hansel1992memorization}, in the sense that in
that case a critical number of samples is required for the generalization error
to be better than random guessing.  Below we will evaluate condition \eqref{eq:stability_0} explicitly
for several examples. 

$\bullet$ {\bf Almost exact recovery fixed point:} Another fixed point of \eqref{eq:SE_0} that is worth
our particular attention is the one corresponding to almost exact recovery, meaning with average error per coordinate going to $0$ as $n\to\infty$, where
$q^*=\rho$. A~sufficient and necessary condition for this to be a fixed point is that 
$\lim_{q\to\rho}\Psi_{P_{\rm out}}'(q;\rho) = +\infty$.
This means that the integral of the Fisher information of the output channel diverges:
\begin{equation*} 
       \int dy d\omega   \frac{e^{-\frac{\omega^2}{2\rho}}}{\sqrt{2\pi \rho}}
	   \frac{ P_{\rm out}'(y|\omega)^2}{   P_{\rm out}(y|\omega)}
       = +\infty  \, ,
\end{equation*}
where $P'_{\rm out}(y|\omega)$ denotes the partial derivative w.r.t.\ $\omega$.
This typically means that the output channel should be noiseless. For
example,  for the
Gaussian channel with noise variance~$\Delta$, the above expression
equals $1/\Delta$. For the probit channel where $P_{\rm out}(y|z) =
{\rm erfc}(-yz/\sqrt{2\Delta})/2$ the above expression at small
$\Delta$ is proportional to $1/\sqrt{\Delta}$. 
  
Stability of the almost exact recovery fixed point 
depends non-trivially on both the properties of the output
channel, and of the prior. Below we give 
several examples where almost exact recovery either is or is not possible. 
In what follow we call the region of parameters for which ${\rm MMSE}=0$,
i.e.\ $q^*=\rho$ is the extremizer in \eqref{RSf_0}, the {\it almost exact recovery} phase. 


$\bullet$ {\bf Hard phase:} As can be anticipated from the statement of our main algorithmic
result, there are regions of parameters for which the error
reached by GAMP is asymptotically equal to the optimal
error, and regions where it is not. We will call {\it hard phase} the region of
parameters where ${\rm MMSE} < {\rm MSE}_{\rm AMP}$ with a strict inequality. Focusing on the
ratio $\alpha$ between the number of samples and the dimensionality, we will
denote $\alpha_{\rm IT}$ the ratio for which the hard
phase appears, and $\alpha_{\rm AMP}>\alpha_{\rm IT}$ the ratio for which
it disappears. In other words, the hard phase is an interval $(\alpha_{\rm IT}, \alpha_{\rm AMP})$, and is associated to a first order phase transition in the Bayes-optimal posterior probability distribution.

It remains a formidable open question of average computational
complexity whether in the setting of this paper (and for problems that are NP-complete in the worst case) there exists an efficient algorithm that achieves better
performance than GAMP in the hard phase. The authors are not aware of
any, and tend to conjecture that there is none.

\subsection{Sensing compressively with non-linear outputs}

\begin{figure*}[h!!]
\centering
\includegraphics[width=1\linewidth]{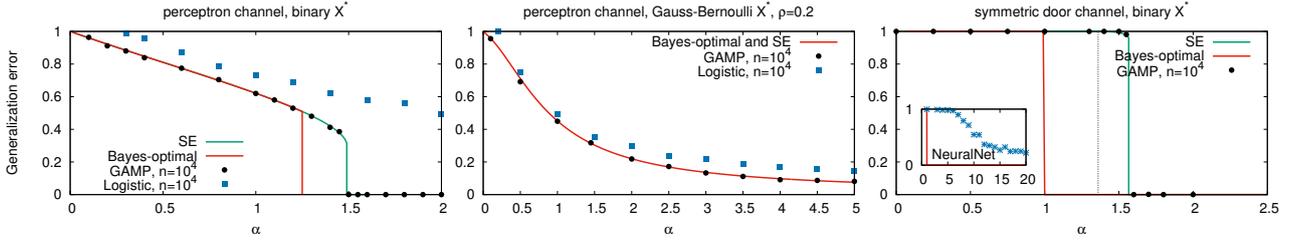}
\caption{\label{fig:classif_0}
 Optimal generalization error in three classification problems versus
 the sample complexity $\alpha$, the size of the training set being $\alpha n$. The red line is the Bayes-optimal
  generalization error \eqref{Egen_generic_0}
  while the green line shows the (asymptotic)
  performances of GAMP as predicted by
  the state evolution \eqref{eq:SE_0}. For comparison,
  we also show the results of GAMP (black dots) and, in blue, the
  performance of a standard out-of-the-box solver.
\textbf{Left:} Perceptron, with $\varphi (x)={\rm sign}(x)$ and a binary Rademacher signal. While a perfect generalization is information-theoretically possible starting from
$\alpha_{\rm IT}\approx1.249$, the state evolution predicts that GAMP will achieve
such perfect prediction only above $\alpha_{\rm AMP}\approx1.493$. The results of a logistic regression with fine-tuned regularizations with the software scikit-learn
\cite{scikit-learn} are shown for comparison. \textbf{Middle:} Perceptron with Gauss-Bernoulli distribution of the weights. No phase transition is observed in
this case, but a smooth decrease of the error with $\alpha$. The results of a
logistic regression are very close to
optimal. \textbf{Right:} The symmetric door activation rule with parameter
$K$ chosen in order to observe the same number of occurrence of the two
classes. In this case there is a sharp phase transition from as bad as
random to perfect generalization at $\alpha_{\rm IT}=1$. GAMP identifies the rule
perfectly only starting from $\alpha_{\rm AMP}\approx1.566$. The
non-informative fixed point is stable up to $\alpha_c=1.36$ (dashed line). Interestingly, this non
linear rule seems very hard to learn for standardly used solvers. Using Keras \cite{chollet2015},
a neural network with $2$ hidden layers was able to learn only approximately the
rule, only for considerably larger training set sizes and much larger number
of iterations than GAMP.}
\end{figure*}

Existing literature covers in detail the case of noiseless compressed sensing,
i.e.\ when the output function $\varphi(z)=z$. Representative
sparse prior distribution is the Gauss-Bernoulli (GB) distribution $P_0=\rho {\cal N}(0,1) + (1-\rho) \delta_{0}$, where $\rho$ is the average
fraction of non-zeros, which are in this case standard Gaussians. The phase diagram of this case is well known
see e.g. \cite{wu2010renyi,krzakala2012statistical}. In noiseless compressed
sensing with random i.i.d.\ matrices and GB prior, 
almost exact recovery of the signal is possible for $\alpha>\alpha_{\rm IT}=\rho$
and GAMP recovers the signal for
$\alpha>\alpha_{\rm AMP,CS}$ where  $\alpha_{\rm AMP,CS}$ is plotted in
Fig.~\ref{fig:phase_0} (left) with a dotted red line, thus delimiting the hard phase
of compressed sensing. We note that the Donoho-Tanner phase transition \cite{donoho2005sparse}
known as the performance limit of the LASSO $\ell_1$ regularization is slightly
higher than $\alpha_{\rm AMP,CS}$. 

$\bullet$ {\bf Sign-less output channel:} The phase diagram of noiseless
compressed sensing changes intriguingly when only the absolute value
of the output is measured, i.e.\ when $\varphi (z)=|z|$ instead of
$\varphi (z)=z$. Such an output channel is reminiscent of the widely studied
phase retrieval problem \cite{fienup1982phase} where the signal is complex valued and only
the amplitude is observed. Generalization of our results to the complex case would require extensions, as done for the algorithmic aspects in \cite{maleki2013asymptotic}. The real-valued case was
studied under the name {\it sparse recovery from quadratic
  measurements} in the literature,
e.g.~\cite{soltanolkotabi2017structured} and references therein, when
the number of non-zero variables grows slower than linearly with the
dimension~$n$. Our results give access to the
phase diagram of sparse recovery from quadratic (or equivalently
sign-less) measurements, that is presented in Fig.~\ref{fig:phase_0}
(left) for the GB prior.

We observe that the information-theoretical phase transition
$\alpha_{\rm IT}$ is the
same in the sign-less sparse recovery as in the canonical linear case,
i.e.\ almost exact recovery is possible whenever $\alpha>\rho$. However, the
algorithmic phase transition $\alpha_{\rm AMP}$ above which GAMP is able to find the
sparse signal\footnote{We note that in order to break the symmetry
  that prevents GAMP to find the
  signal in constant number of iteration steps, we mismatch infinitesimally the output function
  $\varphi$ used in the algorithm from the symmetric one used to
  generate the data. Another way to deal with this issue is related
  to a spectral initialization as discussed recently in \cite{mondelli2017fundamental}.} is strikingly larger for the sign-less case (red line in
left panel of Fig.~\ref{fig:phase_0}). We note that even for a dense
signal $\rho=1$ almost exact recovery is algorithmically possible only for
$\alpha> \alpha_{\rm AMP}(\rho=1) \approx 1.128$. For very sparse signals, small~$\rho$, the
situation is even more striking because measurement rate of at least
$\alpha>\alpha_c = 1/2$ is needed for algorithmically tractable almost exact recovery
for every $\rho$. This is in sharp contrast with the canonical
compressed sensing where $\alpha_{\rm AMP,CS} \to 0 $ as $\rho \to
0$. The nature of this algorithmic difficulty of GAMP is related to
the symmetry of the output channel thanks to which the non-informative
fixed point is stable for $\alpha<\alpha_c=1/2$. Summarizing this
result in one sentence, tractable compressive sensing is impossible
(for $\alpha<1/2$) if we have lost the signs. We remind that this result
holds in the setting of the present paper, i.e.\ in particular when the
sparsity $\rho$ is of constant order. For signals where $\rho=\smallO(1)$ the situation is
expected to be different \cite{soltanolkotabi2017structured}.

$\bullet$ {\bf ReLU output channel:} Another case of output channel that
attracted our interest is the rectified linear unit (ReLU),
$\varphi(z)=\max(0,z)$, as widely used in multi-layer feed-forward
neural networks. In the present single-layer case reconstruction with
the ReLU output is interesting mathematically. With GB
signals, roughly half of the measurements are given without noise, but
the only information we have about the other half is its sign. A
straightforward upper-bound for both information-theoretic and
tractable almost exact recovery is simply twice as many measurements than
needed in the canonical noiseless compressed sensing. It is
interesting to ask whether this bound is tight. Results in the present
paper imply that for the information-theoretic performance this bound
indeed is tight. However, the phase transition $\alpha_{\rm AMP}$ above which almost exact
recovery is possible with the GAMP algorithm is strictly lower than twice the phase
transition of compressed sensing; both are depicted in the central
panel on Fig.~\ref{fig:phase_0}. This implies that while the negative
outputs are not useful information-theoretically, they do help to
achieve better performance algorithmically.

\subsection{Perceptron and alike}\

$\bullet$ {\bf Binary and Gauss-Bernoulli perceptron:} One of the most studied problems that fits in the setting of
the present paper is the problem of perceptron
\cite{rosenblatt1957perceptron}, where $\varphi(z)={\rm sign}(z)$, that has been analyzed for random patterns $\boldsymbol{\Phi}$ in the statistical
physics literature, see
\cite{ReviewTishby,RevModPhys,engel2001statistical} for reviews. We
plot in Fig.~\ref{fig:classif_0} the optimal generalization error \eqref{Egen_generic_0} as
follows from our results for the binary
perceptron, i.e.\ weights taken from the Rademacher distribution $P_0=\frac{1}{2}\delta_{+1}+\frac{1}{2}\delta_{-1}$ (left panel), and for the GB perceptron where $P_0=\rho{\cal N}(0,1)+(1-\rho)\delta_{0}$ (central panel).  
The information-theoretically optimal value of the generalization
error that we report and prove agrees with existing predictions obtained by the
non-rigorous replica method from
\cite{gyorgyi1990first,sompolinsky1990learning,opper1991generalization}. Notably,
we see that for the GB case the optimal generalization error
decreases smoothly as $\alpha$ increases, while for the binary case
the generalization error has a first order (i.e.\ discontinuous) phase transition towards
perfect generalization at $\alpha_{\rm IT}\approx1.249$ as predicted already in
\cite{gyorgyi1990first}. Our results provide rigorous validation
for these old conjectures.

Furthermore, our results together with recent literature on GAMP provide a refreshing clarification of the algorithmic
questions. It is natural to ask for what region of
parameters the optimal generalization error can be provably achieved with efficient algorithms.
This question remained unanswered until now. Indeed,
for the spherical perceptron the optimal generalization error was computed
in \cite{opper1991generalization}, and argued empirically on small
instances to be achievable with a TAP-like algorithm
\cite{opper1996mean}.
The state evolution of
GAMP together with our formulas for the generalization
error (\eqref{Egen_generic_0} for the average optimal one and with $q^{t}$ replacing $q^*$ in this formula for GAMP) imply that the optimal generalization error {\it is} indeed
achievable asymptotically for all $\alpha$ in the GB perceptron.

For the binary perceptron the optimal generalization error was computed in
\cite{gyorgyi1990first,sompolinsky1990learning}. By comparison with
the state evolution of GAMP we obtain that it can also be asymptotically achieved 
by GAMP, but this time only outside of the hard phase $(\alpha_{\rm
  IT},\alpha_{\rm AMP})$ with $\alpha_{\rm AMP}\approx1.493$. 
The past literature was unclear on the algorithmic question, Ref.~\cite{gyorgyi1990first} identified the spinodal of the replica-symmetric solution to be at $\alpha\approx1.493$, but did not attribute it any algorithmic nor physical meaning. Ref.~\cite{sompolinsky1990learning} argues that metastable states exist at least up to $\alpha_{\rm RSB}\approx1.628$ and speculates that Gibbs sampling based algorithms will not be able to reach perfect generalization before that point \cite{ReviewTishby}. Taking our results into account, the main algorithmic question that remains open is whether efficient algorithms can reach perfect generalization for $\alpha_{\rm IT}<\alpha<\alpha_{\rm AMP}$. 

$\bullet$ {\bf Symmetric door:} Out of interest we explored an example of binary output channel for
which $P_{\rm out}(y|z)$ is even in the argument $z$, so that the
non-informative fixed point $q^*=0$ exists. Specifically we analyzed the
{\it symmetric door} channel with $\varphi(z)={\rm sign}(|z|-K)$ and
Rademacher prior $P_0$. In existing
literature such a perceptron was studied with the replica method in the context of lossy data
compression \cite{hosaka2002statistical}. In Fig.~\ref{fig:phase_0} (right panel)
we report the phase diagram in
terms of the stability line of the non-informative fixed point
$\alpha_c$ (below which GAMP is not better than random guesses), the information-theoretic phase transition towards perfect
generalization $\alpha_{\rm IT}$, and the phase transition of GAMP to perfect
generalization $\alpha_{\rm AMP}$. 

A simple counting lower bound states that for binary outputs and weights $X_i^*$ perfect generalization is not possible for
$\alpha<1$. Thus it is interesting to notice that the symmetric
door channel is able to saturate this lower-bound for $K\approx 0.6745$ for
which the probability of $y_\mu=1$ is $1/2$. This saturation was
already remarked in \cite{hosaka2002statistical}. Our results also,
however, imply that in that case the perfect generalization will not
be achievable with GAMP (and we conjecture no other efficient) algorithm unless
$\alpha>\alpha_{\rm AMP}\approx1.566$. The generalization error that GAMP
provides for this case is depicted in Fig.~\ref{fig:classif_0} (right).

\subsection{Empirical comparison with general purpose algorithms}

In this section we argue that many cases that fit into the setting of
the present paper could serve as useful benchmarks for existing
machine learning algorithms. We believe that
the situation is perhaps similar to Shannon coding theorems that have
driven algorithmic developments in error correcting codes, achieving
the Shannon bound being the primary goal in many works in communications. In machine learning, classification is
a natural task and algorithms are usually benchmarked using open
access databases. In current state-of-the-art applications of machine
learning we usually have very little insight about what is the sample
complexity, i.e.\ how many samples are truly needed so that a given
generalization error can be achieved. In our setting the situation is
different: We can present samples $(y_\mu, \boldsymbol{\Phi}_\mu)$ to generic out-of-the-box classification algorithms  and see how their
performances compare to the information-theoretic optimal
performance and to the one of the GAMP algorithm that is
fine-tuned to the problem.

In Fig.~\ref{fig:classif_0} we present examples of state-of-the-art
classification algorithms compared to our results. In the left and
center panels we compare the optimal and GAMP performances to a simple logistic regression, fine-tuned
by manually optimizing the ridge penalty (for $\ell_2$
regularization) and LASSO penalty (for a sparsity-enhancing $\ell_1$
regularization) with the software scikit-learn \cite{scikit-learn}.
We observe that while for the GB case the logistic
regression is comparable to the performance of GAMP, for binary weights perfect generalization is not achieved close to the
GAMP phase transition.


In the right panel of Fig.~\ref{fig:classif_0} we study classification
for labels generated by the symmetric door channel. A general purpose
algorithm would not know about the form of the channel. A neural
network with only two hidden units is in principle able to represent
the corresponding function (each of the hidden neurons can learn one
of the two planes that separate data in the symmetric door function). 
A more intriguing question is whether a more generic multi-layer neural
network is indeed able to learn this rule and how many samples does it
need? In the example used in Fig.~\ref{fig:classif_0}, using the software Keras
\cite{chollet2015} with a tensorflow backend, we show the performance of a network
with two hidden layers, ReLU activation and dropout (the details for
this particular run can be found on the github repository \cite{githubrepo}).
The symmetric door function thus
provides a challenging benchmark that could be used to
study how to improve performance of the general purpose multi-layer
neural network classifiers. In the SI we provide additional examples
comparing the optimal performance to general-purpose algorithms for
regression.


%
%

%
\section{Methods and proofs}
\label{sec:proofs}
In this section we give the main theorem for the free
entropy and main ideas of the proof. An essential tool is the adaptive interpolation method recently introduced in \cite{barbier_stoInt} which is a powerful evolution of the Guerra and Toninelli
interpolation method developed for spin glasses
\cite{guerra2002thermodynamic}. Reference \cite{barbier_stoInt} analyzed simpler inference problems. In particular the proof for the upper bound in \cite{barbier_stoInt} does not apply to GLMs and requires non-trivial new ingredients. One such new ingredient is to work with a potential $f_{\rm RS}(q,r;\rho)$ depending on two parameters $(q,r)$ instead of a single one as in \cite{barbier_stoInt}. This allows us to use convexity arguments that are crucial in order to finish the proof, see the last section ``Matching bounds and end of proof''. We stress that the present analysis heavily relies on properties of Bayes-optimal inference that translate into remarkable
identities between correlation functions (called Nishimori identities by physicists; see SI for their formulation) valid for {\it all} values of parameters. 
These identities are used in the derivation of \eqref{eq:der_f_t_0} and
\eqref{concentration_0} below, which are two essential steps of our
proof. The formula from Theorem~\ref{th:RS_1layer_0} relies on the
Nishimori identities and does not hold out
of the Bayes-optimal setting. 




\subsection{Main theorems}
For the proof it is necessary to work with a slightly different model
with an additive regularizing Gaussian noise with variance $\Delta\ge 0$:
\begin{align}\label{measurements_proof_0}
	Y_\mu = \varphi\Big(\frac{1}{\sqrt{n}} [\boldsymbol{\Phi} \bX^*]_{\mu}, A_\mu\Big) + \sqrt{\Delta} Z_\mu, \qquad 1\leq \mu \leq m,
\end{align}
where $(Z_\mu)\iid \mathcal{N}(0,1)$, and $(A_\mu)\iid P_A$ are r.v.\ that represent the stochastic part of $\varphi$. It is also instrumental to think
of the measurements as the outputs of a ``channel'' $Y_\mu \sim P_{\rm out}( \cdot  | \frac{1}{\sqrt{n}} [\boldsymbol{\Phi} \bX^*]_{\mu})$
%
with transition density $P_{\rm out}( y | z ) 
	= (2\pi \Delta)^{-1/2} \int dP_A(a)\exp\{-\frac{1}{2 \Delta}( y - \varphi(z, a))^2\}$ if $\Delta > 0$, or $P_{\rm out}( y | z ) 
	=\int dP_A(a) \bbf{1}( y = \varphi(z,a))$ else, where $\bbf{1}(\cdot)$ is the indicator function.
%
Our main theorem holds under the following rather general hypotheses:
%

\begin{enumerate}[label=(h\arabic*),noitemsep]
	\item \label{hyp:third_moment} The prior distribution $P_0$ admits a finite third moment and has at least two points in its support.
	\item \label{hyp:um} The sequence $(\E [ \vert\varphi(\frac{1}{\sqrt n}[\boldsymbol{\Phi}\bX^*]_1, \bA_1)\vert^{2+\gamma}])_{n \geq 1}$ is bounded for some $\gamma>0$.
	\item \label{hyp:phi_general} The r.v. $(\Phi_{\mu i})$ are independent with zero mean, unit variance and finite third moment bounded with $n$. 
	\item \label{hyp:cont_pp} 
		For almost-all values of $a$ (w.r.t.\ the distribution $P_A$), the function $x~\mapsto~\varphi(x,a)$ is continuous almost everywhere. 
	\item \label{hyp:delta_pos} ($\Delta > 0$) {\it or} ($\Delta = 0$ and $\varphi$ takes values in $\N$).
\end{enumerate}
In {\it general}, when $\varphi$ is continuous the condition $\Delta >0$ (but arbitrarily small) is necessary for the existence of a finite limit of the free entropy (for particular choices of $(\varphi, P_A)$ this might not be needed, e.g. $\varphi(z,A)=z+A$ with $A\sim {\cal N}(0,\sigma^2)$). We also assume that the kernel $P_{\rm out}$ is {\it informative}, i.e. there exists $y$ such that $P_{\rm out}(y \, | \, \cdot)$ is not equal almost everywhere to a constant. If $P_{\rm out}$ is not informative, it is not difficult to show that estimation is then impossible. 


We define the set of the critical points of $f_{\rm RS}$, \eqref{result-one_0}, also called ``state evolution fixed points'' (as it is clear from \eqref{eq:SE_0}):
	\begin{equation*} \label{fixed_points}
		\Gamma \!\equiv \!\Big\{ (q,r) \in [0,\rho] \times (\R_+ \cup \{+\infty\})  \Big|
			\begin{array}{lll}
				q &=& 2 \psi_{P_0}'(r) \\
				r &=& 2 \alpha \Psi_{P_{\rm out}}'(q;\rho)
			\end{array}
		\Big\}\,.
	\end{equation*}
 Define $f_n \equiv \E f_n(\bY,\bbf{\Phi}) = \frac{1}{n} \E \ln \cZ(\bY,\bbf{\Phi})$. Then the main theorem of this paper is stated as follows:
%
\begin{theorem}[Replica-symmetric free entropy] \label{th:RS_1layer_0}
	Suppose that \ref{hyp:third_moment}-\ref{hyp:um}-\ref{hyp:phi_general}-\ref{hyp:cont_pp}-\ref{hyp:delta_pos} hold. 
	 Then, for the GLM \eqref{measurements_proof_0},
	\begin{equation*}
		\lim_{n\to\infty } f_n
		={\adjustlimits \sup_{q \in [0,\rho]} \inf_{r \geq 0}} f_{\rm RS}(q,r) 
		= \sup_{(q,r) \in \Gamma} f_{\rm RS}(q,r) \,.
	\end{equation*}
\end{theorem}

Moreover, as one can see in the SI, the ``$\sup \inf$'' and the
supremum over $\Gamma$ above are achieved over the same couples.
Under stronger assumptions on $P_0$ and $P_{\rm out}$, one can show (see Theorem~6 in the SI) that $f_n(\bbf{Y},\bbf{\Phi})$ concentrates around its mean $f_n$ and thus obtains convergence in probability~\eqref{RSf_0}.

An immediate corollary of Theorem~\ref{th:RS_1layer_0} is the limiting
expression of the mutual information $I(\bX^*;\bY|\bbf{\Phi})\equiv \E
\ln P(\bY,\bX^*| \bbf{\Phi}) - \E\ln(P(\bY|\bbf{\Phi})
P(\bX^*))$ between the observations and the unknown vector:
\begin{corollary}[Mutual information] \label{coro_0} Under the same hypotheses as in Theorem \ref{th:RS_1layer_0}, the mut. info. for the GLM \eqref{measurements_proof_0} verifies
\begin{align*}
		\lim_{n \to \infty} {\textstyle \frac{1}{n}} I(\bX^*; \bY \, | \bbf{\Phi}) =
		{\adjustlimits \inf_{q \in [0,\rho]} \sup_{r \geq 0}} 
		\, i_{\rm RS}(q,r)= \inf_{(q,r) \in \Gamma} i_{\rm RS}(q,r) \,,\\
		i_{\rm RS}(q,r)\equiv  \alpha \Psi_{P_{\rm out}}(\rho;\rho) - \alpha\Psi_{P_{\rm out}}(q;\rho)- \psi_{P_0}(r) + rq/2 \,.
	\end{align*}
\end{corollary}

Finally, we gather our main results related to the optimal errors in a single theorem, see the SI for more details, including results on the optimality of the GAMP algorithm:
\begin{theorem}[Optimal errors] \label{th:errors_0} Assume
 the same hypotheses as in Theorem \ref{th:RS_1layer_0}. Then formula \eqref{Egen_generic_0} for the generalization error is true as $n,m\to \infty$, $m/n\to\alpha$ whenever the maximizer $q^*(\alpha)$ of \eqref{RSf_0} is unique, which is the case for almost every $\alpha$. If moreover all the moments of $P_0$ are finite, then formula \eqref{asymptOverlap_0} for the overlap as well as the following matrix-MMSE formula
\begin{align}
    \frac{1}{n^2}\, \E \Big[\big\| 
      \bX^* \bX^{* \intercal} -  \E_{P(\bx|\bY,\bbf{\Phi})}[\bx \bx^{\intercal}]
      \big\|_{\rm F}^2
    \Big] \to \rho^2-q^*(\alpha)^2  \label{mMMSE_0}
\end{align}
 are true, where $\|-\|_{\rm F}$ is the Frobenius norm.
\end{theorem}
There are cases of GLMs (e.g. the sign-less output channel $\bY=|\bbf{\Phi}\bX^*|/\sqrt{n}+\bZ$) where the sign of $\bX^*$ simply cannot be estimated (thus the absolute value in \eqref{asymptOverlap_0}). This is why our general theorem is related to an error metric \eqref{mMMSE_0} insensitive to this $\pm$ symmetry. Nevertheless formula \eqref{eq:MMSEresult_0} for the signal MSE is formally valid when there is no such sign symmetry.

\subsection{Proof by the adaptive interpolation method}
We now give the main ideas behind the proof of Theorem~\ref{th:RS_1layer_0}. We
defer to the SI the details, as well as those of Corollary~\ref{coro_0} and Theorem \ref{th:errors_0}.

A word about notation. The r.v.\ $\bY$ (and also $\boldsymbol{\Phi}$,
$\bX^*$, $\bA$, $\bZ$) are called {\it quenched} because
once the measurements are acquired they are fixed.
The expectation w.r.t.\ {\it all} quenched r.v.\ will be denoted by $\E$ {\it without} subscript. In
contrast, expectation of {\it annealed} variables w.r.t.\ a posterior
distribution at fixed quenched variables is denoted by {\it Gibbs brackets} $\langle -\rangle$.
%
%

%
\subsubsection{Two scalar inference channels}
An important role in the proof is played by two simple {\it scalar}
inference channels. The free entropy is expressed in terms of the free
entropies of these channels. This ``decoupling property'' stands at the root of the replica approach in statistical physics. 

The first scalar channel is an additive Gaussian channel. 
Suppose that
we observe $Y_0 = \sqrt{r}\, X_0 + Z_0$
%
where $X_0 \sim P_0$ and $Z_0 \sim \cN(0,1)$ are independent. Consider the inference problem consisting of retrieving $X_0$ from the observation $Y_0$.
The free entropy associated with this channel is the expectation of
the logarithm of the normalization factor of the associated posterior
$dP(x|Y_0)$, that is given by \eqref{psi0_0} (up to a constant).

The second scalar channel that appears naturally in the problem is linked to the channel $P_{\rm out}$ through the following inference model. Suppose that $V,W^* \iid \cN(0,1)$ where $V$ is {\it known} while the inference problem is to recover the unknown $W^*$ from the observation $\tilde{Y}_0 \sim P_{\rm out} ( \cdot \, | \sqrt{q}\, V + \sqrt{\rho - q} \,W^*)$ 
%
%
where $\rho >0$ and $q \in [0, \rho]$. The free entropy for this
model, again related to the average logarithm of the normalization factor of the posterior of $w$
given $\tilde Y_0$ and $V$, is exactly \eqref{PsiPout_0}.
\subsubsection{Interpolating estimation problem}
To carry on the proof, we introduce an ``interpolating estimation
problem'' that interpolates between the original problem $Y_\mu \sim
P_{\rm out}( \cdot  | \frac{1}{\sqrt{n}} [\boldsymbol{\Phi}
\bX^*]_{\mu})$ at $t=0$, $t\in[0,1]$ being the interpolation
parameter, and the two scalar problems described above at $t=1$. For
$t\in(0,1)$ the interpolating estimation problem is a mixture of the
original and the scalar problems. This interpolation scheme is inspired by the interpolation paths used by Talagrand to study the perceptron, see \cite{talagrand2010meanfield1}. Thanks to a novel ingredient specific to the adaptive interpolation method \cite{barbier_stoInt}, it allows to obtain in a unified manner a complete proof of the replica formula for the free entropy and this in the whole phase diagram. 

Let $q_\epsilon \! : \! [0,1] \to [0,\rho]$, $r_\epsilon \! : \! [0,1]\to [0, r_{\rm max}]$, $r_{\rm max} \! \equiv \!  2\alpha\Psi^\prime_{P_{\rm out}}(\rho;\rho)$, be two continuous ``interpolating functions'' parametrized by $\epsilon \! =\! (\epsilon_1, \epsilon_2) \!  \in \! \mathcal{B}_n\equiv[s_n, 2s_n]^2$, with $(s_n)_{n\geq 1} \! \in \!(0, 1/2]^{\mathbb{N}}$ a sequence that tends to zero slowly enough. Set $R_1(t, \epsilon) \equiv \epsilon_1+\int_0^t r_\epsilon(v) dv$, $R_2(t, \epsilon) \equiv \epsilon_2 + \int_0^t q_{\epsilon}(v)dv$ and define $S_{t,\mu} = S_{t,\mu}(\bX^*,W^*_\mu,V_\mu,\boldsymbol{\Phi})$ as
\begin{equation*}
	\textstyle S_{t,\mu} \! \equiv \! \sqrt{\frac{1-t}{n}}\, [\boldsymbol{\Phi} \bX^*]_\mu + \sqrt{R_2(t, \epsilon)} \,V_{\mu}+ \sqrt{\rho t \! - \!R_2(t, \epsilon) \! + \! 2 s_n} \,W_{\mu}^*
\end{equation*}
where $(V_{\mu}), (W^*_{\mu}) \iid \cN(0,1)$. Consider the following observation channels, with two types of observations obtained through 
\begin{align}
	\label{2channels_0}
	\Big\{
		\begin{array}{llll}
			Y_{t,\mu}  &\sim & P_{\rm out}(\ \cdot \ | \, S_{t,\mu})\,, & \text{for}\ 1 \leq \mu \leq m, \\
			Y'_{t,i} &=& \sqrt{R_1(t, \epsilon)}\, X^*_i + Z'_i\,, & \text{for}\ 1 \leq i \leq n,
		\end{array}		
\end{align}
where $(Z_i')\iid {\cal N}(0,1)$. We assume that $\bV=(V_{\mu})_{\mu=1}^m$ is known and that the inference problem is to recover both $\bW^* = (W_{\mu}^*)_{\mu=1}^m$ and $\bX^*=(X_i^*)_{i=1}^n$ from the ``t-dependent'' observations $\bY_{t}=(Y_{t,\mu})_{\mu=1}^m$ and $\bY'_{t}=(Y_{t,i}')_{i=1}^n$. 

We now understand that $(R_1, R_2)$ and $1-t$ appearing in the first and second set of measurements in \eqref{2channels_0} play the role of signal-to-noise ratios (snr) in the interpolating problem, with $t$ giving more and more ``power'' to the scalar inference channels when increasing. Here is the first crucial ingredient of our interpolation scheme. In classical interpolations, these snr would all take a trivial form, i.e.\ be linear in $t$, but here, the non-trivial integral dependency in $t$ of $(R_1, R_2)$ allows for much more flexibility when choosing the interpolation path. This allows us to actually choose the ``optimal interpolation path''. This will become clear below as well as the role of the ``small perturbation'' parameters $(\epsilon_1, \epsilon_2)$.

Define $u_y(x)\equiv \ln P_{\rm out}(y|x)$ and, with a slight abuse of notations, we also define the quantity $s_{t,\mu} = s_{t, \mu}(\bx,w_\mu,V_\mu,\boldsymbol{\Phi}) \equiv S_{t, \mu}(\bx,w_\mu,V_\mu,\boldsymbol{\Phi})$, the expression above with $\bX^*,W_\mu^*$ replaced by $\bx,w_\mu$.
%
%
%
%
We introduce the {\it interpolating Hamiltonian} $\cH_t=\cH_t(\bx,\bw;\bY_{t},\bY_t',\boldsymbol{\Phi},\bV)$
%
\begin{align*}
	\cH_t
	& \!\equiv\!
	- \sum_{\mu=1}^{m}
	u_{Y_{t,\mu}}( s_{t, \mu} ) 
	+ \frac{1}{2} \sum_{i=1}^{n}\big(Y_{t,i}'-\sqrt{t\, r}\,x_i\big)^2
\end{align*}
and the corresponding ($t$-dependent) {\it Gibbs bracket} $\langle - \rangle_t$ which is the expectation w.r.t.\ the joint posterior distribution of $(\bx,\bw)$ given the observations $\bY_{t},\bY_t'$ (and $\boldsymbol{\Phi},\bV$), defined as 
\begin{align*}
	\langle L(\bx,\bw) \rangle_t \equiv \frac{1}{\cZ_t(\bY_t,\bY_t',\boldsymbol{\Phi},\bV)}\int dP_0(\bx){\cal D}\bw L(\bx,\bw)e^{-\cH_t} \, ,
\end{align*}
for every continuous bounded test function $L$. 
Here $\cZ_t \equiv \int dP_0(\bx){\cal D}\bw
\exp\{-\cH_t(\bx,\bw;\bY_{t},\bY_t',\boldsymbol{\Phi},\bV)\}$ is the
appropriate normalization, ${\cal D}\bw$ is the standard Gaussian
measure. 
%
%
Finally we introduce
\begin{align*}
	f_{n,\epsilon}(t) \equiv  \frac{1}{n} \E \ln \cZ_t(\bY_t,\bY'_t,\boldsymbol{\Phi},\bV) 	
\end{align*}
%
which is the interpolating free entropy. One verifies that
\begin{equation}\label{eq:f0_f1_0}
	\begin{cases}
			f_{n,\epsilon}(0) = f_{n} - \frac12  +\mathcal{O}(s_n)\,,\\
			f_{n, \epsilon}(1) = \psi_{P_0}(\int_0^1 r_\epsilon(t)dt) - \frac{1}{2}(1 + \rho\int_0^1 r_\epsilon(t) dt)  + \frac{m}{n} \Psi_{P_{\rm out}}(\int_0^1 q_\epsilon(t) dt;\rho) + \mathcal{O}(s_n)\,,
		\end{cases}	
\end{equation}
where $\vert\mathcal{O}(s_n)\vert \leq C s_n$ for a constant $C>0$. 
Now comes another crucial property of the interpolating model: It is
such that at $t=0$ we recover the original problem and 
$f_{n,\epsilon}(0)=f_n-1/2 + \mathcal{O}(s_n)$ (the constant $1/2$ comes from the purely noisy
measurements of the second channel in \eqref{2channels_0}), while at
$t=1$ we have two scalar inference channels and thus the
associated terms $\psi_{P_0}$ and $\Psi_{P_{\rm out}}$ appear in
$f_{n, \epsilon}(1)$. These are precisely the terms appearing in the free
entropy potential \eqref{result-one_0}. 
%
\subsubsection{Entropy variation along the interpolation}
From the understanding of the previous section, it is natural to evaluate the variation of entropy along the interpolation, which allows to ``compare'' the original and purely scalar models thanks to the identity 
\begin{equation}
f_n =f_{n, \epsilon}(1)-\int_0^1 \frac{df_{n,\epsilon}(t)}{dt}dt+\frac12 +\mathcal{O}(s_n)\,.	 \label{f0_f1_int_0}
\end{equation}
%
Then by choosing the optimal interpolation path thanks to the non-trivial snr dependencies in $t$, we are able to show the equality between the replica formula and the free entropy $f_n$.

We thus compute the $t$-derivative of the free entropy (see the SI for the details of this calculation). It is given by
\begin{align}\label{eq:der_f_t_0}
		&\frac{df_{n,\epsilon}(t)}{dt} = \frac{r_\epsilon(t)}{2}( q_\epsilon(t) - \rho)- \frac{1}{2} 
		\E \Big \langle 
			 \Big(
				\frac{1}{n} \!\sum_{\mu=1}^{m}u'_{Y_{t,\mu}}(S_{t,\mu}) u'_{Y_{t,\mu}}(s_{t,\mu})
				-  r_\epsilon(t)
			 \Big) 
			\big(
				Q - q_\epsilon(t)
			\big)
		\Big\rangle_{ t} + \smallO_n(1)\,,
\end{align}
where $\smallO_n(1)$ is a quantity that goes to $0$ in the $n,m \to \infty$ limit, uniformly in $t$, $\epsilon$, and in the interpolating functions $q_\epsilon$, $r_\epsilon$. The {\it overlap} is $Q=Q_n \equiv \sum_{i=1}^n X_i^* x_i/n$.

We now state a crucial result in an informal way and refer to the SI for precise statements. 
Formally, the overlap concentrates around its mean (for all $t\in[0,1]$), a behaviour called ``replica-symmetric'' in statistical physics. In order to make this statement rigorous, one has to include the $\epsilon$-dependent small perturbation in \eqref{2channels_0} which effectively adds  ``side-information'' 
 about $\bX^*$ (e.g. think of $t=0$) without affecting the asymptotic free entropy density. This perturbation forces the overlap to concentrate. We prove that: If for each $t$ the map $R^t : (\epsilon_1, \epsilon_2)\in \mathcal{B}_n\mapsto (R_1(t, \epsilon), R_2(t, \epsilon))\in R^t(\mathcal{B}_n)$ is a $\mathcal{C}^1$ diffeomorphism whose Jacobian has determinant greater or equal to $1$, then we have for $s_n =\frac{1}{2} n^{-1/16}$ (see Proposition 4 of Sec. 4.3 in the SI for the precise statement)
\begin{align}\label{concentration_0}
\frac{1}{s_n^{2}}\int_{\mathcal{B}_n}d\epsilon\int_0^1dt\,\E\big\langle \big(Q - \E\langle Q\rangle_t \big)^2  \big\rangle_t  
= \mathcal{O}(n^{-1/8})\,.	 
\end{align}
As will be seen below it is possible to choose interpolating functions that satisfy the required condition.

\subsubsection{Canceling the remainder}
Note from \eqref{eq:f0_f1_0} and \eqref{result-one_0} that the first term appearing in \eqref{eq:der_f_t_0} is precisely the missing one to obtain the expression of the potential
on the r.h.s. of \eqref{f0_f1_int_0}. Thus we would like to ``cancel'' the Gibbs bracket in \eqref{eq:der_f_t_0}. 
This term is called {\it remainder}. In order to prove the replica formula, we have to show that this remainder vanishes:
Thanks to the freedom of choice of interpolation paths $(r_\epsilon, q_\epsilon)$, we are able to do so by ``adapting'' the interpolation. Thus we would like to choose $q_\epsilon(t) = \E \left\langle Q \right\rangle_t \approx Q$ because of \eqref{concentration_0}. However, $\E \left\langle Q \right\rangle_t$ is a function of $\int_0^t q_\epsilon(v)dv$ (and of $t,\epsilon$ and $\int_0^t r_\epsilon(v)dv$). The equation $q_\epsilon(t) = \E \left\langle Q \right\rangle_t$ is therefore a first order differential equation over $t \mapsto \int_0^t q_\epsilon(v)dv$. Assuming for the moment that this differential equation has a solution, 
the Cauchy-Schwarz inequality applied to the remainder together with \eqref{concentration_0} allows to show that the absolute value of this remainder {\it integrated} over $(\epsilon, t)\in \mathcal{B}_n\times[0,1]$ is $\mathcal{O}(s_n^2 n^{-1/16})$.
Combining this result with \eqref{eq:f0_f1_0} and \eqref{f0_f1_int_0} leads to the following {\it fundamental sum rule} (Proposition 5 of Sec. 4.3 in SI):
	\begin{align}\label{equality_f_fRS_0}
		 f_n =  \frac{1}{s_n^{2}}\int_{\mathcal{B}_n}d\epsilon{\textstyle\big\{\psi_{P_0}(\int_0^1 r_\epsilon(t) dt) + \alpha \Psi_{P_{\rm out}}(\int_0^1 q_\epsilon(t) dt;\rho)
		- \frac{1}{2} \int_0^1 r_\epsilon(t)q_\epsilon(t) dt\big\}} +\smallO_n(1)\,.	
	\end{align}
	%
%
%

\subsubsection{Matching bounds and end of proof}
We now possess all the necessary tools to finish the proof of Theorem~\ref{th:RS_1layer_0}. We first prove that $\lim_{n\to\infty} f_n=\sup_{r\ge0}\inf_{q\in[0,\rho]}f_{\rm RS}(q,r)$. Then in the SI, we show that $i)$ this is also equal to $\sup_{q\in[0,\rho]}\inf_{r\ge0}f_{\rm RS}(q,r)$ which gives the first equality of the theorem; $ii)$ that this $\sup\inf$ is attained at the supremum of the state evolution fixed points, which gives the second equality.

$\bullet$ {\bf Lower bound:} We choose the constant function $r_\epsilon(t) = r$ for $t\in [0,1]$. In the SI we show, using the Cauchy-Lipschitz theorem and the Liouville formula, that the differential equation $q_\epsilon(t) = \mathbb{E}\langle Q \rangle_t$  posseses a unique solution and that the map 
$R^t\! : \!(\epsilon_1, \epsilon_2) \! \mapsto \! (\epsilon_1 +rt, \epsilon_2 + \int_0^t q_\epsilon(v)dv)$ is a $\mathcal{C}^1$ diffeomorphism with Jacobian greater than $1$ (so \eqref{concentration_0} is valid).
Identity \eqref{equality_f_fRS_0} then implies $\liminf_{n\to\infty} f_n \geq \inf_{q \in [0,\rho]} f_{\rm RS}(q,r)$ for all $r\in [0, r_{\rm max}]$. Thus $\liminf_{n \to \infty} f_n \geq \sup_{r\in [0, r_{\rm max}]} \inf_{q \in [0,\rho]} f_{\rm RS} (q,r)$. In the SI an easy argument shows the r.h.s. is in fact equal to $\sup_{r\geq 0} \inf_{q \in [0,\rho]} f_{\rm RS} (q,r)$.
%

$\bullet$ {\bf Upper bound:} We choose the interpolating functions as solutions of the following system of 1st order differential equations:
$
r_\epsilon(t) = 2\alpha \Psi^\prime_{P_{\rm out}}(\mathbb{E}\langle Q \rangle_t)$, 
$q_\epsilon(t) = \mathbb{E}\langle Q \rangle_t$.
Again, applying the Cauchy-Lipschitz theorem and the Liouville formula
we show in the SI that this system admits a unique solution and the map
$R^t:(\epsilon_1, \epsilon_2) \mapsto (R_1(t, \epsilon), R_2(t, \epsilon))$ is a $\mathcal{C}^1$ diffeomorphism with determinant greater or equal to $1$. So with this choice of interpolating functions \eqref{concentration_0} is valid and we have \eqref{equality_f_fRS_0}. We show in the SI (Proposition 18) that $\Psi_{P_{\rm out}}(q;\rho)$ is convex in $q$ and thus $g:q\in [0, \rho]\mapsto 2\alpha\Psi_{P_{\rm out}}(q;\rho) - r_\epsilon(t) q$ is convex too. Since by the differential equations $r_\epsilon(t) = 2\alpha \Psi^\prime_{P_{\rm out}}(q_\epsilon(t))$, the function $g$ must attain its minimum at $q= q_\epsilon(t)$. By Proposition 17 in the SI $\psi_{P_0}(r)$ is convex, thus from Jensen and the last remark, the integrand $\{\cdots\}$ in \eqref{equality_f_fRS_0} is bounded as
\begin{align*}
&{\textstyle \psi_{P_0}(\int_0^1 r_\epsilon(t) dt) + \alpha \Psi_{P_{\rm out}}(\int_0^1 q_\epsilon(t) dt;\rho)- \frac{1}{2} \int_0^1 r_\epsilon(t)q_\epsilon(t) dt}\nn
\leq\, &\int_0^1 dt{\textstyle \Big\{ \psi_{P_0}(r_\epsilon(t)) + \alpha \Psi_{P_{\rm out}}(q_\epsilon(t);\rho) - \frac{1}{2} r_\epsilon(t) q_\epsilon(t) \Big\}}
\nonumber \\ 
=\, &\int_0^1 dt {\textstyle\Big\{ \psi_{P_0}(r_\epsilon(t)) + \inf_{q\in [0,\rho]}(\alpha \Psi_{P_{\rm out}}(q;\rho) - \frac{1}{2} r_\epsilon(t) q) \Big\}}
\nonumber \\ 
\leq 
\,&\sup_{r\geq 0}\inf_{q\in [0,\rho]}\Big\{ \psi_{P_0}(r) + \alpha \Psi_{P_{\rm out}}(q;\rho) - \frac{1}{2} r q \Big\}
\end{align*}
which implies  $\limsup_{n\to \infty} f_n \leq \sup_{r\geq 0} \inf_{q \in [0,\rho]} f_{\rm RS} (q,r)$.

\section*{Acknowledgments}
  J.B. acknowledges funding from the SNSF (grant 200021-156672). F.K. and L.Z. acknowledge
  funding from the ERC under the European Union's 7th Framework
  Programme Grant Agreement 307087-SPARCS and under Horizon 2020 Research and Innovation Programme Grant Agreement 714608-SMiLe. F.K. and N.M. acknowledge support from the ANR-PAIL. Part of this work was done while L.M. was visiting EPFL.
{%
	\singlespacing
	\bibliographystyle{unsrt_abbvr}
	\bibliography{./refs}
}
\newpage


\part{Supplementary informations}
\section{Setting}
\label{sec:partI}
\subsection{Generalized linear estimation: Problem statement}
We give a formal description of the observation model to which our
results apply. The generalized linear {\it model} covers both the {\it
  estimation} (or inference) problem and the supervised {\it learning}
problems (see Sec.~\ref{estimation}). 

Let $n,m \in \N^*$.
Let $P_0$ be a probability distribution over $\R$ and
let $(X^*_i)_{i=1}^n \iid P_0$ be the components of a signal vector $\bX^*$ (this is also denoted $\bX^*\iid P_0$). 
We fix a function $\varphi: \R \times \R^{k_A} \to \R$ and consider $(\bA_\mu)_{\mu =1}^m \iid P_A$, where $P_A$ is a probability distribution over $\R^{k_A}$, $k_A \in \N$. 
We acquire $m$ measurements through
\begin{align}\label{measurements}
	Y_\mu = \varphi\Big(\frac{1}{\sqrt{n}} [\boldsymbol{\Phi} \bX^*]_{\mu}, \bA_\mu\Big) + \sqrt{\Delta} Z_\mu\,, \qquad 1\leq \mu \leq m\,,
\end{align}
where $(Z_\mu)_{\mu=1}^m\iid \mathcal{N}(0,1)$ is an additive Gaussian noise, $\Delta \ge 0$, and $\boldsymbol{\Phi}$ is a $m \times n$ measurement matrix 
with independent entries that have zero mean and unit variance.
The estimation problem is to recover $\bX^*$ from the knowledge of $\bY=(Y_\mu)_{\mu=1}^m$, $\varphi$, $\boldsymbol{\Phi}$, $\Delta$, $P_0$ and $P_A$ (the realization of the random stream $\bA$ itself, if present in the model, is unknown). 
We use the notation $[\boldsymbol{\Phi} \bX^*]_{\mu}=\sum_{i=1}^n \Phi_{\mu i}X_i^*$.
When $\varphi(x, \bA) = \varphi(x)= x$ we have a random linear estimation problem, whereas if, say,
$\varphi(x) = {\rm sgn}(x)$ we have a noisy single layer perceptron. Sec.~\ref{estimation} discusses various examples related to non-linear estimation 
and supervised learning. 

Denote the prior over the signal as
$dP_0(\bx) = \prod_{i=1}^n dP_0(x_i)$, and similarly $dP_A(\ba) = \prod_{\mu=1}^m dP_A(\ba_\mu)$.
It is also fruitful to think of the measurements as the outputs of a ``channel'',
\begin{equation}\label{eq:channel}
	Y_\mu \sim P_{\rm out}\Big( \cdot \, \Big| \frac{1}{\sqrt{n}} [\boldsymbol{\Phi} \bX^*]_{\mu} \Big)\,.
\end{equation}
When $\Delta > 0$ the transition kernel $P_{\rm out}$ admits a transition density 
with respect to (w.r.t.) Lebesgue's measure, given by
\begin{align}\label{transition-kernel}
P_{\rm out}( y | x ) 
	=
	 \frac{1}{\sqrt{2\pi \Delta}} \int dP_A(\ba) e^{-\frac{1}{2 \Delta}( y - \varphi(x,\ba))^2}\,.
\end{align}
When $\Delta = 0$, we will only consider discrete channels where $\varphi$ takes values in $\N$\protect\footnotemark. 
\footnotetext{
	Notice that this allows to study any channel whose outputs belong to a countable set $S$ by applying a injection $u: S \to \N$ to the outputs.
}
In that case $P_{\rm out}$ admits a transition density with respect the counting measure on $\N$ given by (here $\bbf{1}(\cdot)$ is the indicator function)
\begin{align}\label{transition-kernel-discrete}
P_{\rm out}( y | x ) 
	=
	\int dP_A(\ba) \bbf{1}( y = \varphi(x,\ba))\,.
\end{align}
Note that for deterministic models, $\bA$ in~\eqref{measurements} is absent and thus the associated $\int dP_A(\ba)$ integral in~\eqref{transition-kernel}-\eqref{transition-kernel-discrete} simply disappears. 
In fact~\eqref{measurements} is sometimes called a ``random function representation'' of a transition kernel $P_{\rm out}$. Our analysis 
uses both representations~\eqref{measurements} and~\eqref{eq:channel}.

Throughout this paper we often adopt the 
language of statistical mechanics. In particular the random variables $\bY$ (and also $\boldsymbol{\Phi}$, $\bX^*$, $\bA$, $\bZ$) are called {\it quenched} variables
because once the measurements are acquired they have a ``fixed realization''. An expectation taken w.r.t.\ {\it all} quenched random variables appearing in an expression will simply be denoted by $\E$ {\it without} subscript. Subscripts are only used when the expectation carries over a subset of random variables appearing in an expression or when some confusion could arise. 

A fundamental role is played by the posterior distribution of (the signal) $\bx$ given the quenched
measurements $\bY$ (recall that $\bX^*$, $\bA$ and $\bZ$ are unknown). According to the Bayes formula this posterior is 
\begin{align}\label{joint_bayes}
	dP(\bx| \bY,\boldsymbol{\Phi}) &= \frac{1}{\cZ(\bY,\boldsymbol{\Phi})} dP_0(\bx)  \prod_{\mu=1}^m P_{\rm out}\Big(Y_\mu\Big|\frac{1}{\sqrt{n}} [\boldsymbol{\Phi} \bx]_{\mu}\Big)\\
	&=\frac{1}{\cZ(\bY,\boldsymbol{\Phi})} dP_0(\bx) e^{-\cH(\bx;\bY,\boldsymbol{\Phi})} \label{bayes}
\end{align}
where the {\it Hamiltonian} is defined as
\begin{align} \label{Hamil}
	\cH(\bx;\bY,\boldsymbol{\Phi}) &\defeq - \sum_{\mu = 1}^m \ln P_{\rm out}\Big( Y_{\mu} \Big| \frac{1}{\sqrt{n}}[\boldsymbol\Phi \bx]_{\mu} \Big)	
\end{align}
and the  {\it partition function} (the normalization factor) is defined as 
\begin{align}
	\cZ(\bY,\boldsymbol{\Phi}) &\defeq \int dP_0(\bx) e^{-\cH(\bx;\bY,\boldsymbol{\Phi})}\,.	\label{Z_1}
\end{align}
From the point of view of statistical mechanics~\eqref{bayes} is a Gibbs distribution and the integration over 
$dP_0(\bx)$ in the partition function is best thought as a ``sum over annealed or fluctuating degrees of freedom''. Let us introduce a standard statistical mechanics notation for the expectation w.r.t.\ the posterior~\eqref{joint_bayes}, the so called {\it Gibbs bracket} $\langle - \rangle$ defined as
\begin{align}
	\langle g(\bx)\rangle \defeq \E[g(\bX)|\bY,\boldsymbol{\Phi}]= \int dP(\bx|\bY,\boldsymbol{\Phi}) g(\bx) \label{GibbsBracket}
\end{align}
for any continuous bounded function $g$.
The main quantity of interest here is the associated averaged {\it free entropy} (or minus the averaged {\it free energy})
\begin{align}
	f_n \defeq \frac{1}{n} \E\ln \cZ(\bY,\boldsymbol{\Phi}) \,. \label{f}
\end{align}
It is perhaps useful to stress that $\mathcal{Z}(\bY,\boldsymbol{\Phi})$ is nothing else than the density of $\bY$ conditioned on $\boldsymbol{\Phi}$ so we have the explicit representation (used later on)
\begin{align}\label{fff}
	f_n &= \frac{1}{n} \E_{\boldsymbol\Phi}\int d\bY \mathcal{Z}(\bY,\boldsymbol{\Phi}) \ln \cZ(\bY,\boldsymbol{\Phi}) \nonumber\\
		&= \frac{1}{n} \E_{\boldsymbol\Phi}\int d\bY dP_0(\bX^*) e^{-\cH(\bX^*;\bY,\boldsymbol{\Phi})}
	\ln\int dP_0(\bx)\,  e^{-\cH(\bx;\bY,\boldsymbol{\Phi})} \,,
\end{align}
where $d\bY = \prod_{\mu=1}^m dY_\mu$. Thus $f_n$ is minus the conditional entropy $-H(\bY|\boldsymbol{\Phi})/n$ of the measurements. One of the main contributions of this paper is the derivation, thanks to the adaptive interpolation method, of the thermodynamic limit $\lim_{n\to\infty}f_n$ in the ``high-dimensional regime'', namely when $n,m \to \infty$ while $m / n \to \alpha > 0$ ($\alpha$ is sometimes referred to as the ``measurement rate'' or ``sampling rate'').
\subsection{The teacher-student scenario}\label{sec:teacherStudent}
We now describe an important conceptual setting, the {\it teacher-student scenario} (also called planted model), that allows to then define the optimal generalization error. We voluntarily employ terms coming from machine learning instead of the signal processing terminology used until here. 

First the teacher randomly generates a {\it classifier} $\bX^*\in \mathbb{R}^n$ (the signal in the estimation problem) with $\bX^*\iid P_0$ and an ensemble of $m$ {\it patterns} (row-vectors) $\boldsymbol\Phi_\mu\in\mathbb{R}^n$ for $\mu=1,\ldots,m$ such that $\boldsymbol\Phi_\mu\iid{\cal N}(\bbf{0},\bbf{I}_n)$. The teacher then chooses a model $(\varphi, P_A, \Delta)$ or equivalently $P_{\rm out}$, which are linked through~\eqref{transition-kernel}-\eqref{transition-kernel-discrete}. The teacher then output {\it labels} $Y_\mu\in\mathbb{R}$ through~\eqref{measurements} or~\eqref{eq:channel} for $\mu = 1,\ldots,m$. 

The student is given the distribution $P_0$, the model $(\varphi, P_A, \Delta)$ or equivalently $P_{\rm out}$ and the training data composed of the pattern-label pairs $\{(Y_\mu;\boldsymbol{\Phi}_\mu)\}_{\mu=1}^m$ generated by the teacher. His (supervised) learning task is then to predict the labels associated with new, yet unseen, patterns from all this knowledge. 


How does the teacher may evaluate the student's prediction capabilities? The teacher starts by randomly generating a new line of the matrix, or pattern,
$\boldsymbol{\Phi}_{\rm new}$. Then, still using the same $\bX^*$, he generates the associated new label $Y_{\rm new} \sim P_{\rm out}(\cdot\,|\,\boldsymbol{\Phi}_{\rm new} \cdot \bX^*/\sqrt{n} )$. He is now ready to evaluate the student generalization performance. For that purpose, an important quantity is the {\it generalization error} (or prediction error). 
If we denote $\widehat{Y}(\bbf{\Phi}_{\rm new}, \bbf{\Phi},\bY)$ the estimator used by the student (which is thus a measurable function of the observations), the generalization error is defined as
\begin{align}
	\mathcal{E}_{\rm gen}(\widehat{Y})
	\defeq \EE\big[\big(Y_{\rm new} - \widehat Y(\boldsymbol{\Phi}_{\rm new},\bbf{\Phi},\bY)\big)^2\big] \,.
\end{align}
The optimal generalization error is then defined as the minimum of $\mathcal{E}_{\rm gen}$ over all estimators $\widehat{Y}(\bbf{\Phi}_{\rm new}, \bbf{\Phi},\bY)$:
\begin{align}\label{eq:error_gen_opt}
	\mathcal{E}_{\rm gen}^{\rm opt} \defeq \min_{\widehat Y} \mathcal{E}_{\rm gen}(\widehat{Y})
	= {\rm MMSE}(Y_{\rm new}| \boldsymbol{\Phi}_{\rm new},\bbf{\Phi},\bY)=\EE\big[\big(Y_{\rm new} - \E[ Y_{\rm new} | \boldsymbol{\Phi}_{\rm new},\bbf{\Phi},\bY]\big)^2\big] \,.
\end{align}
Here, and for the rest of the paper, we define the minimum mean-square error (MMSE) function as follows: Given two random variables $\bA,\bB$, the MMSE in estimating $\bA$ given $\bB$ is defined as
\begin{equation}\label{eq:def_mmse}
	\MMSE(\bA|\bB) \defeq \E \big[\|\bA - \E[\bA|\bB]\|^2\big] \,,
\end{equation}
where $\E[\bA|\bB]$ is the expectation of $\bA$ with respect to its posterior given $\bB$.

{\it A word about notations:} Let us emphasize on the link between the different notations that we use in the present supplementary material and in the main text. E.g., the expectation w.r.t. to the posterior of $Y_{\rm new}$ appearing in \eqref{eq:error_gen_opt} can be written equivalently as:
\begin{align}
\E[ Y_{\rm new} | \boldsymbol{\Phi}_{\rm new},\bbf{\Phi},\bY]=\E_{P_A(\ba)}\E_{P( \bx| \bbf{\Phi},\bY)} \varphi\Big(\frac{\bbf{\Phi}_{\rm new}\cdot \bx}{\sqrt{n}},\ba\Big)=\E_{P_A(\ba)}\Big\langle \varphi\Big(\frac{\bbf{\Phi}_{\rm new}\cdot \bx}{\sqrt{n}},\ba\Big)\Big\rangle\,.
\end{align}
To see that, just write:
\begin{align}
\E[ Y_{\rm new} | \boldsymbol{\Phi}_{\rm new},\bbf{\Phi},\bY]&\defeq \int dY_{\rm new}\,Y_{\rm new}P(Y_{\rm new}|\boldsymbol{\Phi}_{\rm new},\bbf{\Phi},\bY)\nn
&=\int dY_{\rm new}\,Y_{\rm new}P_{\rm out}\Big(Y_{\rm new}\Big|\frac{1}{\sqrt{n}}\boldsymbol{\Phi}_{\rm new}\cdot \bx\Big)\,dP(\bx|\bY,\bbf{\Phi})\nn
&=\Big\langle \int dY_{\rm new}\,Y_{\rm new}P_{\rm out}\Big(Y_{\rm new}\Big|\frac{1}{\sqrt{n}}\boldsymbol{\Phi}_{\rm new}\cdot \bx\Big) \Big\rangle
\nn
&=\Big\langle \int dP_A(\ba)dY_{\rm new}\,Y_{\rm new}\,\frac{1}{\sqrt{2\pi \Delta}}   e^{-\frac{1}{2 \Delta}\big\{ Y_{\rm new} - \varphi\big(\frac{1}{\sqrt{n}}\boldsymbol{\Phi}_{\rm new}\cdot \bx,\ba\big)\big\}^2} \Big\rangle	\nn
&=\E_{P_A(\ba)}\Big\langle\varphi\Big(\frac{1}{\sqrt{n}}\boldsymbol{\Phi}_{\rm new}\cdot \bx,\ba\Big)\Big\rangle	\,.
\end{align}
Here we used definition \eqref{transition-kernel} for the transition kernel, but using instead \eqref{transition-kernel-discrete} would lead to the same identity.
\subsection{Two scalar inference channels} \label{sec:scalar_inf}
An important role in our proof of the asymptotic expression of the free entropy is played by simple {\it scalar} inference channels. As we will see, 
the free entropy is expressed in terms of the free entropy of these
channels. This ``decoupling property'' results from the mean-field approach in statistical physics, used through in the replica method to perform a formal calculation of the free entropy of the model \cite{mezard1990spin,mezard2009information}.
Let us now introduce these two scalar denoising models. 

The first one is an additive Gaussian channel. Let $r \geq 0$, which plays the role of a signal-to-noise ratio (snr). Suppose that $X_0 \sim P_0$ and that we observe
\begin{equation}\label{eq:additive_scalar_channel}
	Y_0 = \sqrt{r}\, X_0 + Z_0\, ,
\end{equation}
where $Z_0 \sim \cN(0,1)$ independently of $X_0$. Consider the inference problem consisting of retrieving $X_0$ from the observations $Y_0$.
The associated posterior distribution is
\begin{align}
	dP(x|Y_0) = \frac{dP_0(x) e^{\sqrt r\, Y_0 x - r x^2/2}}{\int dP_0(x)e^{\sqrt r\, Y_0 x - r x^2/2}}	\,.
\end{align}
In this expression all the $x$-independent terms have been simplified between the numerator and the normalization. The free entropy associated with this channel is just the expectation
of the logarithm of the normalization factor
\begin{align} \label{psi0}
	\psi_{P_0}(r) \defeq \E \ln \int dP_0(x)e^{\sqrt r \,Y_0 x - r x^2/2} \,.
\end{align}
The basic properties of $\psi_{P_0}$ are presented in Appendix~\ref{appendix_scalar_channel1} .

The second scalar channel that appears naturally in the problem is linked to the transition kernel $P_{\rm out}$ through the following inference model. 
Suppose that $V,W^* \iid \cN(0,1)$ where $V$ is {\it known} while the inference problem is to recover the unknown $W^*$ from the following observation
\begin{equation}\label{eq:Pout_scalar_channel}
	\widetilde{Y}_0 \sim P_{\rm out} \big( \cdot \, | \,\sqrt{q}\, V + \sqrt{\rho - q} \,W^* \big)\,,
\end{equation}
where $\rho >0$, $q \in [0, \rho]$. 
Notice that the channel~\eqref{eq:Pout_scalar_channel} is equivalent to
$\widetilde{Y}_0 = \varphi(\sqrt{q}\,V + \sqrt{\rho - q}\,W^*, \bA) + \sqrt{\Delta} Z$ with $\Delta \ge 0$ and where $(\bA, Z )\sim P_A \otimes \cN(0,1)$, independently of $V,W^*$.
The free entropy for this model, again related to the normalization of the posterior $dP(w|\widetilde Y_0,V)$, is
\begin{align}\label{PsiPout}
	\Psi_{P_{\rm out}}(q;\rho) =\Psi_{P_{\rm out}}(q) \defeq
	\E \ln \int {\cal D}w P_{\rm out}\big(\widetilde{Y}_0 | \sqrt{q}\, V + \sqrt{\rho - q}\, w\big)\, ,
\end{align}
where ${\cal D}w \defeq dw(2\pi)^{-1/2} e^{-w^2/2}$ is the standard Gaussian measure. In \eqref{PsiPout} above, $P_{\rm out}$ denotes either the transition density with respect to Lebesgue's measure (given by~\eqref{transition-kernel}) in the case $\Delta>0$, or the density with respect to the counting measure over $\N$ (given by~\eqref{transition-kernel-discrete}), in the case of a ``discrete'' channel ($\varphi$ takes values in $\N$ and $\Delta=0$). 
We prove in Appendix~\ref{appendix_scalar_channel2} that this function is convex, differentiable and non-decreasing w.r.t.\ its first argument.
\section{Main results}
\subsection{Replica-symmetric formula and mutual information}\label{RSformula-andhyp}
Let us now introduce our first main result, a single-letter {\it replica-symmetric formula} for the asymptotic free entropy 
of model~\eqref{measurements},~\eqref{eq:channel}. The result holds under the following rather general hypotheses. We will consider two cases, that is when there is some Gaussian noise ($\Delta>0$, see~\ref{hyp:delta_pos} below) and the case without Gaussian noise ($\Delta = 0$, see~\ref{hyp:delta_0} below):
\begin{enumerate}[label=(h\arabic*),noitemsep]
	\item \label{hyp:third_moment} The prior distribution $P_0$ admits a finite third moment and has at least two points in its support.
	\item \label{hyp:um} There exists $\gamma>0$ such that the sequence $(\E [ \vert\varphi(\frac{1}{\sqrt n}[\boldsymbol{\Phi}\bX^*]_1, \bA_1)\vert^{2+\gamma}])_{n \geq 1}$ is bounded.
	\item \label{hyp:phi_general} The random variables $(\Phi_{\mu i})$ are independent with zero mean, unit variance and finite third moment that is bounded with $n$. 
	\item \label{hyp:cont_pp} 
		For almost-all values of $\ba \in \R^{k_A}$ (w.r.t.\ $P_A$), the function $x \mapsto \varphi(x,\ba)$ is continuous almost everywhere. 
\end{enumerate}
We will also assume that one of the two following hypotheses hold:
\begin{enumerate}[label=(h5.\alph*),noitemsep]
	\item \label{hyp:delta_pos} $\Delta > 0$.
	\item \label{hyp:delta_0} $\Delta = 0$ and $\varphi$ takes values in $\N$.
\end{enumerate}

\begin{remark}
The above hypotheses are here stated using the ``random function'' representation of~\eqref{measurements}. In many cases, it can be useful to state them using the ``transition kernel'' representation of~\eqref{eq:channel}.
The hypotheses~\ref{hyp:um} and~\ref{hyp:cont_pp} are respectively equivalent\protect\footnotemark  \ to:
\footnotetext{
	The implications~\ref{hyp:um} $\Leftrightarrow$ \hyperref[hyp:prime_m]{(h2')} and~\ref{hyp:cont_pp} $\Rightarrow$ \hyperref[hyp:prime_cc]{(h4')} are obvious. If \hyperref[hyp:prime_cc]{(h4')} holds one can show, by inverting cumulative distribution functions, that there exists a function $\varphi: \R \times [0,1] \to \R$ such that~\eqref{measurements} holds for $A_{\mu} \iid P_A = {\rm Unif}([0,1])$ and that~\ref{hyp:cont_pp} is verified.
}
\vspace{-2mm}
\begin{enumerate}[noitemsep]
	\item[(h2')]\label{hyp:prime_m} There exists $\gamma > 0$ such that $\E [|Y_1|^{2+\gamma}]$ remains bounded with $n$.
	\item[(h4')]\label{hyp:prime_cc} $x \in \R \mapsto P_{\rm out}(\cdot | x)$ is continuous almost everywhere for the weak convergence.
\end{enumerate}
\end{remark}

Under the above hypothesis~\ref{hyp:delta_pos} (respectively~\ref{hyp:delta_0}), the transition kernel $P_{\rm out}$ admits a density with respect to Lebesgue's measure on $\R$ (resp.\ the counting measure on $\N$) that will be denoted by $P_{\rm out}( \cdot | x )$.
We will call the kernel $P_{\rm out}$ {\it informative} if there exists $y \in \R$ (resp.\ $y \in \N$) such that $P_{\rm out}(y \, | \, \cdot)$ is not equal almost everywhere to a constant. 
If $P_{\rm out}$ is not informative, it is not difficult to show that estimation is then impossible. 

Let us define the {\it replica-symmetric potential} (or just potential). Call
$\rho := \E[(X^*)^2]$ where $X^*\sim P_0$. Then the potential is
\begin{align} \label{frs}
	f_{\rm RS}(q,r;\rho) =  f_{\rm RS}(q,r) \defeq  \psi_{P_0}(r)
	+ \alpha \Psi_{P_{\rm out}}(q;\rho) - \frac{rq}{2} \,.  
\end{align} 
We define also $f_{\rm RS}(\rho, +\infty) = \lim_{r \to \infty} f_{\rm RS}(\rho,r)$.
From now on denote $\psi_{P_0}'(r)$ and $\Psi_{P_{\rm out}}'(q)=\Psi_{P_{\rm out}}'(q;\rho)$ the derivatives of $\psi_{P_0}(r)$ and $\Psi_{P_{\rm out}}(q;\rho)$ w.r.t.\ their first argument.
We need also to define the set of the critical points of $f_{\rm RS}$:
\begin{equation} \label{fixed_points}
	\Gamma \defeq \left\{ (q,r) \in [0,\rho] \times (\R_+ \cup \{+\infty\}) \, \middle|
		\begin{array}{lll}
			q &=& 2 \psi_{P_0}'(r) \\
			r &=& 2 \alpha \Psi_{P_{\rm out}}'(q;\rho)
		\end{array}
	\right\}\,,
\end{equation}
where, with a slight abuse of notation, we define $\psi'_{P_0}(+\infty) = \lim_{r \to \infty} \psi'_{P_0}(r)$ and $\Psi'_{P_{\rm out}}(\rho) = \lim_{q \to \rho} \Psi'_{P_{\rm out}}(q)$. These limits are well defined by convexity of $\psi_{P_0}$ and $\Psi_{P_{\rm out}}$. 
The elements of $\Gamma$ are called ``fixed points of the state evolution''.
Our first main result is
\begin{thm}[Replica-symmetric formula for the free entropy] \label{th:RS_1layer}
	Suppose that hypotheses~\ref{hyp:third_moment}-\ref{hyp:um}-\ref{hyp:phi_general}-\ref{hyp:cont_pp} hold. Suppose that either hypothesis~\ref{hyp:delta_pos} or~\ref{hyp:delta_0} holds.
	Then, for the generalized linear estimation model~\eqref{measurements},~\eqref{eq:channel} the thermodynamic limit of the free entropy~\eqref{f} verifies
	\begin{align}
		f_{\infty} &:= \lim_{n\to\infty }f_n 
		=  {\adjustlimits \sup_{q \in [0,\rho]} \inf_{r \geq 0}} f_{\rm RS}(q,r) 
		= \sup_{(q,r) \in \Gamma} f_{\rm RS}(q,r) \,.
\label{eq:rs_formula}
	\end{align}
	%
		Moreover, if $P_{\rm out}$ is informative,
		then the ``$\sup \inf$'' and the supremum over $\Gamma$ in~\eqref{eq:rs_formula} are achieved over the same couples $(q,r)$.
\end{thm}


An immediate corollary of Theorem~\ref{th:RS_1layer} is the limiting expression of the mutual information between the signal and the observations. To state the result, we need to introduce two mutual informations associated to the two scalar channels presented in Sec.~\ref{sec:scalar_inf}, namely
\begin{equation}\label{eq:mi_P0}
I_{P_0}(r) \defeq I(X_0; \sqrt{r}\,X_0 + Z_0) = \frac{r \rho}{2} - \psi_{P_0}(r)
\end{equation}
for the channel~\eqref{eq:additive_scalar_channel} and
\begin{equation}\label{eq:mi_Pout}
\mathcal{I}_{P_{\rm out}}(q) \defeq I(W^*; \widetilde{Y}_0|V) =\Psi_{P_{\rm out}}(\rho) - \Psi_{P_{\rm out}}(q)
\end{equation}
for the channel~\eqref{eq:Pout_scalar_channel}.

\begin{corollary}[Single-letter formula for the mutual information] \label{cor:mi}
	The thermodynamic limit of the mutual information for model~\eqref{measurements},~\eqref{eq:channel} between the observations and the hidden variables verifies
	\begin{equation}\label{eq:lim_i}
		i_{\infty}\defeq\lim_{n \to \infty} \frac{1}{n} I(\bX^*; \bY \, | \bbf{\Phi}) =
		{\adjustlimits \inf_{q \in [0,\rho]} \sup_{r \geq 0}} 
		\, i_{\rm RS}(q,r)= \inf_{(q,r) \in \Gamma} i_{\rm RS}(q,r) \,,
	\end{equation}
	where
	\begin{equation}\label{irspot}
		i_{\rm RS}(q,r) \defeq 
			I_{P_0}(r) + \alpha \mathcal{I}_{P_{\rm out}}(q) - \frac{r}{2}(\rho - q) \,.
	\end{equation}
\end{corollary}
\begin{proof}
	This follows from a simple calculation:
	\begin{align}
		\frac{1}{n} I(\bX^*;\bY|\bbf{\Phi})
		&= \frac{1}{n}H(\bY|\bbf{\Phi}) - \frac{1}{n} H(\bY|\bX^*,\bbf{\Phi})
		= -f_n
		+ \frac{1}{n} \E \ln P(\bY|\bX^*, \bbf{\Phi}) \nonumber
		\\
		&= -f_n + \frac{m}{n} \E \ln P_{\rm out}(Y_1 \, | \, \bbf{\Phi}_1 \cdot \bX^* / \sqrt{n}) \,. \label{eq:i_f0}
	\end{align}
	By the central limit theorem (that we can apply under hypotheses~\ref{hyp:third_moment}-\ref{hyp:phi_general}) we have 
	$$
	S_n := 
	\frac{1}{\sqrt{n}} \bbf{\Phi}_1 \cdot \bX^*
	=
	\frac{1}{\sqrt{n}} \sum_{i=1}^n \Phi_{1,i} X^*_i \xrightarrow[n \to \infty]{(d)}
	\cN(0,\rho) \,.
	$$
	Now, under the hypotheses~\ref{hyp:um}-\ref{hyp:cont_pp} and either~\ref{hyp:delta_pos} or~\ref{hyp:delta_0} it is not difficult to verify that
	\begin{align*}
		\E \ln P_{\rm out}(Y_1 \, | \, \bbf{\Phi}_1 \cdot \bX^* / \sqrt{n}) 
&=
\E \int dY P_{\rm out}(Y|S_n) \ln P_{\rm out}(Y|S_n)
\\
&\xrightarrow[n \to \infty]{}
\E \int dY P_{\rm out}(Y|\sqrt{\rho}\,V) \ln P_{\rm out}(Y|\sqrt{\rho}\,V)
= \Psi_{P_{\rm out}}(\rho)
	\end{align*}
	where $V \sim \cN(0,1)$. We conclude, using~\eqref{eq:i_f0}:
	\begin{equation}\label{eq:i_f}
		\frac{1}{n} I(\bX^*;\bY|\bbf{\Phi}) = - f_n + \alpha \Psi_{P_{\rm out}}(\rho) + o_n(1) 
	\end{equation}
	where $\lim_{n\to\infty}o_n(1) =0$.
\end{proof}

The next proposition, proved in Appendix~\ref{app:proof_q_star}, states that for almost every $\alpha > 0$ there is one unique optimizer $q^*$ in~\eqref{eq:rs_formula} (or equivalently in~\eqref{eq:lim_i}):
\begin{proposition}\label{prop:q_star}
	Assume that the assumptions of Theorem~\ref{th:RS_1layer} hold and that $P_{\rm out}$ is informative.
	Define
	\begin{align}
		D^* \defeq \big\{ \alpha > 0 \, \big| \,~\eqref{eq:rs_formula} \ \text{(or equivalently \eqref{eq:lim_i}) admits a unique optimizer} \ q^*(\alpha) \big\}\,.\label{Dstar}
	\end{align}
	the set $D^*$ is equal to $\R_+^*$ minus some countable set. Moreover $\alpha \mapsto q^*(\alpha)$ is continuous on $D^*$.
\end{proposition}

As an application of Theorem~\ref{th:RS_1layer} we can compute the free entropy of the ``planted perceptron'' on the hypercube and the sphere.
This perceptron model has already been studied in physics \cite{gardner1989three} and more recently in statistics, where it is known as ``one-bit compressed sensing'' \cite{boufounos20081,xu2014bayesian}.
The limit of the free entropy follows from an application of Theorem~\ref{th:RS_1layer} with $\varphi(x) = {\rm sgn}(x)$ and
$P_0 = \frac{1}{2} \delta_{-1} + \frac{1}{2} \delta_1$ (for the hypercube) or $P_0 = \cN(0,1)$ (for the sphere).
For $\mu \in \{1, \dots, m \}$ we define
\begin{align}
	S_{\mu} \defeq \Big\{ \bx \in \R^n \, \Big| \, 
		{\rm sgn}(\bx \cdot \bbf{\Phi}_\mu ) = {\rm sgn}(\bX^* \cdot \bbf{\Phi}_\mu )
	\Big\} \,.
\end{align}
We will use the notation $\cN(x) = \mathbb{P}(Z \leq x)$ for $Z \sim \cN(0,1)$.
Let $\mathbb{S}_n$ be the unit sphere in $\R^n$ and $\mu_n$ the uniform probability measure on $\mathbb{S}_n$.
\begin{corollary}[Free entropy of the planted perceptron] \label{cor:capacity} 
	Let $Z,V \iid \mathcal{N}(0,1)$. We have
	\begin{align}
		&\frac{1}{n} \E \ln \Big( \# \bigcap_{\mu = 1}^m S_{\mu} \cap \{-1,1\}^n \Big)
		\nonumber
		\\
		&
		\xrightarrow[n \to \infty]{}
		\ln(2) + {\adjustlimits \sup_{q \in [0,1)} \inf_{r \geq 0}} \Big\{ \E \ln \cosh(\sqrt{r}Z + r) + 2 \alpha 
			\E \Big[ \cN\Big(\frac{\sqrt{q}\,V}{\sqrt{1-q}}\Big) \ln \cN\Big(\frac{\sqrt{q}\,V}{\sqrt{1-q}}\Big) \Big]
		- \frac{r (q+1)}{2}  \Big\}\,,
		\\
		&\frac{1}{n} \E \ln \mu_n \Big( \bigcap_{\mu = 1}^m S_{\mu} \cap \mathbb{S}_n \Big)\nn
		&\xrightarrow[n \to \infty]{}
		\sup_{q \in [0,1)} \Big\{ \frac{1}{2} \ln(1-q) + 2 \alpha 
			\E \Big[ \cN\Big(\frac{\sqrt{q}\,V}{\sqrt{1-q}}\Big) \ln \cN\Big(\frac{\sqrt{q}\,V}{\sqrt{1-q}}\Big) \Big]
		+ \frac{q}{2}  \Big\}\,.
	\end{align}
\end{corollary}

\subsection{Optimal reconstruction (or estimation) error}


We first consider the problem of estimating $\bX^*$ given $\bY$ and $\bbf{\Phi}$. The following theorem states that the optimizer $q^*(\alpha)$ of the replica-symmetric formula~\eqref{eq:rs_formula} gives the asymptotic correlation between the planted solution $\bX^*$ and a typical sample from the posterior distribution $P(\cdot\,|\,\bY,\bbf{\Phi})$:

\begin{thm}[Limit of the overlap]\label{th:overlap}
	Assume that all the moments of $P_0$ are finite and that $P_{\rm out}$ is informative.
	Assume that~\ref{hyp:third_moment}-\ref{hyp:um}-\ref{hyp:phi_general}-\ref{hyp:cont_pp} hold and that either~\ref{hyp:delta_pos} or~\ref{hyp:delta_0} holds.
	Then for all $\alpha \in D^*$,
	\begin{align}\label{eq:lim_overlap}
		\frac{1}{n} \big|\bx\cdot \bX^*\big|=\frac{1}{n} \Big| \sum_{i=1}^n x_i X_i^* \Big| \xrightarrow[n \to \infty]{} q^*(\alpha)\,, \qquad \text{in probability},
	\end{align}
	where $\bx = (x_1,\dots,x_n)$ is sampled from the posterior distribution of the signal $P(\cdot\,|\, \bY,\bbf{\Phi})$ given by~\eqref{joint_bayes}, independently of everything else.
\end{thm}

Theorem~\ref{th:overlap} is proved in Sec.~\ref{sec:proof_overlap}. Notice that in all generality it is only possible to estimate $\bX^*$ up to its sign (think for instance to $\bY = | \bbf{\Phi}\bX^*|/\sqrt{n} + \sqrt{\Delta} \bZ$), this is why the absolute values in~\eqref{eq:lim_overlap} are needed. For this reason, the usual MSE on $\bX^*$
$$
{\rm mse}(\widehat{\bX}) \defeq \frac{1}{n} \E \Big[\big\| \bX^*  - \widehat{\bX}(\bY,\bbf{\Phi}) \|^2 \Big]
$$
is not (in all generality) an appropriate error metric. Indeed, in the case where $\bY = | \bbf{\Phi} \bX^* |/\sqrt{n} + \sqrt{\Delta} {\bZ}$, where $\bbf{\Phi},\bX^*,\bZ$ have all independent $\cN(0,1)$ entries, then $\E[\bX^*| \bY,\bbf{\Phi}] = 0$ and $\min_{\widehat{\bX}} {\rm mse}(\widehat{\bX}) = 1$. This means that the minimum mean-square error is always equal to the variance and thus, in this sense, it is never possible to estimate the signal better than trivial estimators.
For this reason, the appropriate error metric for the reconstruction problem is the MSE on $\bX^* \bX^{*\intercal}$.
From Theorem~\ref{th:overlap} one deduces the limit of the MMSE in estimating $\bX^* \bX^{* \intercal}$:

\begin{corollary}[Matrix minimum mean-square error]
	Under the same conditions as in Theorem \ref{th:overlap}, for all $\alpha \in D^*$ we have
	\begin{align}
		{\rm MMSE}_n \defeq 
		\frac{1}{n^2} \E \Big[\big\| 
			\bX^* \bX^{* \intercal} -  \E [\bX^* \bX^{*\intercal} | \bY,\bbf{\Phi} ]
			\big\|_{\rm F}^2
		\Big]
		\xrightarrow[n \to \infty]{} \rho^2 - q^*(\alpha)^2 \,,
	\end{align}
	where $\|-\|_{\rm F}$ denotes the Frobenius norm.
\end{corollary}
\subsection{Optimal generalization (or prediction) error}\label{sec:optimal_gen}
In order to express the optimal generalization error we introduce the following function (recall that $\widetilde Y_0$ is drawn from the channel~\eqref{eq:Pout_scalar_channel}):
\begin{align}
	\mathcal{E}(q) &\defeq{\rm MMSE}(\widetilde Y_0|V)=
	\EE\big[\big(\widetilde Y_0-\EE[\widetilde Y_0|V]\big)^2\big]\label{Egen_1}\\
	&=\E_{V}\int dY \,Y^2 P_{\rm out}(Y|\sqrt{\rho}\,V) 
	- \E_{V}\Big[\E_W\Big[\int dY \,Y P_{\rm out}(Y|\sqrt{q}\,V + \sqrt{\rho - q} \,W)\Big]^2\Big] 
	\label{Egen_final}	
	\\
	&=
	\EE \big[
		\varphi(\sqrt{\rho}\,V, \bA)^2
	\big]
	-
	\EE_V\big[ \EE_{W,\bA} \big[
			\varphi(\sqrt{q}\, V + \sqrt{\rho - q}\, W, \bA )
		\big]^2
	\big]+\Delta
	\label{Egen_final_Pout}	
\end{align}
where $V,W \iid \cN(0,1)$, $\bA \sim P_A$ are independent random variables, $\E_{W,\bA}$ denotes the expectation w.r.t.\ $W$ and $\bA$ only and $\EE_W[-]^2=(\EE_W[-])^2$. We recall $\rho \defeq \mathbb{E}[(X^*)^2]$ with $X^*\sim P_0$. 
With a slight abuse of notation $\int dY$ denotes in~\eqref{Egen_final} either the integration w.r.t.\ Lebesgue's measure on $\R$ in the case $\Delta > 0$ or the integration w.r.t.\ the counting measure on $\N$ (in the case $\Delta = 0$).

Recall the teacher-student setting of Sec.~\ref{sec:teacherStudent}: The generalization error is related to the estimation of a new output $Y_{\rm new} \sim P_{\rm out}(\cdot\,|\,\boldsymbol{\Phi}_{\rm new} \cdot \bX^*/\sqrt{n} )$ where $\boldsymbol{\Phi}_{\rm new}$ is a new row of the matrix, and is defined by~\eqref{eq:error_gen_opt}.
\begin{thm}[Optimal generalization error]\label{th:gen}
	Assume that $P_{\rm out}$ is informative, that~\ref{hyp:third_moment}-\ref{hyp:um}-\ref{hyp:phi_general}-\ref{hyp:cont_pp} hold and that either~\ref{hyp:delta_pos} or~\ref{hyp:delta_0} hold.
	Then for all $\alpha \in D^*$ we have
	\begin{equation}\label{eq:gen0}
		\mathcal{E}_{\rm gen}^{\rm opt}(\alpha)
		\xrightarrow[n \to \infty]{}
		\mathcal{E}(q^*(\alpha)) 
	\end{equation}
	where 
	$q^*(\alpha)$ is the optimizer of the replica-symmetric formula~\eqref{eq:rs_formula}, see Proposition~\ref{prop:q_star}.
\end{thm}

Theorem \ref{th:gen} follows from a more general result, that we state now.
Let $f : \R \to \R$ and consider the generalized optimal generalization error
\begin{equation}\label{eq:def_e_gen_f}
	\mathcal{E}_{f,n}(\alpha) \defeq {\rm MMSE}(f(Y_{\rm new})|\bbf{\Phi}_{\rm new}, \bY, \bbf{\Phi}) =\E \big[\big(f(Y_{\rm new}) - \E[f(Y_{\rm new})|\bbf{\Phi}_{\rm new}, \bY, \bbf{\Phi}]\big)^2 \big]
\end{equation}
which is the minimum mean-square error on $f(Y_{\rm new})$. In particular $\mathcal{E}_{\rm gen}^{\rm opt}(\alpha) = \mathcal{E}_{f,n}(\alpha)$ for $f:x \mapsto x$.
We define also
\begin{align}\label{eq:def_e_f}
	\mathcal{E}_f(q) &\defeq{\rm MMSE}(f(\widetilde Y_0)|V)= \EE\big[\big(f(\widetilde Y_0)-\EE[f(\widetilde Y_0)|V]\big)^2\big] \\
	&=\EE \big[
		f\big(\varphi(\sqrt{\rho}\,V, \bA) + \sqrt{\Delta} Z\big)^2 
	\big]
	-
	\EE_V\big[ \EE_{W,Z,\bA} \big[
			f\big(\varphi(\sqrt{q}\, V + \sqrt{\rho - q}\, W, \bA ) + \sqrt{\Delta} Z\big)
		\big]^2
	\big]\,,\label{eq:def_e_f_varphi}
\end{align}
where $\widetilde{Y}_0$ is the output of the second scalar channel \eqref{eq:Pout_scalar_channel}.
\begin{thm}[Generalized optimal generalization error]\label{th:gen_f}
	Let $f: \R \to \R$ be a measurable function such that $\E[|f(Y_1)|^{2+\gamma}]$ remains bounded as $n$ grows, for some $\gamma>0$.
	Assume that $P_{\rm out}$ is informative and that~\ref{hyp:third_moment}-\ref{hyp:um}-\ref{hyp:phi_general}-\ref{hyp:cont_pp} hold and that either~\ref{hyp:delta_pos} or~\ref{hyp:delta_0} holds.
	Then for all $\alpha \in D^*$ we have
	\begin{equation}\label{eq:gen0}
		\mathcal{E}_{f,n}(\alpha)
		\xrightarrow[n \to \infty]{}
		\mathcal{E}_f(q^*(\alpha)) 
	\end{equation}
	where $q^*(\alpha)$ is the optimizer of the replica-symmetric formula~\eqref{eq:rs_formula}, see Proposition~\ref{prop:q_star}.
\end{thm}
Theorem~\ref{th:gen_f} is proved in Sec.~\ref{appendix:i_mmse_gen}.

\subsection{Optimality of the generalized approximate message-passing
algorithm}
\label{GAMP}
\subsubsection{The generalized approximate message-passing
algorithm}
%
While the main results presented until now are information-theoretic, our next one
concerns the performance of a popular algorithm to solve random instances of generalized linear problems, called generalized approximate message-passing (GAMP). We shall not re-derive its properties here,
and instead refer to the original papers for details. This approach has a long history, especially in statistical physics \cite{thouless1977solution,mezard1989space,Kaba,Baldassi26062007}, error correcting codes \cite{richardson2008modern}, and graphical
models \cite{wainwright2008graphical}. For a modern derivation in the
context of linear models, see
\cite{donoho2009message,krzakala2012statistical,vila2013expectation}. The case of generalized linear models was discussed by Rangan in
\cite{GAMP}, and has been used for classification purpose in
\cite{ziniel2014binary}.

We first need to define two so-called threshold functions that are associated to the two scalar channels~\eqref{eq:additive_scalar_channel} and~\eqref{eq:Pout_scalar_channel}. 
The first one is the posterior mean of the signal in channel~\eqref{eq:additive_scalar_channel} with signal-to-noise ratio $r$:
\begin{align}
	g_{P_0}(y,r) \defeq \E[X_0|Y_0 = y] \,.	
\end{align}
The second one is the posterior mean of $W^*$ in channel~\eqref{eq:Pout_scalar_channel} with ``noise level'' $\eta = \rho - q$:
\begin{align}
	g_{P_{\rm out}}(\widetilde{y},v,\eta) \defeq \E[W^*|\widetilde{Y}_0 = \widetilde{y},\sqrt{q}\,V = v] \,. \label{gpout}
\end{align}
These functions act componentwise when applied to vectors.

Given initial estimates $(\widehat{\bx}^0,\mathbf{v}^0)$ for the means
and variances of the elements of the signal vector $\bX^*$, GAMP takes as input the observation vector $\bY$ and then iterates the following equations with initialization $g_{\mu}^{0} =0$ for all $\mu=1,\ldots,m$ (we denote by $\overline \bu$ the average over all the components of the vector $\bu$ and $\boldsymbol{\Phi}^\intercal$ is the transpose of the matrix $\boldsymbol{\Phi}$): From $t=1$ until convergence,
\begin{align}
	\left\{
		\begin{array}{lllr}
			V^{t}  &=& \overline {\mathbf{v}^{t-1}} \\
			\boldsymbol{\omega}^{t}  &=& \boldsymbol{\Phi} \widehat{\bx}^{t-1}/\sqrt n - V^{t} \mathbf{g}^{t-1} &\label{omega} \\
			g_{\mu}^{t} &=& g_{P_{\rm out}}(Y_{\mu},\omega^{t}_{\mu},V^{t}) & \forall\ \mu=1,\ldots m\\
			\lambda^t &=&  \alpha \,\overline {  g^2_{P_{\rm out}}(\bY,\boldsymbol{\omega}^{t},V^{t})} &\\
			{\mathbf R}^t &=&  \widehat{\bx}^{t-1} +
			(\lambda^t)^{-1} \boldsymbol{\Phi}^\intercal   \bg^{t}/\sqrt{n} &\\
			\widehat{x}_i^{t} &=&  g_{P_0}(R^{t}_i,\lambda^t) & \forall \ i= 1,\ldots n \\
			\mathrm{v}_i^{t} &=&  (\lambda^t)^{-1} \,\partial_R g_{P_0}(R,\lambda^t)|_{R=R^{t}_i} & \forall\ i= 1,\ldots n
		\end{array}
	\right.
\end{align}
One of the strongest asset of GAMP is that its performance can be tracked rigorously in the limit $n,m\to\infty$ while $m/n\to\alpha$ via a procedure known as state evolution (SE), see
\cite{bayati2011dynamics,bayati2015universality} for the linear case,
and \cite{GAMP,javanmard2013state} for the generalized one. In our
notations, state evolution tracks the asymptotic value of the overlap between
the true hidden value $\bX^*$ and its estimate by GAMP $\widehat{\bx}^t$ defined as
$q^t\defeq\lim_{n\to\infty}\bX^*\cdot \widehat{\bx}^t/n$ (that is related to the asymptotic mean-square error (MSE) $E^t$ between $\bX^*$ and its estimate $\widehat{\bx}^t$ by $E^t=\rho-q^t$, where recall that $\rho\defeq\EE[(X^*)^2]$ with $X^*\sim P_0$) via:
\begin{align} \label{GAMP_SE}
	\left\{
		\begin{array}{lll}
			q^{t+1} &=&  2 \psi'_{P_0}( r^t) \,, \\
			r^t &=& 2 \alpha  \Psi'_{P_{\rm out}}(q^t;\rho) \,.
		\end{array}
	\right.
\end{align}
From Theorem~\ref{th:RS_1layer} we realize that the fixed points of these equations
correspond to the critical points of the asymptotic free entropy
in~\eqref{eq:rs_formula}. In fact, in the replica heuristic, the optimizer $q^*$ of the potential is conjectured to give the optimal value of the overlap, a fact that was proven for the linear channel \cite{BarbierDMK16,BarbierMDK17,reeves2016replica}.
We will see in Sec.~\ref{estimation} that $q^t \xrightarrow[t \to \infty]{} q^*$ for a large set of parameters.
\subsubsection{Estimation and generalization error of GAMP}
Perhaps more surprisingly, one can use GAMP in the teacher-student
scenario described in Sec.~\ref{sec:teacherStudent} in order to provide an estimation of a new output $Y_{\rm new} \sim P_{\rm out}(\cdot\,|\,\boldsymbol{\Phi}_{\rm new} \cdot \bX^*/\sqrt{n} )$ where $\boldsymbol{\Phi}_{\rm new}$ is a new row of the matrix. 
As ${\bf \widehat{x}}^t$ is the GAMP estimate of the posterior expectation of $\bX^*$, with estimated variance $\mathbf{v}^t$, the natural heuristic is to consider for the posterior probability distribution of the random variable $\boldsymbol{\Phi}_{\rm new} \cdot \bX^*/\sqrt{n}$ a Gaussian with mean $\boldsymbol{\Phi}_{\rm new} \cdot {\bf \widehat{x}}^{t-1}/\sqrt{n}$ and variance $V^t=E^t=\rho-q^t$ (the fact that the variance and MSE are equal follows from the Nishimori identity of Proposition~\ref{prop:nishimori} but applied to GAMP instead of the Gibbs measure, see e.g. \cite{REVIEWFLOANDLENKA} where this is shown). This allows to estimate the posterior mean of the output, which leads to the GAMP prediction (recall the $P_{\rm out}$ definition~\eqref{transition-kernel}-\eqref{transition-kernel-discrete}):
\begin{align}
	\widehat Y^{\text{\scriptsize\rm GAMP},t}	\defeq \int  y\, P_{\rm out}\Big(y \,\Big| \, \frac{1}{\sqrt{n}}\boldsymbol{\Phi}_{\rm new} \cdot {\bf \widehat{x}}^{t-1} + \sqrt{\rho-q^t}\, w\Big) {\cal D}wdy\,, \label{Cgamp}
\end{align}
%
%
%
%
%
where ${\cal D}w$ denotes the standard Gaussian measure. The following claim, from \cite{GAMP}, gives the precise estimation error of GAMP.
It is stated there as a claim because some steps of the proof are missing. The paper \cite{javanmard2013state} affirms in its abstract to prove the claim of \cite{GAMP}, but without further details. For these reasons, we believe that the claim holds, however we prefer to state it here as a claim (instead of a theorem).

\begin{claim}[GAMP estimation error,~\cite{GAMP}]\label{claim:gamp}
	We have almost surely for all $t \in \N$, 
	\begin{align}\label{eq:lim_gamp}
		\lim_{n \to \infty} \frac{1}{n} {\widehat{\bx}}^t \cdot \bX^* =
		\lim_{n \to \infty} \frac{1}{n} \|{\widehat{\bx}}^t\|^2 = q^t
		\,,
	\end{align}
as well as
\begin{equation}\label{eq:mse_gamp}
	\lim_{n \to \infty} \frac{1}{n^2} \E \Big[\big\| \bX^* \bX^{*\intercal} - {\widehat{\bx}}^t ({\widehat{\bx}}^t)^{\intercal} \big\|^2 \Big]
	= \rho^2 - (q^t)^2 \,.
\end{equation}
\end{claim}
Compairing \eqref{eq:mse_gamp} with the MMSE given by Corollary \ref{Cor:mMMSE}, we see that if $\lim_{t\to\infty}q^t= q^*(\alpha)$, then GAMP achieves the MMSE. Provided that Claim~\ref{claim:gamp} holds 
we can deduce the generalization error of GAMP:

\begin{proposition}[GAMP generalization error]\label{claim:gamp2} 
	 Suppose that hypotheses~\ref{hyp:third_moment}-\ref{hyp:um}-\ref{hyp:cont_pp} hold. Moreover suppose that either~\ref{hyp:delta_pos} or~\ref{hyp:delta_0} holds. Assume that $(\Phi_{\mu i}) \iid \cN(0,1)$, and that $x \mapsto P_{\rm out}(\cdot | x)$ is continuous almost everywhere for the Wasserstein distance of order $2$.
	Let $t \in \N$.
	Assume that the limit~\eqref{eq:lim_gamp} holds in probability and that there exists $\eta>0$ such that $\E [|\widehat Y^{\text{\scriptsize\rm GAMP},t} |^{2+\eta}]$ remains bounded (as $n$ grows).
	Then we have
\begin{align}\label{eq:lim_err_gamp}
	\lim_{n\to\infty} {\cal E}^{\text{\scriptsize\rm GAMP},t}_{\rm gen} &\defeq \lim_{n\to\infty} \EE\big[\big(Y_{\rm new} - \widehat Y^{\text{\scriptsize\rm GAMP},t}\big)^2\big] =  {\cal E}(q^t)\,.
\end{align}
\end{proposition}

\begin{remark}
If we modify slightly the GAMP estimator of \eqref{Cgamp} by changing the first $y$ into $f(y)$, it is not difficult to show (following the steps of Proposition \ref{claim:gamp2}) that this new estimator achieves an asymptotic error of $\mathcal{E}_{f}(q^t)$, given by \eqref{eq:def_e_f}, for estimating $f(Y_{\rm new})$.
\end{remark}
Proposition~\ref{claim:gamp2} is proved in Sec.~\ref{appendix:proof_gen2}.
We see that this formula matches the one for the Bayes-optimal generalization error, see Theorem~\ref{th:gen}, up to the fact that instead of $q^*(\alpha)$ (the optimizer of the replica formula~\eqref{eq:rs_formula}) appearing in the optimal error formula, here it is $q^t$ which appears.
Thus clearly, when $q^t$ converges to $q^*(\alpha)$ (we shall see that this is the case in many situations in the examples of Sec.~\ref{estimation}) this yields a very interesting and non trivial
result: {\it GAMP achieves the Bayes-optimal generalization error} in a plethora of models (a task again often believed to be intractable) and this for large sets of parameters.
\subsection{Optimal denoising error}

Another interesting error measure to study is the following ``denoising error''. Assume that the observations are noisy, i.e.\ $\Delta > 0$ in~\eqref{measurements}.
The goal here is to denoise the observations $Y_{\mu}$ and estimate the signal which in this case is $\varphi\big(\frac{1}{\sqrt{n}} [\boldsymbol{\Phi} \bX^*]_{\mu}, \bA_\mu\big)$.

The minimum denoising error (in $L^2$ sense) is actually a simple corollary from the replica-symmetric formula of Theorem~\ref{th:RS_1layer} and follows from a so-called ``I-MMSE relation'', see Proposition \ref{prop:immse}. We will need the joint posterior distribution of $(W^*,\bA)$ given $(V,\widetilde{Y}_0)$ for the scalar channel~\eqref{eq:Pout_scalar_channel}. So we define the Gibbs bracket for the scalar channel by (here $\ba\in \mathbb{R}^{k_A}$):
\begin{align}
	\langle g(w,\ba)\rangle_{\rm sc} &\defeq \E[g(W^*,\bA)|\widetilde Y_0,V] =\frac{ \int {\cal D}w dP_A(\ba)g(w,\ba)e^{-\frac{1}{2\Delta}\big\{\widetilde Y_0 -\varphi(\sqrt{q}\,V + \sqrt{\rho - q}\, w,\ba)\big\}^2}}{\int {\cal D}w dP_A(\ba) e^{-\frac{1}{2\Delta}\big\{\widetilde Y_0 -\varphi(\sqrt{q}\,V + \sqrt{\rho - q}\, w,\ba)\big\}^2}}	\label{GibbsBracket_sc}\,,
\end{align}
%
for any continuous bounded function $g$. When the function depends only on $w$ it may be re-written as
\begin{align}
	\langle g(w)\rangle_{\rm sc} =\frac{ \int {\cal D}w g(w)P_{\rm out}\big(\widetilde Y_0\big|\sqrt{q}\,V + \sqrt{\rho - q}\, w\big)}{\int {\cal D}w P_{\rm out}\big(\widetilde Y_0\big|\sqrt{q}\,V + \sqrt{\rho - q}\, w\big)}\,.	\label{44_}
\end{align}

\begin{corollary}[Optimal denoising error] \label{Cor:mMMSE}
	Suppose that hypotheses~\ref{hyp:third_moment}-\ref{hyp:um}-\ref{hyp:phi_general}-\ref{hyp:cont_pp} hold. Suppose that either hypothesis~\ref{hyp:delta_pos} or~\ref{hyp:delta_0} holds.
	Then for almost every $\Delta > 0$, for any optimal couple $(q^*,r^*)$ of~\eqref{eq:rs_formula},
	\begin{align}\label{den_err}
		\lim_{n\to\infty}\frac{1}{m} {\rm MMSE}\Big(
			\varphi\Big(\frac{1}{\sqrt{n}}\bbf{\Phi}\bX^*,\bA\Big)\Big|\bbf{\Phi},\bY	
		\Big) 
		&= {\rm MMSE}\big(\varphi(\sqrt{q^*}\,V + \sqrt{\rho - q^*}\,W^*,\bA)\big|\widetilde Y_0,V\big) \nn
		&=\E \big[\varphi(\sqrt{\rho}\,V,\bA)^2\big] - \E \big[\big\langle \varphi(\sqrt{q^*}\,V + \sqrt{\rho - q^*}\,w,\ba)\rangle_{\rm sc}^2\big] 
			\,, 
		\end{align}
		where $\langle - \rangle_{{\rm sc}}$ acts jointly on $(w,\ba)$ and is defined by~\eqref{GibbsBracket_sc}, and $V,W^* \iid {\cal N}(0,1)$.
	\end{corollary}

	Note that the joint posterior over both the signal $\bX^*$ and the random stream $\bA$ is simply expressed as
\begin{align}\label{jointPost}
dP(\bx,\ba| \bY,\boldsymbol{\Phi})\propto dP_0(\bx)dP_A(\ba)\prod_{\mu=1}^m e^{-\frac{1}{2 \Delta}\big\{ Y_\mu - \varphi\big(\frac{1}{\sqrt{n}} [\boldsymbol{\Phi} \bx]_{\mu},\ba_\mu\big)\big\}^2}\,.
\end{align}
	The proof of Corollary~\ref{Cor:mMMSE} is presented in Sec.~\ref{appendix:i_mmse_den}. 

\section{Application to concrete situations}\label{estimation}

In this section, we show how our main results can be applied to several models
of interest in fields ranging from machine learning to signal
processing, and unveil several interesting new phenomena in learning
of generalized linear models. For various specific cases of prior $P_0$
and output $P_{\rm out}$, we evaluate numerically the free
entropy potential (\ref{frs}), its stationary points $\Gamma$ and identify which of them
gives the information-theoretic results, i.e.\ is the optimizer in
\eqref{eq:rs_formula}. We also identify which of the stationary points corresponds
to the result obtained asymptotically by the GAMP algorithm, i.e.\ the
fixed point of the state evolution \eqref{GAMP_SE}. Finally we compute the
corresponding generalization error \eqref{eq:lim_err_gamp}.  We stress that in
this section the results are based on numerical investigation of the
resulting formulas: We do not aim at rigor that would involve precise
bounds and more detailed analytical control for the corresponding integrals. 
\subsection{Generic observations}
%
	Using the functions $g_{P_{\rm out}}$ and $g_{P_0}$
introduced in Sec.~\ref{GAMP} we can rewrite the fixed point
equations \eqref{fixed_points} as
	\begin{align}
		q &= 2 \psi_{P_0}'(r) = \E[g_{P_0}(Y_0,r)^2]\,,
		\label{fixed_points_explicit1} \\
		r &= 2 \alpha \Psi_{P_{\rm out}}'(q) = \frac{\alpha}{\rho - q} \E[g_{P_{\rm out}}(\widetilde{Y}_0,\sqrt{q}\,V,\rho - q)^2]\,,
		\label{fixed_points_explicit2}
	\end{align}  
	where the expectation in \eqref{fixed_points_explicit1} corresponds to the scalar channel \eqref{eq:additive_scalar_channel} and the expectation in \eqref{fixed_points_explicit2} corresponds to the second scalar channel \eqref{eq:Pout_scalar_channel}.

\paragraph{Non-informative fixed point and its stability:} It is interesting to analyze under what conditions $q^*=0$ is the
optimizer of \eqref{eq:rs_formula}. Notice that $q^*=0$
corresponds to the error on the recovery of the signal as large as
it would be if we had no observations at our disposition. 
Theorem~\ref{th:RS_1layer} gives that any optimal couple $(q^*,r^*)$ of \eqref{eq:rs_formula} should be a fixed point of the state evolution equations \eqref{fixed_points_explicit1}--\eqref{fixed_points_explicit2}.
A sufficient condition for $(q,r) = (0,0)$ to be a fixed point of
\eqref{fixed_points_explicit1}--\eqref{fixed_points_explicit2} is that:
\begin{enumerate}[label=(\alph*)]
\item The transition density $P_{\rm out}(y|z)$ is even in the argument
$z$. 
\item The prior $P_0$ has zero mean. 
\end{enumerate}
In order to see this, notice that if $P_{\rm out}(y|z)$ is even in $z$ then from the
definition \eqref{gpout} of the function $g_{P_{\rm out}}$ we have $g_{P_{\rm
out}}(y,0,\rho)=0$ and consequently from
\eqref{fixed_points_explicit2} we have $\Psi'_{P_{\rm out}}(0)=0$.
From the second point, notice that we have $\psi'_{P_0}(0) = \frac{1}{2}\E_{P_0}[X_0]^2 = 0$.

We assume now that the transition density $P_{\rm out}(y|z)$ is even in the argument $z$ and that the prior $P_0$ has zero mean. 
In order for $q=0$ to be the global maximizer $q^*$ of (\ref{eq:rs_formula}) or
to be a relevant fixed point of the state evolution (\ref{GAMP_SE}) (relevant in the sense that GAMP might indeed converge to it in a practical setting) we need
$q=0$ to be a stable fixed point of the above equations
\eqref{fixed_points_explicit1}--\eqref{fixed_points_explicit2}. 
We therefore need to expand \eqref{fixed_points_explicit1}--\eqref{fixed_points_explicit2} around
$q=0$, and doing so, we obtain that $q=0$ is stable if
	\begin{equation}\label{eq:stability}
	2\alpha \Psi_{P_{\rm out}}''(0) \times 2 \psi_{P_0}''(0) = \alpha
	\E\big[
		\big(
			\langle w^2 \rangle_{\rm sc}	
			- \langle w \rangle_{\rm sc}^2
			-1
			\big)	^2
	\big]
	< 1\,,
	\end{equation}
	where the expectation corresponds to the scalar channel \eqref{eq:Pout_scalar_channel} with $q=0$ and the Gibbs bracket $\langle - \rangle_{\rm sc}$ is given by \eqref{GibbsBracket_sc}.
		The expectation quantifies how the observation of
                $\widetilde{Y}_0$ in the scalar channel
                \eqref{eq:Pout_scalar_channel} modifies the variance
                of $W^*$ (which is $1$ without any observation). 
Rewriting this condition more explicitly into a form that is
convenient for numerical evaluation we get (recalling \eqref{44_}, $q=0$ and condition (a))
\begin{equation}\label{eq:stability2}
           \alpha \int dy \frac{ \big(\int {\cal D}z
		  (z^2 -1 )P_{\rm out}(y|\sqrt{\rho}z) \big)^2
           }{ \int {\cal D}z  P_{\rm
               out}(y|\sqrt{\rho}z) }  < 1 \, ,
\end{equation}
where recall that ${\cal D}z$ is a standard Gaussian measure.
We conjecture that the condition \eqref{eq:stability2} delimits precisely the region where polynomial-time algorithms do not perform better than ``random guessing'' (see the discussion below, where we will make this stability condition explicit for several examples of symmetric output channels). Note that the condition \eqref{eq:stability2} also appears in a recent work \cite{mondelli2017fundamental} as a barrier for performance of spectral algorithms.

\paragraph{Exact recovery fixed point:} Another particular fixed point of \eqref{fixed_points_explicit1}--\eqref{fixed_points_explicit2} that
we observe is the one corresponding to exact recovery $q^*=\rho$. A
sufficient and necessary condition for this to be a fixed point is that 
$\lim_{q\to\rho}\Psi_{P_{\rm out}}'(q) = +\infty$.
Heuristically, this means that the integral of the Fisher information of the output channel should diverge:
\begin{equation} 
       \int dy d\omega   \frac{e^{-\frac{\omega^2}{2\rho}}}{\sqrt{2\pi \rho}}
	   \frac{ P_{\rm out}'(y|\omega)^2}{   P_{\rm out}(y|\omega)}
       = +\infty  \, ,
\end{equation}
where $P'_{\rm out}(y|\omega)$ denotes the partial derivative w.r.t.\ $\omega$.
This typically means that the channel should be noiseless. For the
Gaussian channel with noise variance $\Delta$, the above expression
equals $1/\Delta$. For the probit channel where $P_{\rm out}(y|z) =
{\rm erfc}(-yz/\sqrt{2\Delta})/2$ the above expression at small
$\Delta$ is proportional to $1/\sqrt{\Delta}$. 

Stability of the exact recovery fixed point was also
investigated, but we did not obtain any unified expression. The
stability depends non-trivially on both the properties of the output
channel, but also on the properties of the prior. Below we give 
several examples where exact recovery either is or is not possible,
or where there is a phase transition between the two regimes. 


%
\subsection{Phase diagram of perfect learning}

In this section we consider {\it deterministic (noiseless) output
channels} and ask: How many measurements are needed in order to
perfectly recover the signal? 

Our crucial point is to compare with the well explored phase diagram of Bayesian
(noiseless) compressed sensing in the case of the linear channel
\cite{krzakala2012statistical,krzakala2012probabilistic}. As the
number of samples (measurements) varies we encounter five different regimes of parameters:
\begin{itemize}
	\item The {\it tractable recovery} phase: This is the region in the parameter space where GAMP achieves perfect reconstruction.
	\item The {\it non-informative} phase: Region where perfect
          reconstruction is information-theoretically impossible and
          moreover even the Bayes-optimal estimator is as bad as a
          random guess based on the prior information and on the
          knowledge of the output function.
	\item The {\it no recovery} phase: Region where perfect
          reconstruction is information-theoretically impossible, but
          an estimator positively correlated with the ground truth exists.
	\item The {\it hard} phase: Region where the perfect reconstruction
          is information-theoretically possible, but where GAMP is
          unable to achieve it. At the same time, in this region GAMP leads to a
          better generalization error than the one corresponding to the
          non-informative fixed point. It remains a challenging
          open question whether polynomial-time algorithms can
          achieve perfect reconstruction in this regime.
	\item The {\it hard non-informative} phase: This phase
          corresponds to the region where perfect reconstruction is
          information-theoretically possible but where GAMP only
          achieves an error as bad as randomly guessing, given by the
          trivial fixed point. In this phase as well, the existence of
          polynomial-time exact recovery algorithms is an open
          question. This phase does not exist for the linear
          channel. 
\end{itemize}

Some of the codes used in this section can be consulted online on the
github repository \cite{githubrepo}.

\subsubsection{The linear channel}
The case of exact recovery of a sparse signal after it passed trough a
noiseless linear channel, i.e.\ $\varphi(x)=x$, is studied in the literature in great details, especially in the context
of compressed sensing \cite{candes2006near,donoho2009message}. For a
signal with a fraction $\rho$ of non-zero entries it is found that as
soon as $\alpha>\rho$, perfect reconstruction is theoretically
possible, although it may remain computationally difficult. The whole
field of compressed sensing builds on the realization that, using
the $\ell_1$ norm minimization technique, one can efficiently recover
the signal for larger $\alpha$ values than the so-called
Donoho-Tanner transition \cite{donoho2009observed,donoho2009message}.

In the context of the present paper, when the empirical distribution of the signal
is known, one can fairly easily beat the $\ell_1$
transition and reconstruct the signal up to lower values of $\alpha$ using the Bayesian GAMP algorithm
\cite{donoho2009message,GAMP,krzakala2012statistical,krzakala2012probabilistic}. In this case, three different
phases are present \cite{krzakala2012statistical,krzakala2012probabilistic}: $i)$ For $\alpha<\rho$, perfect reconstruction is impossible; $ii)$
for $\rho<\alpha<\alpha_s$ reconstruction is possible, but not with
any known polynomial-complexity algorithm; $iii)$ for $\alpha>\alpha_s$, the
so-called {\it spinodal transition} computed with state evolution, GAMP
provides a polynomial-complexity algorithm able to reach perfect
reconstruction. The line $\alpha_s(\rho)$ depends on the
distribution of the signal. For a Gauss-Bernoulli signal with a fraction $\rho$ of
non-zero (Gaussian) values we compare the GAMP performance to the
optimal one in Fig.~\ref{fig:phase} (left and right). This is the same figure as in the main text. We copy it here so that the SI is self-contained. 
%
\begin{figure*}[t]
\hspace{-1cm}
\includegraphics[width=.44\linewidth]{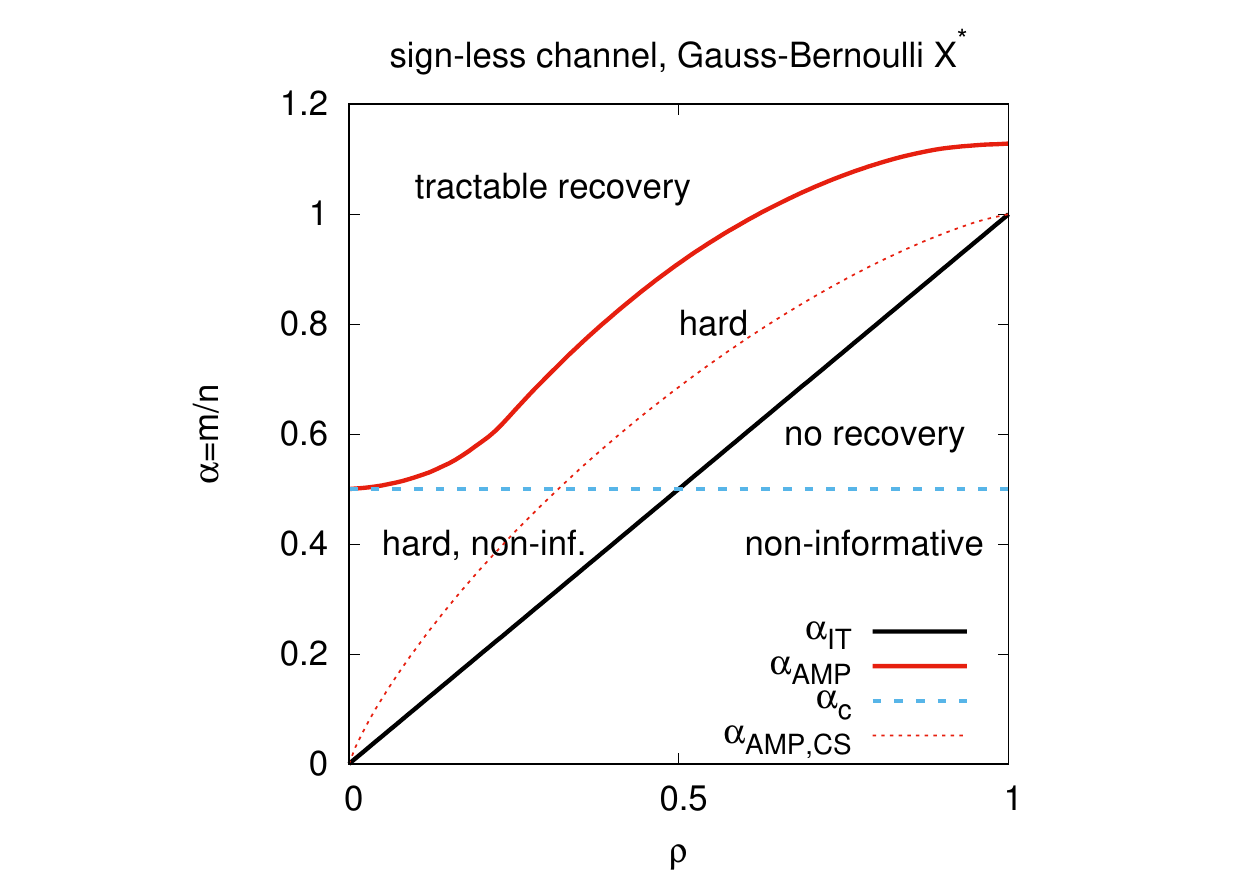}
\hspace{-2.5cm}
\includegraphics[width=.44\linewidth]{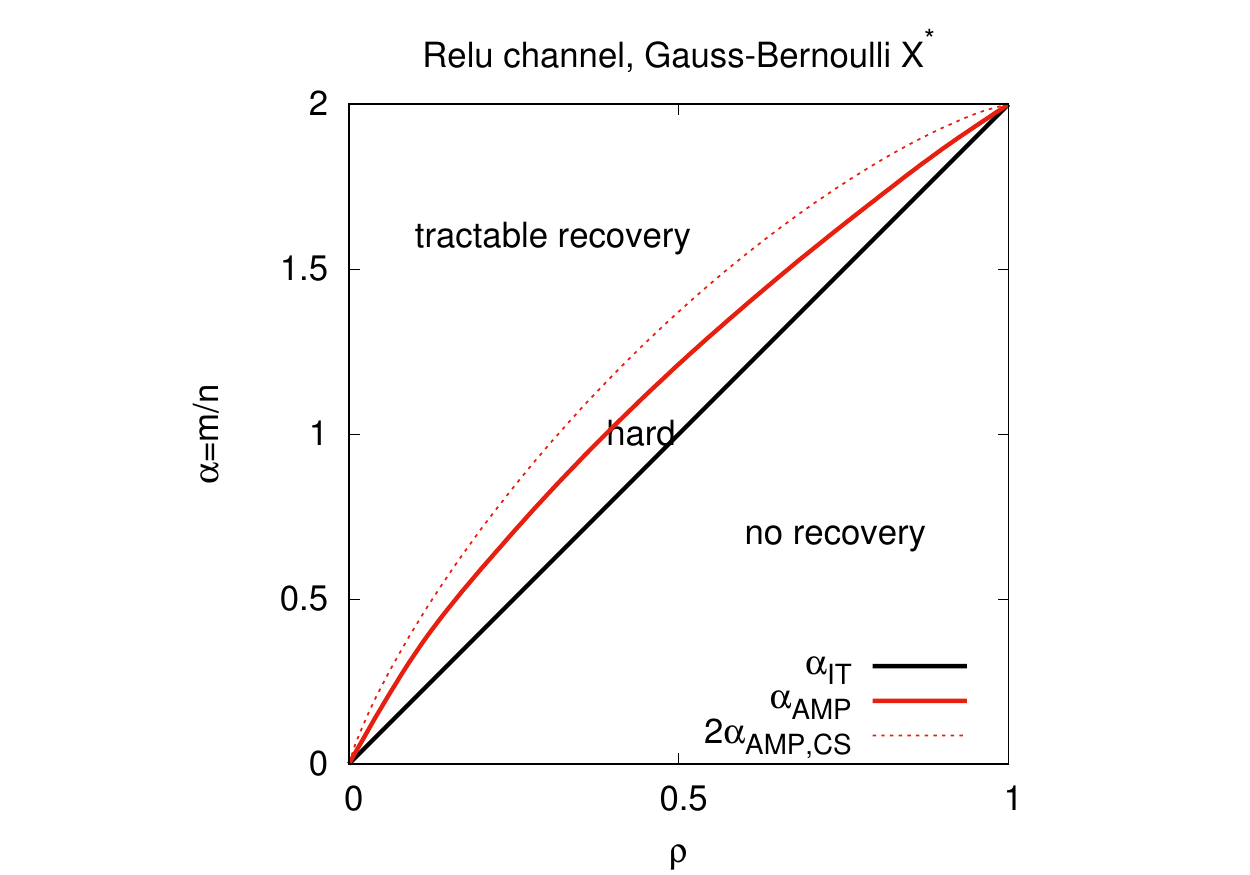}
\hspace{-2.5cm}
\includegraphics[width=.44\linewidth]{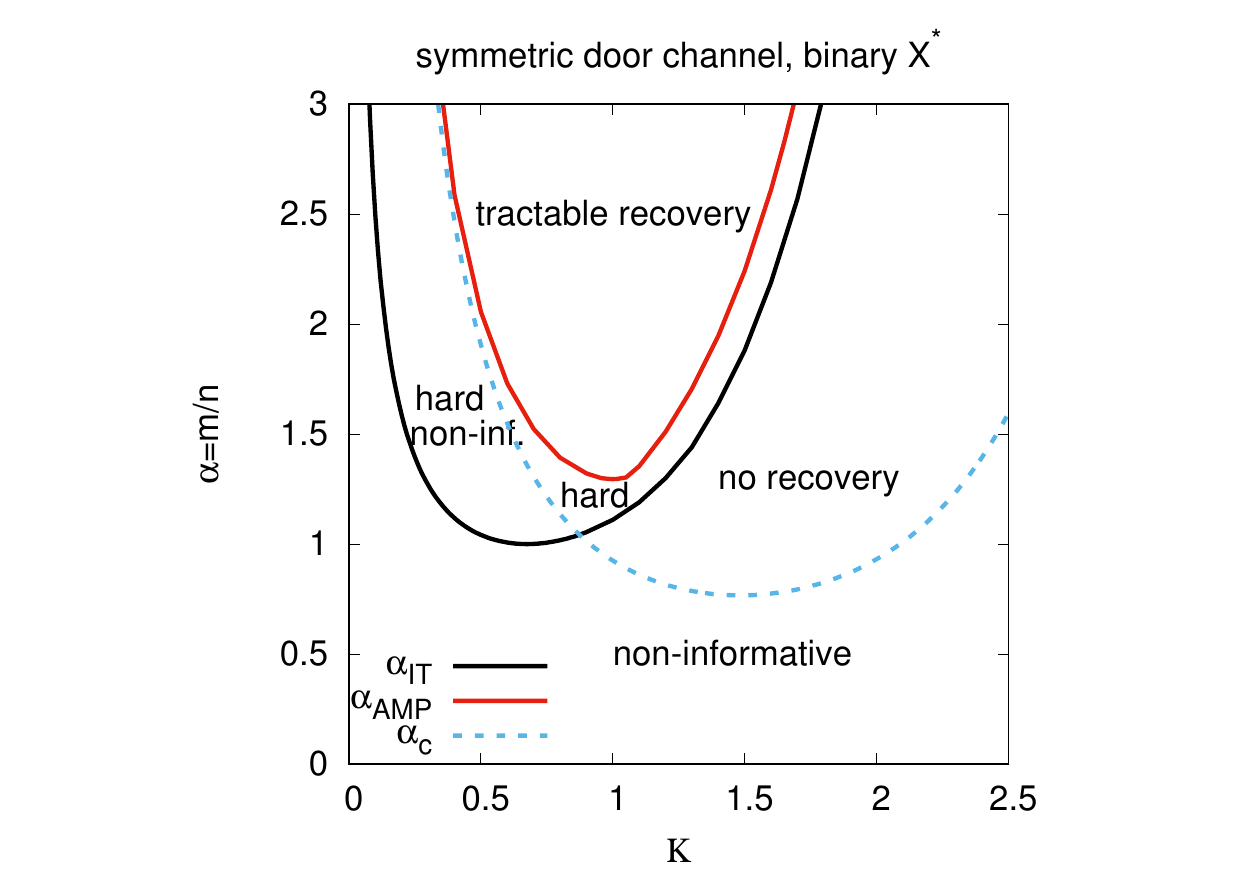}
\caption{\label{fig:phase} Phase diagrams showing boundaries
  of the region where exact recovery is possible (in absence of
  noise). \textbf{Left:} The case of sign-less sparse recovery, $\varphi (x)=|x|$ with a Gauss-Bernoulli signal, as a
  function of the ratio between number of samples/measurements and the
  dimension $\alpha=m/n$, and the fraction of non-zero components
  $\rho$. Evaluating the free entropy for this case, we
  find that a recovery of the signal is information-theoretically impossible for
  $\alpha<\alpha_{\rm IT}=\rho$. Recovery becomes possible starting from
  $\alpha>\rho$, just as in the canonical compressed
  sensing. Algorithmically the sign-less case is much
  harder. Evaluating \eqref{eq:stability2} we conclude that GAMP is not able to perform better than a
  random guess as long as $\alpha<\alpha_c=1/2$.
For larger values of $\alpha$, the
  inference using GAMP leads to better results than a purely random guess.
  GAMP can exactly recover the signal and generalize perfectly only for values of $\alpha$
  larger than $\alpha_{\rm AMP}$ (full red line).
  The dotted red line shows for comparison the
  algorithmic phase transition of the canonical compressed sensing. 
  \textbf{Center:} Analogous to the left panel, for the ReLU output function,
  $\varphi (x)=\max(0,x)$. 
  Here it is always possible to perform better than random
  guessing using GAMP. The dotted red line shows the algorithmic phase transition
  when using information only about the non-zero observations. 
\textbf{Right:} Phase diagram for
  the symmetric door output function $\varphi(z)={\rm sign}(
  |z|-K)$ for a Rademacher signal, as
  a function of $\alpha$ and $K$. The stability line $\alpha_c$ is
  depicted in dashed blue, the information-theoretic phase transition
  to exact recovery $\alpha_{\rm IT}$ in black, and the algorithmic
  one $\alpha_{\rm AMP}$ in red.}
\end{figure*}
\subsubsection{The rectified linear unit (ReLU) channel}
Let us start by discussing the case of a generalized linear model with the
ReLU output channel, i.e.\ $\varphi(x)=\max(0,x)$, with a signal
coming from a Gauss-Bernoulli distribution  $P_0= \rho \cN(0,1) + (1-\rho)\delta_0$ with a fraction $\rho$ of
non-zero (Gaussian) values. We are motivated by the omnipresent use of
the ReLU activation function in deep learning, and explore its
properties for GLMs that can be seen as a simple single layer neural
network. 

Our analysis shows that a
perfect generalization (and thus a perfect reconstruction of the
signal as well) is possible whenever the number of samples per
dimension (measurement rate) $\alpha > 2\rho$, and impossible
when $\alpha < 2 \rho$. This is very intuitive, since half of the
measurements (those non-zero) are giving as much information as in the
linear case, thus the factor $2$. 

How hard is it to actually solve the problem with an efficient algorithm? The answer
is given by applying the state evolution analysis to GAMP, which tells
us that only for even larger values of $\alpha$, beyond the spinodal transition, does GAMP reach a perfect recovery. Notice,
however, that this spinodal transition occurs at a significantly lower
measurement rate $\alpha$ than one would reach just keeping the
non-zero measurements. This shows that, actually, these zero
measurements contain a useful information for the algorithm. The
situation is shown in the center panel of Fig.~\ref{fig:phase}: The
zero measurements do not help information-theoretically but they, however, do help algorithmically. 

%
\subsubsection{The sign-less channel}
\label{signless}
We now discuss the sign-less channel where only the absolute value of
the linear mixture is observed, i.e.\ $\varphi (x)=|x|$. This case can be seen as the
real-valued analog of the famous phase retrieval problem. We again consider the signal
to come from a Gauss-Bernoulli distribution with a fraction $\rho$ of
non-zero (Gaussian) values. 

Sparse phase retrieval has been well explored
in the literature in the regime where the number $s$ of non-zeros is
sub-leading in the dimension, $s = \smallO(n)$. This case is known to
present a
large algorithmic gap. While analogously to
compressed sensing exact recovery
is information-theoretically possible for a number of measurement $\Omega(s \ln(n/s))$, best known algorithms
achieve it only with $\Omega(s^2/\ln n)$ measurements \cite{oymak2015simultaneously}, see also \cite{soltanolkotabi2017structured}
and references therein for a good discussion of other related
literature. This is sometimes referred to as the $s^2$ barrier. We are not aware of a
study where, as in our setting, the sparsity is $s= \rho n$ and the
number of measurements is $\alpha n$ with $\alpha = \Omega(1)$. Our
analysis in this regime hence sheds a
new light on the hardness of the problem of recovering a sparse signal
from sign-less measurements. 

Our analysis of the mutual information shows that a perfect
reconstruction is information-theoretically possible as soon as
$\alpha >\rho$: In other words, the problem is --information-theoretically-- as easy, or as hard as the compressed sensing one. This is maybe less surprising when one thinks of the following algorithm: Try all $2^m$ choices of the possible signs for the $m$
outputs, and solve a compressed sensing problem for each of
them. Clearly, this should yields a perfect solution only in the case
of the actual combination of signs. 

Algorithmically, however, the problem is much harder than for the
linear output channel. As
shown in the left side of Fig.~\ref{fig:phase}, for small $\rho$ one requires a
much larger fraction $\alpha$ of measurements in order for GAMP to
recover the signal. For the linear channel the algorithmic
transition $\alpha_s(\rho)\to 0$ as $\rho \to 0$, while for the sign-less
channel we get $\alpha_s(\rho)\to 1/2$ as $\rho \to 0$. In other words
if one looses the signs one cannot perform recovery in compressed
sensing with less than $n/2$ measurements.  

What we observe in this example for $\alpha<1/2$ is in the statistical
physics literature on neural networks known as retarded learning
\cite{hansel1992memorization}. This appears in problems where the 
$\varphi(x)$ function is symmetric, as seen at the beginning of this section: There is always a critical point of the mutual information with an overlap value $q=0$. For this problem, this critical point is actually ``stable'' (meaning
that it is actually a local minimum in $q$ in the mutual information \eqref{eq:lim_i}) for all $\alpha<1/2$
independently of $\rho$. 
To see that, we have to go back to \eqref{eq:stability}. For the absolute value channel,
the posterior distribution of $W^*$ given $\widetilde{Y}_0 = \sqrt{\rho}\,|W^*|$ in the second scalar channel \eqref{eq:Pout_scalar_channel} is a mean of two Dirac masses on $-W^*$ and $W^*$. Thus, the posterior variance is $\langle w^2 \rangle_{\rm sc} - \langle w \rangle_{\rm sc}^2 = (W^*)^2$. Consequently, \eqref{eq:stability} leads to the stability condition $\alpha < 1/2$.

This has
the two following implications: $i)$ In the non-informative phase, when
$\alpha<1/2$ and $\rho>\alpha$, the minimum at $q=0$ is actually the
global one. In this case, the MMSE on $\bX^*$ and the generalization error are the
ones given by using $0$ as a guess for each element of $\bX^*$; in other words, there is no useful information that one can exploit and no
algorithmic approach can be better than a random guess. $ii)$ In the
hard non-informative phase when $\alpha<1/2$, GAMP initialized at
random, i.e.\ close to the $q=0$ fixed point, will remain there. This suggests that in this region, even if a perfect
reconstruction is information-theoretically possible, it will still be very hard to beat a
random guess with a tractable algorithm. 
%

\subsubsection{The symmetric door channel}
The third output channel we study in detail is the symmetric door channel,
where $\varphi (x)={\rm sgn}(|x|-K)$. In case of channels with discrete
set of outputs exact recovery is only possible when the prior is also
discrete. In the present case we consider the signal to be Rademacher,
where each element is chosen at random
between $1$ and $-1$, i.e.\ $P_0=\frac{1}{2}\delta_{+1}+\frac{1}{2}\delta_{-1}$. This channel was studied previously using the
replica method in the context of optimal data compression \cite{hosaka2002statistical}.

This output channel is in the class of symmetric channels for which
overlap $q=0$ is a fixed point. This fixed point is stable for
$\alpha<\alpha_c(K)$. Exact recovery is information-theoretically
possible above $\alpha_{\rm IT}(K)$ and tractable with the GAMP
algorithm above the spinodal transition $\alpha_s(K)$. The
values of these three transition lines are depicted in the right panel of Fig.~\ref{fig:phase}. 

We note that $\alpha_{\rm IT}\ge 1$ is a generic bound on exact recovery for every
$K$, required by a simple counting argument. While a-priori it is not
clear whether this bound is saturated for some $K$, we observe that it
is for $K=0.67449$ such that half of the observed measurements are
negative and the rest positive. This is, however, not efficiently
achievable with GAMP. The saturation of the $\alpha_{\rm IT}\ge 1$ bound
was remarked previously in the context of the work
\cite{hosaka2002statistical} on optimal data compression. Our work
predicts that this information-theoretic result will not be achievable
with known efficient algorithms.

\subsection{Examples of optimal generalization error}
\begin{figure*}[t]
\centering
\hspace*{-0.029\linewidth}
\includegraphics[width=1.05\linewidth]{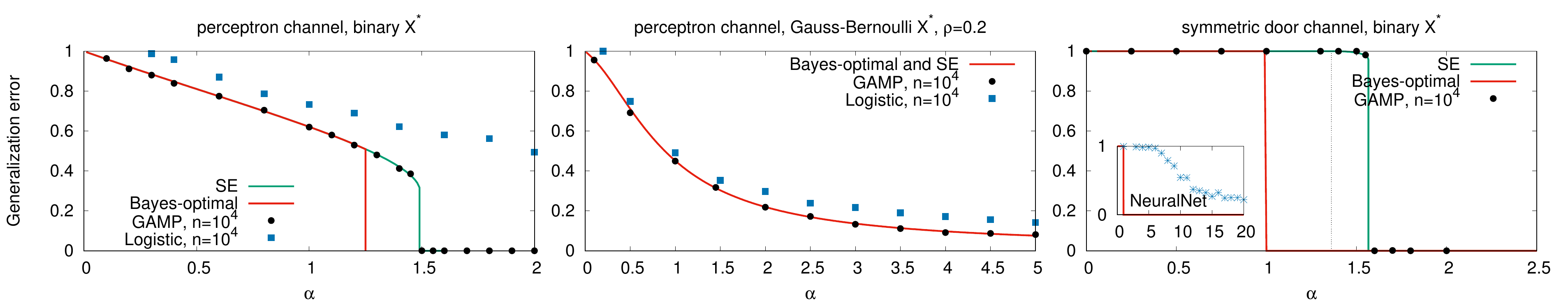}
\caption{\label{fig:classif} Generalization error in three
  classification problems as a function of the number of data-samples
  per dimension $\alpha$. The red line is the Bayes-optimal
  generalization error, while the green one shows the (asymptotic)
  performances of GAMP as predicted by the state evolution (SE), when
  different. For comparison, we also show the result of GAMP (black
  dots) and, in blue, the performance of a standard out-of-the-box
  solver, both tested on a single randomly generated instance.
  \textbf{Left:} Perceptron, with $\varphi (x)={\rm sgn}(x)$ and a
  Rademacher ($\pm 1$) signal. While a perfect generalization is
  information-theoretically possible starting from $\alpha=1.249(1)$,
  the state evolution predicts that GAMP will allow such perfect
  prediction only from $\alpha=1.493(1)$. The results of a logistic
  regression with fine-tuned ridge penalty with the software
  scikit-learn \cite{scikit-learn} are shown for
  comparison. \textbf{Middle:} Perceptron with Gauss-Bernoulli
  coefficients for the signal. No phase transition is observed in this
  case, but a smooth decrease of the error with $\alpha$. The results
  of a logistic regression with fine-tuned $\ell_1$ sparsity-enhancing
  penalty (again with \cite{scikit-learn}) are very close to
  optimal. \textbf{Right:} The symmetric door activation rule with
  parameter $K=0.67449$ chosen in order to observe the same number of
  occurrence of the two classes. In this case there is a sharp phase
  transition at $\alpha=1$ from a situation where it is impossible to
  learn the rule, so that the generalization is not better than a
  random guess, to a situation where the optimal generalization error
  drops to zero. However, GAMP identifies the rule perfectly only
  starting from $\alpha_s=1.566(1)$ (GAMP error stays 1 up to
  $\alpha_{\rm stab}=1.36$, see the black dashed
  curve). Interestingly, this non-linear rule seems very hard to learn
  for other existing algorithms. Using Keras \cite{chollet2015}, a
  neural network with two hidden layers was able to learn
  approximately the rule, but only for much larger training set sizes
  (shown in inset, the Keras/tensorflow code for this particular run
  can be found on the github repository \cite{githubrepo}).}
\end{figure*}

Besides the formula for the mutual information, the main result of this paper is the Theorem \ref{th:gen} for the optimal generalization error, and formula \eqref{eq:lim_err_gamp}
for the generalization error achieved by the GAMP algorithm. In this
section we evaluate both these generalization errors for several
cases of priors and output functions. We study both regression problems, where
the output is real-valued, and classification problems, where the
output is discrete. 

While in realistic regression and classification problems the matrix $\boldsymbol{\Phi}$
corresponds to the data, and is thus not i.i.d.\ random, we view the
practical interest of our theory as a benchmark for state-of-the art
algorithms. Our work provides an exact asymptotic analysis of
optimal generalization error and sample complexity for a range of
simple rules where a teacher uses random data to generate labels. The challenge for state-of-the-art multi-purpose algorithms is to try to
match as closely as possible the performance that can be obtained with
GAMP that is fine-tuned to the specific form of the output and prior.

%
%

%
\subsubsection{Threshold output: The perceptron}
The example of non-linear output that is the most widely explored in the
literature is the threshold output, where the deterministic output (or
``activation'') function is $\varphi(x)={\rm sgn }(x)$. This output in
the teacher-student setting of the present paper is known as the
perceptron problem \cite{gardner1989three}, or equivalently, the
one-bit compressed sensing in signal processing
\cite{boufounos20081}. Its solution has been discussed in details within the replica
  formalism (see for instance
  \cite{baum1991transition,opper1991generalization,Kaba,xu2014bayesian})
  and we confirm all of these heuristic computations within our approach. Let $V\sim{\cal N}(0,1)$. The formula \eqref{Egen_final}, \eqref{eq:gen0} for the generalization error then reduces to (recall $q^*=q^*(\alpha)$ is an optimizer of \eqref{eq:rs_formula})
\begin{align}
\lim_{n\to\infty} {\cal E}_{\rm gen}^{\rm opt}
&= 1- \int {\cal D}V \Big(\frac{2}{\sqrt{\pi}} \int_0^{V\sqrt{\frac{q^*}{2(\rho -q^*)}}} dt\, e^{-t^2/2} \Big)^2
=1-\E\Big[{\rm erf}\Big(V\sqrt{\frac{q^*}{2(\rho-q^*)}}\Big)^2\Big]\,. \label{Egen_final_binary}
\end{align}

In Fig.~\ref{fig:classif} (left) we plot the optimal generalization
error of the perceptron with a Rademacher signal, the state evolution prediction of the generalization error of
the GAMP algorithm, together with the error actually achieved by GAMP
on one randomly generated instance of the problem. We also compare these to
the performance of a standard logistic regression. As expected from existing literature \cite{gardner1989three,ReviewTishby} we confirm that in this
case the information-theoretic transition appears at a number of samples
per dimension $\alpha_{\rm IT} =1.249(1)$, while the algorithmic transition is at $\alpha_s=1.493(1)$. Logistic regression does not seem to be able
to match the performance on GAMP in this case. 

In Fig.~\ref{fig:classif} (center) we plot the generalization error
for a Gauss-Bernoulli signal with density $\rho=0.2$. Cases as this one
were studied in detail in the context of one-bit compressed sensing
\cite{xu2014bayesian} and GAMP was found to match the optimal
generalization performance with no phase transitions observed, which is confirmed by our analysis. In this case the
logistic regression is rather close to the performance of GAMP.

\subsubsection{Symmetric Door}
The next classification problem, i.e.\ discrete output rule, we study is
the symmetric door function $\varphi(x)={\rm sgn}(|x|-K)$. 
In this case the generalization error \eqref{Egen_final} becomes (here again $V\sim{\cal N}(0,1)$)
\begin{align}
\lim_{n\to\infty} {\cal E}_{\rm gen}^{\rm opt} =1-\E_V\Big[ \Big\{ {\rm{erf}}{\Big(\frac{K-\sqrt{q^*}\,V}{\sqrt{2(\rho-q^*)}}\Big)}
-{\rm{erf}}{\Big(-\frac{K+\sqrt{q^*}\,V}{\sqrt{2(\rho-q^*)}}\Big)}-1\Big\}^2\Big ]\,. \label{EsymDoor}
\end{align}

In Fig.~\ref{fig:classif} (right) we plot the generalization error for $K=0,67449$
such that $1/2$ of the outputs are $1$ and $1/2$ are $-1$. 
The symmetric door output is an example of function for which the
optimal generalization error for $\alpha<\alpha_{\rm IT}=1$ (for that
specific value of $K$, see phase diagram in the right panel of
Fig.~\ref{fig:phase}) is as bad as if we were guessing randomly. The
GAMP algorithm still achieves such a bad generalization until $\alpha_{\rm stab}=1.36$, and achieves perfect generalization only for $\alpha>
\alpha_{s} = 1.566(1)$. 

Interestingly, labels created from this very simple symmetric door rule
seem to be very challenging to learn for general purpose
algorithms. We tried to optimize parameters of a two-layers neural
network and only managed to get the performances shown in the inset of
Fig.~\ref{fig:classif} (right). It is an interesting theoretical
challenge whether a deeper neural network can learn this simple rule
from fewer samples.

\begin{figure*}[t]
\centering
\hspace*{-0.029\linewidth}
\includegraphics[width=1.05\linewidth]{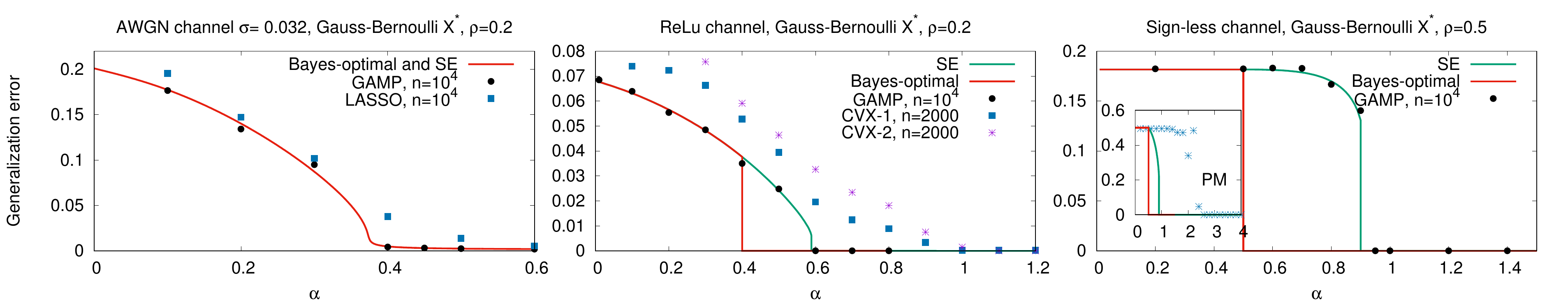}
\caption{The generalization error for three regression problems is
  plotted as a function of the number of samples per dimension
  $\alpha$. The red line is again the Bayes-optimal generalization
  error, while the green one shows the (asymptotic) performances of
  GAMP as predicted by the state evolution (SE), when
  different. Again, we also show the result of GAMP on a particular
  instance (black dots) and, in blue, the performance of an
  out-of-the-box solver.  \textbf{Left:} White Gaussian noise output
  and a Gauss-Bernoulli signal. For this choice of noise, there is no
  sharp transition (as opposed to what happens at smaller noises). The
  results of a LASSO with fine-tuned $\ell_1$ sparsity-enhancing
  penalty (with \cite{scikit-learn}) are very close to
  optimal. \textbf{Middle:} Here we analyze a ReLU output function
  $\varphi (x)=\max(0,x)$, still with a Gauss-Bernoulli signal. Now
  there is an information-theoretic phase transition at
  $\alpha=2\rho=0.4$, but GAMP requires $\alpha_s=0.589(1)$ to reach
  perfect recovery. We show for comparison the results of maximum
  likelihood estimation performed with CVXPY ---a powerful
  python-embedded language for convex optimization
  \cite{cvxpy}--- using two methods that are both amenable to convex
  optimization: In CVX-1 we use only the non-zero values of $\bY$, and
  perform a minimization of the $\ell_1$ norm of $\bx$ subject to
  $Y_{\mu}={\bbf \Phi}_{\mu} \cdot \bx$ for $\mu\in\{1,\ldots,m\}$ such that $Y_{\mu}\neq 0$, while in CVX-2, we use all the dataset, with the constraint
  that $Y_{\mu}={\bf \Phi}_{\mu}  \cdot \bx$ for $\mu\in\{1,\ldots,m\}$ such that $Y_{\mu}\neq 0$ (as before) and the
  additional restriction ${\bf \Phi}_{\mu} \cdot \bx \le 0$ for $\mu\in\{1,\ldots,m\}$ such that $Y_{\mu}= 0$. In both case, a perfect generalization is obtained only
  for $\alpha\gtrapprox1$. \textbf{Right:} The sign-less output
  function $\varphi (x)=|x|$. The information-theoretic perfect
  recovery starts at $\alpha=\rho=0.5$, but the problem is again
  harder algorithmically for GAMP that succeeds only above
  $\alpha_s=0.90(1)$. Again, the problem appears to be hard for other
  solvers. In inset, we show the performance for the estimation
  problem using PhaseMax\cite{phasemax}, which is able to learn the
  rule only using about four times as many measurements than needed
  information-theoretically.
  \label{fig:regression}}
\end{figure*}

\subsubsection{Linear regression}
The additive white Gaussian noise (AWGN) channel, or linear regression, is defined by $\varphi(x,A)=x + \sigma A$ with $A\sim{\cal N}(0,1)$. This models the (noisy) linear regression problem, as well as noisy random linear estimation and compressed sensing. In this case \eqref{Egen_final} leads to
\begin{align}
\lim_{n\to\infty} {\cal E}_{\rm gen}^{\rm opt} = \rho - q^* + \sigma^2\,.  \label{ElinReg}
\end{align}
This result agrees with the generalization error analyzed heuristically
in \cite{ReviewTishby} in the limit $\sigma\to
0$. Fig.~\ref{fig:regression} (left) depicts the generalization error
for this example. The performance of GAMP in this case is very close to the
one of LASSO. 
\subsubsection{Rectified linear unit (ReLU)}
In Fig.~\ref{fig:regression} (center) we analyze the generalization
error for  the ReLU output function, $\varphi(x)=\max(0,x)$. This
channel models the behavior of a single neuron with the rectified linear unit activation \cite{lecun2015deep}  widely used in
  multilayer neural networks. In this case \eqref{Egen_final} becomes after simple algebra and Gaussian integration by parts (again $V\sim{\cal N}(0,1)$),
\begin{align}
\lim_{n\to\infty} {\cal E}_{\rm gen}^{\rm opt} 
= \frac{\rho}{2}-\frac{q^*}{4}\Big(1+\E_V\Big[V^2{\rm erf}\Big(V\sqrt{\frac{q^*}{2(\rho-q^*)}}\Big)^2\Big]\Big) 
- \frac{(\rho-q^*)^{3/2}}{\sqrt{\rho+q^*}} \Big(\frac{1}{2\pi}+\frac{q^*}{\rho\pi}\sqrt{\frac{\rho+q^*}{\rho-q^*}}\Big)\,.\label{Erelu}
\end{align}

For sparse Gauss-Bernoulli signals in Fig.~\ref{fig:regression}
(center) we observe again the information-theoretic transition to
perfect generalization to be distinct from the algorithmic one. At the
same time our test with existing algorithms were not able to closely match the
performance of GAMP. This hence also remains an interesting benchmark.

%
%
%
%
%

\subsubsection{Sign-less channel}
In Fig.~\ref{fig:regression} (right) we analyze the generalization
error for the sign-less output function  where $\varphi(x)=|x|$. This models
  a situation similar to compressed sensing, except that the sign of
  the output has been lost. This is a real-valued analog of the phase
  retrieval problem as discussed in Sec.~\ref{signless}. In this case the
  generalization error \eqref{Egen_final} becomes (again $V\sim{\cal N}(0,1)$)
\begin{align}
\lim_{n\to\infty} {\cal E}_{\rm gen}^{\rm opt} =
\rho-\E_V\big[ b(V\sqrt{q^*},\rho-q^*)^2 \big ]\,,\label{Eabs}
\end{align}
where 
\begin{align}
b(x,y)=\sqrt{\frac {2y}{\pi}}e^{-\frac{x^2}{2y}}+\frac x2 {{\rm erfc}}{\Big(-\frac{x}{\sqrt{2y}}\Big)}-\frac x2 \Big\{1+{\rm erf}{\Big(-\frac{x}{\sqrt{2y}}\Big)}\Big\}\,.
\end{align}

Our comparison with the performance of a state-of-the-art algorithm
PhaseMax \cite{phasemax} suggests that also for this simple benchmark there is room for
improvement in term of matching the performance of GAMP. 

\subsubsection{Sigmoid, or logistic regression}
Let us also consider an output function with auxiliary randomization. After having generated the classifier $\bX^*$, the teacher randomly associates the label $+1$ to the pattern $\boldsymbol \Phi_\mu$ with probability $f_{\lambda}(n^{-1/2}\boldsymbol \Phi_\mu \cdot \bX^*)$, where 
$f_\lambda(x)=(1+\exp(-\lambda x))^{-1} \in[0,1]$ is the sigmoid of parameter $\lambda>0$, and the label $-1$ with probability $1-f_{\lambda}(n^{-1/2}\boldsymbol \Phi_\mu \cdot \bX^*)$. One of the (many) possible ways for the teacher to do so is by selecting $\varphi(x,A)= \bbf{1}(A \le f_\lambda(x)) -\bbf{1}(A > f_\lambda(x))$, where $\bbf{1}(E)$ is the indicator function of the event $E$. He then generates a stream of uniform random numbers $\bA\iid {\cal U}_{[0,1]}$ and obtains the labels through \eqref{measurements} (with $\Delta=0$). Let $V,w\iid{\cal N}(0,1)$. In this setting the error \eqref{Egen_final} becomes
\begin{align}
\lim_{n\to\infty} {\cal E}_{\rm gen}^{\rm opt} 
&= 2-4\,\E_V\big[\big\{\E_wf_\lambda(\sqrt{q^*}\,V+\sqrt{\rho-q^*}\,w) \big\}^2 \big]\,. \label{34}
\end{align}
This formula reduces to \eqref{Egen_final_binary} 
when $\lambda\to\infty$ as it should. 

\section{Proof of the replica formula by the adaptive interpolation method}
\label{sec:interpolation}
We now prove Theorem~\ref{th:RS_1layer}. Our main tool will be an interpolation method recently introduced in \cite{barbier_stoInt} and called ``adaptive interpolation method''. Here we 
formulate the method as a direct evolution of the Guerra and Toninelli interpolation method developed in the context of spin glasses \cite{guerra2002thermodynamic}. In contrast with the discrete and more pedestrian version of the adaptive interpolation method presented in \cite{barbier_stoInt}, here we employ a continuous approach which is more straightforward (see \cite{barbier_stoInt} for the links between the discrete and continuous versions of the method) and that has also been recently used in \cite{2017arXiv170910368B} for studying non-symmetric tensor estimation.

We will prove Theorem~\ref{th:RS_1layer} under the following hypotheses:
\begin{enumerate}[label=(H\arabic*),noitemsep]
	\item \label{hyp:bounded} The support of the prior distribution $P_0$ is included in $[-S,S]$, for some $S >0$.
	\item \label{hyp:c2} $\varphi$ is a bounded $\cC^2$ function with bounded first and second derivatives w.r.t.\ its first argument.
	\item \label{hyp:phi_gauss2} $(\Phi_{\mu i}) \iid \cN(0,1)$.
\end{enumerate}
These stronger assumptions will then be relaxed in Appendix~\ref{Appendix-approx} to the weaker assumptions \ref{hyp:third_moment}-\ref{hyp:um}-\ref{hyp:phi_general}-\ref{hyp:cont_pp} and \ref{hyp:delta_pos}
or \ref{hyp:delta_0}.
Since the observations \eqref{measurements} are equivalent to the rescaled observations
\begin{align}\label{measurements_2}
	\widetilde{Y}_\mu \defeq \Delta^{-1/2}\,Y_{\mu}  = \Delta^{-1/2}\,\varphi\Big(\frac{1}{\sqrt{n}} [\boldsymbol{\Phi} \bX^*]_{\mu}, \bA_\mu\Big) + Z_\mu\,, \qquad 1\leq \mu \leq m\,,
\end{align}
the variance $\Delta$ of the Gaussian noise can be ``incorporated'' inside the function $\varphi$.
Thus, it suffices to prove Theorem~\ref{th:RS_1layer} for $\Delta=1$ and we suppose, for the rest of the proof, that we are in this equivalent case.

\subsection{Interpolating estimation problem}\label{interp-est-problem}
We introduce an ``interpolating estimation problem'' that interpolates between the original problem \eqref{eq:channel} at $t=0$, $t\in[0,1]$ being the interpolation parameter, and the two scalar problems described in Sec.~\ref{sec:scalar_inf} at $t=1$ which are analytically tractable. For $t\in(0,1)$ the interpolating estimation problem is a mixture of the original and scalar problems. This interpolation scheme is inspired from the interpolation paths used by Talagrand to study the perceptron, see \cite{talagrand2010meanfield1}. 
There are two major differences between the ``non-planted perceptron'' studied by Talagrand, and the ``planted perceptron'' that we are investigating:
\begin{itemize}
	\item In the planted case, the presence of a planted solution forces (under small perturbations) the correlations to vanish for all values of the parameters, see \cite{andrea2008estimating,korada2009exact}. In the non-planted case, proving such decorrelation is much more involved, and is proved only in a limited region of the parameter space, see \cite{talagrand2010meanfield1}.
	\item However, in the planted case, there can be arbitrarily many solutions to the state evolution equations \eqref{fixed_points} (see Remark~21 in~\cite{wu2012optimal}), whereas in the region studied by \cite{talagrand2010meanfield1}, there is only one solution. For this reason, our interpolation method needs to be more sophisticated in order to interpolate with the ``right fixed point''.
\end{itemize}
We fix a sequence $(s_n)_{n \geq 1}\in(0,1/2]^{\N}$ that converges to $0$ as $n$ goes to infinity ($s_n$ will be chosen in Sec.~\ref{sec:overlap_concentration} below to be equal to $\frac{1}{2} n^{-1/16}$). We define $\mathcal{B}_n \defeq [s_n,2 s_n]^2$.
For all $\epsilon = (\epsilon_1, \epsilon_2) \in \mathcal{B}_n$, we consider two continuous ``interpolation functions''
$q_{\epsilon}: [0,1] \to [0,\rho]$ and $r_{\epsilon}: [0,1] \to [0,r_{\rm max}]$,
where $r_{\rm max} \defeq 2\alpha \sup_{q\in[0,\rho]} \Psi'_{P_{\rm out}}(q;\rho)= 2\alpha \Psi'_{P_{\rm out}}(\rho;\rho)$ (recall that by Proposition \ref{prop:psi_convex_reg}, $\Psi_{P_{\rm out}}'$ is non-decreasing). 
We define also for all $t\in[0,1]$ and all $\epsilon \in \mathcal{B}_n$
\begin{align}\label{R1R2}
	R_{1}(t,\epsilon)\defeq\epsilon_1+\int_0^t r_{\epsilon}(v)dv\,,\qquad 
	R_2(t,\epsilon)\defeq \epsilon_2+\int_0^t q_{\epsilon}(v)dv\,.
\end{align}
We will be mainly interested in functions $r_{\epsilon}$, $q_{\epsilon}$ that satisfy some regularity properties. We will use the following definition:
\begin{definition}[Regularity]\label{def:reg}
	We say that the families of functions $(q_{\epsilon})_{\epsilon \in \mathcal{B}_n}$ and $(r_{\epsilon})_{\epsilon \in \mathcal{B}_n}$, taking values respectively in $[0,\rho]$ and $[0,r_{\rm max}]$, are regular if for all $t \in [0,1]$ the mapping
	\begin{equation}
		R^t:
		\left|
		\begin{array}{ccc}
			(s_n,2 s_n)^2 & \to & R^t \big((s_n,2 s_n)^2\big) \\
			\epsilon & \mapsto & \big(R_1(t,\epsilon), R_2(t,\epsilon )\big)
		\end{array}
		\right.
	\end{equation}
	is a $\cC^1$ diffeomorphism, whose Jacobian is greater or equal to $1$.
\end{definition}
Define
\begin{align}
	S_{t,\mu} \defeq \sqrt{\frac{1-t}{n}}\, [\boldsymbol{\Phi} \bX^*]_\mu  + \sqrt{R_2(t,\epsilon)} \,V_{\mu} + \sqrt{\rho t -R_2(t,\epsilon)+2s_n} \,W_{\mu}^*	
\end{align}
where $V_{\mu}, W^*_{\mu} \iid \cN(0,1)$. Consider the following observation channels, with two types of observations obtained through 
\begin{align}
	\label{2channels}
	\left\{
		\begin{array}{llll}
			Y_{t,\mu}  &\sim & P_{\rm out}(\ \cdot \ | \, S_{t,\mu})\,,\qquad  &1 \leq \mu \leq m, \\
			Y'_{t,i} &=& \sqrt{R_1(t,\epsilon)}\, X^*_i + Z'_i\,, \qquad &1 \leq i \leq n,
		\end{array}	
	\right.
\end{align}
where $(Z_i')_{i=1}^n\iid {\cal N}(0,1)$. We assume that $\bV=(V_{\mu})_{\mu=1}^m$ is {\it known}. Then the inference problem is to recover both unknowns $\bW^* = (W_{\mu}^*)_{\mu=1}^m$ and $\bX^*=(X_i^*)_{i=1}^n$ from the knowledge of $\bV$, $\boldsymbol{\Phi}$ and the ``time-dependent'' observations $\bY_{t}=(Y_{t,\mu})_{\mu=1}^m$ and $\bY'_{t}=(Y_{t,i}')_{i=1}^n$. 

We now understand that $R_1(t,\epsilon)$ appearing in the second set of measurements in \eqref{2channels}, and the terms $1-t$, $R_{2}(t,\epsilon)$ and $\rho t -R_2(t,\epsilon)+2s_n$ appearing in the first set all play the role of signal-to-noise ratios in the interpolating model, with $t$ giving more and more ``power'' (or weight) to the scalar inference channels when increasing. Here is the first crucial and novel ingredient of our interpolation scheme. In the classical interpolation method, these signal intensities would all take a trivial form (i.e.\ would be linear in $t$) but here, the non-trivial (integral) dependency in $t$ of the intensities through the use of the interpolation functions $q$ and $r$ allows for much more flexibility when choosing the interpolation path. This will allow us to actually choose the ``optimal interpolation path'' (this will become clear soon).

Define $u_y(x) \defeq \ln P_{\rm out}(y|x)$ and, with a slight abuse of notations,
\begin{align}
	s_{t,\mu} = s_{t, \mu}(\bx,w_\mu) &\defeq \sqrt{\frac{1-t}{n}} [\boldsymbol{\Phi}\bx]_{\mu}  + \sqrt{R_{2}(t,\epsilon)}\, V_{\mu} + \sqrt{\rho t -R_2(t,\epsilon)+2s_n}\, w_{\mu} \,.
	\label{stmu}
\end{align}
%
We introduce the {\it interpolating Hamiltonian}
\begin{align}
	\cH_{t,\epsilon}(\bx,\bw;\bY_t,\bY'_t,\boldsymbol{\Phi},\bV)
	& \defeq
	- \sum_{\mu=1}^{m}
	\ln P_{\rm out} ( Y_{t,\mu} |s_{t, \mu}) + \frac{1}{2} \sum_{i=1}^{n}\big(Y'_{t,i}  - \sqrt{R_{1}(t,\epsilon)}\, x_i\big)^2 .
\end{align}
The dependence in $\boldsymbol{\Phi}$ and $\bV$ of the Hamiltonian is through the $(s_{t,\mu})_{\mu=1}^m$. 
It becomes, when the observations are replaced by their expression \eqref{2channels}, 
\begin{align}
	\cH_{t,\epsilon}(\bx,\bw;\bY_t,\bY_t',\boldsymbol{\Phi},\bV)=- \sum_{\mu=1}^{m}
	u_{Y_{t,\mu}}( s_{t, \mu} ) 
	+ \frac{1}{2} \sum_{i=1}^{n}\big( \sqrt{R_{1}(t,\epsilon)}\, (X_i^*-x_i) +Z_i'\big)^2 \,.
	\label{interpolating-ham}
\end{align}
We also introduce the corresponding Gibbs bracket $\langle - \rangle_{n,t,\epsilon}$ which is the expectation operator w.r.t.\ the $(t,\epsilon)$-dependent posterior distribution of $(\bX^*,\bW^*)$ given $(\bY_{t},\bY_t',\boldsymbol{\Phi},\bV)$. It is defined as 
\begin{align}
	\label{gibbs}
	\langle g(\bx,\bw) \rangle_{n,t,\epsilon} \defeq \frac{1}{\cZ_{t,\epsilon}(\bY_t,\bY_t',\boldsymbol{\Phi},\bV)}\int dP_0(\bx){\cal D}\bw \, g(\bx,\bw)\,e^{-\cH_{t,\epsilon}(\bx,\bw;\bY_t,\bY_t',\boldsymbol{\Phi},\bV)} \, ,	
\end{align}
for every continuous bounded function $g$ on $\R^n \times \R^m$. In \eqref{gibbs}
${\cal D}\bw = (2\pi)^{-m/2}\prod_{\mu=1}^m dw_\mu e^{-w_\mu^2/2}$ is the $m$-dimensional standard Gaussian distribution and 
$\cZ_{t,\epsilon}(\bY_t,\bY_t',\boldsymbol{\Phi},\bV)$ is the appropriate normalization (or {\it partition function}):
\begin{align} \label{Zt}
	\cZ_{t,\epsilon}(\bY_t,\bY_t',\boldsymbol{\Phi},\bV) \defeq \int dP_0(\bx){\cal D}\bw \, e^{-\cH_{t,\epsilon}(\bx,\bw;\bY_t,\bY_t',\boldsymbol{\Phi},\bV)}\,.
\end{align}
Finally the {\it interpolating free entropy} is 
\begin{align}
	f_{n,\epsilon}(t) \defeq \frac{1}{n} \E \ln \cZ_{t,\epsilon}(\bY_t,\bY'_t,\boldsymbol{\Phi},\bV) 
	\,.	\label{ft}
\end{align}
Note that the presence of the {\it perturbation} $\epsilon=(\epsilon_1,\epsilon_2)$ induces only a small change in the free entropy, namely of the order of $s_n$:
\begin{lemma}[Small free entropy variation under perturbation]\label{lem:perturbation_f}
	For all $\epsilon_1,\epsilon_2 \in[s_n,2s_n]$,
	\begin{align}
		|f_{n,\epsilon}(0) - f_{n,\epsilon=(0,0)}(0) | \leq C s_n 
	\end{align}	
	for some constant $C$ that only depends on $S$, $\alpha$ and $\varphi$.
\end{lemma}
\begin{proof}
	Let us compute 
	\begin{align}
		\Big|\frac{d f_{n,\epsilon}(0)}{d {\epsilon_1}} \Big| = \frac{1}{2}|\E\langle Q\rangle_{n,0,\epsilon}| \leq \frac{S^2}{2},
	\end{align}
	by hypothesis \ref{hyp:bounded}.
	Next we compute
	\begin{align}
	\Big|\frac{d f_{n,\epsilon}(0)}{d {\epsilon_2}}\Big|= \frac{1}{2n}\sum_{\mu=1}^m\big|\E \big[u'_{Y_{0,\mu}}(S_{0,\mu})\langle u'_{Y_{0,\mu}}(s_{0,\mu}) \rangle_{n,0,\epsilon}\big]\big|\,.
	\end{align}
	This identity is obtained using very similar steps as in Sec.~\ref{appendix_interpolation} to which we refer. Under hypothesis \ref{hyp:c2} this quantity is bounded by a constant that only depends on $\alpha$ and $\varphi$. Then by the mean value theorem we obtain $|f_{n,\epsilon}(0) - f_{n,(0,0)}(0) | \leq C \|\epsilon\|\le 2 \sqrt{2} s_n C$ for some constant $C$ that only depends on $S$, $\alpha$ and $\varphi$.
\end{proof}
One verifies easily, using the Lemma \ref{lem:perturbation_f}, that for all $\epsilon \in \mathcal{B}_n$
\begin{equation}\label{eq:f0_f1}
	\left\{
		\begin{array}{lll}
			f_{n,\epsilon}(0) &=& f_{n,(0,0)}(0)+ {\cal O}(s_n)\\
			&=& f_n - \frac12 + {\cal O}(s_n) \,,\\
			f_{n,\epsilon}(1) &=& \psi_{P_0}(R_1(1,\epsilon)) - \frac{1}2(1 + \rho R_1(1,\epsilon)) + \frac{m}{n} \Psi_{P_{\rm out}}(R_2(1,\epsilon);\rho+2s_n)\\
							  &=& \psi_{P_0}(\int_0^1 r_{\epsilon}(t)dt) - \frac{1}2(1 + \rho\int_0^1r_{\epsilon}(t)dt) + \frac{m}{n} \Psi_{P_{\rm out}}(\int_0^1 q_{\epsilon}(t) dt;\rho) + {\cal O}(s_n)\,.
		\end{array}
	\right.
\end{equation}
where $f_n$ is given by \eqref{fff} and where ${\cal O}(s_n)$ denotes a quantity that is bounded by $C s_n$ for some constant $C>0$ that only depends on $S$, $\varphi$ and $\alpha$.
For the last equality we used Proposition \ref{prop16} in Appendix \ref{appendix_scalar_channel1} which says that $\psi_{P_0}$ is $\frac{\rho}{2}$-Lipschitz and, similarly to Proposition \ref{prop:psi_convex_reg}, it is not difficult to verify that $(q_1,q_2) \mapsto \Psi_{P_{\rm out}}(q_1;q_2)$ is $\cC^1$ on the compact set $\{ (q_1,q_2)  \, | \, 0 \leq q_1 \leq q_2 \leq \rho+1\}$ and is thus Lipschitz. 
We emphasize a crucial property of the interpolating model: It is such that at $t=0$ we recover the original model and thus $f_{n,\epsilon}(0)\approx f_n-1/2$ (the trivial constant comes from the purely noisy measurements of the second channel in \eqref{2channels}), while at $t=1$ we have the two scalar inference channels and thus the associated terms $\psi_{P_0}$ and $\Psi_{P_{\rm out}}$ discussed in Sec.~\ref{sec:scalar_inf} appear in $f_{n,\epsilon}(1)$. These are precisely the terms appearing in the potential \eqref{frs}.
\subsection{Free entropy variation along the interpolation path}
From the understanding of the previous section, it is at this stage very natural to evaluate the variation of free entropy along the interpolation path, which allows to ``compare'' the original and purely scalar models thanks to the identity 
\begin{align}
f_n = f_{n,\epsilon}(0)+\frac{1}{2}+{\cal O}(s_n)=f_{n,\epsilon}(1)-\int_0^1\frac{df_{n,\epsilon}(t)}{dt} dt +\frac{1}{2}+{\cal O}(s_n)	 \,,\label{f0_f1_int}
\end{align}
where the first equality follows from \eqref{eq:f0_f1}. As discussed above, part of the potential \eqref{frs} appears in $f_{n,\epsilon}(1)$. If the interpolation is properly done, the missing terms required to obtain the potential on the r.h.s.\ of \eqref{f0_f1_int} should naturally appear. Then by choosing the optimal interpolation path thanks to the non-trivial snr dependencies in $t$ (i.e.\ by selecting the proper interpolating functions $q$ and $r$), we will be able to show the equality between the replica formula and the free entropy $\lim_{n \to \infty} f_{n}$.

We thus now compute the $t$-derivative of the free entropy along the interpolation path (see Appendix~\ref{appendix_interpolation} for the proof). Let $u'_y(x)$ be the derivative (w.r.t.\ $x$) of $u_y(x)$. Then we have the following.

\begin{proposition}[Free entropy variation] \label{prop:der_f_t}
	The derivative of the free entropy \eqref{ft} verifies, for all $\epsilon \in \mathcal{B}_n$ and all $t \in (0,1)$
	\begin{align}\label{eq:der_f_t}
		\frac{df_{n,\epsilon}(t)}{dt}\! &= \!- \frac{1}{2} 
		\E \Big \langle 
			\Big(
				\frac{1}{n} \sum_{\mu=1}^{m}u'_{Y_{t,\mu}}(S_{t,\mu}) u'_{Y_{t,\mu}}(s_{t,\mu})
				- r_{\epsilon}(t)
			\Big)
			\big(
				Q - q_{\epsilon}(t)
			\big)
		\Big\rangle_{n,t,\epsilon}\!
		+ \!\frac{r_{\epsilon}(t)}{2}(q_{\epsilon}(t) -\rho)  \!+\! \smallO_n(1)\,,
	\end{align}
	where $\smallO_n(1)$ is a quantity that goes to $0$ in the $n,m \to \infty$ limit, uniformly in $t \in (0,1)$, $\epsilon \in \mathcal{B}_n$ and uniformly in the choice of the functions $q_{\epsilon}$ and $r_{\epsilon}$.
	The {\it overlap} is 
	\begin{align}\label{eq:def_overlap_Q}
		Q_n=Q \defeq \frac1n\bX^*\cdot \bx=\frac1n \sum_{i=1}^n X_i^* x_i
	\end{align}
	where $\bx$ is a sample from the posterior of model \eqref{2channels} associated with the Gibbs bracket $\langle -\rangle_{n,t,\epsilon}$, see \eqref{gibbs}.
\end{proposition}
\subsection{Overlap concentration and fundamental sum rule}\label{sec:overlap_concentration}
The next lemma plays a key role in our proof. Essentially it states that the overlap concentrates around its mean, a behavior called ``replica symmetric'' in statistical physics.
Similar results have been obtained in the context of the analysis of spin glasses \cite{talagrand2010meanfield1,coja2017information}. Here we use a formulation taylored to Bayesian 
inference problems as developed in the context of LDPC codes, random linear estimation \cite{BarbierMDK17} and Nishimori symmetric spin glasses \cite{Macris2007,korada2010,korada2009exact}. 

\begin{proposition}[Overlap concentration] \label{concentration}
	Assume that the interpolation functions $(q_\epsilon)$, $(r_{\epsilon})$ are regular, see Definition \ref{def:reg}. Let $s_n = \frac{1}{2}n^{-1/16}$ for all $n\ge 1$. Under assumptions~\ref{hyp:bounded},~\ref{hyp:c2} and \ref{hyp:phi_gauss2} there exists a constant $C(\varphi, S, \alpha)$ that depends only on $S$, $\varphi$ and $\alpha$ such that
	\begin{align}
		\frac{1}{s_n^2}\int_{{\cal B}_n} d\epsilon\int_0^1dt\, \E\big\langle \big(Q - \E\langle Q\rangle_{n, t, \epsilon}\big)^2\big \rangle_{n, t, \epsilon}  \le \frac{C(\varphi, S, \alpha)}{n^{1/8}}\,.
	\end{align}
\end{proposition}

%
Proposition \ref{concentration} follows from Proposition \ref{L-concentration} proved in Appendix~\ref{appendix-overlap}, combined with \eqref{boundFLuctLQ} and Fubini's theorem. Note from \eqref{eq:f0_f1} and \eqref{frs} that the second term appearing in \eqref{eq:der_f_t} is precisely the missing one that is required in order to obtain the expression of the potential on the r.h.s.\ of \eqref{f0_f1_int}. Thus in order to prove Theorem~\ref{th:RS_1layer} we would like to ``cancel'' the Gibbs bracket in \eqref{eq:der_f_t}, which is the so called {\it remainder} (once integrated over $t$). This is made possible thanks to the adaptive interpolating functions. 

One possible way to cancel the remainder is to choose $q_\epsilon(t) = \E \left\langle Q \right\rangle_{n, t, \epsilon}$, which is approximately equal to $Q$ because it concentrates by Proposition~\ref{concentration}. However, $\E \left\langle Q \right\rangle_{n, t, \epsilon}$ depends on $\int_0^t q_\epsilon(v)dv$ (and on $t$, $\int_0^t r_\epsilon(v)dv$ and $\epsilon$ too). The equation $q_\epsilon(t) = \E \left\langle Q \right\rangle_{n, t, \epsilon}$ is therefore a first order differential equation over $t \mapsto \int_0^t q_\epsilon(v)dv$. We will see in details in Sec. \ref{subsec:lower-upper} that it possesses a solution, but for the moment we just assume it exists in order to derive the following {\it fundamental sum rule}, which is a core identity in the proof scheme:



\begin{proposition}[Fundamental sum rule]\label{prop:cancel_remainder}
	Assume that the interpolation functions $(q_\epsilon)$ and $(r_{\epsilon})$ are regular (see Definition \ref{def:reg}).
	Assume that for all $t \in [0,1]$ and $\epsilon \in \mathcal{B}_n$ we have $q_{\epsilon}(t) = \E \langle Q \rangle_{n,t,\epsilon}$. Then
	\begin{align}\label{64}
		f_n &=
		\frac{1}{s_n^2} \int_{\mathcal{B}_n} \Big\{{\textstyle
		\psi_{P_0}\big( \int_0^1 r_{\epsilon}(t)dt\big) + \alpha \Psi_{P_{\rm out}}\big(\int_0^1 q_{\epsilon}(t) dt;\rho\big) - \frac{1}{2}\int_0^1 q_{\epsilon}(t)r_{\epsilon}(t)   dt }\Big\} d\epsilon + \smallO_n(1)\,,
	\end{align}
	where $\smallO_n(1)$ denotes a quantity that goes to $0$ as $n \to \infty$ uniformly w.r.t. the choice of the interpolation functions.
\end{proposition}
\begin{proof}
	By the Cauchy-Schwarz inequality
	\begin{align*}
		&\Big(\frac{1}{s_n^2}
		\int_{{\cal B}_n} d\epsilon\int_0^1 dt\, \E \Big\langle 
			\Big(
				\frac{1}{n} \sum_{\mu=1}^{m}u'_{Y_{t,\mu}}(S_{t,\mu}) u'_{Y_{t,\mu}}(s_{t,\mu})
				- r_{\epsilon}(t)
			\Big)
			\big(
				Q - q_{\epsilon}(t)
			\big)
		\Big\rangle_{n, t, \epsilon}
	\Big)^2
		\nn
		\leq 
		\frac{1}{s_n^2}\int_{{\cal B}_n} &d\epsilon\int_0^1 dt\, 
			\E \Big\langle 
				\Big(
					\frac{1}{n} \sum_{\mu=1}^{m}u'_{Y_{t,\mu}}(S_{t,\mu}) u'_{Y_{t,\mu}}(s_{t,\mu})
					- r_{\epsilon}(t)
				\Big)^2
		\Big\rangle_{n, t, \epsilon}\times\frac{1}{s_n^2}\int_{{\cal B}_n}d\epsilon\int_0^1dt\, \E \big\langle 
				\big(
					Q - q_{\epsilon}(t)
				\big)^2
			\big\rangle_{n, t, \epsilon}\,.\nonumber
	\end{align*}
	The first term of this product is bounded by some constant $C(\varphi,\alpha)$ that only depend on $\varphi$ and $\alpha$, see Appendix~\ref{appendix-boundedness-uterms}.
	The second term is bounded by $C(\varphi,S,\alpha) n^{-1/8}$ by Proposition \ref{concentration}, since we assumed that for all $\epsilon \in \mathcal{B}_n$ and all $t \in [0,1]$ we have $q_{\epsilon}(t) = \E \langle Q \rangle_{n,t,\epsilon}$.
	We have therefore
	$$
		\Big|\frac{1}{s_n^2}
		\int_{{\cal B}_n} d\epsilon\int_0^1 dt\, \E \Big\langle 
			\Big(
				\frac{1}{n} \sum_{\mu=1}^{m}u'_{Y_{t,\mu}}(S_{t,\mu}) u'_{Y_{t,\mu}}(s_{t,\mu})
				- r_{\epsilon}(t)
			\Big)
			\big(
				Q - q_{\epsilon}(t)
			\big)
		\Big\rangle_{n, t, \epsilon}
		\Big| \leq \frac{C(\varphi,S,\alpha)}{n^{1/16}}\,.
	$$
	Therefore from \eqref{eq:der_f_t} 
	\begin{align}
		\frac{1}{s_n^2}\int_{{\cal B}_n} d\epsilon \int_0^1dt \frac{df_{n,\epsilon}(t)}{dt}  = 
		\frac{1}{2s_n^2}\int_{{\cal B}_n}d\epsilon \int_0^1 dt
		\big\{q_{\epsilon}(t)r_{\epsilon}(t) 
		- r_{\epsilon}(t)\rho \big\}
		+ \smallO_n(1)+{\cal O}(n^{-1/16})
		\,.
		\label{eq:id_fluctuation}
	\end{align}
	Here the small terms are going to $0$ both uniformly w.r.t. to the choice of $q_{\epsilon}$ and $r_{\epsilon}$.
	When replacing \eqref{eq:id_fluctuation} in \eqref{f0_f1_int} and combining it with \eqref{eq:f0_f1} we reach
	%
	the claimed identity \eqref{64}, but up to the fact that $\Psi_{P_{\rm out}}(\int_0^1 q_{\epsilon}(t) dt;\rho)$ is multiplied by $m/n$ instead of $\alpha$.
	Recalling that $m/n \to \alpha$ as $m,n \to \infty$  allows to finish the argument (notice that $\Psi_{P_{\rm out}}$ is continuous and hence bounded on $[0,\rho]$, see Proposition \ref{prop:psi_convex_reg}).
\end{proof}

We are now ready to prove matching bounds.

\subsection{Lower and upper matching bounds}\label{subsec:lower-upper}
We now possess all the necessary tools to prove Theorem~\ref{th:RS_1layer} in three steps. 
\begin{enumerate}[label=(\roman*)]
	\item\label{item:step1} We prove that, under assumptions~\ref{hyp:bounded},~\ref{hyp:c2} and \ref{hyp:phi_gauss2}, $\lim_{n\to\infty}f_n=\sup_{r\ge0}\inf_{q\in[0,\rho]}f_{\rm RS}(q,r)$. 
	\item Under hypothesis~\ref{hyp:c2}, the function $\Psi_{P_{\rm out}}$ is convex, Lipschitz and non-decreasing (Proposition~\ref{prop:psi_convex}). We thus apply Corollary~\ref{Cor:supinf_supinf} of Appendix~\ref{appendix_sup_inf} to get $\sup_{r\ge0}\inf_{q\in[0,\rho]}f_{\rm RS}(q,r) = \sup_{q\in[0,\rho]}\inf_{r\ge0}f_{\rm RS}(q,r)$. We then deduce from~\ref{item:step1} that $\lim_{n\to\infty}f_n=\sup_{q\in[0,\rho]}\inf_{r\ge0}f_{\rm RS}(q,r)$ under~\ref{hyp:bounded}-\ref{hyp:c2}-\ref{hyp:phi_gauss2}.
	\item Finally, the approximation arguments given in Appendix~\ref{Appendix-approx} permit to relax~\ref{hyp:bounded}-\ref{hyp:c2} to the weaker hypotheses~\ref{hyp:third_moment}-\ref{hyp:um} and allow to replace the Gaussian assumption \ref{hyp:phi_gauss2} on $\bbf{\Phi}$ by \ref{hyp:third_moment}-\ref{hyp:phi_general}-\ref{hyp:cont_pp}. The fact that for discrete channels the Gaussian noise can then be removed, allowing to replace \ref{hyp:delta_pos} (i.e. $\Delta >0$ treated until here) to \ref{hyp:delta_0} (i.e. $\Delta = 0$ and $\varphi$ takes values in $\N$), is proven in Sec. \ref{sec:relax_discrete}. This proves the first equality of Theorem~\ref{th:RS_1layer}. The last equality in \eqref{eq:rs_formula} and the remaining part of Theorem~\ref{th:RS_1layer} follow then from Lemma~\ref{lem:sup_inf_2}.
	\end{enumerate}


	It thus remains to tackle~\ref{item:step1}, but before that we need a definition. For $t\in[0,1]$ and $\epsilon \in \mathcal{B}_n$, we write $R^t(\epsilon) = (R_1(t,\epsilon),R_2(t,\epsilon))$.
	The quantity $\E \langle Q \rangle_{n,t,\epsilon}$ is a function of $n,t,R^t(\epsilon)$ that we write $\E \langle Q \rangle_{n,t,\epsilon} = F_n\big(t,R^t(\epsilon)\big)$, where
	$F_n$ is a function defined on 
	\begin{equation}
		D_n \defeq \Big\{ (t,r_1,r_2) \in [0,1] \times \R_+ \times \R_+ \, \Big| \, r_2 \leq \rho t + 2s_n\Big\}\,.
	\end{equation}
	The following proposition, proven in Appendix \ref{appendix-proofPropEqDiff}, will be useful.
\begin{proposition}\label{prop:F_equadiff}
$F_n$ is a continuous function from $D_n$ to $[0,\rho]$. 
	Let $D_n^{\circ}$ denotes the interior of $D_n$. $F_n$ admits partial derivatives with respect to its second and third argument on $D_n^{\circ}$. These partial derivatives are both continuous and non-negative on $D_n^{\circ}$.
	\end{proposition}

	Let us now start with the lower bound.

	\subsubsection{Lower bound} 

\begin{proposition}[Lower bound] \label{prop:lower_bound} The free entropy \eqref{f} verifies 
	\begin{align}
		\liminf_{n \to \infty} f_n \geq {\adjustlimits\sup_{r\geq 0} \inf_{q \in [0,\rho]}} f_{\rm RS} (q,r)\,. \label{}
	\end{align}
\end{proposition}
\begin{proof}	
	We consider, for $(\epsilon_1, \epsilon_2) \in \mathcal{B}_n$ and a fixed value $r \in [0, r_{\rm max}]$, the following 1st order differential equation:
	\begin{equation}\label{eq:equadiff1}
		y(0) = (\epsilon_1, \epsilon_2) \qquad \text{and} \qquad \forall \ t \in [0,1], \quad y'(t) = \big(r, F_n(t,y(t))\big)\,.
	\end{equation}
	By the Cauchy-Lipschitz Theorem (see for instance Theorem~3.1 in Chapter~V from \cite{hartmanordinary}) this equation admits a (unique) solution that we write $y(\cdot,\epsilon) = \big(y_1(\cdot, \epsilon), y_2(\cdot, \epsilon)\big)$. The hypotheses of the Cauchy-Lipschitz Theorem are verified, because of Proposition \ref{prop:F_equadiff}.
	We define then, for all $t \in [0,1]$,
	$$
	r_{\epsilon}(t) = y_1'(t,\epsilon) = r
	\qquad \text{and} \qquad
	q_{\epsilon}(t) = y_2'(t,\epsilon) = F_n(t,y(t,\epsilon)) \in [0,\rho]\,.
	$$
	We have therefore $R_1(t,\epsilon) = \epsilon_1 + \int_0^t y_1'(s,\epsilon) ds = y_1(t,\epsilon)$ and similarly $R_2(t,\epsilon) = y_2(t,\epsilon)$. We obtain that for all $t\in [0,1]$,
	$$
	q_{\epsilon}(t) 
	= F_n(t,y(t,\epsilon))
	= F_n\big(t,(R_1(t,\epsilon),R_2(t,\epsilon))\big)
	= \E \langle Q \rangle_{n,t,\epsilon}\,.
	$$
	Let us show now that the functions $(q_{\epsilon})$ and $(r_{\epsilon})$ are regular (see Definition \ref{def:reg}). Let $t \in [0,1]$. The function $R^t:\epsilon \mapsto (R_1(t,\epsilon),R_2(t,\epsilon)) = y(t,\epsilon)$ is the flow of \eqref{eq:equadiff1} and is thus injective (by unicity of the solution) and $\cC^1$ because of the regularity properties (see Proposition \ref{prop:F_equadiff}) of $F_n$. The Jacobian of the flow is given by the Liouville formula (see Corollary~3.1 in Chapter~V from \cite{hartmanordinary}):
	$$
	{\rm det}\Big(\frac{\partial R^t}{\partial \epsilon}(\epsilon)\Big)
	= 
	\exp\Big(\int_0^t dv
		\frac{\partial F_n}{\partial y_2}(v,y(v,\epsilon))
	\Big)
	\geq 1,
	$$
	because by Proposition \ref{prop:F_equadiff} we have $\partial_{y_2} F_n \geq 0$.
	We obtain (by the local inversion Theorem) that $R^t$ is a $\cC^1$ diffeomorphism, and since its Jacobian is greater or equal to $1$ the functions $(q_{\epsilon})$ and $(r_{\epsilon})$ are regular.

	We have seen that for all $\epsilon \in \mathcal{B}_n$ and all $t \in [0,1]$, $q_{\epsilon}(t) = \E \langle Q \rangle_{t,n,\epsilon}$, so we can apply Proposition~\ref{prop:cancel_remainder} to get
	\begin{align*}
		f_n
		 &= 
		\frac{1}{s_n^2} \int_{\mathcal{B}_n} \Big\{{\textstyle
		\psi_{P_0}(r) + \alpha \Psi_{P_{\rm out}}\big(\int_0^1 q_{\epsilon}(t) dt;\rho\big) - \frac{r}{2}\int_0^1 q_{\epsilon}(t) dt }\Big\} d\epsilon + \smallO_n(1)
		\\
		 &=
\frac{1}{s_n^2} \int_{\mathcal{B}_n} f_{\rm RS}\big({\textstyle \int_0^1 q_{\epsilon}(t) dt},r\big) d\epsilon + \smallO_n(1)
\\
		 & \geq \inf_{q \in [0,\rho]} f_{\rm RS}(q,r)+ \smallO_n(1)
	\end{align*}
	and thus $\liminf_{n \to \infty} f_n \geq \inf_{q \in [0,\rho]} f_{\rm RS}(q,r)$.
	This is true for all $r \in [0,r_{\rm max}]$ so we get 
	\begin{equation}\label{eq:liminf_fn}
	\liminf_{n \to \infty} f_n \geq {\adjustlimits \sup_{r \in [0,r_{\rm max}]}\inf_{q \in [0,\rho]}} f_{\rm RS}(q,r)\,.
	\end{equation}
	Let $r \geq r_{\rm max}$. We have for all $q \in [0,\rho]$, $\partial_q f_{\rm RS}(q,r) = \alpha \Psi_{P_{\rm out}}'(q) - \frac{r}{2} \leq 0$,
	because $r \geq r_{\rm max} \geq 2 \alpha \Psi_{P_{\rm out}}'(q)$. Therefore for all $r \geq r_{\rm max}$, $\,\inf_{q \in [0,\rho]} f_{\rm RS}(q,r) = f_{\rm RS}(\rho,r)$ and
	$$
	\frac{\partial}{\partial r} \inf_{q \in [0,\rho]} f_{\rm RS}(q, r)
	=
	\frac{\partial}{\partial r} f_{\rm RS}(\rho,r) = \psi_{P_0}'(r) - \frac{\rho}{2} \leq 0,
	$$
	because by Proposition \ref{prop16}, $\psi_{P_0}$ is $\frac{\rho}{2}$-Lipschitz. The function $r \mapsto \inf_{q \in [0,\rho]} f_{\rm RS}(q,r)$ is therefore non-increasing on $[r_{\rm max},+\infty)$. Going back to \eqref{eq:liminf_fn}, we conclude
	\begin{equation}
		\liminf_{n \to \infty} f_n 
		\geq {\adjustlimits \sup_{r \in [0,r_{\rm max}]}\inf_{q \in [0,\rho]}} f_{\rm RS}(q,r)
		= {\adjustlimits\sup_{r \geq 0}\inf_{q \in [0,\rho]}} f_{\rm RS}(q,r)\,.
	\end{equation}
\end{proof}

\subsubsection{Upper bound}

\begin{proposition}[Upper bound]\label{prop:upper_bound} The free entropy \eqref{f} verifies
	\begin{align}
		\limsup_{n \to \infty} f_n \leq {\adjustlimits \sup_{r\geq 0} \inf_{q \in [0,\rho]}} f_{\rm RS} (q,r)\,.
	\end{align}		
\end{proposition}
\begin{proof}
	We consider, for $(\epsilon_1, \epsilon_2) \in \mathcal{B}_n$, the following order-1 system of differential equations:
	\begin{equation}\label{eq:equadiff2}
		y(0) = (\epsilon_1, \epsilon_2) \qquad \text{and} \qquad \forall\ t \in [0,1], \quad y'(t) = 
		\Big(
		\begin{array}{c}
			2\alpha \Psi'_{P_{\rm out}}\big(F_n(t,y(t))\big) \\
			F_n(t,y(t))
		\end{array}
	\Big)\,.
	\end{equation}
	By Proposition \ref{prop:psi_convex_reg} the function $\Psi'_{P_{\rm out}}$ is $\cC^1$ and takes values in $[0,r_{\rm max}]$. By Proposition \ref{prop:F_equadiff}, the function $F_n$ is continuous, bounded and admits partial derivatives w.r.t. its second and third arguments, that are continuous.
	We can therefore apply the Cauchy-Lipschitz Theorem as in the proof of Proposition \ref{prop:lower_bound}:
	The equation \eqref{eq:equadiff2} admits a (unique) solution that we write $y(\cdot,\epsilon) = \big(y_1(\cdot, \epsilon), y_2(\cdot, \epsilon)\big)$.
	We define then, for all $t \in [0,1]$,
	$$
	r_{\epsilon}(t) = y_1'(t,\epsilon) = 2 \alpha \Psi'_{P_{\rm out}}\big(F_n(t,y(t,\epsilon))\big) \in [0,r_{\rm max}]
	\qquad \text{and} \qquad
	q_{\epsilon}(t) = y_2'(t,\epsilon) = F_n(t,y(t,\epsilon)) \in [0,\rho]\,.
	$$
	We have therefore $R_1(t,\epsilon) = \epsilon_1 + \int_0^t y_1'(s,\epsilon) ds = y_1(t,\epsilon)$ and similarly $R_2(t,\epsilon) = y_2(t,\epsilon)$. We obtain that for all $t\in [0,1]$,
	$$
	q_{\epsilon}(t) 
	= F_n(t,y(t,\epsilon))
	= F_n\big(t,(R_1(t,\epsilon),R_2(t,\epsilon))\big)
	= \E \langle Q \rangle_{n,t,\epsilon}\,.
	$$
	Let us show now that the functions $(q_{\epsilon})$ and $(r_{\epsilon})$ are regular (see Definition \ref{def:reg}). Let $t \in [0,1]$. The function $R^t:\epsilon \mapsto (R_1(t,\epsilon),R_2(t,\epsilon)) = y(t,\epsilon)$ is the flow of \eqref{eq:equadiff2} and is thus injective and $\cC^1$ because of the regularity properties (see Proposition \ref{prop:F_equadiff}) of $F_n$. The Jacobian of the flow is again given by the Liouville formula:
	$$
	{\rm det}\Big(\frac{\partial R^t}{\partial \epsilon}(\epsilon)\Big)
	= 
	\exp\Big(\int_0^t dv
		2\alpha\frac{\partial F_n}{\partial y_1}(v,y(v,\epsilon))
\Psi''_{P_{\rm out}}\big(F_n(v,y(v,\epsilon))\big)
		+
		\int_0^t dv
		\frac{\partial F_n}{\partial y_2}(v,y(v,\epsilon))
	\Big)
	\geq 1,
	$$
	because by Proposition \ref{prop:F_equadiff}, $\frac{\partial F_n}{\partial y_1}$ and $\frac{\partial F_n}{\partial y_2}$ are both non negative and since $\Psi_{P_{\rm out}}$ is convex (see Proposition \ref{prop:psi_convex_reg}), we have also $\Psi_{P_{\rm out}}'' \geq 0$.
	We obtain (by the local inversion Theorem) that $R^t$ is a $\cC^1$ diffeomorphism. Its Jacobian is greater or equal to $1$, and the functions $(q_{\epsilon})$ and $(r_{\epsilon})$ are therefore regular.

	We have seen that for all $\epsilon \in \mathcal{B}_n$ and all $t \in [0,1]$, $q_{\epsilon}(t) = \E \langle Q \rangle_{t,n,\epsilon}$, so we can apply Proposition~\ref{prop:cancel_remainder} to get
	\begin{align}
		f_n
		 &= 
		\frac{1}{s_n^2} \int_{\mathcal{B}_n} \Big\{{\textstyle
		\psi_{P_0}\big( \int_0^1 r_{\epsilon}(t)dt\big) + \alpha \Psi_{P_{\rm out}}\big(\int_0^1 q_{\epsilon}(t) dt;\rho\big) - \frac{1}{2}\int_0^1 q_{\epsilon}(t)r_{\epsilon}(t)   dt }\Big\} d\epsilon + \smallO_n(1)
		\nn
		 &\leq
		\frac{1}{s_n^2} \int_{\mathcal{B}_n} \int_0^1 \Big\{
		\psi_{P_0}(  r_{\epsilon}(t) ) + \alpha \Psi_{P_{\rm out}}(q_{\epsilon}(t) ;\rho) - \frac{1}{2} q_{\epsilon}(t)r_{\epsilon}(t)\Big\} dt d\epsilon + \smallO_n(1)
		\label{eq:upper_jensen}
	\end{align}
	by Jensen's inequality, because by Propositions \ref{prop16} and \ref{prop:psi_convex_reg} the functions $\psi_{P_0}$ and $\Psi_{P_{\rm out}}$ are convex.

	Let us fix $\epsilon \in \mathcal{B}_n$ and $t \in [0,1]$.
	By definition of $r_{\epsilon}$ and $q_{\epsilon}$, we have
	\begin{equation}\label{eq:fix_point}
	r_{\epsilon}(t) 
	= 2\alpha \Psi_{P_{\rm out}}'\big(F_n(t,y(t,\epsilon))\big)
	= 2 \alpha \Psi_{P_{\rm out}}'\big(q_{\epsilon}(t)\big).
	\end{equation}
	The function $g : q \in [0,\rho] \mapsto 2 \alpha \Psi_{P_{\rm out}}(q;\rho) -  r_{\epsilon}(t) q$ is convex by Proposition \ref{prop:psi_convex_reg}. By equation \eqref{eq:fix_point} above, we see that $g'(q_{\epsilon}(t)) = 0$ and therefore:
	$$
\alpha \Psi_{P_{\rm out}}(q_{\epsilon}(t) ;\rho) - \frac{1}{2} q_{\epsilon}(t)r_{\epsilon}(t)
=
\inf_{q \in [0,\rho]} \big\{
\alpha \Psi_{P_{\rm out}}(q;\rho) - \frac{1}{2} q\, r_{\epsilon}(t)
\big\}.
	$$
	This holds for all $\epsilon \in \mathcal{B}_n$ and all $t \in [0,1]$. Plugging this back in \eqref{eq:upper_jensen}, we get:
	\begin{align*}
		f_n
		 &\leq
		 \frac{1}{s_n^2} \int_{\mathcal{B}_n} \int_0^1 \inf_{q \in [0,\rho]} \Big\{
		\psi_{P_0}(  r_{\epsilon}(t) ) + \alpha \Psi_{P_{\rm out}}(q ;\rho) - \frac{1}{2} q\, r_{\epsilon}(t)\Big\} dt d\epsilon + \smallO_n(1)
		\\
		 &\leq
		 {\adjustlimits \sup_{r \geq 0} \inf_{q \in [0,\rho]}} \Big\{
		 \psi_{P_0}(  r ) + \alpha \Psi_{P_{\rm out}}(q ;\rho) - \frac{1}{2} q \,r \Big\} + \smallO_n(1)\,.
	\end{align*}
	This proves Proposition \ref{prop:upper_bound}.
\end{proof}

From the arguments given at the beginning of the section, this ends the proof of Theorem~\ref{th:RS_1layer}.

\section{Proofs of the limits of optimal errors} \label{app_perc}
	
\subsection{Optimal generalization error: Proof of Theorem~\ref{th:gen_f}}\label{appendix:i_mmse_gen}

\subsubsection{Formal derivation and proof idea: A teacher-student scenario with side information} 
Before proving Theorem~\ref{th:gen_f} rigorously, we find useful to provide a conceptual framework allowing to formally derive the generalization error (a framework that will actually serve as a basis for the rigorous derivation presented in the next section). In order to obtain the (generalized) optimal generalization error, we need first to assume that the new ``test labels'' are also observed by the student in the teacher-student scenario of Sec.~\ref{sec:teacherStudent} but with a very low signal-to-noise ratio. The presence of this side information will allow us to use the I-MMSE relation (Proposition \ref{prop:immse}) to obtain the generalization error when small, but non-zero, information about the test labels is known by the student. Then, by formally taking the limit of vanishing side information on the resulting expression (and {\it assuming} that the large $n$ and vanishing side information limits commute), we will recover the generalization error. We thus now introduce the following ``train-test'' observation model.

The set of patterns and labels are 
divided into two sets by the teacher: The {\it training set} ${\cal S}^{\rm tr}$ of 
size $m$ that will be used as the main source of information by the student in order to then generalize, and the {\it test set} ${\cal S}^{\rm te}$ of size $m'=\epsilon n$ that will be 
used by the teacher in order to evaluate the performance of the student, but also by the student as small additional side information. Let us be more precise: The teacher gives to the student both the patterns and associated labels of the training set, namely ${\cal S}^{\rm tr}\defeq\{(Y_\mu; \boldsymbol{\Phi}_\mu)\}_{\mu = 1}^{m}$ (recall the labels are given by \eqref{measurements}, \eqref{eq:channel}). 
For the test set, the test patterns to classify are given to the student but 
the associated labels are (almost) not: Let $\epsilon,\lambda \geq 0$. Instead of the test labels $\{\widetilde{Y}_\mu\}_{\mu=1}^{m'}$ (that should be totally unknown to the student in the ideal setting), what is given to the student is 
\begin{align}\label{noiseU}
U_{\mu} =  \sqrt{\lambda}\, Y'_{\mu} +  Z'_{\mu} \,, \qquad \text{for} \qquad  1 \leq \mu \leq m'=\epsilon n \,,	
\end{align}
where $Z'_{\mu} \iid \cN(0,1)$, and $Y'_{\mu}$ is given by
\begin{align}\label{Y'f}
Y'_{\mu} = f(\widetilde{Y}_{\mu})\,, \qquad
\widetilde{Y}_{\mu} \sim P_{\rm out}\Big( \cdot \, \Big| \, \frac{\bbf{\Phi}'_{\mu} \cdot \bX^*}{\sqrt{n}}  \Big)\,,	
\end{align}
where $f: \R \to \R$ is a continuous bounded function and $\bbf{\Phi}'_{\mu} \iid \cN(0,\bbf{I}_n)$ independently of everything else. 
We will first prove Theorem \ref{th:gen_f} for continuous bounded functions $f$, and then relax this at the end of the proof.
The test set given to the student, in addition of the training set, is ${\cal S}^{\rm te}={\cal S}^{\rm te}(\lambda,\epsilon)\defeq\{(U_\mu=\sqrt{\lambda}\,Y'_\mu + Z'_\mu; \boldsymbol{\Phi}'_\mu)\}_{\mu = 1}^{\epsilon n}$ where {\it $\lambda$ is typically very small}. Indeed, we are particularly interested in the case $\lambda, \epsilon \to 0$ when the student has {\it no} information about the test labels, which is the ideal setting we want to study. But in order to employ the I-MMSE relation we consider instead very small $\lambda >0$.

The learning of the classifier $\bX^*$ given ${\cal S}^{\rm tr}$ {\it and} ${\cal S}^{\rm te}$ is a slight extension of model \eqref{measurements}.
%
%
%
Define $\bY' =(Y_\mu')_{\mu= 1}^{m'}$ as the vector of test labels (before they are corrupted by additional noise through \eqref{noiseU}). 
Then the (generalized) optimal generalization error {\it with side information} (i.e. at $\lambda,\epsilon > 0$) in this ``train-test'' observation model is
\begin{align}
	{\cal E}_{f,n}^{\rm side}(\lambda,\epsilon) &\defeq
	 \min_{\widehat{\bY}'}\frac{1}{\epsilon n}\EE\big[\big\|\bY' - \widehat{\bY}'({\cal S}^{\rm te},{\cal S}^{\rm tr})\big\|^2\big]= \frac{1}{\epsilon n}\EE\big[\big\|\bY' - \E\big[\bY'\big|{\cal S}^{\rm te},{\cal S}^{\rm tr}\big]\big\|^2\big]\,.\label{defGenerror_side}
\end{align}
%
The ``true'' generalization error \eqref{eq:def_e_gen_f} is recovered by defining instead ${\cal S}^{\rm te}=\bbf{\Phi}'$ or equivalently letting $\lambda,\epsilon \to 0$, i.e. when only ${\cal S}^{\rm tr}$ and the test patterns are given to the student: $\lim_{\lambda,\epsilon\to 0} {\cal E}_{f,n}^{\rm side}(\lambda,\epsilon)={\cal E}_{f,n}$. Note that the $f$ function only plays a role in the test set, while the labels of the training data are generated through the ``pure model'' \eqref{measurements}, \eqref{eq:channel}.

From there one can use the I-MMSE relation of Proposition \ref{prop:immse} in order to formally compute the limiting $n\to\infty$ expression of \eqref{defGenerror_side}. Indeed, 
\begin{align}\label{96}
	\frac{\partial}{\partial \lambda}	\frac{1}{n} I(\bY'; \sqrt{\lambda}\bY' + \bZ'| \bY,\bbf{\Phi},\bbf{\Phi}')= \frac\epsilon 2 {\cal E}_{f,n}^{\rm side}(\lambda,\epsilon)\,.
\end{align}
Fortunately, by a straightforward extension of the interpolation 
method presented in Sec.~\ref{sec:interpolation} one can generalize Theorem~\ref{th:RS_1layer} to take into account this additional side information and access this mutual information (see the end of the section for the proof):
\begin{lemma}\label{lemma5.1} 
	For all $\epsilon,\lambda \geq 0$ we have
	\begin{equation}\label{eq:mutual_side}
	\frac{1}{n} I(\bY'; \sqrt{\lambda}\bY' + \bZ'| \bY,\bbf{\Phi},\bbf{\Phi}')
	\xrightarrow[n \to \infty]{}
	{\adjustlimits \inf_{q \in [0,\rho]} \sup_{r \geq 0}}\, \tilde{i}_{\rm RS}(q,r,\lambda)
	- i_{\infty} \,,
	\end{equation}
	where $i_\infty$ is given by Corollary \ref{cor:mi} and
	\begin{align}
\tilde{i}_{\rm RS}(q,r,\lambda)
&\defeq i_{\rm RS}(q,r) + \epsilon I(f(Y^{(q)});\sqrt{\lambda} \,f(Y^{(q)}) + Z'| V)\label{itilde_irs}\\
&=I_{P_0}(r) + \alpha \mathcal{I}_{P_{\rm out}}(q;\rho) + \epsilon I(f(Y^{(q)});\sqrt{\lambda}\, f(Y^{(q)}) + Z'| V) - \frac{r}{2}(\rho -q)\,.
	\end{align}
	Recall that $Y^{(q)}$ is sampled from the ``second scalar channel'' \eqref{eq:Pout_scalar_channel}: $Y^{(q)} \sim P_{\rm out}(\cdot \,|\, \sqrt{q}\, V + \sqrt{\rho - q}\,W^*)$, where $V,W^* \iid \cN(0,1)$. 
\end{lemma}

Define the following MMSE function:
\begin{equation} \label{eq:def_M}
M_f:(\lambda,q) \mapsto \MMSE\big(f(Y^{(q)})\big|  \sqrt{\lambda}\, f(Y^{(q)}) + Z' , V\big) \,.
\end{equation}
By concavity arguments detailed in the next section, we have almost everywhere $$\lim_{n\to\infty}\frac{\partial}{\partial \lambda} \frac 1n I(\bY'; \sqrt{\lambda}\bY' + \bZ'| \bY,\bbf{\Phi},\bbf{\Phi}')=\frac{\partial}{\partial \lambda}\lim_{n\to\infty}\frac1n I(\bY'; \sqrt{\lambda}\bY' + \bZ'| \bY,\bbf{\Phi},\bbf{\Phi}')=\frac{\partial}{\partial \lambda} {\adjustlimits \inf_{q \in [0,\rho]} \sup_{r \geq 0}} \,\tilde{i}_{\rm RS}(q,r,\lambda)$$ using Lemma \ref{lemma5.1} for the last equality. Assuming $\partial_\lambda\inf_{q \in [0,\rho]}\sup_{r\ge 0} \tilde{i}_{\rm RS}(q,r,\lambda)= \partial_\lambda\tilde{i}_{\rm RS}(q,r,\lambda)|_{(q^*_{\lambda},r^*_{\lambda})}$, where $(q^*_{\lambda},r^*_{\lambda})$ is an optimal couple,
\eqref{96} and the last identity combined lead to
\begin{align}
	\frac\epsilon2 \lim_{n\to\infty}	{\cal E}_{f,n}^{\rm side}(\lambda,\epsilon)=\epsilon \frac{\partial}{\partial \lambda}I\big(f(Y^{(q)}); \sqrt{\lambda} f(Y^{(q)}) + Z'\big|V\big)\Big|_{q^*_{\lambda}} = \frac\epsilon2 M_f(\lambda,q^*_{\lambda})\,,
\end{align}
using again the I-MMSE relation for the last equality. Thus $\lim_{n\to\infty}	{\cal E}_{f,n}^{\rm side}(\lambda,\epsilon)= M_f(\lambda,q^*_{\lambda})$.

A formal calculation of the vanishing side information limit of $M_f(\lambda,q^*_{\lambda})$ gives back ${\cal E}_f(q^*(\alpha))$ (recall \eqref{eq:def_e_f} and $q^*(\alpha)$ is the optimizer of the replica-symmetric formula~\eqref{eq:rs_formula}), so that $\lim_{\lambda,\epsilon \to 0} \lim_{n\to\infty}	{\cal E}_{f,n}^{\rm side}(\lambda,\epsilon)={\cal E}_f(q^*(\alpha))$. It is very natural to believe that the vanishing side information limit of $\lim_{n\to\infty}{\cal E}_{f,n}^{\rm side}(\lambda,\epsilon)$ should give back the true asymptotic generalization error. So if one could justify the commutation of limits $$\lim_{\lambda,\epsilon \to 0} \lim_{n\to\infty}	{\cal E}_{f,n}^{\rm side}(\lambda,\epsilon)= \lim_{n\to\infty}	\lim_{\lambda,\epsilon \to 0} {\cal E}_{f,n}^{\rm side}(\lambda,\epsilon)= \lim_{n\to\infty}	{\cal E}_{f,n}$$ this would end the proof. 
We prove this point in the next section.
%
%
%
\\

\begin{proof}[Proof of Lemma \ref{lemma5.1}:]
Extending the interpolation 
method presented in Sec.~\ref{sec:interpolation}, one can generalize Theorem~\ref{th:RS_1layer} to take into account this additional side information. This gives directly
\begin{align} \label{fn_perc}
	\frac{1}{n} I(\bX^*; \bY,\sqrt{\lambda} \bY' + \bZ' | \bbf{\Phi},\bbf{\Phi}') \xrightarrow[n \to \infty]{}
	\tilde{I}_{\infty}(\alpha,\epsilon,\lambda) \defeq 
	{\adjustlimits \inf_{q \in [0,\rho]}\sup_{r\ge 0}}\, \tilde{I}_{\rm RS}(q,r,\lambda) 
\end{align}
where $\tilde{I}_{\rm RS}(q,r,\lambda)$ is given by
\begin{align}
\tilde{I}_{\rm RS}(q,r,\lambda)
\defeq I_{P_0}(r) + \alpha \mathcal{I}_{P_{\rm out}}(q;\rho) + \epsilon I(W^*; \sqrt{\lambda} f(Y^{(q)}) + Z'|V) - \frac{r}{2}(\rho -q)\,.
\end{align}
	Conditionally on $(V,f(Y^{(q)}))$, the random variables $W^*$ and $\sqrt{\lambda} f(Y^{(q)}) + Z'$ are independent, therefore
	$$
	I\big(f(Y^{(q)}); \sqrt{\lambda}f(Y^{(q)}) + Z' \big| V\big)
	=
	I\big(W^*,f(Y^{(q)}); \sqrt{\lambda}f(Y^{(q)}) + Z' \big| V\big) \,.
	$$
	Now, by the chain rule of the mutual information we have
	$$
	I\big(W^*,f(Y^{(q)}); \sqrt{\lambda}f(Y^{(q)}) + Z' \big| V\big)
	=
	I\big(W^*; \sqrt{\lambda}f(Y^{(q)}) + Z' \big| V\big)
	+
	I\big(f(Y^{(q)}); \sqrt{\lambda}f(Y^{(q)}) + Z' \big| V,W^*\big) \,.
	$$
	We obtain that 
	\begin{equation}\label{eq:decomp_i_scalar}
	I\big(W^*; \sqrt{\lambda}f(Y^{(q)}) + Z' \big| V\big)
	= I\big(f(Y^{(q)}); \sqrt{\lambda}f(Y^{(q)}) + Z' \big| V\big) - I\big(f(Y^{(q)}); \sqrt{\lambda}f(Y^{(q)}) + Z' \big| V,W^*\big) \,.
\end{equation}
	Notice that the last mutual information in the above equation does not depend on $q$ nor $r$. Therefore we have:
	\begin{equation}\label{eq:endi}
	{\adjustlimits \inf_{q \in [0,\rho]}\sup_{r\ge 0}}\, \tilde{I}_{\rm RS}(q,r,\lambda) 
	=
	- \epsilon I\big(f(Y^{(q)}); \sqrt{\lambda}f(Y^{(q)}) + Z' \big| V,W^*\big) + 
	{\adjustlimits \inf_{q \in [0,\rho]}\sup_{r\ge 0}}\, \tilde{i}_{\rm RS}(q,r,\lambda) \,.
\end{equation}
	Now, by the chain rule, we have
	\begin{equation}\label{eq:step_chain}
	\frac{1}{n} I(\bX^*; \bY,\sqrt{\lambda} \bY' + \bZ' | \bbf{\Phi},\bbf{\Phi}')
=
\frac{1}{n} I(\bX^*; \bY| \bbf{\Phi})
+
\frac{1}{n} I(\bX^*; \sqrt{\lambda} \bY' + \bZ' | \bY, \bbf{\Phi},\bbf{\Phi}') \,.
	\end{equation}
	The limit of the left-hand side is given by \eqref{fn_perc}. By Corollary~\ref{cor:mi}, we have $\lim_{n\to\infty}I(\bX^*; \bY| \bbf{\Phi})/n = i_{\infty}$. It remains to investigate the last term of the equation above.
	By the arguments used to prove \eqref{eq:decomp_i_scalar}, we have
	\begin{align}
		I(\bX^*; \sqrt{\lambda} \bY' + \bZ' | \bY, \bbf{\Phi},\bbf{\Phi}')
		&=
I(\bY'; \sqrt{\lambda} \bY' + \bZ' | \bY, \bbf{\Phi},\bbf{\Phi}') 
-
I(\bY'; \sqrt{\lambda} \bY' + \bZ' | \bY, \bbf{\Phi},\bX^*,\bbf{\Phi}')
\nn
&=
I(\bY'; \sqrt{\lambda} \bY' + \bZ' | \bY, \bbf{\Phi},\bbf{\Phi}') 
-
I(\bY'; \sqrt{\lambda} \bY' + \bZ' |\bX^*,\bbf{\Phi}') \,.
\label{eq:decomp_i2}
	\end{align}
	We have $ I(\bY'; \sqrt{\lambda} \bY' + \bZ' |\bX^*,\bbf{\Phi}')/n = \epsilon I(Y_1'; \sqrt{\lambda} Y_1' + Z_1'| \bX^*, \bbf{\Phi}_1')$ and it is not difficult to show, using similar computations as in the proof of Corollary \ref{cor:mi}, that
	$$
I(Y_1'; \sqrt{\lambda} Y_1' + Z_1'| \bX^*, \bbf{\Phi}_1')
\xrightarrow[n \to \infty]{}
I\big(f(Y^{(q)}); \sqrt{\lambda}f(Y^{(q)}) + Z' \big| V,W^*\big) \,,
	$$
	(recall that the right-hand side does not depend on $q$).	
	Combining this with \eqref{eq:decomp_i2}, \eqref{eq:step_chain}, \eqref{fn_perc}, Corollary~\ref{cor:mi} and \eqref{eq:endi}, we obtain the desired result.
\end{proof}
\subsubsection{Proof of Theorem~\ref{th:gen_f}}

In order to compute the limit of the (generalized) generalization error, we work in the teacher-student scenario with side-information discussed in the previous section.

\begin{lemma}\label{lem:D_full_measure}
	For all $\alpha,\lambda >0$ the set
	\begin{equation}\label{eq:sets_gen}
	D_{\alpha,\lambda} \defeq	\big\{ \epsilon \geq 0 \, \big| \, 
\text{the infimum in} \ \eqref{eq:mutual_side} \ \text{is achieved at a unique} \ q^*_{\alpha,\epsilon,\lambda} 
\big\}
	\end{equation}
	is equal to $[0,+\infty)$ minus some countable set.
	Moreover, $\epsilon \mapsto q^*_{\alpha,\lambda,\epsilon}$ is continuous on $D_{\alpha,\epsilon}$.
\end{lemma}
\begin{proof}
	This follows from the same arguments than the proof of Proposition~\ref{prop:q_star}. 
\end{proof}

\begin{lemma} \label{lem:lim_mmse_side_gen}
	For all $\alpha,\lambda >0$, we have for all $\epsilon \in D_{\alpha,\lambda} \setminus \{ 0\}$
\begin{align*}
	\lim_{n \to \infty} \MMSE(Y'_1 | \bY,\bU,\bbf{\Phi},\bbf{\Phi}')
&= M_f(\lambda,q_{\alpha,\epsilon,\lambda}^*) \,,
\end{align*}
where  $q^*_{\alpha,\epsilon,\lambda}$ is the unique minimizer of \eqref{eq:mutual_side}.
\end{lemma}
\begin{proof}
	Let us fix $\alpha,\epsilon>0$.
	Consider the function
	\begin{align}
		h_{\alpha,\epsilon}: \lambda \mapsto {\adjustlimits \inf_{q \in [0,\rho]} \sup_{r \geq 0}}\, \tilde{i}_{\rm RS}(q,r,\lambda)\label{hfunc} \,.
	\end{align}
	Corollary~4 from \cite{milgrom2002envelope} gives that $h_{\alpha,\epsilon}$ is differentiable at $\lambda$ if and only if
	$$
	\Big\{ 
	\epsilon \frac{\partial}{\partial \lambda} I\big(f(Y^{(q)}); \sqrt{\lambda} f(Y^{(q)}) + Z' \big| V\big)
	= \frac{\epsilon}{2} M_f(\lambda,q)
\, \Big| \, q \ \text{minimizer of \eqref{eq:mutual_side} (or equivalently of \eqref{hfunc})}  \Big\}
	$$
	is a singleton (the equality comes from the I-MMSE relation from Proposition \ref{prop:immse}). In such case, Corollary~4 from \cite{milgrom2002envelope} also gives that
	\begin{equation}\label{eq:derhh}
		h_{\alpha,\epsilon}'(\lambda) = 
	\frac{\epsilon}{2} M_f(\lambda,q)
	\,,
\end{equation}
for all $q$ minimizer of \eqref{hfunc}.
So if now $\epsilon \in D_{\alpha,\lambda} \setminus \{0\}$, then the minimizer is unique and thus $h_{\alpha,\epsilon}$ is differentiable at $\lambda$, with derivative $h_{\alpha,\lambda}'(\lambda) = \epsilon M_f(\lambda,q^*_{\alpha,\epsilon,\lambda})/2$.
However, by \eqref{eq:mutual_side} in Lemma \ref{lemma5.1}, $h_{\alpha,\epsilon}$ is the pointwise limit on $\R_+$ of the sequence of concave functions
	$$
	(h_n)_{n\geq 1} = 
	\Big( \lambda \mapsto \frac{1}{n} I\big(\bY'; \sqrt{\lambda} \bY' + \bZ' \big|\bY,\bbf{\Phi},\bbf{\Phi}' \big) + i_{\infty} \Big)_{n \geq 1} \,.
	$$
	Consequently, a standard convex analysis result gives that $h_n'(\lambda) \xrightarrow[n \to \infty]{} h_{\alpha,\epsilon}'(\lambda)$. By the I-MMSE relation (Proposition~\ref{prop:immse}) we have $h'_n(\lambda) = \epsilon \MMSE(Y'_1 | \bY,\bU,\bbf{\Phi},\bbf{\Phi}')/2$ and we conclude using the fact that $\epsilon \neq 0$.
\end{proof}

\begin{lemma}\label{lem:lim_eps_lam}
	For all $\alpha \in D^*$ given by \eqref{Dstar},
	$$
	\lim_{\lambda \to 0} \lim_{\epsilon \to 0} M_f(\lambda,q^*_{\alpha,\epsilon,\lambda})
	= \mathcal{E}_f(q^*(\alpha)) \,.
	$$
\end{lemma}
\begin{proof}
	Let $\alpha \in D^*$ and $\lambda > 0$. We have by definition of $D_{\alpha,\lambda}$, of $D^*$ and using the link between $\tilde i_{\rm RS}$ and $i_{\rm RS}$ given by \eqref{itilde_irs}, that $0 \in D_{\alpha,\lambda}$. By Lemma~\ref{lem:D_full_measure} above, we have
	$$
	q^*_{\alpha,\epsilon,\lambda} \xrightarrow[\epsilon \to 0, \ \epsilon \in D_{\alpha,\lambda}]{} q^*_{\alpha,0,\lambda} = q^*(\alpha) \,.
	$$
	Analogously to Proposition \ref{prop:gen_continuous}, $M_f(\lambda, \cdot)$ is continuous on $[0,\rho]$, thus
	$\lim_{\epsilon \to 0} M_f(\lambda,q^*_{\alpha,\epsilon,\lambda}) = M_f(\lambda,q^*(\alpha))$. And we obtain the result by taking $\lim_{\lambda \to 0}M_f(\lambda,q^*(\alpha))=\mathcal{E}_f(q^*(\alpha))$, using that $M_f(\cdot,q)$ is continuous for $q \in [0,\rho]$ fixed (by Proposition \ref{prop:immse}) and by comparing \eqref{eq:def_M} and \eqref{eq:def_e_f}.
\end{proof}

In order to simplify the proof, we assume that $m = \alpha n$.
By definition of the generalization error \eqref{eq:def_e_gen_f} and of the labels $\bY'$ given by \eqref{Y'f},
$$
\mathcal{E}_{f,n}(\alpha) \defeq \MMSE(Y'_1 | \bY,\bbf{\Phi},\bbf{\Phi}') \,.
$$
\begin{lemma}[Lower bound on the generalization error]
	For all $\alpha \in D^*$,
	$$
	\liminf_{n \to \infty} \mathcal{E}_{f,n}(\alpha) \geq \mathcal{E}_f(q^*(\alpha)) \,.
	$$
\end{lemma}
\begin{proof}
	Let $\alpha \in D^*$, $\lambda > 0$ and $\epsilon \in D_{\alpha,\lambda} \setminus \{0\}$. Obviously,
$$
\mathcal{E}_{f,n}(\alpha) \geq \MMSE(Y'_1|\bY,\bU,\bbf{\Phi},\bbf{\Phi}') \xrightarrow[n \to \infty]{} M_f(\lambda,q^*_{\alpha,\epsilon,\lambda}) \,,
$$
where we used Lemma~\ref{lem:lim_mmse_side_gen}.
Consequently $\liminf\limits_{n \to \infty} \mathcal{E}_{f,n}(\alpha) \geq M_f(\lambda,q^*_{\alpha,\epsilon,\lambda})$ and we obtain the lower bound by letting $\epsilon,\lambda \to 0$ and using Lemma~\ref{lem:lim_eps_lam}.
\end{proof}

Let us now prove the converse upper bound.

\begin{lemma}\label{lem:upper_gen}
	There exists a constant $C>0$ (that only depend on $f$) such that for all $\alpha,\lambda > 0$ and all $\epsilon \in D_{\alpha,\lambda} \setminus \{0\}$
	$$
	\limsup_{n \to \infty} \mathcal{E}_{f,n}(\alpha + \epsilon) \leq M_f(\lambda,q^*_{\alpha,\epsilon,\lambda}) + C \lambda \,.
	$$
\end{lemma}
\begin{proof}
	We will let the signal-to-noise ratio (snr) of the observation of $Y'_1$ go to zero. Let us denote by $\lambda_1$ this snr: $U_1 = \sqrt{\lambda_1}\, Y'_1 + Z'_1$. We will let $\lambda_1$ go from $\lambda$ to $0$ while the other snr for the observations of $U_{\mu}$ for $\mu = 2, \dots,\epsilon n$ will remain equal to $\lambda$. Recall that we denote $\bU=(U_\mu)_{\mu=1}^{\epsilon n}$. Using Proposition~9 from \cite{guo2011estimation}, 
	\begin{align*}
	\Big|
	\frac{\partial}{\partial \lambda_1}
\MMSE(Y'_1|\bY,\bU,\bbf{\Phi},\bbf{\Phi}')
\Big|
&=
\E\big[\Var(Y'_1|\bY,\bU,\bbf{\Phi},\bbf{\Phi}')^2\big]
\leq \E \big[(Y_1')^4\big]  \leq \|f\|_{\infty}^4 \,.
	\end{align*}
	We define $C \defeq \|f\|_{\infty}^4$.
	Consequently, by the mean value theorem,
\begin{align}
	\big| 
\MMSE(Y'_1|\bY,\bU,\bbf{\Phi},\bbf{\Phi}')
-
\MMSE(Y'_1| \bY,(U_{\mu})_{\mu =2}^{\epsilon n},\bbf{\Phi},\bbf{\Phi}')
\big| \leq C \lambda \,.\label{113__}
\end{align}
Since $(U_{\mu})_{\mu =2}^{\epsilon n}$ contains less information than $(\widetilde{Y}_{\mu})_{\mu =2}^{\epsilon n}$ because of the additional Gaussian noise and the application of the function $f$, we have
\begin{equation}\label{eq:ineq_info}
\MMSE(Y'_1| \bY,(U_{\mu})_{\mu =2}^{\epsilon n},\bbf{\Phi},\bbf{\Phi}')
\geq
\MMSE(Y'_1| \bY,(\widetilde{Y}_{\mu})_{\mu =2}^{\epsilon n},\bbf{\Phi},\bbf{\Phi}')
= \mathcal{E}_{f,n}(\alpha + \epsilon-1/n)\ge\mathcal{E}_{f,n}(\alpha + \epsilon)   \,.
\end{equation}
The last identity combined with \eqref{113__} leads to
\begin{align}
\MMSE(Y'_1|\bY,\bU,\bbf{\Phi},\bbf{\Phi}')	+ C\lambda\geq \mathcal{E}_{f,n}(\alpha + \epsilon) \,.
\end{align}
By Lemma~\ref{lem:lim_mmse_side_gen} we know that $\lim_{n\to\infty}\MMSE(Y'_1|\bY,\bU,\bbf{\Phi},\bbf{\Phi}') = M_f(\lambda, q^*_{\alpha,\epsilon,\lambda})$. Thus we conclude by taking the limsup in the inequality above.
\end{proof}

\begin{corollary}[Upper bound on the generalization error]
	For all $\alpha \in D^*$,
	$$
	\limsup_{n \to \infty} \mathcal{E}_{f,n}(\alpha) \leq \mathcal{E}_f(q^*(\alpha)) \,.
	$$
\end{corollary}
\begin{proof}
	Let $\alpha \in D^*, \lambda >0$ and $\epsilon_1 >0$ such that $\alpha - \epsilon_1 \in D^*$. Since by Lemma \ref{lem:D_full_measure} the set $D_{\alpha-\epsilon_1,\lambda}$ is dense in $\R_+$, we can find $\epsilon_2 \in D_{\alpha - \epsilon_1, \lambda}$ such that $0 < \epsilon_2 \leq \epsilon_1$.
	Using Lemma~\ref{lem:upper_gen} above, we have
	$$
	\limsup_{n \to \infty} \mathcal{E}_{f,n}(\alpha -\epsilon_1 + \epsilon_2) \leq M_f(\lambda,q_{\alpha - \epsilon_1, \epsilon_2,\lambda}^*)+ C \lambda \,.
	$$
	Now, using the fact that $\epsilon_2 \leq \epsilon_1$ we have
	$$
	\limsup_{n \to \infty} \mathcal{E}_{f,n}(\alpha)
	\leq
	\limsup_{n \to \infty} \mathcal{E}_{f,n}(\alpha - \epsilon_1 + \epsilon_2)
	\leq M_f(\lambda,q^*_{\alpha - \epsilon_1,\epsilon_2,\lambda})+ C \lambda \,.
	$$
	Now, by Lemma~\ref{lem:lim_eps_lam} we have
	$$
	\lim_{\lambda \to 0} \lim_{\epsilon_2 \to 0} M_f(\lambda,q^*_{\alpha - \epsilon_1,\epsilon_2,\lambda})+ C \lambda = \mathcal{E}_f(q^*(\alpha-\epsilon_1))
	$$
	which leads to $\limsup_{n\to\infty} \mathcal{E}_{f,n}(\alpha) \leq \mathcal{E}_f(q^*(\alpha-\epsilon_1))$.
	We conclude by letting $\epsilon_1 \to 0$ (recall that by Proposition~\ref{prop:q_star} $D^*$ is dense in $\R_+$ so it is possible to find $\epsilon_1>0$ arbitrary small such that $\alpha-\epsilon_1 \in D^*$), using the continuity of $\mathcal{E}_f$ (by Proposition~\ref{prop:gen_continuous}) and the continuity of $q^*$ (by Proposition~\ref{prop:q_star}).
\end{proof}

\begin{proof}[Proof of Theorem \ref{th:gen_f}:]
	For the moment we have proven Theorem \ref{th:gen_f} when $f$ is continuous and bounded. We are going to relax this assumption by approximation.
	Let $f: \R \to \R$ such that $\E [ |f(Y_{\rm new})|^{2+\gamma} ]$ remains bounded as $n$ goes to infinity, for some $\gamma > 0$.
	Let $\epsilon > 0$. By density of the continuous and bounded functions in the space $L^2(\R)$ equipped with the law of $Y^{(q)} \sim P_{\rm out}(\cdot \,|\, \sqrt{q}\,V + \sqrt{\rho-q} \,W)$ ($V,W \iid \cN(0,1)$), we can find a continuous bounded function $\widetilde{f}:\R \to \R$ such that $\E [ (f(Y^{(q)}) - \widetilde{f}(Y^{(q)}))^2] \leq \epsilon$.
	\begin{lemma}\label{lem:CLT_f}
		For all $q \in [0,\rho]$ (because the law of $Y^{(q)}$ does not depend on $q$), we have
		\begin{equation}
f(Y_{\rm new})-\widetilde{f}(Y_{\rm new}) \xrightarrow[n \to \infty]{(d)} f(Y^{(q)})-\widetilde{f}(Y^{(q)}) \,.
		\end{equation}
	\end{lemma}
	\begin{proof}
		Let $(\bA_{\rm new},Z_{\rm new}) \sim P_A \otimes \cN(0,1)$ such that $Y_{\rm new} = \varphi(\bbf{\Phi}_{\rm new} \cdot \bX^* /\sqrt{n},\bA_{\rm  new}) + \sqrt{\Delta} Z_{\rm new}$.
	By the central limit theorem (that we apply under \ref{hyp:third_moment}-\ref{hyp:phi_general} and using \ref{hyp:cont_pp}) 
	\begin{equation}\label{eq:clt0}
	\varphi\Big(\frac{\bbf{\Phi}_{\rm new}\cdot \bX^*}{\sqrt{n}},\bA_{\rm new}\Big)
	\xrightarrow[n \to \infty]{(d)}
	\varphi(\sqrt{\rho} Z,\bA_{\rm new}) \,,
	\end{equation}
	where $Z \sim \cN(0,1)$ is independent from $\bA_{\rm new}$.
	Under \ref{hyp:delta_0} this proves Lemma \ref{lem:CLT_f}, because in that case $Y_{\rm new} =  \varphi(\bbf{\Phi}_{\rm new} \cdot \bX^* /\sqrt{n},\bA_{\rm  new})$ takes values in $\N$.
	Under \ref{hyp:delta_pos} we let $g: \R \to \R$ be a continuous bounded function and we write $h \defeq f - \widetilde{f}$. Then
	\begin{align*}
		\E \big[ g \circ h(Y_{\rm new})\big]
	&= \E \Big[ g \circ h \Big(
		\varphi\Big(\frac{\bbf{\Phi}_{\rm new}\cdot \bX^*}{\sqrt{n}},\bbf{A}_{\rm new}\Big) + \sqrt{\Delta} Z_{\rm new}
	\Big) \Big]
	\\
	&= \E \Big[ 
	\frac{1}{\sqrt{2\pi\Delta}}\int g \circ h(z) \exp\Big\{-\frac{1}{2\Delta}\Big(z - \varphi\Big(\frac{\bbf{\Phi}_{\rm new}\cdot \bX^*}{\sqrt{n}},\bbf{A}_{\rm new}\Big)\Big)^2\Big\}
	dz \Big]\,.
	\end{align*}
	The function $x \mapsto \int g \circ h(z) \frac{e^{-\frac{1}{2\Delta}(z-x)^2}}{\sqrt{2\pi \Delta}} dz$ is continuous and bounded: \eqref{eq:clt0} then gives that
	\begin{align*}
		\E \big[ g \circ h(Y_{\rm new})\big]
		\xrightarrow[n \to \infty]{}
		& \E \Big[ 
		\frac{1}{\sqrt{2\pi\Delta}}\int g \circ h(z) \exp\Big\{-\frac{1}{2\Delta}\Big(z - \varphi\big(\sqrt{\rho}Z,\bbf{A}_{\rm new}\big)\Big)^2\Big\}
	dz \Big] = \E \big[ g \circ h(Y^{(q)})\big],
	\end{align*}
	which concludes the proof by the Portemanteau Theorem.  
	\end{proof}

	The sequence $\big((f(Y_{\rm new})-\tilde{f}(Y_{\rm new}))^2\big)_{n \geq 0}$ is uniformly integrable because bounded in $L^{1+\gamma}$ with $\gamma>0$. Consequently, Lemma \ref{lem:CLT_f} above implies
	$$
	\E \big[
		(f(Y_{\rm new}) - \widetilde{f}(Y_{\rm new}))^2
	\big] \xrightarrow[n \to \infty]{}\big\| f(Y^{(q)}) - \widetilde{f}(Y^{(q)}) \big\|_{L^2}^2=
	\E \big[
		(f(Y^{(q)}) - \widetilde{f}(Y^{(q)}))^2
	\big] \leq \epsilon \,.
	$$
	Therefore, we can find $n_0 \in \N$ such that for all $n\geq n_0$, $\| f(Y_{\rm new}) - \widetilde{f}(Y_{\rm new}) \|_{L^2}^2= \E [ (f(Y_{\rm new}) - \widetilde{f}(Y_{\rm new}))^2 ] \leq 2 \epsilon$.
	If we now apply Theorem \ref{th:gen_f} for $\widetilde{f}$, we can find $n_1 \geq n_0$ such that for all $n \geq n_1$, $| \mathcal{E}_{\widetilde{f},n}^{1/2} - \mathcal{E}_{\widetilde{f}}(q^*(\alpha))^{1/2}| \leq \sqrt{\epsilon}$. Let $n \geq 1$, and compute
	\begin{align*}
		\Big| \mathcal{E}_{f,n}^{1/2} -\mathcal{E}_{\widetilde{f},n}^{1/2} \Big|
		&=
		\Big| \big\| f(Y_{\rm new}) - \E[f(Y_{\rm new})| \bY,\bbf{\Phi},\bbf{\Phi}_{\rm new}] \big\|_{L^2}
		- \big\| \widetilde{f}(Y_{\rm new}) - \E[\widetilde{f}(Y_{\rm new})| \bY,\bbf{\Phi},\bbf{\Phi}_{\rm new}] \big\|_{L^2} \Big|
		\\
		&\leq
		\big\| f(Y_{\rm new}) - \widetilde{f}(Y_{\rm new}) \big\|_{L^2}
		+ \big\| \E [f(Y_{\rm new}) - \widetilde{f}(Y_{\rm new})| \bY,\bbf{\Phi},\bbf{\Phi}_{\rm new}] \big\|_{L^2} 
		\\
		&\leq
		2 \big\| f(Y_{\rm new}) - \widetilde{f}(Y_{\rm new}) \big\|_{L^2} \leq 2 \sqrt{2 \epsilon} \,,
	\end{align*}
	where we successively used the triangular inequality twice for the first inequality ($|| a- b| - |x -y| | \le |a-x + y-b| \le |a-x| + |y-b|$) and Jensen's inequality for the second. By the same arguments we have also
	$| \mathcal{E}_{f}(q)^{1/2} - \mathcal{E}_{\widetilde{f}}(q)^{1/2} | \leq 2 \sqrt{\epsilon}$ for all $q \in [0,\rho]$. We conclude that for all $n \geq n_1$, 
	\begin{align}
		\big|\mathcal{E}_{f,n}^{1/2} - \mathcal{E}_{f}(q^*(\alpha))^{1/2}\big| &\leq \big|\mathcal{E}_{f,n}^{1/2} - \mathcal{E}_{\widetilde f,n}^{1/2}\big| + \big|\mathcal{E}_{\widetilde f}(q^*(\alpha))^{1/2} - \mathcal{E}_{f}(q^*(\alpha))^{1/2}\big|+ \big|\mathcal{E}_{\widetilde f,n}^{1/2} - \mathcal{E}_{\widetilde f}(q^*(\alpha))^{1/2}\big| \nn
	&\leq (2 \sqrt{2} + 3) \sqrt{\epsilon}\,,
	\end{align}
	which proves Theorem \ref{th:gen_f}.
\end{proof}

\subsection{Generalization error of GAMP: Proof of Proposition~\ref{claim:gamp2}}\label{appendix:proof_gen2}

Let us decompose:
\begin{align}
	\mathcal{E}_{\rm gen}^{{\scriptsize\rm GAMP}, t}
	&\defeq
	\E \big[
		\big(Y_{\rm new} - \widehat{Y}^{{\scriptsize\rm GAMP},t}\big)^2
	\big]
	= 
	\E \big[ Y_{\rm new}^2 \big]
	+ \E \big[ \big(\widehat{Y}^{{\scriptsize\rm GAMP},t}\big)^2 \big]
	- 2 \E \big[ Y_{\rm new} \widehat{Y}^{{\scriptsize\rm GAMP},t} \big] \,.
	\label{eq:dec_gen_gamp}
\end{align}

\begin{lemma}\label{lem:CLT_GAMP}
	We have
	\begin{align}
		\E \big[ Y_{\rm new} \widehat{Y}^{{\scriptsize\rm GAMP},t} \big] 
\xrightarrow[n \to \infty]{}
\E_{V}\Big[\E_W\Big[\int dY \,Y P_{\rm out}(Y|\sqrt{q^t}\,V + \sqrt{\rho - q^t} \,W)\Big]^2\Big] \,.
	\end{align}
\end{lemma}
\begin{proof}
Start by writing
	\begin{align*}
\E \big[ Y_{\rm new} \widehat{Y}^{{\scriptsize\rm GAMP},t} \big] 
=
\E \int y\, y' \,
P_{\rm out}\Big(y \Big| \frac{\bbf{\Phi}_{\rm new} \cdot \bX^*}{\sqrt{n}} \Big)  
P_{\rm out}\Big(y' \Big| \frac{\bbf{\Phi}_{\rm new} \cdot {\widehat{\bx}}^{t}}{\sqrt{n}} + \sqrt{\rho - q^t}\, W \Big)
dy dy'
	\end{align*}
	where $W \sim \cN(0,1)$ is independent of everything else. $\bbf{\Phi}_{\rm new} \sim \cN(0,\bbf{I}_n)$ is independent of $\bX^*$ and ${\widehat{\bx}}^t$, so, conditionally on $\bX^*,{\widehat{\bx}}^t$ we have
	$$
	\Big(
	\frac{\bbf{\Phi}_{\rm new} \cdot \bX^*}{\sqrt{n}},
\frac{\bbf{\Phi}_{\rm new} \cdot {\widehat{\bx}}^t}{\sqrt{n}} 
	\Big)
	\sim 
	\cN \bigg(0,\frac{1}{n}
		\begin{psmallmatrix}
			 \|\bX^*\|^2 & {\widehat{\bx}}^t \cdot \bX^* \\
			 {\widehat{\bx}}^t \cdot \bX^* & 
			  \| {\widehat{\bx}}^t \|^2 
		\end{psmallmatrix}
	\bigg) \,.
	$$
	We assumed that \eqref{eq:lim_gamp} holds, i.e. $ \bX^* \cdot {\widehat{\bx}}^t/n \to q^t$ and $\|{\widehat{\bx}}^t\|^2/n \to q^t$, in probability. By the law of large numbers $ \|\bX^*\|^2/n \to \rho$ in probability. Consequently,
	$$
	\Big(
	\frac{\bbf{\Phi}_{\rm new} \cdot \bX^*}{\sqrt{n}},
\frac{\bbf{\Phi}_{\rm new} \cdot {\widehat{\bx}}^t}{\sqrt{n}} 
	\Big)
	\xrightarrow[n \to \infty]{(d)}
	\cN \Big(0,
		\begin{psmallmatrix}
			\rho & q^t \\
			q^t& q^t
		\end{psmallmatrix}
	\Big) \,.
	$$
	Since $x \mapsto P_{\rm out}(\cdot | x)$ is continuous almost everywhere for the Wasserstein distance of order $2$, the function
	$h: (a,b) \mapsto \E_W \int y y' P_{\rm out}(y | a) P_{\rm out}(y' | b + \sqrt{\rho-q^t}\,W) dy dy'$ with $W\sim {\cal N}(0,1)$
	is continuous almost everywhere. Therefore
	\begin{equation}\label{eq:cv_law_H}
	H_n
	\defeq
	h\Big(
	\frac{\bbf{\Phi}_{\rm new} \cdot \bX^*}{\sqrt{n}},
\frac{\bbf{\Phi}_{\rm new} \cdot {\widehat{\bx}}^t}{\sqrt{n}}
	\Big)
	\xrightarrow[n \to \infty]{(d)}
	h(\sqrt{q^t} \,Z_0 + \sqrt{\rho - q^t}\, Z_1, \sqrt{q^t} \,Z_0) \,,
\end{equation}
 where $Z_0, Z_1\iid {\cal N}(0,1)$. We have by Jensen's inequality
	\begin{align*}
		\E \big[ \big|H_n\big|^{1+\eta} \big]
		&\leq \E \big[ \big|Y_{\rm new} \widehat{Y}^{{\scriptsize\rm GAMP},t}\big|^{1+\eta} \big]
		\leq \E \Big[ \Big(\frac{1}{2}Y_{\rm new}^2  + \frac{1}{2} (\widehat{Y}^{{\scriptsize\rm GAMP},t})^2 \Big)^{1+\eta} \Big]
		\\
		&\leq \frac{1}{2} \E |Y_1|^{2+2\eta}
		+\frac{1}{2} \E \big|\widehat{Y}^{{\scriptsize\rm GAMP},t}\big|^{2+2\eta} \,.
	\end{align*}
	By assumption, there exists $\eta>0$ such that the two last terms above remain bounded with $n$: $H_n$ is therefore bounded in $L^{1+\eta}$ and is therefore uniformly integrable.
	From \eqref{eq:cv_law_H} we thus get
\begin{align*}
\E \big[ Y_{\rm new} \widehat{Y}^{{\scriptsize\rm GAMP},t} \big] 
= \E [H_n]
\xrightarrow[n \to \infty]{}
&\E 
\big[h(\sqrt{q^t} \,Z_0 + \sqrt{\rho - q^t} \,Z_1, \sqrt{q^t} \,Z_0) \big]
\\
= \,&\E_{V}\Big[\E_W\Big[\int dY \,Y P_{\rm out}(Y|\sqrt{q^t}\,V + \sqrt{\rho - q^t} \,W)\Big]^2\Big] \,.
\end{align*}
\end{proof}

Following the arguments of Lemma \ref{lem:CLT_GAMP} one can also show that
\begin{align*}
	\E \big[ \big(\widehat{Y}^{{\scriptsize\rm GAMP},t}\big)^2 \big] 
	&\xrightarrow[n \to \infty]{}
\E_{V}\Big[\E_W\Big[\int dY \,Y P_{\rm out}(Y|\sqrt{q^t}\,V + \sqrt{\rho - q^t} \,W)\Big]^2\Big]  \,,
\\
\E \big[ Y_{\rm new}^2 \big] 
&\xrightarrow[n \to \infty]{}
\E_{V}\int dY \,Y^2 P_{\rm out}(Y|\sqrt{\rho}\,V)  \,.
\end{align*}
This proves (together with \eqref{eq:dec_gen_gamp} and Lemma \ref{lem:CLT_GAMP}) Proposition \ref{claim:gamp2}.

\subsection{Limit of the overlap: Proof of Theorem~\ref{th:overlap}}\label{sec:proof_overlap}

Recall the definition of the overlap \eqref{eq:def_overlap_Q}: $Q_n\defeq  \bX^*\cdot \bx/n$,
where $\bx = (x_1, \dots, x_n)$ is a sample from the posterior distribution $P(\bX^*|\,\bY,\bbf{\Phi})$, independently of everything else.
In this section we will show that $|Q_n|$ converges in probability to $q^*(\alpha)$, when $\alpha \in D^*$ given by \eqref{Dstar}. We will first show an upper-bound in Sec.~\ref{sec:upper_bound_Q} below, before proving the converse lower-bound in Sec.~\ref{sec:lower_bound_Q}.

\subsubsection{Upper bound on the overlap}\label{sec:upper_bound_Q}
\begin{proposition}[Upper bound on the overlap] \label{prop:upper_bound_overlap}
	For all $\alpha \in D^*$ and for all $\epsilon > 0$,
	$$
	\P \big(|Q_n| \geq q^*(\alpha) + \epsilon \big) \xrightarrow[n \to \infty]{} 0 \,.
	$$
\end{proposition}

Let us fix $\alpha \in D^*$ and let $p \geq 1$.
In order to obtain an upper bound on the overlap, we consider an observation model with some (small) extra information (that takes the form of a tensor of order $2p$) in addition of the original model \eqref{eq:channel}, i.e. we observe
\begin{align}\label{pert_tensoir}
	\begin{cases}
	&\bY \sim P_{\rm out}( \cdot \, | \,  \bbf{\Phi} \bX^*/\sqrt{n} ) \,,
	\\
	&\bY' = \sqrt{\frac{\lambda}{n^{2p-1}}} (\bX^*)^{\otimes 2p} + \bZ' \,,	
	\end{cases}
\end{align}
where $\lambda \geq 0$, $\bZ'=(Z'_{i_1 \dots i_{2p}})_{1 \leq i_1, \dots, i_{2p} \leq n} \iid \cN(0,1)$ and $(\bX^*)^{\otimes 2p} = (X_{i_1} \dots X_{i_{2p}})_{1 \leq i_1, \dots, i_{2p} \leq n}$. In order to prove Proposition \ref{prop:upper_bound_overlap} we need the two results below, which are proven after the proof of Proposition \ref{prop:upper_bound_overlap}.
\begin{proposition}[Mutual information of the perturbed model]\label{prop:mi_side_tensor}
	For all $\lambda \geq 0$, the mutual information for model \eqref{pert_tensoir} verifies
	\begin{align}
		\lim_{n \to \infty} \frac{1}{n} &I \big(\bX^*; \bY,\bY' \big| \bbf{\Phi}\big)		=I(\lambda)\,,\label{MI_pert_tens}
		\end{align}
		where the right-hand-side is 
		\begin{align}
		I(\lambda) \!\defeq\! {\adjustlimits \inf_{q \in [0,\rho]} \sup_{r \geq 0}}
		\Big\{
			I_{P_0}(r \!+\! 2p\lambda q^{2p-1}) + \alpha \mathcal{I}_{P_{\rm out}}(q) - \frac{r}{2}(\rho-q) + \frac{2p-1}{2} \lambda q^{2p} -\rho p \lambda q^{2p-1}
			+\frac{\lambda}{2} \rho^{2p}
		\Big\}. \label{eq:defF}
	\end{align}
\end{proposition}

\begin{lemma}\label{lem:derF}	
	The function $I$ defined above by \eqref{eq:defF} is concave on $\R_+$. Its left- and right-derivatives are given by
	\begin{align*}
		I'(\lambda^+) &= \min \Big\{ \frac{1}{2} \big( \rho^{2p} - q_*(\lambda)^{2p} \big) \, \Big| \, q_*(\lambda) \ \text{achieves the infimum in \eqref{eq:defF}} \Big\} \,,
	\\
	I'(\lambda^-) &= \max \Big\{ \frac{1}{2}\big(\rho^{2p} - q_*(\lambda)^{2p} \big) \, \Big| \, q_*(\lambda) \ \text{achieves the infimum in \eqref{eq:defF}} \Big\} \, .
	\end{align*}
\end{lemma}

We are now in position to prove Proposition \ref{prop:upper_bound_overlap}.

\begin{proof}[Proof of Proposition \ref{prop:upper_bound_overlap}:]
By the I-MMSE relation of Proposition~\ref{prop:immse},
$$
\frac{1}{n} \frac{\partial}{\partial \lambda} I\big(\bX^*; \bY,\bY'\big|\bbf{\Phi}\big)
=
\frac{1}{n} \frac{\partial}{\partial \lambda} I\big((\bX^*)^{\otimes 2p}; \bY,\bY'\big|\bbf{\Phi}\big)
= \frac{1}{2n^{2p}} \MMSE\big( (\bX^*)^{\otimes 2p} \big| \bY,\bY', \bbf{\Phi}\big) \,.
$$
Using Proposition~\ref{prop:mi_side_tensor} and Lemma~\ref{lem:derF} above we obtain by concavity that
$$
\frac{1}{2n^{2p}} \MMSE\big( (\bX^*)^{\otimes 2p} \big| \bY,\bY', \bbf{\Phi}\big) =
\frac{1}{n}\frac{\partial}{\partial \lambda} I\big(\bX^*; \bY,\bY'\big|\bbf{\Phi}\big)
\xrightarrow[n \to \infty]{}
I'(\lambda) = \frac{1}{2} \big(\rho^{2p} - q_*(\lambda)^{2p} \big) \,,
$$
for all  $\lambda > 0$ for which the infimum of \eqref{eq:defF} is achieved at a unique $q_*(\lambda)$. Consequently,
\begin{equation}\label{eq:liminf_mmse_si}
\liminf_{n \to \infty}
\frac{1}{n^{2p}} \MMSE\big( (\bX^*)^{\otimes 2p} \big| \bY, \bbf{\Phi}\big)
\geq
\liminf_{n \to \infty}
\frac{1}{n^{2p}} \MMSE\big( (\bX^*)^{\otimes 2p} \big| \bY,\bY', \bbf{\Phi}\big) = \rho^{2p} - q_*(\lambda)^{2p} \,.
\end{equation}
Let us now suppose that $\alpha \in D^*$. In that case, there exists a unique $q_*(\lambda=0) = q^*(\alpha)$ that achieves the infimum in \eqref{eq:defF}. Consequently, $I'(0^+) = \frac{1}{2} (\rho^{2p} - q^*(\alpha)^{2p})$. By concavity, $I'(\lambda) \to I'(0^+)$ as $\lambda \to 0$, which gives $q_*(\lambda) \to q^*(\alpha)$. By taking the $\lambda \to 0$ limit in \eqref{eq:liminf_mmse_si} above we get
$$
\liminf_{n \to \infty}
\frac{1}{n^{2p}} \MMSE\big( (\bX^*)^{\otimes 2p} \big| \bY, \bbf{\Phi}\big)
\geq
\rho^{2p} - q^*(\alpha)^{2p} \,.
$$
One verifies easily that 
\begin{align}
\frac{1}{n^{2p}} \MMSE\big( (\bX^*)^{\otimes 2p} \big| \bY, \bbf{\Phi}\big) = \rho^{2p} - \E\big[Q_n^{2p}\big] + o_n(1)\,,	
\end{align}
 so we deduce that 
$$
\limsup_{n \to \infty} \E \big[Q_n^{2p}\big] \leq q^*(\alpha)^{2p} \,.
$$
Let $\epsilon > 0$. By Markov's inequality we have
$$
\P\big(|Q_n| \geq q_*(\alpha) + \epsilon \big) \leq \frac{\E\big[Q_n^{2p}\big]}{(q_*(\alpha) + \epsilon)^{2p}}\,.
$$
By taking the $\limsup$ in $n$ on both sides we obtain
$$
\limsup_{n \to \infty}\P\big(|Q_n| \geq q_*(\alpha) + \epsilon \big) \leq \frac{q^*(\alpha)^{2p}}{(q_*(\alpha) + \epsilon)^{2p}}\,,
$$
and Proposition~\ref{prop:upper_bound_overlap} follows by taking the $p \to \infty$ limit in the inequality above.	
\end{proof}

We now prove the two preliminary results used in the proof of Proposition \ref{prop:upper_bound_overlap}.

\begin{proof}[Proof of Proposition \ref{prop:mi_side_tensor}:]
	The proof is very similar to the one of Theorem~\ref{th:RS_1layer} (and Corollary~\ref{cor:mi}), by the adaptive interpolation method (see Sec.~\ref{sec:interpolation}), so we provide only the main arguments and omit to write the small perturbation (i.e.\ the $\epsilon_1, \epsilon_2$ present in Sec.~\ref{sec:interpolation}) for simplicity.

	In order to tackle model \eqref{pert_tensoir} we need first to study a simpler one, namely when we have access to the simultaneous observations $\bY \sim P_{\rm out}(\cdot \,|\, \bbf{\Phi} \bX^* / \sqrt{n})$ and $\bY'' = \sqrt{\gamma}\, \bX^* + \bZ''$.
	Define, for $\gamma \geq 0$, the free entropy (expected log-partition function) of this model:
	\begin{equation}
		F_n(\gamma) \defeq \frac{1}{n}
		\E \ln 		
		\int dP_0(\bx) 
		\exp\Big(\sum_{i=1}^n  \sqrt{\gamma} Z''_i x_i + \gamma x_i X_i^* - \frac{\gamma}{2} x_i^2 \Big)
		\prod_{\mu = 1}^m P_{\rm out}\Big(Y_{\mu}\Big|\frac{1}{\sqrt{n}} \bbf{\Phi}_{\mu} \cdot \bx\Big) \,,
	\end{equation}
	where $Z''_i \iid \cN(0,1)$ are independent of everything else. Let us define
	\begin{equation}\label{eq:def_FRS}
	F_{\rm RS}(\gamma) \defeq {\adjustlimits \sup_{q \in [0, \rho]} \inf_{r \geq 0}} \Big\{ \psi_{P_0}(r + \gamma) + \alpha \Psi_{P_{\rm out}}(q) - \frac{rq}{2} \Big\} \,.
	\end{equation}
	A slight and easy modification of the Theorem~\ref{th:RS_1layer} gives that for all $\gamma \geq 0$
	\begin{equation}\label{eq:rs_side_scalar}
	F_n(\gamma) \xrightarrow[n \to \infty]{} F_{\rm RS}(\gamma) \,.
\end{equation}
$F_n$ is a convex function of $\gamma$ (this can be checked by relating it to the mutual information like in Corollary \ref{cor:mi} and then using the I-MMSE relation of Proposition \ref{prop:immse}), thus $F_{\rm RS}$ is too. The function $F_{\rm RS}$ is therefore continuous on $\R_+$. $F_n$ is also a non-decreasing function of $\gamma$ (this is again checked using the I-MMSE relation). By Dini's second theorem we obtain that the convergence of \eqref{eq:rs_side_scalar} is uniform over all compact subsets of $\R_+$.

Now that we have studied this simpler model, we come back to the analysis of \eqref{pert_tensoir}. We proceed by interpolation as in Sec.~\ref{interp-est-problem}. Let $q: [0,1] \to [0,\rho]$ be a continuous interpolating function. For $t \in [0,1]$, consider the following ``interpolating estimation model'':
\begin{align}
	\label{3channels}
	\left\{
		\begin{array}{llll}
			Y_{\mu} &\sim & P_{\rm out}(\ \cdot \ | \, \bbf{\Phi} \bX^* / \sqrt{n})\,,\qquad  &1 \leq \mu \leq m\,, \\
			\bY'_{t}  &= & \sqrt{\frac{\lambda(1-t)}{n^{2p-1}}} (\bX^* )^{\otimes 2p} + \bZ' \,, \\
			Y''_{t,i} &=& \sqrt{ 2p \lambda \int_0^t {q}(v)^{2p-1} dv}\, X^*_i + Z''_i\,, \qquad &1 \leq i \leq n\,.
		\end{array}	
	\right.
\end{align}
Define the corresponding interpolating free entropy:
$$
f_n(t) \defeq \frac{1}{n}
		\E \ln 	
		\int dP_0(\bx) e^{H_{n,t}(\bx)}
		\prod_{\mu = 1}^m P_{\rm out}\Big(Y_{\mu}\Big| \frac{1}{\sqrt{n}}\bbf{\Phi}_{\mu} \cdot \bx \Big)  \,,
$$
where the Hamiltonian of the model is
\begin{align*}
	H_{n,t}&(\bx) 
= \sum_{i=1}^n \Big\{\sqrt{ 2p \lambda \int_0^t {q}(v)^{2p-1}dv}\,Z_i'' x_i + 2p \lambda \int_0^t {q}(v)^{2p-1}dv\, x_i X_i^* - p \lambda \int_0^t {q}(v)^{2p-1}dv\, x_i^2\Big\}
\\
&+ \sum_{i_1, \dots, i_{2p}}\Big\{ \sqrt{\frac{\lambda (1-t)}{n^{2p-1}}}Z_{i_1 \dots i_{2p}}' x_{i_1} \dots x_{i_{2p}} + \frac{\lambda (1-t)}{n^{2p-1}} x_{i_1} \dots x_{i_{2p}} X^*_{i_1} \dots X^*_{i_{2p}} - \frac{\lambda (1-t)}{2n^{2p-1}} x_{i_1}^2 \dots x_{i_{2p}}^2 \Big\}\,.
\end{align*}
We aim at computing $f_n \defeq f_n(0)$. We have $f_n(1) = F_n(2p \lambda \int_0^1 q(t)^{2p-1}dt)$. Similarly to Proposition~\ref{prop:der_f_t} one can compute (see \cite{barbier_stoInt} where this computation is done):
\begin{equation}\label{eq:der_f_tensor}
f_n'(t) =
-\frac{\lambda}{2}
\E \big\langle Q_t^{2p} - 2p q(t)^{2p-1} Q_t\big\rangle_t
\end{equation}
where $Q_t =  \sum_{i=1}^n X_i^* x_i/n$ is the overlap between the planted solution $\bX^*$ and $\bx = (x_1, \dots, x_n)$, a sample from the posterior distribution $P(\bX^*|\bY,\bY'_t,\bY''_t)$. Similarly as in Sec.~\ref{interp-est-problem} the Gibbs bracket $\langle - \rangle_t$ denotes the expectation w.r.t. this $t$-dependent posterior acting on $\bx$, $\E$ is w.r.t. the quenched variables $\bY,\bY'_t,\bY''_t$. By convexity of the function $x \mapsto x^{2p}$, we have for all $a,b \in \R$, $a^{2p} - 2p a b^{2p-1} \geq (1-2p) b^{2p}$. Consequently, if we choose $q$ to be a constant function, i.e.\ $q(t) = q$ for all $t \in [0,1]$, we have
$$
f_n'(t) \leq \frac{\lambda}{2} (2p-1) q^{2p} \,.
$$
This gives
$$
f_n = f_n(0) = f_n(1) - \int_0^1 f_n'(t) dt \geq F_n(2p\lambda q^{2p - 1}) - \frac{\lambda}{2} (2p-1) q^{2p} \,.
$$
By taking the $\liminf$ in $n$ on both sides, we obtain $\liminf_{n \to \infty} f_n \geq F_{\rm RS}(2p \lambda q^{2p-1}) - \frac{\lambda}{2} (2p-1) q^{2p}$ using \eqref{eq:rs_side_scalar} and since this holds for all $q \in [0,\rho]$ we get
$$
\liminf_{n \to \infty} f_n \geq \sup_{q \in [0,\rho]} \Big\{ F_{\rm RS}(2p \lambda q^{2p-1}) - \frac{\lambda}{2} (2p-1) q^{2p} \Big\} \,.
$$

Let us now prove the converse upper-bound. One can show as in Sec.~\ref{sec:overlap_concentration} that the overlap $Q_t$ concentrates around its expectation: Proposition \ref{concentration} applies. This perturbation does not change the free entropy in the limit $n\to\infty$ nor the following derivation, so we do not track it explicitely for the sake of simplicity. Let us go back to \eqref{eq:der_f_tensor}. Therefore, using this concentration and then choosing $q(t) = \tilde q(t) = \E\langle Q_t\rangle_t$ as done in Sec.~\ref{subsec:lower-upper}, we obtain that
$$
f_n'(t) = \frac{\lambda}{2} (2p-1) \tilde q(t)^{2p} + o_n(1) \,.
$$
Consequently,
\vspace{-5mm}
\begin{align*}
	f_n &= f_n(1) - \int_0^1 f_n'(t) dt
= F_n\Big(2p \lambda \int_0^1  \tilde q(t)^{2p-1}dt\Big) - \frac{\lambda}{2} (2p-1) \int_0^1  \tilde q(t)^{2p}dt + o_n(1) 
\\
&\leq
F_n\Big(2p \lambda \int_0^1  \tilde q(t)^{2p-1}dt\Big) - \frac{\lambda}{2} (2p-1) \Big(\int_0^1 \tilde q(t)^{2p-1}dt\Big)^{\frac{2p}{2p-1}} + o_n(1) 
\\
&\leq
\sup_{q \in [0,\rho]} \Big\{
F_n(2p \lambda q^{2p-1}) - \frac{\lambda}{2} (2p-1) q^{2p} \Big\}+ o_n(1) \,.
\end{align*}
We use now the fact that the convergence in \eqref{eq:rs_side_scalar} is uniform over all compact sets to get the upper-bound: $\limsup_{n \to \infty} f_n \leq \sup_{q \in [0,\rho]} \big\{ F_{\rm RS}(2p \lambda q^{2p-1}) - \frac{\lambda}{2} (2p-1) q^{2p} \big\}$.
We conclude that 
\begin{equation}\label{eq:lim_f_tensor}
	\lim_{n \to \infty} f_n = \sup_{q \in [0,\rho]} \Big\{ F_{\rm RS}\big(2p \lambda q^{2p-1}\big) - \frac{\lambda}{2} (2p-1) q^{2p} \Big\} \,.
\end{equation}
We are now going to simplify the right-hand side of the above equation.
\begin{lemma}\label{lem:diff_FRS}
	$F_{\rm RS}$ is a convex function on $\R_+$, whose left- and right-derivatives at $\gamma \geq 0$ are:
	\begin{align*}
		F_{\rm RS}'(\gamma^+) &= \max \Big\{ \frac{1}{2}q_*(\gamma) \, \Big| \, q_*(\gamma) \ \text{achieves the supremum in \eqref{eq:def_FRS}} \Big\} \,,
	\\
	F_{\rm RS}'(\gamma^-) &= \min \Big\{ \frac{1}{2}q_*(\gamma) \, \Big| \, q_*(\gamma) \ \text{achieves the supremum in \eqref{eq:def_FRS}} \Big\} \, .
	\end{align*}
	In particular, $F_{\rm RS}$ is differentiable at $\gamma \geq 0$ if and only if the supremum in \eqref{eq:def_FRS} is achieved at a unique $q_*(\gamma)$. 
\end{lemma}
\begin{proof}
	We already know that $F_{\rm RS}$ is convex (as a limit of convex functions, see \eqref{eq:rs_side_scalar}). We have
	\begin{equation}\label{eq:proof_der_FRS}
	F_{\rm RS}(\gamma) = 
	{\adjustlimits\sup_{q \in [0,\rho]} \inf_{r \geq 0} }
	\Big\{ \psi_{P_0}(r+\gamma)  + \alpha \Psi_{P_{\rm out}}(q) - \frac{rq}{2} \Big\}
	=
	\sup_{q \in [0,\rho]} 
	\Big\{ \alpha \Psi_{P_{\rm out}}(q) - g(\gamma, q/2) \Big\}
\end{equation}
where $g(\gamma,x) = \sup_{r \geq 0} \big\{ xr - \psi_{P_0}(\gamma + r)\big\}$
is the Legendre transform of $r \mapsto \psi_{P_0}(\gamma + r)$. Let us now compute $\frac{\partial g}{\partial \gamma}(\gamma,x)$. If $x \leq \psi_{P_0}'(\gamma)$, then the supremum in $r$ is achieved at $r=0$, $g(x,\gamma) = -\psi_{P_0}(\gamma)$. If now $x > \psi_{P_0}'(\gamma)$ then 
$$
g(\gamma,x) 
= \sup_{r \geq -\gamma} \big\{ xr - \psi_{P_0}(\gamma + r)\big\}
= - \gamma x +\sup_{r \geq 0} \big\{ xr - \psi_{P_0}( r)\big\} \,.
$$
The first equality comes from the fact that the supremum can not be achieved on $[-\gamma,0]$ because for all $r \in [-\gamma,0]$, $x > \psi_{P_0}'(\gamma) \geq \psi_{P_0}'(\gamma+r)$. We obtain
$$
g(\gamma,x) =
\left\{
	\begin{array}{ll}
	- \psi_{P_0}(\gamma) & \text{if} \ x \leq \psi'_{P_0}(\gamma)\,, \\
	- x \gamma +  g(0,x) & \text{if} \ x > \psi'_{P_0}(\gamma)\,.
\end{array}
\right.
$$
From there, we conclude that $\frac{\partial g}{\partial \gamma}(\gamma,x) = - \max\big(\psi_{P_0}'(\gamma),x\big)$.
By Lemma~\ref{lem:sup_inf_2}, every optimal couple $(q_*(\gamma),r_*(\gamma))$ satisfy $q_*(\gamma) = 2 \psi_{P_0}'(\gamma + r_*(\gamma))$. This implies (by convexity of $\psi_{P_0}$) that $q_*(\gamma) / 2 \geq \psi_{P_0}'(\gamma)$. 
Using Corollary~4 from \cite{milgrom2002envelope} $F_{\rm RS}$ we get that
$$
F_{\rm RS}'(\gamma^+)= 
\max \Big\{ -\frac{\partial g}{\partial \gamma}(\gamma,q_*(\gamma)) \, \Big| \, q_*(\gamma) \ \text{maximizer of \eqref{eq:def_FRS}} \Big\}
=
\max \Big\{ \frac{1}{2}q_*(\gamma) \, \Big| \, q_*(\gamma) \ \text{maximizer of \eqref{eq:def_FRS}} \Big\}
$$
and analogously for $F_{\rm RS}'(\gamma^-)$.
\end{proof}
\begin{lemma}\label{lem:sup_mq} We have
	$$
	\sup_{\tilde{q} \in [0,\rho]} \big\{ F_{\rm RS}(2p \lambda \tilde{q}^{2p-1}) - (2p-1) \frac{\lambda}{2} \tilde{q}^{2p} \big\}
	=
	{\adjustlimits \sup_{q \in [0, \rho]} \inf_{r \geq 0}} \Big\{ \psi_{P_0}(r + 2p \lambda q^{2p-1}) + \alpha \Psi_{P_{\rm out}}(q) - \frac{rq}{2} - (2p-1) \frac{\lambda}{2} q^{2p}\Big\} \,.
	$$
\end{lemma}
\begin{proof}
	Consider the equality above. The inequality l.h.s $\ge$ r.h.s.
	is obvious because it suffices to restrict the supremum over $(q,\tilde{q}) \in [0,\rho]^2$ to the supremum over the couples $(q,q)$ for $q \in [0,\rho]$. 

	Let us prove now the converse inequality. 
	Let us do the change of variable $x = \tilde{q}^{2p-1}$ and define $H(x) \defeq F_{\rm RS}(2 p \lambda x) - (2p-1) \frac{\lambda}{2} x^{2p/(2p-1)}$.
	$F_{\rm RS}$ is left- and right-differentiable everywhere, so is $H$. We have
	\begin{align}
	H'(x) 
	&= 2p \lambda  F_{\rm RS}'(2p\lambda x) - \frac{\lambda}{2}  2p x^{1/(2p-1)}
	=
	p \lambda \big( 2 F_{\rm RS}'(2p \lambda x) - x^{1/(2p-1)} \big)
	\label{eq:derivH}
	\end{align}
	at the points at which $H$ is differentiable, and analogously for the left- and right-derivatives of $H$.
	Let $x \in [0, \rho^{2p-1}]$ be a point at which $H$ achieves its supremum over $[0,\rho^{2p-1}]$. Let us distinguish 3 cases:
	\begin{itemize}
		\item Case 1: $ x = 0$. In that case, we have $H'(0^+) \leq 0$ and thus $F_{\rm RS}'(0^+) \leq 0$. Using Lemma~\ref{lem:diff_FRS}, we obtain that the only $q \in [0,\rho]$ that achieves the supremum in \eqref{eq:def_FRS} is $q = 0 = x^{1/(2p-1)}$.
		\item Case 2: $0 < x < \rho^{2p-1}$. We have then $H'(x^-) \geq 0$ and $H'(x^+) \leq 0$. Using \eqref{eq:derivH}, we deduce that
			$$
			2 F_{\rm RS}'\big((2p \lambda x)^+\big) \leq x^{1/(2p-1)}
			\leq 2 F_{\rm RS}'\big((2p \lambda x)^-\big) \,.
			$$
			$F_{\rm RS}$ is convex, so the above inequalities collapses into equalities and we get that $F_{\rm RS}$ is differentiable at $2p\lambda x$ with derivative given by $F_{\rm RS}'(2p\lambda x) = x^{1/(2p-1)}/2$. Lemma~\ref{lem:diff_FRS} above gives then that the supremum in \eqref{eq:def_FRS} is achieved uniquely at $q = x^{1/(2p-1)}$.
		\item Case 3: $x = \rho^{2p-1}$. Using the same arguments than in Case 1, we obtain also $q = x^{1/(2p-1)}$.
	\end{itemize}
	Conclusion: In all 3 cases above, $q = x^{1/(2p - 1)}$ achieves the supremum in \eqref{eq:def_FRS}. Recall that we used the change of variable $x = \tilde{q}^{2p-1}$. Consequently, if $\tilde{q} \in [0,\rho]$ achieves the supremum of $\tilde{q} \mapsto F_{\rm RS}(2p\lambda \tilde{q}^{2p-1}) - (2p-1) \lambda \tilde{q}^{2p}$, then $\tilde{q}$ achieves also the supremum in \eqref{eq:def_FRS}. This proves the converse bound.
\end{proof}

By Lemma~\ref{lem:sup_mq} and \eqref{eq:lim_f_tensor} above, we get that 
$$
f_n \xrightarrow[n \to \infty]{} 
	{\adjustlimits\sup_{q \in [0, \rho]} \inf_{r \geq 0}} \Big\{ \psi_{P_0}(r + 2p \lambda q^{2p-1}) + \alpha \Psi_{P_{\rm out}}(q) - \frac{rq}{2} - (2p-1) \frac{\lambda}{2} q^{2p}\Big\} \,.
	$$
	Proposition~\ref{prop:mi_side_tensor} follows then by rewriting the above limit in terms of mutual information, as we did to deduce Corollary~\ref{cor:mi} from Theorem~\ref{th:RS_1layer}.
\end{proof}

\begin{proof}[Proof of Lemma \ref{lem:derF}:]
	The proof follows exactly the same steps than the one of Lemma~\ref{lem:diff_FRS}, so we omit it for the sake brevity.
\end{proof}

\subsubsection{Lower bound using the generalization error}\label{sec:lower_bound_Q}

Let us fix $\alpha \in D^*$. 
The sequence of the overlaps $\big(Q_n \big)_{n \geq 1}$ is tight (because bounded in $L^1$). 
By Prokhorov's Theorem we know that the sequence of the laws of $\big(Q_n \big)_{n \geq 1}$ is relatively compact. We can thus consider a subsequence along which it converges in law, to some random variable $Q$. In order to simplify the notations (and because working with an extraction does not change the proof) we will assume in the sequel that
$$
Q_n \xrightarrow[n \to \infty]{(d)} Q \,,
$$
for some random variable $Q$.
We aim now at showing that $|Q| = q^*(\alpha)$ almost-surely.

\begin{lemma}[Upper bound on the overlap]\label{lem:as_bound_Q}
	$|Q| \leq q^*(\alpha)$ almost-surely.
\end{lemma}
\begin{proof}
	Let $\epsilon > 0$. The set $[0,q^*(\alpha) + \epsilon]$ is closed, so by Portemanteau's Theorem
	$$
	\P\big(|Q| \leq q^*(\alpha) + \epsilon \big) \geq \limsup_{n \to \infty} \P \big(|Q_n| \leq q^*(\alpha) + \epsilon \big) = 1 \,,
	$$
	by Proposition~\ref{prop:upper_bound_overlap}. So $\P\big(|Q| \leq q^*(\alpha)+ \epsilon\big) = 1$ for all $\epsilon >0$ which gives $\P\big(|Q| \leq q^*(\alpha)\big) = 1$.
\end{proof}

We are going to prove the converse lower bound using Theorem~\ref{th:gen_f}. Let $f: \R \to \R$ be a continuous bounded function.
Theorem \ref{th:gen_f} gives $\mathcal{E}_{f,n}(\alpha) \xrightarrow[n \to \infty]{} \mathcal{E}_f(q^*(\alpha))$. The function $\mathcal{E}_f$ can be written as
\begin{align*}
\mathcal{E}_f(q) 
	&=
	\frac{1}{2} \E \Big[
		h_f\big(\sqrt{q} Z_0 + \sqrt{\rho - q} Z_1, \sqrt{q}Z_0 + \sqrt{\rho - q} Z_1'\big)
	\Big]
\end{align*}
where $Z_0,Z_1,Z_1' \iid \cN(0,1)$ and $h_f: (a,b) \in \R^2 \mapsto \int (f(y_1) -f(y_2))^2 P_{\rm out}(y_1|a) P_{\rm out}(y_2|b) dy_1 dy_2 $. 
	By a central limit argument, we have:
	\begin{lemma}\label{lem:CLT_Q}
	$$
	\Big(\frac{\bx \cdot \bbf{\Phi}_{\rm new}}{\sqrt{n}}, \frac{\bX^* \cdot \bbf{\Phi}_{\rm new}}{\sqrt{n}} \Big)
	\xrightarrow[n \to \infty]{(d)} (Z_1,Z_2) \,,
	$$
	where $(Z_1,Z_2)$ is sampled, conditionally on $Q$, from $\cN\Big(0, 
	\begin{psmallmatrix}
		\rho & Q \\
		Q & \rho
	\end{psmallmatrix} \Big)$.
	\end{lemma}
	\begin{proof}
		Notice that $\bx$ and $\bX$ are independent of $\bbf{\Phi}_{\rm  new}$. If $(\Phi_{{\rm new}, 1}, \dots, \Phi_{{\rm new}, n}) \iid \cN(0,1)$, then Lemma~\ref{lem:CLT_Q} is obvious because in that case
		$$
	\Big(\frac{\bx \cdot \bbf{\Phi}_{\rm new}}{\sqrt{n}}, \frac{\bX^* \cdot \bbf{\Phi}_{\rm new}}{\sqrt{n}} \Big)
	\sim \mathcal{N}\left(0,
		\frac{1}{n}
		\begin{psmallmatrix}
			\|\bx\|^2 & \bx \cdot \bX^* \\
			\bx \cdot \bX^* & \|\bX^*\|^2
	\end{psmallmatrix} \right)
		\quad \text{and} \quad
		\frac{1}{n}
		\begin{psmallmatrix}
			\|\bx\|^2 & \bx \cdot \bX^* \\
			\bx \cdot \bX^* & \|\bX^*\|^2
		\end{psmallmatrix}
		\xrightarrow[n \to \infty]{(d)}
		\begin{psmallmatrix}
			\rho & Q \\
			Q & \rho
		\end{psmallmatrix}
		.
		$$
		Let us now suppose that the entries of $\bbf{\Phi}_{\rm  new}$ are not i.i.d.\ standard Gaussian (but still verify hypothesis~\ref{hyp:phi_general}). 
		Let $g_1, \dots, g_n \iid \cN(0,1)$.
		Let $L: \R^2 \to \R$ be a bounded $\cC^3$ function, with bounded partial derivatives. We have to show that
		\begin{equation}\label{eq:goal_CLT_Q}
			\E \Big[ L 
			\Big(\frac{\bx \cdot \bbf{\Phi}_{\rm new}}{\sqrt{n}}, \frac{\bX^* \cdot \bbf{\Phi}_{\rm new}}{\sqrt{n}} \Big) \Big]
			\xrightarrow[n \to \infty]{} \E \big[L(Z_1,Z_2) \big] \,.
\end{equation}
	We have seen above that
	$\E \big[L 
	\big(\frac{\bx \cdot \bbf{g}}{\sqrt{n}}, \frac{\bX^* \cdot \bbf{g}}{\sqrt{n}} \big) \big]
		\xrightarrow[n \to \infty]{} \E [L(Z_1,Z_2)]$. 
		We now apply Theorem \ref{th:lindeberg} (Theorem~2 from \cite{korada2011lindeberg}) conditionally on $\bx,\bX^*$ to obtain
		$$
		\E \Big[ L 
		\Big(\frac{\bx \cdot \bbf{\Phi}_{\rm new}}{\sqrt{n}}, \frac{\bX^* \cdot \bbf{\Phi}_{\rm new}}{\sqrt{n}} \Big) \Big]
	=
	\E \Big[L 
	\Big(\frac{\bx \cdot \bbf{g}}{\sqrt{n}}, \frac{\bX^* \cdot \bbf{g}}{\sqrt{n}} \Big) \Big]
	+ {\cal O}_n(n^{-1/2}) \,,
		$$
		which proves \eqref{eq:goal_CLT_Q} and therefore Lemma~\ref{lem:CLT_Q}.
	\end{proof}

\begin{proposition}We have\label{prop:lim_gen_f}
	$$
	\mathcal{E}_{f,n}(\alpha) \xrightarrow[n \to \infty]{} \frac{1}{2} \E \big[h_f(Z_1,Z_2)\big] \,,
	$$
	where $(Z_1,Z_2)$ is defined in Lemma~\ref{lem:CLT_Q} above.
\end{proposition}

\begin{proof}
	We have
	\vspace{-4mm}
	\begin{align*}
		\mathcal{E}_{f,n} &=
		\E \Big[
			\Big(
				f(Y_{\rm new}) - \E \big[ f(Y_{\rm new}) \big| \bbf{\Phi}_{\rm new}, \bbf{\Phi},\bY \big]			\Big)^2
		\Big]
		\\
		&=
		\frac{1}{2}
		\E \left[
			\int
			(f(y_{\rm new}) - f(y))^2  
			P_{\rm out}\big(y_{\rm  new} \big| \bbf{\Phi}_{\rm new} \cdot \bX^* / \sqrt{n} \big) 
			P_{\rm out}\big(y \big| \bbf{\Phi}_{\rm new} \cdot \bx / \sqrt{n} \big) 
			dy_{\rm new} dy
		\right]
		\\
		&= 
		\frac{1}{2} \E \left[h_f\Big(\frac{\bx \cdot \bbf{\Phi}_{\rm new}}{\sqrt{n}}, \frac{\bX^* \cdot \bbf{\Phi}_{\rm new}}{\sqrt{n}}\Big)\right] \,.
	\end{align*}
	By Lemma~\ref{lem:CLT_Q} above, we have $\big(\frac{\bx \cdot \bbf{\Phi}_{\rm new}}{\sqrt{n}}, \frac{\bX^* \cdot \bbf{\Phi}_{\rm new}}{\sqrt{n}} \big)
	\xrightarrow[n \to \infty]{(d)} (Z_1,Z_2)$.
	Using \ref{hyp:cont_pp} (and the fact that either \ref{hyp:delta_pos} or \ref{hyp:delta_0} hold), we can find a Borel set $S \subset \R$ of full Lebesgue's measure such that $x \mapsto P_{\rm out}(y|x)$ is continuous on $S$, for all $y \in \R$. By dominated convergence (recall that $f$ is assumed to be bounded), we obtain that $h_f$ is continuous on $S \times S$.
	The set of discontinuity points of $h_f$ has thus zero measure for the law of $(Z_1,Z_2)$. Indeed if we condition on $Q$:
	\begin{itemize}
		\item if $|Q| < \rho$, then $(Z_1,Z_2)$ has a density over $\R^2$.
		\item if $Q = \rho$, then $Z_1 = Z_2$ almost surely, but $h_f$ is continuous on $\big\{ (s,s) \, \big| \, s \in S \big\}$ that has full Lebesgue's measure on the diagonal $\big\{ (x,x) \, \big| \, x \in \R \big\}$.
		\item if $Q = - \rho$, then $Z_1 = -Z_2$ almost surely and we use then similar arguments as for the previous point.
	\end{itemize}
	We have therefore:
	$$
h_f\Big(\frac{\bx \cdot \bbf{\Phi}_{\rm new}}{\sqrt{n}}, \frac{\bX^* \cdot \bbf{\Phi}_{\rm new}}{\sqrt{n}}\Big)
\xrightarrow[n \to \infty]{(d)} h_f(Z_1,Z_2) \,,
	$$
	and Lemma \ref{prop:lim_gen_f} follows from the fact that $h_f$ is bounded.
\end{proof}

Let us now define:
\begin{equation}
	H_f:
	\left|
	\begin{array}{ccc}
		[-\rho,\rho] & \to & \R \\
		q & \mapsto & \frac{1}{2} \E \big[h_f(G^{(q)})\big]
	\end{array}
	\right.
\end{equation}
where $G^{(q)} \sim \cN\big(0, 
	\begin{psmallmatrix}
		\rho & q \\
		q & \rho
	\end{psmallmatrix} \big)$.
	Notice that $H_f$ is equal to the function $\mathcal{E}_f$ on $[0,\rho]$.
	By Proposition~\ref{prop:lim_gen_f} above and Theorem~\ref{th:gen_f}, we have:
	\begin{equation}\label{eq:comp_Q}
		H_f(q^*(\alpha)) = \lim_{n \to \infty} \mathcal{E}_{f,n}(\alpha) = \E \big[H_f(Q)\big] \,.
\end{equation}
	\begin{lemma}\label{lem:absQ}
		For all $q \in [-\rho,\rho]$, $H_f(q) \geq H_f(|q|)$.
	\end{lemma}
	\begin{proof}
		Let $q \in [0,\rho]$ and $Z_0,Z_1,Z_1' \iid \cN(0,1)$.
		$$
		H_f(-q) = \frac{1}{2} \E \Big[ h_f\big(\sqrt{q} Z_0 + \sqrt{\rho -q} Z_1, - \sqrt{q} Z_0 + \sqrt{\rho-q} Z_1' \big)\Big] \,.
		$$
		Let us denote by $\E_{Z_1}$ and $\E_{Z_1'}$ the expectations with respect to $Z_1$ and $Z_1'$. By replacing $h_f$ by its expression, we have
		\begin{align*}
			H_f(-q) 
			&= \frac{1}{2} \E
			\int (f(y)-f(y'))^2 P_{\rm out}(y| \sqrt{q} Z_0 + \sqrt{\rho-q} Z_1) P_{\rm out}(y'| -\sqrt{q} Z_0 + \sqrt{\rho-q} Z_1') dy dy'
			\\
			&= \frac{1}{2} \E
			\int (f(y)-f(y'))^2 \E_{Z_1} P_{\rm out}(y| \sqrt{q} Z_0 + \sqrt{\rho-q} Z_1) \E_{Z_1'} P_{\rm out}(y'| -\sqrt{q} Z_0 + \sqrt{\rho-q} Z_1') dy dy'
			\\
			&= \frac{1}{2} \E
			\int (f(y)-f(y'))^2 \widetilde{P}_{\rm out}(y| Z_0) \widetilde{P}_{\rm out}(y'| -Z_0) dy dy' \,,
		\end{align*}
		where $\widetilde{P}_{\rm out}(y|z) = \E_{Z_1} P_{\rm out}(y|\sqrt{q}z + \sqrt{\rho-q} Z_1)$. Let now $Y$ and $Y'$ be two random variables that are independent conditionally on $Z_0$ and distributed as
		$$
		Y \sim \widetilde{P}_{\rm out}(\cdot | Z_0) 
		\qquad \text{and} \qquad
		Y' \sim \widetilde{P}_{\rm out}(\cdot | -Z_0)  \,.
		$$
		Then we have
		\begin{align*}
			H_f(-q) 
			&= \frac{1}{2} \E \Big[\big(f(Y)-f(Y')\big)^2\Big]
			= \frac{1}{2} \E \Big[\big(f(Y) - \E[f(Y)|Z_0] + \E[f(Y)|Z_0]-f(Y')\big)^2\Big]
			\\
			&= \frac{1}{2} \E \Big[\big(f(Y) - \E[f(Y)|Z_0] \big)^2\Big]
			+\frac{1}{2} \E \Big[\big( \E[f(Y)|Z_0]-f(Y')\big)^2\Big] \,,
		\end{align*}
		because $Y$ and $Y'$ are independent conditionally on $Z_0$. The conditional expectation $\E[f(Y)|Z_0]$ is $Z_0$-measurable, therefore $\E \big[( \E[f(Y)|Z_0]-f(Y'))^2\big] \geq \E \big[( \E[f(Y')|Z_0]-f(Y'))^2\big] = \E \big[( \E[f(Y)|Z_0]-f(Y))^2\big]$. We conclude
		$$
		H_f(-q) \geq \E \big[( \E[f(Y)|Z_0]-f(Y))^2\big] = H_f(q) \,.
		$$
	\end{proof}
	
	We have now all the tools needed to prove Theorem~\ref{th:overlap}. Using Lemma~\ref{lem:absQ} and \eqref{eq:comp_Q} above, we get that $\E H_f(|Q|) \leq \E H_f(Q) = H_f(q^*(\alpha))$. Since $H_f$ is equal to $\mathcal{E}_f$ on $[0,\rho]$ this gives
	\begin{equation}\label{eq:comp_E}
	\E \big[\mathcal{E}_f(|Q|) \big] \leq \mathcal{E}_f(q^*(\alpha)) \,.
\end{equation}
If $q^*(\alpha) = 0$, then Theorem \ref{th:overlap} follows simply from Proposition \ref{prop:upper_bound_overlap}. We suppose now that $q^*(\alpha) > 0$ and consider $\epsilon \in (0,q^*(\alpha))$.
	We define $p(\epsilon) = \P\big(|Q| \leq q^*(\alpha) - \epsilon\big)$. We are going to show that $p(\epsilon) = 0$.
	We assumed that $P_{\rm out}$ is informative, so by Proposition \ref{prop:gen_strict} and Proposition \ref{prop:exists_gen_dec} in Appendix \ref{appendix_scalar_channel2}, there exists a continuous bounded function $f : \R \mapsto \R$ such that $\mathcal{E}_f$ is strictly decreasing on $[0,\rho]$. In the following, $f$ is assumed to be such a function.
	We have
	\begin{align*}
		\E\big[\mathcal{E}_f(|Q|)\big]
		&=
		\E \Big[
			\bbf{1}\big(|Q| \leq q^*(\alpha) - \epsilon \big) \mathcal{E}_f(|Q|)
			+
			\bbf{1}\big(|Q| > q^*(\alpha) - \epsilon \big) \mathcal{E}_f(|Q|)
		\Big]
		\\
		&\geq
		p(\epsilon) \mathcal{E}_f(q^*(\alpha) - \epsilon)
			+
			(1-p(\epsilon)) \mathcal{E}_f(q^*(\alpha)) \,.
	\end{align*}
	because $\mathcal{E}_f$ is non-increasing and because $|Q| \leq q^*(\alpha)$ almost-surely (Lemma~\ref{lem:as_bound_Q}). Combining this with \eqref{eq:comp_E} leads to
	$$
	p(\epsilon) \mathcal{E}_f(q^*(\alpha)) \geq p(\epsilon) \mathcal{E}_f(q^*(\alpha) - \epsilon) \,.
	$$
	Since $\mathcal{E}_f$ is strictly decreasing: $\mathcal{E}_f(q^*(\alpha)) < \mathcal{E}_f(q^*(\alpha) - \epsilon)$, which implies $p(\epsilon) = 0$. This is true for all $\epsilon >0$, consequently $|Q| \geq q^*(\alpha)$ almost-surely. We get (using Lemma~\ref{lem:as_bound_Q}) that
	$$
	|Q| = q^*(\alpha)\,,  \quad \text{almost-surely.}
	$$
	We conclude that the only possible limit in law of the tight sequence $\big(|Q_n|\big)_{n \geq 1}$ is $q^*(\alpha)$. Therefore $|Q_n| \to q^*(\alpha)$ in law and in probability because $q^*(\alpha)$ is a constant.

\subsection{Denoising error: Proof of Corollary~\ref{Cor:mMMSE}}\label{appendix:i_mmse_den}

Start by noticing that the denoising error, i.e. the right hand side of \eqref{den_err}, is obtained through the I-MMSE theorem, see Proposition \ref{prop:immse}, applied to $\widetilde{i}_n\defeq I(\bX^*,\bA; \bY \, | \bbf{\Phi})/n$:
\begin{align}\label{175}
	\frac{\partial \widetilde{i}_n}{\partial \Delta^{-1}} = 
	\frac{1}{2n} {\rm MMSE}\Big(
			\varphi\Big(\frac{1}{\sqrt{n}}\bbf{\Phi}\bX^*,\bA\Big)\Big|\bbf{\Phi},\bY	
		\Big)
	\,.	
	\end{align}
This mutual information is simply computed using our main theorem. Indeed, 
\begin{align*}
\widetilde i_{\infty}\defeq \lim_{n\to\infty}\widetilde{i}_n=\lim_{n\to\infty}\frac1n H(\bY|\bbf{\Phi})-\lim_{n\to\infty}\frac1n H(\bY|\bbf{\Phi},\bX^*,\bA)	=-f_{\infty}-\lim_{n\to\infty}\frac1n H(\bY|\bbf{\Phi},\bX^*,\bA)\,.
\end{align*}
One can simply check that $\lim_{n\to\infty} H(\bY|\bbf{\Phi},\bX^*,\bA)/n=\alpha\ln(2\pi\Delta e)/2$ by similar computations as in the proof of Corollary \ref{cor:mi}. Therefore, defining
	\begin{align*}
	\widetilde{i}_{\rm RS}(q,\Delta) \defeq -\frac{\alpha}{2}\ln(2\pi \Delta e) - \alpha \Psi_{P_{\rm out}} (q) - \inf_{r \geq 0} \Big\{ \psi_{P_0}(r) - \frac{qr}{2} \Big\}\,,	
	\end{align*}
we have $\widetilde i_{\infty} = \sup_{q \in [0,\rho]} \widetilde{i}_{\rm RS}(q,\Delta)$ from Theorem \ref{th:RS_1layer}.

One can verify easily that $\widetilde{i}_n$ is a concave differentiable function of $\Delta^{-1}$ (this is again related to the I-MMSE theorem). Thus its limit $\widetilde i_{\infty}$ is also a concave function of $\Delta^{-1}$. Therefore, a standard analysis lemma gives that the derivative of $\widetilde i_n$ w.r.t.\ $\Delta^{-1}$ converges to the derivative of $\widetilde i_{\infty}$ at every point at which $\widetilde i_{\infty}$ is differentiable (i.e.\ almost every points, by concavity): $\lim_{n\to\infty} \partial_{\Delta^{-1}}\widetilde  i_{n}=\partial_{\Delta^{-1}}\widetilde  i_{\infty}=\partial_{\Delta^{-1}}\sup_{q \in [0,\rho]} \widetilde{i}_{\rm RS}(q,\Delta)$. The first limit is given by the limit of the right hand side of \eqref{175}. It thus remains to compute $\partial_{\Delta^{-1}}\sup_{q \in [0,\rho]} \widetilde{i}_{\rm RS}(q,\Delta)$.

Assume for a moment that the $\partial_{\Delta^{-1}}$ and $\sup_{q \in [0,\rho]}$ operations commute. Then we need to compute $\partial_{\Delta^{-1}}\Psi_{P_{\rm out}}(q)$; this follows from the I-MMSE theorem. Indeed, if we denote $S = \varphi(\sqrt{q}\,V + \sqrt{\rho-q\,}W^*,\bA)$, notice that $\Psi_{P_{\rm out}}(q) = - I(S;S+\sqrt{\Delta}Z\,|\,V)-  \ln(2\pi e \Delta)/2$, because $\Psi_{P_{\rm out}}(q) = - H(S+\sqrt{\Delta}Z \,| \,V)$ and $I(S + \sqrt{\Delta} Z ; S \, | \, V) = H(S+\sqrt{\Delta}Z \,| \,V) - H(S+\sqrt{\Delta}Z \,| \,V,S) = - \Psi_{P_{\rm out}}(q) -  \ln(2\pi e \Delta)/2$. Therefore
	\begin{align*}
\frac{\partial \Psi_{P_{\rm out}}(q)}{\partial \Delta^{-1}}
&= 
\frac{\Delta}{2} - \frac{\partial }{\partial \Delta^{-1}}I(S + \sqrt{\Delta} Z ; S \, | \, V) 
= 
\frac{\Delta}{2} - \frac{1}{2}{\rm MMSE}(S \, | \, V,S + \sqrt{\Delta} Z) 
\\
&=
\frac{\Delta}{2} - \frac{1}{2} \Big( \E \big[\varphi(\sqrt{\rho}\,V,\bA)^2\big] - \E \big[\big\langle \varphi(\sqrt{q}\,V + \sqrt{\rho - q}\,w,\ba)\rangle_{\rm sc}^2\big] \Big) \,.
	\end{align*}
	Consequently,
	$$
	\frac{\partial \widetilde{i}_{\rm RS}(q,\Delta)}{\partial \Delta^{-1}} =  \frac{\alpha}{2} \Big( \E \big[\varphi(\sqrt{\rho}\,V,\bA)^2\big] - \E \big[\big\langle \varphi(\sqrt{q}\,V + \sqrt{\rho - q}\,w,\ba)\rangle_{\rm sc}^2\big] \Big) \,.
	$$
	Now, Theorem 1 from~\cite{milgrom2002envelope} gives that at every $\Delta^{-1}$ at which $i_{\infty}$ is differentiable
	\begin{align*}
	\frac{\partial \widetilde i_{\infty}}{\partial \Delta^{-1}}
	=
	\frac{\partial}{\partial \Delta^{-1}}
	\sup_{q \in [0,\rho]} \widetilde i_{\rm RS}(q,\Delta)
	=
	\frac{\alpha}{2} \Big(
			\E \big[\varphi(\sqrt{\rho}\,V,\bA)^2\big] - \E \big[\big\langle \varphi(\sqrt{q^*}\,V + \sqrt{\rho - q^*}\,w,\ba)\rangle_{\rm sc}^2\big] \Big) 
	\end{align*}
	where $q^* \in [0,\rho]$ is a point where the supremum above is achieved, and thus corresponds to an optimal couple in \eqref{eq:rs_formula}.
	As explained above, $\lim_{n\to\infty} \partial_{\Delta^{-1}}\widetilde  i_{n}=\partial_{\Delta^{-1}}\widetilde  i_{\infty}$ at every $\Delta^{-1}$ at which $\widetilde i_{\infty}$ is differentiable, which concludes the proof.

\begin{appendices}
		\section{Some technicalities}\label{appendix_misc}

		\subsection{The Nishimori identity}\label{appendix-nishimori}

\begin{proposition}[Nishimori identity] \label{prop:nishimori}
	Let $(\bX,\bY) \in \R^{n_1} \times \R^{n_2}$ be a couple of random variables. Let $k \geq 1$ and let $\bX^{(1)}, \dots, \bX^{(k)}$ be $k$ i.i.d.\ samples (given $\bY$) from the conditional distribution $P(\bX=\cdot\, | \bY)$, independently of every other random variables. Let us denote $\langle - \rangle$ the expectation operator w.r.t.\ $P(\bX= \cdot\, | \bY)$ and $\E$ the expectation w.r.t.\ $(\bX,\bY)$. Then, for all continuous bounded function $g$ we have
	\begin{align}
	\E \langle g(\bY,\bX^{(1)}, \dots, \bX^{(k)}) \rangle
	=
	\E \langle g(\bY,\bX^{(1)}, \dots, \bX^{(k-1)}, \bX) \rangle\,.	
	\end{align}
\end{proposition}

\begin{proof}
	This is a simple consequence of Bayes formula.
	It is equivalent to sample the couple $(\bX,\bY)$ according to its joint distribution or to sample first $\bY$ according to its marginal distribution and then to sample $\bX$ conditionally to $\bY$ from its conditional distribution $P(\bX=\cdot\,|\bY)$. Thus the $(k+1)$-tuple $(\bY,\bX^{(1)}, \dots,\bX^{(k)})$ is equal in law to $(\bY,\bX^{(1)},\dots,\bX^{(k-1)},\bX)$.
\end{proof}

		\subsection{Unicity of the optimizer $q^*$ of the replica formula: Proof of Proposition~\ref{prop:q_star}}\label{app:proof_q_star}
	The function
	\begin{equation}\label{eq:def_h}
	h: \alpha \mapsto
	\inf_{q \in [0,\rho]} 
	\Big\{
	\alpha \mathcal{I}_{P_{\rm out}}(q) 
	+ \sup_{r \geq 0} \big\{ I_{P_0}(r)  -\frac{r}{2}(\rho - q) \big\}
\Big\}
	\end{equation}
	is concave (as an infimum of linear functions). An ``envelope'' theorem (Corollary~4 from \cite{milgrom2002envelope}) gives that $h$ is differentiable at $\alpha$ if and only if 
	$$
	\Big\{ \mathcal{I}_{P_{\rm out}}(q) \, \Big| \, q \ \text{minimizer of} \ \eqref{eq:def_h} \Big\}
	$$
	is a singleton. 
	We assumed that $P_{\rm out}$ is informative, so Proposition~\ref{prop:psi_stricly_monoton} gives that $\mathcal{I}_{P_{\rm out}}$ is strictly decreasing. We obtain thus that the set of points at which $h$ is differentiable is exactly $D^*$. 
	Since $h$ is concave, $D^*$ is equal to $\R_+^*$ minus a countable set.
	Corollary~4 from \cite{milgrom2002envelope} gives also that
	$h'(\alpha) = \mathcal{I}_{P_{\rm out}}(q^*(\alpha))$, for all $\alpha \in D^*$.
	The function $h$ is concave, so its derivative $h'$ is non-increasing. Since $\mathcal{I}_{P_{\rm out}}$ is strictly decreasing, we obtain that $\alpha \in D^* \mapsto q^*(\alpha)$ is non-decreasing.

	Let now $\alpha_0 \in D^*$. By concavity of $h$, $h'(\alpha) \to h'(\alpha_0)$ when $\alpha \in D^* \to \alpha_0$. Therefore:
	$$
	\mathcal{I}_{P_{\rm out}}(q^*(\alpha)) \xrightarrow[\alpha \in D \to \alpha_0]{} \mathcal{I}_{P_{\rm out}}(q^*(\alpha_0))
	$$
	which implies $q^*(\alpha) \to q^*(\alpha_0)$ by strict monotonicity of $\mathcal{I}_{P_{\rm out}}$.

		\subsection{Continuity properties of the mutual information}

We establish in this section two continuity properties of the mutual information, namely Proposition \ref{prop:free_wasserstein} and Corollary \ref{cor:cont_I_discrete}. Recall definition \eqref{eq:def_mmse} of the MMSE function.
The following proposition comes from \cite{GuoShamaiVerdu_IMMSE} 
and will be repeatedly used in the sequel.
\begin{proposition}[I-MMSE theorem, \cite{GuoShamaiVerdu_IMMSE,guo2011estimation}]\label{prop:immse}
	Let $P_X$ be a probability distribution over $\R^n$ that admits a finite second moment. Let $\bX \sim P_X$ and $\bZ \sim \cN(0,\bbf{I}_n)$ be independent random variables. Then the function
	$$
	I_{P_X}: \left|
	\begin{array}{ccc}
		\R_+ & \to & \R \\
		\lambda & \mapsto & I(\bX; \sqrt{\lambda} \bX + \bZ)
	\end{array}
	\right.
	$$
	is concave, continuously differentiable over $\R_+$, with derivative given by
	$$
	I_{P_X}'(\lambda) = \frac{1}{2} \MMSE(\bX \, | \, \sqrt{\lambda} \bX + \bZ) = \frac{1}{2} \E \Big[\big\| \bX - \E[\bX|\sqrt{\lambda} \bX + \bZ] \big|^2 \Big]\,.
	$$
\end{proposition}

\noindent\textbf{Remark:} We will often apply Proposition \ref{prop:immse} in a ``conditional fashion''. Let $\bU$ be some random variable independent from $\bZ$, then
$$
\frac{\partial}{\partial \lambda} I(\bX; \sqrt{\lambda} \bX + \bZ|\bU) = 
\frac{1}{2} \MMSE(\bX \, | \, \sqrt{\lambda} \bX + \bZ,\bU) = \frac{1}{2} \E \Big[\big\| \bX - \E[\bX|\sqrt{\lambda} \bX + \bZ, \bU] \big|^2 \Big]\,.
$$

	\begin{proposition}\label{prop:free_wasserstein}
	Let $P_{1}$ and $P_{2}$ be two probability distributions on $\R^n$, that admits a finite second moment. We denote by $W_2(P_1,P_2)$ the Wasserstein distance of order 2 between $P_1$ and $P_2$.
		$$
		\big| I(\bX_1; \bX_1 + \bZ) - I(\bX_2;\bX_2 + \bZ) \big| \leq \big(\sqrt{\E\|\bX_1\|^2} + \sqrt{\E\|\bX_2\|^2}\big)W_2(P_1,P_2) \,.
		$$
	\end{proposition}
A similar result was proved in~\cite{wu2012functional} but with a weaker bound for the $W_2$ distance.

	\begin{proof}
		Let $\epsilon > 0$. Let us fix a coupling of $\bX_1 \sim P_1$ and $\bX_2 \sim P_2$ such that
		$$
		\big(\E \| \bX_1 - \bX_2 \|^2\big)^{1/2} \leq W_2(P_1,P_2) + \epsilon \,.
		$$
		Let us consider for $t_1,t_2 \in [0,1]$ the observation model
		$$
		\begin{cases}
			\bY^{(t_1)}_1 &=\  \sqrt{t_1} \bX_1  + \bZ_1 \,, \\
			\bY^{(t_2)}_2 &=\  \sqrt{1-t_2} \bX_2 + \bZ_2 \,,
		\end{cases}
		$$
		where $\bZ_1,\bZ_2 \iid \cN(0,\bbf{I}_n)$ are independent from $(\bX_1,\bX_2)$. 
		Define $J(t_1,t_2) = I(\bX_1,\bX_2; \bY_1^{(t_1)},\bY^{(t_2)}_2)$ and $I(t) = J(t,t)$.
		Let us now differentiate $J$ with respect to $t_1$. Using the chain rule for the mutual information,
		\begin{align*}
			J(t_1,t_2) &= I(\bX_1,\bX_2; \bY_2^{(t_2)}) + I(\bX_1,\bX_2; \bY_1^{(t_1)} | \bY_2^{(t_2)})
		= I(\bX_1,\bX_2; \bY_2^{(t_2)}) + I(\bX_1; \bY_1^{(t_1)} | \bY_2^{(t_2)}) + I(\bX_2; \bY_1^{(t_1)} | \bX_1,\bY_2^{(t_2)})
		\\
		&= I(\bX_1,\bX_2; \bY_2^{(t_2)}) + I(\bX_1; \bY_1^{(t_1)} | \bY_2^{(t_2)}) 
		\end{align*}
		because, conditionally on $\bX_1$, $\bX_2$ and $\bY_1^{(t_1)}$ are independent. The quantity $I(\bX_1,\bX_2; \bY_2^{(t_2)})$ does not depend on $t_1$, therefore by the ``I-MMSE relation'' from Proposition \ref{prop:immse}:
		$$
		\frac{\partial J}{\partial t_1} (t_1,t_2) = \frac{1}{2}\MMSE(\bX_1 | \bY^{(t_1)}_1, \bY^{(t_2)}_2)
		$$
		and similarly
		$$
		\frac{\partial J}{\partial t_2} (t_1,t_2) = -\frac{1}{2}\MMSE(\bX_2 | \bY^{(t_1)}_1, \bY^{(t_2)}_2) \,.
		$$
		Let us write $\bbf{E}_i = \E[\bX_i | \bY_1^{(t)},\bY^{(t)}_2]$ for $i=1,2$, then
		$$
		I'(t) = 
 \frac{1}{2}\MMSE(\bX_1 | \bY^{(t)}_1, \bY^{(t)}_2)
-\frac{1}{2}\MMSE(\bX_2 | \bY^{(t)}_1, \bY^{(t)}_2)
= \frac{1}{2}\E \left[
	\big\|\bX_1 - \bbf{E}_1 \big\|^2
	-
	\big\|\bX_2 - \bbf{E}_2 \big\|^2
\right]
		$$
		so that 
		\begin{align*}
			|I'(t)| &= 
\frac{1}{2}\E \left[
	\big(
	\big\|\bX_1 - \bbf{E}_1 \big\|
	+
\big\|\bX_2 - \bbf{E}_2 \big\| \big)
	\big(
	\big\|\bX_1 - \bbf{E}_1 \big\|
	-
\big\|\bX_2 - \bbf{E}_2 \big\| \big)
\right]
\\
&\leq
\frac{1}{2}\E \left[
	\big\|\bX_1 \big\|^2
	+
\big\|\bX_2  \big\|^2 \right]^{1/2}
\E \left[
	\big(
	\big\|\bX_1 - \bbf{E}_1 \big\|
	-
\big\|\bX_2 - \bbf{E}_2 \big\| \big)^2
\right]^{1/2}
\\
&\leq
\frac{1}{2}\E \left[
	\big\|\bX_1 \big\|^2
	+
\big\|\bX_2  \big\|^2 \right]^{1/2}
\E \left[
	\big\|\bX_1 - \bX_2 + \bbf{E}_2- \bbf{E}_1 \big\|^2
\right]^{1/2}
\\
&\leq
\frac{1}{2} \big( \sqrt{\E \|\bX_1 \|^2}
	+
\sqrt{\E \|\bX_2 \|^2} \big)
		\E \left[
	2\|\bX_1 - \bX_2\|^2 + 2 \| \bbf{E}_2- \bbf{E}_1\|^2
\right]^{1/2}
\\
&\leq
\big( \sqrt{\E \|\bX_1 \|^2}
	+
\sqrt{\E \|\bX_2 \|^2} \big)
\big(W_2(P_1,P_2) + \epsilon \big) \,.
		\end{align*}
		We obtain the result by letting $\epsilon \to 0$.
	\end{proof}

	\begin{proposition}\label{prop:cont_I_discrete}
		Let $P_U$ be a probability distribution over $\N^m$ that admits a finite second moment. Let $\bU \sim P_U$ and $\bZ \sim \cN(0,\bbf{I}_m)$ be two independent random variables. Then $H(\bU) = - \sum_{\bbf{n} \in \N^m} P_U(\bbf{n}) \ln P_U(\bbf{n})$ is finite and for all $\Delta \in (0,1]$,
		$$
		\big| I(\bU;\bU + \sqrt{\Delta}\bZ) - H(\bU) \big| \leq 
 48 m e^{-1/(16\Delta)} \,.
		$$
	\end{proposition}

	\begin{proof}
		Let us define for $\Delta > 0$, $h(\Delta) = I(\bU;\bU + \sqrt{\Delta}\bZ) = I_{P_U}(\Delta^{-1})$.
		By Proposition \ref{prop:immse} we have for all $\Delta > 0$,
		\begin{equation}\label{eq:h_prime}
		h'(\Delta) = - \frac{1}{2\Delta^2} \MMSE(\bU \, | \, \bU + \sqrt{\Delta} \bZ) \,.
	\end{equation}
		We are now going to upper bound $\MMSE(\bU \, | \, \bU + \sqrt{\Delta} \bZ)$ by considering the following estimator:
		$$
		\widehat{\theta}_i = \argmin_{u \in \N} |u - U_i + \sqrt{\Delta} Z_i|,
		$$
		for all $i \in \{1, \dots, m\}$. Note that $\widehat{\theta}_i$ is well-defined almost-surely since there is a.s.\ a unique minimizer above. We have
		$$
		\P(\widehat{\theta}_i \neq U_i) \leq \P\big(\sqrt{\Delta}|Z_i| \geq 1/2 \big) = 2 \P\Big(\cN(0,1) \geq \frac{1}{2 \sqrt{\Delta}}\Big) 
		\leq 2 \frac{1}{\sqrt{2\pi}} 2 \sqrt{\Delta} e^{-1/(8 \Delta)}
		\leq 2 \sqrt{\Delta} e^{-1/(8 \Delta)}\,,
		$$
		by usual bounds on the Gaussian cumulative distribution function. We have then
		\begin{align*}
			\MMSE(\bU \, | \, \bU + \sqrt{\Delta} \bZ)
			&\leq \E \| \bU - \widehat{\theta}\|^2
			= \sum_{i=1}^m \E (U_i - \widehat{\theta}_i)^2
			= \sum_{i=1}^m \E\Big[ \bbf{1}(\widehat{\theta}_i \neq U_i)(U_i - \widehat{\theta}_i)^2 \Big]
			\\
			& \leq \sum_{i=1}^m 2 \E\Big[ \bbf{1}(\widehat{\theta}_i \neq U_i)(U_i - (U_i + \sqrt{\Delta} Z_i))^2 \Big]
			+ 2 \E\Big[ \bbf{1}(\widehat{\theta}_i \neq U_i)(U_i + \sqrt{\Delta} Z_i - \widehat{\theta}_i)^2 \Big]
			\\
			& \leq \sum_{i=1}^m 2 \E\Big[ \bbf{1}(\widehat{\theta}_i \neq U_i)\Delta Z_i^2 \Big]
			+ \frac{1}{2} \E\Big[ \bbf{1}(\widehat{\theta}_i \neq U_i) \Big]
			\\
			& \leq \sum_{i=1}^m 2 \Delta\P(\widehat{\theta}_i \neq U_i)^{1/2} \E[Z_i^4]^{1/2} + \frac{1}{2} \P(\widehat{\theta}_i \neq U_i)
			\\
			&\leq m e^{-1 /(16 \Delta)} \Big(2 \sqrt{6} \Delta^{5/4} + \sqrt{\Delta}  \Big) 
			\leq 6 m e^{-1/(16 \Delta)}
		\end{align*}
		for $\Delta \leq 1$. Plugging this inequality in \eqref{eq:h_prime}, we obtain for all $\Delta \in (0,1]$,
		\begin{equation}\label{eq:bound_h_prime}
		|h'(\Delta)| \leq \frac{3 m}{\Delta^2} e^{-1/(16 \Delta)} \,.
	\end{equation}
		Since $h(1)$ is finite and $\int_0^1 \frac{e^{-1/(16 \Delta)}}{\Delta^2} d\Delta < +\infty$ we obtain that
		\begin{equation} \label{eq:h_bounded}
		\sup_{\Delta \in (0,1]} |h(\Delta)| < + \infty \,.
	\end{equation}
	By definition of $h$:
	\begin{equation}\label{eq:def_hh}
		h(\Delta) = 
		I(\bU;\bU + \sqrt{\Delta}\bZ) = -\frac{m}{2} - \E \ln \sum_{\bU \in \N^m} P_U(\bU) \exp \Big(- \frac{1}{2\Delta} \| \bU + \sqrt{\Delta} \bZ - \bU \|^2\Big) \,.
	\end{equation}
		By the previous equality and \eqref{eq:h_bounded}, the family of (non-negative) random variables
		$$
		\left(
		-\ln \sum_{\bU \in \N^m} P_U(\bU) \exp \Big(- \frac{1}{2\Delta} \| \bU + \sqrt{\Delta} \bZ - \bU \|^2\Big) \right)_{\Delta \in (0,1]}
		$$
		is bounded in $L^1$. Notice that (by dominated convergence) 
		$$
		-\ln \sum_{\bU \in \N^m} P_U(\bU) \exp \Big(- \frac{1}{2\Delta} \| \bU + \sqrt{\Delta} \bZ - \bU \|^2\Big) \xrightarrow[\Delta \to 0]{}
		- \ln \Big( P_U(\bU) e^{-\frac{1}{2}\|\bZ\|^2} \Big)
		= \frac{1}{2} \| \bZ\|^2 - \ln P_U(\bU)
		$$
		almost-surely. This gives (by Fatou's Lemma) that this almost-sure limit is integrable and thus that $H(\bU) = - \E \ln P_U(\bU)$ is finite. Let us now show that $h(\Delta) \xrightarrow[\Delta \to 0]{} H(\bU)$. We have almost-surely
		$$
		\ln \Big(P_U(\bU) e^{-\frac{1}{2} \|\bZ\|^2}\Big)
		\leq \ln \sum_{\bU \in \N^m} P_U(\bU) \exp \Big(- \frac{1}{2\Delta} \| \bU + \sqrt{\Delta} \bZ - \bU \|^2\Big) \leq 0 \,.
		$$
		Since we now know that the left-hand side is integrable (because $H(\bU)$ is finite), we can apply the dominated convergence theorem to obtain that
		$$
\E \ln \sum_{\bU \in \N^m} P_U(\bU) \exp \Big(- \frac{1}{2\Delta} \| \bU + \sqrt{\Delta} \bZ - \bU \|^2\Big)
\xrightarrow[\Delta \to 0]{}
\E \ln \Big(P_U(\bU) e^{-\frac{1}{2} \|\bZ\|^2}\Big) = H(\bU) - \frac{m}{2} \,,
		$$
		which combined with \eqref{eq:def_hh} gives $h(\Delta) \xrightarrow[\Delta \to 0]{} H(\bU)$. Now, using the bound on the derivative of $h$ \eqref{eq:bound_h_prime} we conclude that for all $\Delta \in (0,1]$,
		$$
		|h(\Delta)- H(\bU)| \leq 3 m \int_0^{\Delta} \frac{e^{-1/(16 t)}}{t^2} dt
		= 3 m \Big[16 e^{-1/(16t)}\Big]_0^\Delta = 48 m e^{-1/(16\Delta)} \,.
		$$
	\end{proof}

	\begin{corollary}\label{cor:cont_I_discrete}
		Let $\bU$ be a random variable over $\N^m$ with finite second moment, let $\bX$ be a random variable over $\R^n$ and let $\bZ \sim \cN(0,\bbf{I}_m)$. We assume $(\bU, \bX)$ to be independent from $\bZ$. Then, for all $\Delta \in (0,1]$,
		$$
		\big|I(\bX; \bU + \sqrt{\Delta} \bZ) - I(\bX; \bU)\big|
		\leq 100 m e^{-1/(16\Delta)} \,.
		$$
	\end{corollary}
	\begin{proof}
	We have by the chain rule of the mutual information:
	$$
	I(\bU;\bU + \sqrt{\Delta} \bZ)
	=
	I(\bU,\bX;\bU + \sqrt{\Delta} \bZ)
	=
	I(\bX;\bU + \sqrt{\Delta} \bZ)
	+
	I(\bU;\bU + \sqrt{\Delta} \bZ| \bX) \,.
	$$
	By applying Proposition \ref{prop:cont_I_discrete} twice, we get
	$$
	|I(\bU;\bU + \sqrt{\Delta} \bZ) - H(\bU)| , \
	|I(\bU;\bU + \sqrt{\Delta} \bZ|\bX) - H(\bU|\bX)| 
		\leq 48 m e^{-1/(16\Delta)} \,.
	$$
	Since $I(\bX;\bU) = H(\bU) - H(\bU|\bX)$ we obtain the desired inequality.
	\end{proof}

		\subsection{A simple consequence of hypotheses \ref{hyp:third_moment}-\ref{hyp:um}-\ref{hyp:phi_general}-\ref{hyp:cont_pp}}

\begin{proposition}\label{prop:varphi_l2}
	Assume that hypotheses \ref{hyp:third_moment}-\ref{hyp:um}-\ref{hyp:phi_general}-\ref{hyp:cont_pp} hold. Then there exists $\eta > 0$ such that
	$$
	\E \big[\varphi(\sqrt{\rho} Z, \bbf{A})^{2+\eta} \big] < \infty \,,
	$$
	where the expectation above is with respect to $(Z,\bbf{A}) \sim \cN(0,1) \otimes P_A$.
\end{proposition}
\begin{proof}
	By the Central Limit Theorem (using the fact that the third moments of $(X^*_i \Phi_{1,i})$ are bounded with $n$, because of hypotheses~\ref{hyp:third_moment} and~\ref{hyp:phi_general}) we have
	$
	\big([\bbf{\Phi}\bX^*]_1/\sqrt{n},\bA_1\big)\xrightarrow[n \to \infty]{(d)}
	(\sqrt{\rho} G,\bA_1)$.
	This implies that 
	\begin{equation}\label{eq:conv_law2}
		\varphi\left(\frac{[\bbf{\Phi}\bX^*]_1}{\sqrt{n}},\bA_1\right)
		\xrightarrow[n \to \infty]{(d)}
		\varphi(\sqrt{\rho} G,\bA_1) \,,
	\end{equation}
	because $\varphi(\cdot, \bbf{A}_1)$ is almost-surely continuous almost-everywhere, by assumption~\ref{hyp:cont_pp}.
	The sequence of random variables $\big(\varphi([\bbf{\Phi}\bX^*]_1/\sqrt{n},\bA_1)\big)_n$ is by assumption~\ref{hyp:um} bounded in $L^{2+\eta}$ for some $\eta>0$. By \eqref{eq:conv_law2} we conclude that $\E[\varphi(\sqrt{\rho}G,\bA_1)^{2+\eta}] < \infty$.
\end{proof}

		\subsection{Derivative of the interpolating free entropy: Proof of Proposition~\ref{prop:der_f_t}}\label{appendix_interpolation}
Recall $u'_{y}(x)$ is the $x$-derivative of $u_{y}(x)=\ln P_{\rm out}(y|x)$. Moreover denote $P_{\rm out}'(y|x)$ and $P_{\rm out}''(y|x)$ the first and second $x$-derivatives, respectively, of $P_{\rm out}(y|x)$. We will first prove that for all $t \in (0,1)$
\begin{align}
	\frac{df_{n,\epsilon}(t)}{dt} = &- \frac{1}{2} 
		\E \Big\langle 
			\Big(
				\frac{1}{n} \sum_{\mu=1}^{m}u'_{Y_{t,\mu}}(S_{t,\mu}) u'_{Y_{t,\mu}}(s_{t,\mu})
				- r(t)
			\Big)
			\big(
				Q - q(t)
			\big)
		\Big\rangle_{n,t,\epsilon}
		+ \frac{r(t)}{2}(q(t)-\rho) -\frac{A_n}{2}\,,
		\label{eq:der_f_t_raw}	
\end{align}
where recall $Q\defeq\sum_{i=1}^{n} X^*_i x_i/n$ and
\begin{align}
	A_{n,\epsilon}\defeq
		\E \Big[
			\frac{1}{\sqrt{n}} \sum_{\mu=1}^{m} \frac{P_{\rm out}''(Y_{t,\mu} | S_{t,\mu})}{P_{\rm out}(Y_{t,\mu} | S_{t,\mu})} 
			\Big( \frac{1}{\sqrt{n}} \sum_{i=1}^{n} \big((X^*_i)^2 - \rho\big) \Big) \frac{1}{n} \ln \cZ_{t,\epsilon} \label{An}
		\Big]	\,.
\end{align}
Once this is done, we will prove that $A_{n,\epsilon}$ goes to $0$ as $n \to \infty$ uniformly in $t \in [0,1]$, in order to obtain Proposition~\ref{prop:der_f_t}.
\subsubsection{Proof of \eqref{eq:der_f_t_raw}}
Recall definition \eqref{ft} which becomes, when written as a function of the interpolating Hamiltonian \eqref{interpolating-ham},
	\begin{align}
		f_{n,\epsilon}(t)= \frac{1}{n} \E_{\boldsymbol{\Phi},\bV}\int d\bY_t d\bY_t' dP_0(\bX^*) {\cal D}\bW^*  e^{-\cH_{t,\epsilon}(\bX^*,\bW^*;\bY_t,\bY'_t,\boldsymbol{\Phi},\bV)} 
		\ln \int dP_0(\bx) {\cal D}\bw\, e^{-\cH_{t,\epsilon}(\bx,\bw;\bY_t,\bY'_t,\boldsymbol{\Phi},\bV)}\,.
\end{align}
We will need the Hamiltonian $t$-derivative ${\cal H}'_{t,\epsilon}$ given by
	\begin{align}
		\cH_{t,\epsilon}'(\bX^*,\bW^*;\bY_t,\bY'_t,\boldsymbol{\Phi},\bV)
	&=
	- \sum_{\mu=1}^{m}
	\frac{dS_{t,\mu}}{dt} u'_{Y_{t,\mu}}( S_{t,\mu})
	- \frac{r(t)}{2\sqrt{R_1(t)}} \sum_{i=1}^{n} X_i^* (Y'_{t,i}  - \sqrt{R_1(t)} X_i^*)\,.  \label{117_}
	\end{align}
	The derivative of the interpolating free entropy thus reads, for $0 < t < 1$,
	\begin{align}
		\frac{df_{n,\epsilon}(t)}{dt} = -\underbrace{\frac{1}{n}  \E \big[\cH_{t,\epsilon}'(\bX^*,\bW^*;\bY_t,\bY'_t,\boldsymbol{\Phi},\bV)\ln \cZ_{t,\epsilon}
		 \big] }_{T_1}
		 - \underbrace{\frac{1}{n} \E \big\langle \cH_{t,\epsilon}'(\bx,\bw;\bY_t,\bY'_t,\boldsymbol{\Phi},\bV) \big\rangle_{n,t,\epsilon} }_{T_2}\label{106}
	\end{align}
	where recall the definition of $\cZ_{t,\epsilon}=\cZ_{t,\epsilon}(\bY_t,\bY'_t,\boldsymbol{\Phi},\bV)$ given by \eqref{Zt}. 

Let us compute $T_1$. Let $1 \leq \mu \leq m$. Let us start with the following term
	\begin{align}
		\E \Big[\frac{dS_{t,\mu}}{dt} &u'_{Y_{t,\mu}}(S_{t,\mu})  \ln \cZ_{t,\epsilon}  \Big] \nn=\,
	&\frac{1}{2}
	\E \Big[\Big(
			- \frac{[\bbf{\Phi} \bX^*]_{\mu}}{\sqrt{n (1-t)}}
			+ \frac{q(t)}{ \sqrt{R_2(t)}} V_{\mu}
			+ \frac{\rho - q(t)}{\sqrt{\rho t - R_2(t) + 2s_n}} W^*_{\mu}
	\Big)u'_{Y_{t,\mu}}(S_{t,\mu}) \ln \cZ_{t,\epsilon}  \Big]\,. \label{107}
	\end{align}
	Let us compute the first term of the right-hand side of the last identity. By Gaussian integration by parts w.r.t $\Phi_{\mu i}$ we obtain
	\begin{align} 
		&\frac{1}{\sqrt{n(1-t)}}\E\big[
			[\bbf{\Phi} \bX^*]_{\mu}
			u_{Y_{t,\mu}}' ( S_{t,\mu} )
			\ln \cZ_{t,\epsilon}
		\big] \nn
   =\, &\frac{1}{\sqrt{n(1-t)}}\sum_{i=1}^n\E \Big[\int d\bY_t d\bY'_t e^{-\cH_{t,\epsilon}(\bX^*,\bW^*;\bY_t,\bY'_t,\boldsymbol{\Phi},\bV)}
			\Phi_{\mu i} X_i^*
			u_{Y_{t,\mu}}' ( S_{t,\mu} )
		\ln \cZ_{t,\epsilon} \Big]
		\nn
		=\,&\frac{1}{n} \sum_{i=1}^{n} 
		\Big(
			\E\big[
				(X_i^*)^2
				\big(
					u_{Y_{t,\mu}}''( S_{t,\mu} )
					+
					u_{Y_{t,\mu}}' ( S_{t,\mu})^2
				\big)
				\ln \cZ_{t,\epsilon}
			\big]
			+ \E\big\langle
				X_i^* x_i 
				u_{Y_{t,\mu}}' ( S_{t,\mu})
				u_{Y_{t,\mu}}' ( s_{t,\mu} )
			\big\rangle_{n,t,\epsilon}
		\Big)
		\nonumber
		\\
		=\,&
		\E\Big[
			\frac{1}{n} \sum_{i=1}^{n} 
			(X_i^*)^2
			\frac{P_{\rm out}''(Y_{t,\mu} | S_{t,\mu})}{P_{\rm out}(Y_{t,\mu} | S_{t,\mu})}
			\ln \cZ_{t,\epsilon}
		\Big]
		+ 
		\E\Big\langle
			\frac{1}{n} \sum_{i=1}^{n} 
			X_i^* x_i 
			u_{Y_{t,\mu}}' ( S_{t,\mu} )
			u_{Y_{t,\mu}}' ( s_{t,\mu} )
		\Big\rangle_{n,t,\epsilon}\,,
		\label{eq:compA1}
	\end{align}
	where we used the identity 
	\begin{align}
	u_{Y_{t,\mu}}'' ( x )
	+
	u_{Y_{t,\mu}}' ( x )^2
	= \frac{P_{\rm out}''(Y_{t,\mu} | x)}{P_{\rm out}(Y_{t,\mu} | x)}\,. 	
	 \end{align} 
	We now compute the second term of the right hand side of \eqref{107}. Using again Gaussian integrations by parts but this time w.r.t $V_\mu,W_\mu^*\iid {\cal N}(0,1)$ as well as the previous formula, we obtain similarly 
	\begin{align}\label{eq:compA2}
		&\E \Big[
			\Big( \frac{q(t)}{\sqrt{R_2(t)}} V_{\mu} + \frac{\rho -q(t)}{\sqrt{\rho t - R_2(t) + 2s_n}} W^*_{\mu}\Big)
			u_{Y_{t,\mu}}' ( S_{t,\mu} )
			\ln \cZ_{t,\epsilon}
		\Big]\nn
		=\,&\E \Big[\int d\bY_t d\bY'_t e^{-\cH_{t,\epsilon}(\bX^*,\bW^*;\bY_t,\bY'_t,\boldsymbol{\Phi},\bV)} 
			\Big( \frac{q(t)}{\sqrt{R_2(t)}} V_{\mu} + \frac{\rho -q(t)}{\sqrt{\rho t - R_2(t) + 2s_n}} W^*_{\mu}\Big)
			u_{Y_{t,\mu}}' ( S_{t,\mu} )
		\ln \cZ_{t,\epsilon}  \Big]
		 \nn
		=\,&\E\Big[
			\rho \frac{P_{\rm out}''(Y_{t,\mu} | S_{t,\mu})}{P_{\rm out}(Y_{t,\mu} | S_{t,\mu})} 
			\ln \cZ_{t,\epsilon} 
		\Big]
		+\E \big\langle q(t) u'_{Y_{t,\mu}}(S_{t,\mu}) u'_{Y_{t,\mu}}(s_{t,\mu})
		\big\rangle_{n,t,\epsilon}\,.
	\end{align}
	Combining equations \eqref{107}, \eqref{eq:compA1} and \eqref{eq:compA2} together, we have
	\begin{align*}
		&-\E \Big[\frac{dS_{t,\mu}}{dt} u'_{Y_{t,\mu}}(S_{t,\mu}) \ln \cZ_{t,\epsilon}  \Big] 
		\nn
	 = \,&\frac{1}{2}
		\E\Big[
			\frac{P_{\rm out}''(Y_{t,\mu} | S_{t,\mu})}{P_{\rm out}(Y_{t,\mu} | S_{t,\mu})} 
			\Big(
				\frac{1}{n} \sum_{i=1}^n (X_i^*)^2 - \rho
			\Big)
			\ln \cZ_{t,\epsilon}
		\Big]
		+\frac{1}{2} \E
		\Big\langle
			\Big(
				\frac{1}{n} \sum_{i=1}^{n} 
				X_i^* x_i 
				- q(t)
			\Big)
			u_{Y_{t,\mu}}' ( S_{t,\mu} )
			u_{Y_{t,\mu}}' ( s_{t,\mu} )
		\Big\rangle_{n,t,\epsilon}\,.
	\end{align*}
	As seen from \eqref{117_}, \eqref{106} it remains to compute $\E [X^*_j (Y'_{t,j}  - \sqrt{R_1(t)} X^*_j)\ln \cZ_{t,\epsilon}  ]$. Recalling that for $1 \leq j \leq n$, $Y_{t,j}'-\sqrt{R_1(t)} X_j^* = Z_j'$ and then using again a Gaussian integration by parts w.r.t $Z_j' \sim {\cal N}(0,1)$ we obtain
	\begin{align}
		\E \big[X^*_j (Y'_{t,j}  - &\sqrt{R_1(t)} X^*_j)\ln \cZ_{t,\epsilon}  \big]
		=
		\E \big[X^*_j Z_j' \ln \cZ_{t,\epsilon}  \big]
		=
		\E \Big[X^*_j Z_j' \ln \int dP_0(\bx) {\cal D}\bw\, e^{-\cH_{t,\epsilon}(\bx,\bw;\bY_t,\bY'_t,\boldsymbol{\Phi},\bV)}  \Big]
		\nn
		&= \E \Big[X^*_j Z_j' \ln \int dP_0(\bx) {\cal D}\bw\, 
		\exp \Big\{
			\sum_{\mu=1}^{m} u_{Y_{t,\mu}}(s_{t, \mu}) - \frac{1}{2} \sum_{i=1}^{n}\big(\sqrt{R_1(t)}X^*_{i} + Z_i'  - \sqrt{R_1(t)}\, x_i\big)^2
		\Big\}  
	\Big]
	\nn
		&= -\E \big[X^*_j \big\langle \sqrt{R_1(t)}(X_j^* - x_j) +Z_j' \big\rangle_{n,t,\epsilon} \big]
	\nn
		&= -\sqrt{R_1(t)} \big(\rho - \E \langle X^*_j x_j \rangle_{n,t,\epsilon} \big) \,.
	\end{align}
	Thus, by taking the sum,
	\begin{align}
		-\frac{r(t)}{2\sqrt{R_1(t)}}\E\Big[ \frac1n \sum_{i=1}^{n} X^*_i (Y'_{t,i}  - \sqrt{R_1(t)} X^*_i)\ln \cZ_{t,\epsilon}  \Big]
		= \frac{r(t)\rho}{2}-\frac{r(t)}{2}
		\E 
		\Big\langle
			\frac1n\sum_{i=1}^{n} X_i^* x_i
		\Big\rangle_{n,t,\epsilon}
		\,.
\end{align}
Therefore, for all $t \in (0,1)$,
	\begin{align}
		T_1
		= \,&\frac{1}{2}
		\E\Big[
			\frac{1}{\sqrt{n}} \sum_{\mu=1}^m \frac{P_{\rm out}''(Y_{t,\mu} | S_{t,\mu})}{P_{\rm out}(Y_{t,\mu} | S_{t,\mu})} 
			\Big(
				\frac{1}{\sqrt{n}} \sum_{i=1}^n ((X_i^*)^2 - \rho)
			\Big)
			\frac{1}{n}\ln \cZ_{t,\epsilon}
		\Big]
		+ \frac{r(t) \rho}{2}  - \frac{r(t) q(t)}{2}
		\nn
		& +\frac{1}{2} \E
		\Big\langle\Big(
				\frac{1}{n}\sum_{\mu=1}^m u_{Y_{t,\mu}}' ( S_{t,\mu} )
			u_{Y_{t,\mu}}'( s_{t,\mu} )
			-r(t)
			\Big)
			\Big(
				\frac{1}{n} \sum_{i=1}^{n} 
				X^*_i x_i 
				- q(t)
			\Big)			
		\Big\rangle_{n,t,\epsilon}\label{113}\,.
	\end{align}	
	To obtain \eqref{eq:der_f_t_raw}, it remains to show that $T_2 = 0$. This is a direct consequence of the Nishimori identity (see Appendix~\ref{appendix-nishimori}):
	\begin{align}
		T_2 = \frac{1}{n} \E \big\langle \cH_{t,\epsilon}'(\bx,\bw;\bY_t,\bY'_t,\boldsymbol{\Phi}) \big\rangle_{n,t,\epsilon}
		= \frac{1}{n} \E \,\cH_{t,\epsilon}'(\bX^*,\bW^*;\bY_t,\bY'_t,\boldsymbol{\Phi}) = 0\,.
	\end{align}

	For obtaining the Lemma, it remains to show that $A_{n,\epsilon}$ goes to $0$ uniformly in $t \in [0,1]$.
	\subsubsection{Proof that $A_{n,\epsilon}$ vanishes as $n\to\infty$}
	We now consider the final step, that is showing that $A_{n,\epsilon}$ given by \eqref{An} vanishes in the $n\to\infty$ limit uniformly in $t \in [0,1]$ under conditions~\ref{hyp:bounded}-\ref{hyp:c2}-\ref{hyp:phi_gauss2}. First we show that
\begin{align}		
	\E \Big[
			\frac{1}{\sqrt{n}} \sum_{\mu=1}^{m} \frac{P_{\rm out}''(Y_{t,\mu} | S_{t,\mu})}{P_{\rm out}(Y_{t,\mu} | S_{t,\mu})} 
			\Big( \frac{1}{\sqrt{n}} \sum_{i=1}^{n} \big((X^*_i)^2 - \rho\big) \Big)			
		\Big] = 0\,. \label{115}
\end{align}		
Once this is done, we use the fact that $\frac{1}{n} \ln \cZ_{t,\epsilon}$ concentrates around $f_{n,\epsilon}(t)$ to prove that $A_{n,\epsilon}$ converges to $0$ as $n\to\infty$. 
We start by noticing the simple fact that for all $s \in \R, \ \int P_{\rm out}''(y|s)dy = 0$. Consequently, for $\mu \in \{1 ,\dots, m \}$,
\begin{align}
			\E \Big[ \frac{P_{\rm out}''(Y_{t,\mu} | S_{t,\mu})}{P_{\rm out}(Y_{t,\mu} | S_{t,\mu})}		
		\, \Big| \, \bX^*, \bS_t \Big]
		=
		\int dY_{t,\mu} P_{\rm out}''(Y_{t,\mu}|S_{t,\mu}) = 0\,. \label{117}
\end{align}
Thus, using the ``tower property'' of the conditionnal expectation:
\begin{align*}
		\E\Big[\Big( \sum_{i=1}^{n} ((X^*_i)^2 - \rho) \Big) \sum_{\mu=1}^{m} \frac{P_{\rm out}''(Y_{t,\mu} | S_{t,\mu})}{P_{\rm out}(Y_{t,\mu} | S_{t,\mu})} \Big]
		=
		\E\Big[\Big( \sum_{i=1}^{n} \big((X^*_i)^2 - \rho\big) \Big) 
		\E \Big[\sum_{\mu=1}^{m} \frac{P_{\rm out}''(Y_{t,\mu} | S_{t,\mu})}{P_{\rm out}(Y_{t,\mu} | S_{t,\mu})} \, \Big| \, \bX^*, \bS_t \Big] \Big]
		= 0
\end{align*}		
which gives \eqref{115}. 
We now show that $A_{n,\epsilon}$ goes to $0$ uniformly in $t \in[0,1]$ as $n \to \infty$. Using successively \eqref{115} and the Cauchy-Schwarz inequality, we have
	\begin{align}
		|A_{n,\epsilon}| &= \Big|
		\E \Big[
			\frac{1}{\sqrt{n}} \sum_{\mu=1}^{m} \frac{P_{\rm out}''(Y_{t,\mu} | S_{t,\mu})}{P_{\rm out}(Y_{t,\mu} | S_{t,\mu})} 
			\Big( \frac{1}{\sqrt{n}} \sum_{i=1}^{n} \big((X^*_i)^2 - \rho\big) \Big)
			\Big(\frac{1}{n}\ln \cZ_{t,\epsilon} -f_{n,\epsilon}(t)\Big)
		\Big]
		\Big|
		\nn
		&\leq\,
		\E \Big[
			\Big(\frac{1}{\sqrt{n}} \sum_{\mu=1}^{m} \frac{P_{\rm out}''(Y_{t,\mu} | S_{t,\mu})}{P_{\rm out}(Y_{t,\mu} | S_{t,\mu})}\Big)^2
			\Big( \frac{1}{\sqrt{n}} \sum_{i=1}^{n} \big((X^*_i)^2 - \rho\big) \Big)^2
		\Big]^{1/2}
		\E\Big[\Big(\frac{1}{n} \ln \cZ_{t,\epsilon} - f_{n,\epsilon}(t)\Big)^2\Big]^{1/2}\,.\label{119}
	\end{align}
	Using again the ``tower property'' of conditional expectations
	\begin{align}
		\E \Big[
			\Big(\sum_{\mu=1}^{m} &\frac{P_{\rm out}''(Y_{t,\mu} | S_{t,\mu})}{P_{\rm out}(Y_{t,\mu} | S_{t,\mu})}\Big)^2
			\Big( \sum_{i=1}^{n} ((X^*_i)^2 - \rho) \Big)^2
		\Big]
		\nn
		=\,
		&\E \Big[
			\Big( \sum_{i=1}^{n} \big((X^*_i)^2 - \rho\big) \Big)^2
		\E \Big[
			\Big(\sum_{\mu=1}^{m} \frac{P_{\rm out}''(Y_{t,\mu} | S_{t,\mu})}{P_{\rm out}(Y_{t,\mu} | S_{t,\mu})}\Big)^2
			\, \Big| \, \bX^*, \bS_t
		\Big]
	\Big]\,. \label{eq:tower2}
	\end{align}
	Now, using the fact that conditionally on $\bS_t$, the random variables $\big(
\frac{P_{\rm out}''(Y_{t,\mu} | S_{t,\mu})}{P_{\rm out}(Y_{t,\mu} | S_{t,\mu})}
\big)_{1 \leq \mu \leq m}$ are i.i.d.\ and centered, we have
	\begin{align}
		\E \Big[
			\Big(\sum_{\mu=1}^{m} \frac{P_{\rm out}''(Y_{t,\mu} | S_{t,\mu})}{P_{\rm out}(Y_{t,\mu} | S_{t,\mu})}\Big)^2
			\, \Big| \, \bX^*, \bS_t
		\Big]
		\!=\!
		\E \Big[
			\Big(\sum_{\mu=1}^{m} \frac{P_{\rm out}''(Y_{t,\mu} | S_{t,\mu})}{P_{\rm out}(Y_{t,\mu} | S_{t,\mu})}\Big)^2
			\, \Big| \, \bS_t
		\Big]
		\!=\!
		m
		\E \Big[
			\Big(\frac{P_{\rm out}''(Y_{1} | S_{t,1})}{P_{\rm out}(Y_{1} | S_{t,1})}\Big)^2
			\, \Big| \, \bS_t
		\Big] \,. \label{eq:var_simp}
	\end{align}
	Under condition~\ref{hyp:c2}, it is not difficult to show that there exists a constant $C > 0$ such that
	\begin{align}
\E \Big[ \Big(\frac{P_{\rm out}''(Y_{t,1} | S_{t,1})}{P_{\rm out}(Y_{t,1} | S_{t,1})}\Big)^2 \, \Big| \, \bS_t \Big] \leq C\,.
\label{eq:borne_ddp}
	\end{align}
	Combining now \eqref{eq:borne_ddp}, \eqref{eq:var_simp} and \eqref{eq:tower2} we obtain that
	\begin{align*}
		\E \Big[
			\Big(\sum_{\mu=1}^{m} \frac{P_{\rm out}''(Y_{t,\mu} | S_{t,\mu})}{P_{\rm out}(Y_{t,\mu} | S_{t,\mu})}\Big)^2
			\Big( \sum_{i=1}^{n} \big((X^*_i)^2 - \rho\big) \Big)^2
		\Big]
		\leq 
		m C \, \E \Big[
			\Big( \sum_{i=1}^{n} \big((X^*_i)^2 - \rho\big) \Big)^2
		\Big]
		= m n C \, \Var\big((X_1^*)^2\big)\,.
	\end{align*}
	Going back to \eqref{119}, therefore there exists a constant $C' > 0$ such that
	\begin{align}
		|A_{n,\epsilon}|
		\leq\,
		C' \, \E\Big[\Big(\frac{1}{n} \ln \cZ_{t,\epsilon} - f_{n,\epsilon}(t)\Big)^2\Big]^{1/2}\,.
		\label{125}
	\end{align}
	By Theorem~\ref{concentrationtheorem} we have $\E[( n^{-1}\ln\cZ_{t,\epsilon} - f_{n,\epsilon}(t))^2] \to 0$ as $n\to\infty$ uniformly in $t \in [0,1]$. Thus $A_{n,\epsilon}$ goes to $0$ as $n \to \infty$ uniformly in $t \in [0,1]$, $\epsilon$ and w.r.t. the choice of the interpolation functions. This ends the proof of Proposition~\ref{prop:der_f_t}.

		\subsection{Boundedness of an overlap fluctuation}\label{appendix-boundedness-uterms}

In this appendix we show that the ``overlap fluctuation''
\begin{align}\label{over-fluct}
\E \Big\langle 
				\Big(
					\frac{1}{n} \sum_{\mu=1}^{m}u'_{Y_{t,\mu}}(S_{t,\mu}) u'_{Y_{t,\mu}}(s_{t,\mu})
					- r_\epsilon(t)
				\Big)^2
			\Big\rangle_{n,t,\epsilon}\!
			\leq \!
			2 r_{\rm max}^2 + 2 \E \Big\langle 
				\Big(
					\frac{1}{n} \sum_{\mu=1}^{m}u'_{Y_{t,\mu}}(S_{t,\mu}) u'_{Y_{t,\mu}}(s_{t,\mu})
				\Big)^2
			\Big\rangle_{n,t,\epsilon}	
\end{align}
is bounded uniformly in $t$ under hypothesis~\ref{hyp:c2} on $\varphi$. From the representation \eqref{transition-kernel} (recall we consider $\Delta=1$ in Sec.~\ref{sec:interpolation})
\begin{align}
u_{Y_{t, \mu}}(s) & = \ln P_{\rm out}(Y_{t, \mu}|s) = \ln \int dP_A(\ba_\mu) \frac{1}{\sqrt{2\pi}}e^{-\frac{1}{2} (Y_{t,\mu} - \varphi(s, \ba_\mu))^2}
\end{align}
and thus 
\begin{align}
u_{Y_{t, \mu}}^\prime(s) & = \frac{\int dP_A(\ba_\mu) (Y_{t,\mu} - \varphi(s, \ba_\mu)) \varphi^\prime(s, \ba_\mu) e^{-\frac{1}{2} (Y_{t,\mu} - \varphi(s, \ba_\mu))^2}}
{\int dP_A(\ba_\mu) e^{-\frac{1}{2} (Y_{t,\mu} - \varphi(s, \ba_\mu))^2}}
\end{align}
where $\varphi^\prime$ is the derivative w.r.t.\ the first argument. From \eqref{measurements} at $\Delta=1$ we get $\vert Y_{t, \mu} \vert \leq \sup\vert\varphi\vert +\vert Z_\mu\vert$, where the supremum is taken over both arguments of $\varphi$, and thus immediately obtain for all $s\in \mathbb{R}$
\begin{align}\label{bound-u}
\vert u_{Y_{t, \mu}}^\prime(s)\vert \leq (2 \sup\vert \varphi\vert + \vert Z_\mu\vert) \sup\vert\varphi^\prime\vert\,.
\end{align}
From \eqref{bound-u} and \eqref{over-fluct} we see that it suffices to check that 
\begin{align*}
\frac{m^2}{n^2}\mathbb{E}\big[\big((2\sup\vert\varphi\vert + |Z_\mu|)^2(\sup\vert\varphi^\prime\vert)^2\big)^2\big] \leq C(\varphi,\alpha)
\end{align*}
where $C(\varphi,\alpha)$ is a constant depending only on $\varphi$ and $\alpha$. This is easily seen by expanding all squares and using 
that $m/n$ is bounded.

\subsection{Proof of Proposition \ref{prop:F_equadiff}}\label{appendix-proofPropEqDiff}
	The continuity and differentiability properties of $F_n$ follow from the standard theorems of continuity and derivation under the integral sign. The domination hypotheses are easily verified because we are working under hypotheses \ref{hyp:bounded}-\ref{hyp:c2}.

	The overlap $\E \langle Q \rangle_{n,t,\epsilon}$ is related to the minimum mean-square error by
	$$
	\frac{1}{n}\MMSE( \bX^*  | \bY_t,  \bY_t',\bV, \bbf{\Phi})
	= \frac{1}{n} \E \Big[\big\| \bX^* - \E[\bX^*| \bY_t, \bY_t', \bV, \bbf{\Phi}] \big\|^2\Big]
	= \frac{1}{n} \E \Big[\big\| \bX^* - \langle \bx \rangle_{n,t,\epsilon} \big\|^2\Big]
	= \rho - \E \langle Q \rangle_{n,t,\epsilon}.
	$$
	Since the left hand side belongs to $[0,\rho]$, we obtain that $\E \langle Q \rangle_{n,t,\epsilon} \in [0,\rho]$.

	It remains therefore to prove that $\MMSE(\bX^* | \bY_t, \bY'_t,\bV, \bbf{\Phi})$ is separately non-increasing in $R_1$ and $R_2$. 
	$R_1$ only appears in the definition of $\bY'_t$.
	Recall that $\bY_t' = \sqrt{R_1} \,\bX^* + \bZ'$, where $\bZ'$ is a standard Gaussian vector, so $\MMSE(\bX^* | \bY_t, \bY'_t,\bV, \bbf{\Phi})$ is obviously a non-increasing function of $R_1$.

	$R_2$ only plays a role in $\bY_t\sim
	P_{\rm out}( \cdot \, | \,  \sqrt{(1-t)/{n}}\, \bbf{\Phi}\bX^* 
	+ \sqrt{R_2}\, \bV + \sqrt{\rho t - R_2 + 2 s_n}\, \bW^*)$.
Let $0 < r_2 \leq r_2' <\rho t$. Let $\bV' \sim \cN(0,\bbf{I}_m)$, independently of everything else.
	Define
	$$
	\widetilde{\bY}_t
	\sim
	P_{\rm out}\Big( \cdot \, \Big| \,  \sqrt{\frac{1-t}{n}}\, \bbf{\Phi}\bX^* 
		+ \sqrt{r_2}\, \bV + \sqrt{r_2'-r_2}\, \bV' + \sqrt{\rho t - r_2' + 2 s_n} \,\bW^*
\Big),
	$$
	independently of everything else. Now notice that
	\begin{align*}
	\MMSE( \bX^*  |  \bY_t, \bY_t',\bV, \bbf{\Phi} )|_{R_2 = r_2}
	&=
	\MMSE( \bX^* |  \widetilde{\bY}_t,  \bY_t',\bV, \bbf{\Phi} ),\nn
	\MMSE( \bX^* | \bY_t, \, \bY_t',\bV, \bbf{\Phi} )|_{R_2 = r_2'}
	&=
	\MMSE( \bX^* | \widetilde{\bY}_t,  \bY_t',\bV,\bV', \bbf{\Phi} ),
	\end{align*}
	which implies of course that $\MMSE( \bX^* | \bY_t, \bY_t',\bV, \bbf{\Phi} )|_{R_2 = r_2}
	\geq \MMSE( \bX^*  |  \bY_t, \bY_t',\bV, \bbf{\Phi} )|_{R_2 = r_2'}$. We have proved that $\MMSE( \bX^* |  \bY_t,  \bY_t',\bV, \bbf{\Phi} )$ is a non-increasing function of $R_2$.
	\section{Some properties of the scalar channels}

\subsection{The additive Gaussian scalar channel}\label{appendix_scalar_channel1}
We recall some properties (see \cite{GuoShamaiVerdu_IMMSE} and \cite{miolane2018phase} for proofs) of the free entropy of the first scalar channel \eqref{eq:additive_scalar_channel}. 
\begin{proposition}\label{prop16}
	Let $X_0 \sim P_0$ be a real random variable with finite second moment. Let $r \geq 0$ and $Y_0 = \sqrt{r} X_0 + Z_0$, where $Z_0 \sim \cN(0,1)$ is independent from $X_0$.
	Then the function
	$$
	\psi_{P_0}: r \mapsto \E \ln \int dP_0(x)e^{\sqrt r \,Y_0 x - r x^2/2}
	$$
	is convex, differentiable, non-decreasing and $\frac{1}{2}\E[X_0^2]$-Lipschitz on $\R_+$. Moreover, $\psi_{P_0}$ is strictly convex, if $P_0$ is not a Dirac measure.
\end{proposition}

\subsection{The non-linear scalar channel}\label{appendix_scalar_channel2}
We prove here some properties of the free entropy of the second scalar channel \eqref{eq:Pout_scalar_channel}, where $V,W^* \iid \cN(0,1)$ and
\begin{equation}\label{eq:Pout_channel2}
	Y^{(q)} \sim P_{\rm out}\big(\cdot \, \big| \, \sqrt{q}\,V +  \sqrt{\rho -q} \,W^*\big)\,.
\end{equation}
In this channel, the statistician observes $V$ and $Y^{(q)}$ and wants to recover $W^*$.
Recall that by definition $\mathcal{I}_{P_{\rm out}}(q) = I(W^*; Y^{(q)}| V)= \Psi_{P_{\rm out}}(\rho) - \Psi_{P_{\rm out}}(q)$ so the properties we will prove on $\Psi_{P_{\rm out}}$ can be directly translated for $\mathcal{I}_{P_{\rm out}}$, and vice-versa.

	 \begin{proposition}\label{prop:psi_convex_reg}
		 Suppose that for all $x\in\R$, $P_{\rm out}(\cdot \, | \, x)$ is the law of
		 $\varphi(x,A) + \sqrt{\Delta} Z$
		 where $\Delta>0$, $\varphi:\R \times \R^{k_A} \to \R$ is a measurable function and $(Z,A) \sim \cN(0,1) \otimes P_A$, for some probability distribution $P_A$ over $\R^{k_A}$. In that case
		 $P_{\rm out}$ admits a density given by
		 $$
		 P_{\rm out}(y|x)=
	 \frac{1}{\sqrt{2\pi \Delta}} \int dP_A(\ba) e^{-\frac{1}{2 \Delta}( y - \varphi(x,\ba))^2}\,.
		 $$
		Assume that $\varphi$ is bounded and $\cC^2$ with respect to its first coordinate, with bounded first and second derivatives.
		Then $q \mapsto \Psi_{P_{\rm out}}(q)$ is convex, $\cC^2$ and non-decreasing on $[0,\rho]$. 
\end{proposition}

\begin{proof}
	Let $V,W^* \iid \cN(0,1)$ and $Y^{(q)}$ be the output of the scalar channel given by \eqref{eq:Pout_channel2}. Then for all $q \in [0, \rho]$,
	$$
	\Psi_{P_{\rm out}}(q) = \E \ln \int dw\, \frac{e^{-\frac{w^2}{2}}}{{\sqrt{2\pi}}} P_{\rm out}\big( Y^{(q)} | \sqrt{q}\, V + \sqrt{\rho - q}\, w\big)
	$$
	Under the hypotheses we made on $\varphi$, 
	we will be able to use continuity and differentiation under the expectation, because all the domination hypotheses will be easily verified.
	It is thus easy to check that $\Psi_{P_{\rm out}}$ is continuous on $[0, \rho]$.

	We compute now the first derivative. Recall that $\langle - \rangle_{\rm sc}$, defined in \eqref{GibbsBracket_sc}, denotes the posterior distribution of $W^*$ given $Y^{(q)}$. 
	We will use the notation $u_y(x)= \ln P_{\rm out}(y|x)$. 
	For $q \in (0, \rho)$ we have
	\begin{align*}
		\Psi'_{P_{\rm out}}(q) 
		&= \frac{1}{2} \E \Big\langle u'_{Y^{(q)}}(\sqrt{q} \,V + \sqrt{\rho-q}\, w) u'_{Y^{(q)}}(\sqrt{q}\, V + \sqrt{\rho-q}\, W^{*}) \Big\rangle_{\rm sc}
		\\
		&= \frac{1}{2} \E \Big\langle u'_{Y^{(q)}}(\sqrt{q} \,V + \sqrt{\rho-q} \,w) \Big\rangle_{\rm sc}^2 \geq 0 \,,
	\end{align*}
	where $w \sim \langle - \rangle_{\rm sc}$, independently of everything else. 
	$\Psi_{P_{\rm out}}$ is therefore non-decreasing.
	Using the boundedness assumption on $\varphi$ and its derivatives, it is not difficult to check that $\Psi_{P_{\rm out}}'$ is indeed bounded. 

	We will now compute $\Psi''_{P_{\rm out}}$.
	To lighten the notations, we write $u'(w)$ for $u'_{Y^{(q)}}(\sqrt{q}\,V + \sqrt{\rho - q}\,w)$. We compute
	\begin{align}
		\partial_q \E \Big\langle u'(w) u'(W^{*}) \Big\rangle_{\rm sc}
		= \E \Big[ \Big(\frac{1}{2\sqrt{q}}\, V - \frac{1}{2 \sqrt{\rho - q}}\, W^{*}\Big) u'(W^{*}) \Big\langle u'(w) u'(W^{*}) \Big\rangle_{\rm sc} \Big]& \qquad (A)\nonumber
		\\
		+ 2 \E \Big\langle \Big(\frac{1}{2\sqrt{q}}\, V - \frac{1}{2 \sqrt{\rho - q}}\, W^{*}\Big) u''(W^{*}) u'(w) \Big\rangle_{\rm sc}&\qquad (B)\nonumber
		\\
		+ \E \Big\langle \Big(\frac{1}{2\sqrt{q}}\, V - \frac{1}{2 \sqrt{\rho - q}}\, W^{*}\Big) u'(W^{*})^2 u'(w) \Big\rangle_{\rm sc}&\qquad (C) \nonumber
		\\
		- \E \Big\langle u'(W^{*}) u'(w) \Big\rangle_{\rm sc} \Big\langle  \Big(\frac{1}{2\sqrt{q}}\, V - \frac{1}{2 \sqrt{\rho - q}}\, w\Big) u'(w) \Big\rangle_{\rm sc}&\qquad (D) \label{eq:der_abcd}
	\end{align}
	Notice that $(A) = (C)$. We compute, using Gaussian integration by parts and the Nishimori identity (Proposition~\ref{prop:nishimori})
	\begin{align}
		(A) &= 
		\frac{1}{2} \E \Big[ u'(W^{*}) \Big\langle u'(W^{*}) u''(w) \Big\rangle_{\rm sc} \Big]
		+ \frac{1}{2} \E \Big[ u'(W^{*}) \Big\langle u'(W^{*}) u'(w)^2 \Big\rangle_{\rm sc} \Big] \nonumber
		\\
		&\quad - \frac{1}{2} \E \Big[ u'(W^{*}) \Big\langle u'(W^{*}) u'(w) \Big\rangle_{\rm sc} \Big\langle u'(w) \Big\rangle_{\rm sc} \Big]
		\label{eq:compA}
		\\
		(B) &= 
		\E \Big\langle u''(W^{*}) u''(w) \Big\rangle_{\rm sc}
		+ \E \Big\langle u''(W^{*}) u'(w)^2 \Big\rangle_{\rm sc}
		- \E \Big\langle u''(W^{*}) u'(w) \Big\rangle_{\rm sc} \Big\langle u'(w) \Big\rangle_{\rm sc}
		\label{eq:compB}
		\\
		(D) &= 
		- \E \Big\langle  \Big(\frac{1}{2\sqrt{q}} V - \frac{1}{2 \sqrt{\rho - q}} W^{*}\Big) u'(W^{*}) u'(w^{(1)}) u'(w^{(2)}) \Big\rangle_{\rm sc} \nonumber
		\\
		&= 
		-\E \langle u'(W^{*}) u''(w^{(1)}) u'(w^{(2)}) \rangle_{\rm sc}
		- \E \langle u'(W^{*}) u'(w^{(1)})^2 u'(w^{(2)}) \rangle_{\rm sc}
		\nonumber \\
		& \quad + \E \langle u'(W^{*}) u'(w^{(1)}) u'(w^{(2)}) \rangle_{\rm sc} \langle u'(w) \rangle_{\rm sc}
		\label{eq:compD}
	\end{align}
	We now replace \eqref{eq:compA}, \eqref{eq:compB} and \eqref{eq:compD} in \eqref{eq:der_abcd}:
	\begin{align*}
		2 \Psi''_{P_{\rm out}}(q)
		&=
		\E \Big\langle u'(W^{*})^2 u''(w) \Big\rangle_{\rm sc} 
		+ \E \Big\langle u'(W^{*})^2 u'(w)^2 \Big\rangle_{\rm sc}
		- \E \Big\langle u'(W^{*})^2 u'(w^{(1)}) u'(w^{(2)}) \Big\rangle_{\rm sc}
		\\
		&\quad+ \E \Big\langle u''(W^{*}) u''(w) \Big\rangle_{\rm sc}
		+ \E \Big\langle u''(W^{*}) u'(w)^2 \Big\rangle_{\rm sc}
		- \E \Big\langle u''(W^{*}) u'(w^{(1)}) u'(w^{(2)}) \Big\rangle_{\rm sc}
		\\
		&\quad-\E \langle u'(W^{*}) u''(w^{(1)}) u'(w^{(2)}) \rangle_{\rm sc}
		- \E \langle u'(W^{*}) u'(w^{(1)})^2 u'(w^{(2)}) \rangle_{\rm sc}
		+ \E \langle u'(w) \rangle_{\rm sc}^4 \,.
	\end{align*}
	Using the identity $u''_Y(x) + u'_Y(x)^2 = \frac{P_{\rm out}''(Y|x)}{P_{\rm out}(Y|x)}$, this factorizes and gives
	\begin{equation}
		\Psi''_{P_{\rm out}}(q)
		= \frac{1}{2} \E\Big[
			\Big(
				\Big\langle \frac{P_{\rm out}''(Y|\sqrt{q}\,V + \sqrt{\rho-q}\,w)}{P_{\rm out}(Y|\sqrt{q}\,V + \sqrt{\rho-q}w)} \Big\rangle_{\rm sc} - \Big\langle u_{Y^{(q)}}'(\sqrt{q}\,V + \sqrt{\rho-q}\,w) \Big\rangle_{\rm sc}^2
			\Big)^2
		\Big] \geq 0 \,.
	\end{equation}
	$\Psi_{P_{\rm out}}$ is thus convex on $[0, \rho]$.
	It is not difficult to verify (by standard arguments of continuity under the integral) that $\Psi_{P_{\rm out}}''$ is continuous on $[0,\rho]$, which gives that $\Psi_{P_{\rm out}}$ is $\cC^2$ on its domain.
\end{proof}

\begin{proposition} \label{prop:psi_convex}
		 Suppose that for all $x\in\R$, $P_{\rm out}(\cdot \, | \, x)$ is the law of
		 $\varphi(x,A) + \sqrt{\Delta} Z$
		 where $\varphi:\R \times \R^{k_A} \to \R$ is a measurable function and $(Z,A) \sim \cN(0,1) \otimes P_A$, for some probability distribution $P_A$ over $\R^{k_A}$. Assume also that
	\begin{equation}\label{eq:second_m}
		\E[\varphi(\sqrt{\rho}Z,A)^2] < \infty \,,
	\end{equation}
	and that we are in one of the following cases:
\begin{enumerate}[label=(\roman*),noitemsep]
	\item \label{item:psi_gauss} $\Delta > 0$.
	\item \label{item:psi_discrete} $\Delta = 0$ and $\varphi$ takes values in $\N$.
\end{enumerate}
	Then $q \mapsto \Psi_{P_{\rm out}}(q)$ is continuous, convex and non-decreasing over $[0,\rho]$. 
\end{proposition}
Notice that \eqref{eq:second_m} is for instance verified under hypotheses~\ref{hyp:third_moment}-\ref{hyp:um}-\ref{hyp:phi_general}-\ref{hyp:cont_pp}, see Proposition \ref{prop:varphi_l2}.

\begin{proof}
	We deduce Proposition~\ref{prop:psi_convex} from Proposition~\ref{prop:psi_convex_reg} above by an approximation procedure.
	Since $\Psi_{P_{\rm out}} = \Psi_{P_{\rm out}}(\rho) - \mathcal{I}_{P_{\rm out}}$, we will work with the mutual information $\mathcal{I}_{P_{\rm out}}$.
	Let us define $U^{(q)} =  \varphi\big(\sqrt{q}\,V + \sqrt{\rho -q}\,W^*, A \big)$ and $Y^{(q)} = U^{(q)} + \sqrt{\Delta} Z$.

	We start by proving Proposition~\ref{prop:psi_convex} under the assumption~\ref{item:psi_gauss}.
	Let $\epsilon >0$.
	By density of the $\cC^{\infty}$ functions with compact support in $L^2$ (see for instance Corollary~4.2.2 from \cite{bogachev2007measure}), one can find a $\cC^{\infty}$ function $\widehat{\varphi}$ with compact support, such that
	$$
	\E\Big[\big(\varphi(\sqrt{\rho}\,Z,A) - \widehat{\varphi}(\sqrt{\rho}\,Z,A) \big)^2\Big] \leq \epsilon^2 \,.
	$$
	Let us write $\widehat{U}^{(q)} = \widehat{\varphi}(\sqrt{q}\,V + \sqrt{\rho-q}\,W^*,A)$ and $\widehat{Y}^{(q)} = \widehat{U} + \sqrt{\Delta} Z$.
	We have by the chain rule for the mutual information
	\begin{align}
	&I(U^{(q)};Y^{(q)}|V) = I(W^*,U^{(q)} ; Y^{(q)} | V) \nn
	=\,&I(U^{(q)};Y^{(q)}|V,W^*) + I(W^*;Y^{(q)}|V)
	=I(U^{(q)};Y^{(q)}|V,W^*) + \mathcal{I}_{P_{\rm out}}(q)
	\end{align}
	and similarly, $\mathcal{I}_{\widehat{P}_{\rm out}}(q) = I(\widehat{U}^{(q)};\widehat{Y}^{(q)}|V) - I(\widehat{U}^{(q)};\widehat{Y}^{(q)}|V,W^*)$. By Proposition~\ref{prop:free_wasserstein}, there exists a constant $C>0$ such that
	$$
	|I(\widehat{U}^{(q)};\widehat{Y}^{(q)}|V) -I(U^{(q)};Y^{(q)}|V)|  \leq C \epsilon
\quad \text{and} \quad
|I(\widehat{U}^{(q)};\widehat{Y}^{(q)}|V,W^*) - I(U^{(q)};Y^{(q)}|V,W^*)| \leq C \epsilon \,.
	$$
	We get that for all $q \in [0,\rho]$, $|\mathcal{I}_{P_{\rm out}}(q) - \mathcal{I}_{\widehat{P}_{\rm out}}(q)| \leq C\epsilon$. The function $\mathcal{I}_{P_{\rm out}}$  can therefore be uniformly approximated by continuous, concave, non-increasing functions on $[0,\rho]$: $\mathcal{I}_{P_{\rm out}}$ is therefore continuous, concave and non-increasing.

	Let us now prove Proposition~\ref{prop:psi_convex} under the assumption~\ref{item:psi_discrete}. Under this assumption we have $\mathcal{I}_{P_{\rm out}}(q) = I(W^*;U^{(q)}|V)$ and by the case~\ref{item:psi_gauss} we know that the function $i_{\Delta}(q) = I(W^*; U^{(q)} + \sqrt{\Delta}Z|V)$ is concave and non-increasing for all $\Delta >0$. By Corollary~\ref{cor:cont_I_discrete} we obtain that for all $q \in [0, \rho]$ and all $\Delta \in (0,1]$ we have
	$$
	\big|\mathcal{I}_{P_{\rm out}}(q) - i_{\Delta}(q)\big| \leq 100 e^{-1/(16\Delta)} \,,
	$$
	which proves (by uniform approximation) that $\mathcal{I}_{P_{\rm out}}$ is continuous, concave and non-increasing.
	\end{proof}

	\begin{proposition}\label{prop:psi_out_diff}
		Under the same hypotheses than Proposition~\ref{prop:psi_convex} above,
		$\Psi_{\rm out}$ is differentiable over $[0,\rho)$ and for all $q \in [0,\rho)$
		$$
		\Psi_{P_{\rm out}}'(q)=
	\frac{1}{2(\rho - q)} \E \langle w \rangle_{\rm  sc}^2 \,,
	$$
	where we recall that $\langle - \rangle_{\rm sc}$ is defined by \eqref{GibbsBracket_sc}.
	\end{proposition}
	\begin{proof}
	The fact that $\Psi_{P_{\rm out}}$ is differentiable on $[0,\rho)$ follows from differentiation under the expectation sign. In order to see it, we define $X = \sqrt{q}\,V + \sqrt{\rho - q}\,W^*$. Then, for all $q \in [0,\rho)$:
	\begin{align}
		\Psi_{P_{\rm out}}(q) = \E\int dX \frac{1}{\sqrt{2\pi(\rho-q)}} e^{-\frac{(X-\sqrt{q}\,V)^2}{2(\rho - q)}}
		\int dY P_{\rm out}(Y|X) \ln  \int dx \frac{1}{\sqrt{2\pi(\rho-q)}} e^{-\frac{(x-\sqrt{q}\,V)^2}{2(\rho - q)}} P_{\rm out}(Y|x)\,. 
	\end{align}	
	We are now in a good setting to differentiate under the expectation sign. 
	We have for all $q \in (0,\rho)$,
	\begin{equation}\label{eq:diff_gauss}
	\frac{\partial}{\partial q} \left[
	\frac{1}{\sqrt{\rho-q}} e^{-\frac{(X-\sqrt{q}\,V)^2}{2(\rho - q)}}
	\right]
	= \frac{1}{2\sqrt{\rho - q}} \Big(\frac{1}{\rho - q} - \frac{(X-\sqrt{q}V)^2}{(\rho - q)^2} + \frac{V (X - \sqrt{q}V)}{\sqrt{q}(\rho-q)}\Big) e^{-\frac{(X-\sqrt{q}\,V)^2}{2(\rho - q)}} \,.
	\end{equation}
	Thus
	\begin{align*}
		\Psi_{P_{\rm out}}'(q)&=
	\frac{1}{2} \E \left[
 \Big(\frac{1}{\rho - q} - \frac{(X-\sqrt{q}\,V)^2}{(\rho - q)^2} + \frac{V (X - \sqrt{q}\,V)}{\sqrt{q}\,(\rho-q)}\Big) 
\ln  \int dx \frac{1}{\sqrt{2\pi(\rho-q)}} e^{-\frac{(x-\sqrt{q}\,V)^2}{2(\rho - q)}} P_{\rm out}(Y|x) \right]
\\
&+
	\frac{1}{2} \E \left\langle
 \frac{1}{\rho - q} - \frac{(x-\sqrt{q}\,V)^2}{(\rho - q)^2} + \frac{V (x - \sqrt{q}\,V)}{\sqrt{q}\,(\rho-q)}
 \right\rangle_{\rm sc}
	\end{align*}
	where the Gibbs brackets $\langle - \rangle_{\rm sc}$ denotes the expectation with respect to $x \sim P(X|Y^{(q)},V)$. The second term of the sum above is equal to zero. Indeed by the Nishimori identity (Proposition~\ref{prop:nishimori}):
 \begin{align*}
	\E \left\langle
 \frac{1}{\rho - q} - \frac{(x-\sqrt{q}\,V)^2}{(\rho - q)^2} + \frac{V (x - \sqrt{q}\,V)}{\sqrt{q}\,(\rho-q)}
 \right\rangle_{\rm sc}
 &=
 \E \left[
 \frac{1}{\rho - q} - \frac{(X-\sqrt{q}\,V)^2}{(\rho - q)^2} + \frac{V (X - \sqrt{q}\,V)}{\sqrt{q}\,(\rho-q)}
 \right]
 \\
 &=\frac{1}{\rho - q} \E \left[1 - (W^*)^2\right] = 0\,.
 \end{align*}
 We now compute, by Gaussian integration by parts with respect to $V \sim \cN(0,1)$:
 \begin{align*}
	 &	\E \left[
  \frac{V (X - \sqrt{q}\,V)}{\sqrt{q}\,(\rho-q)}
  \ln  \int dx \frac{1}{\sqrt{2\pi(\rho-q)}} e^{-\frac{(x-\sqrt{q}\,V)^2}{2(\rho - q)}} P_{\rm out}(Y^{(q)}|x) \right]
  \\
&=
	\E \left[
  \frac{-1}{\rho-q}
  \ln  \int dx \frac{e^{-\frac{(x-\sqrt{q}\,V)^2}{2(\rho - q)}}}{\sqrt{2\pi(\rho-q)}}  P_{\rm out}(Y^{(q)}|x) \right]
  +
\E \left[
	\frac{(X - \sqrt{q}\,V)^2}{(\rho-q)^2}
	\ln  \int dx \frac{e^{-\frac{(x-\sqrt{q}\,V)^2}{2(\rho - q)}}}{\sqrt{2\pi(\rho-q)}}  P_{\rm out}(Y^{(q)}|x)
\right]
\\
&\ +
\E \left\langle
	\frac{(X - \sqrt{q}\,V)(x - \sqrt{q}\,V)}{(\rho-q)^2}
\right\rangle_{\rm sc}.
\end{align*}
Bringing all together, we conclude:
	$$
	\Psi_{P_{\rm out}}'(q) = 
	\frac{1}{2}
\E \left\langle
	\frac{(X - \sqrt{q}\,V)(x - \sqrt{q}\,V)}{(\rho-q)^2}
\right\rangle_{\rm sc} = 
	\frac{1}{2(\rho - q)} \E \langle w \rangle_{\rm  sc}^2 \,.
	$$
	This derivative is continuous at $q=0$ thus $\Psi_{P_{\rm out}}$ is differentiable at $q=0$ with derivative given by the same expression.
\end{proof}

\begin{proposition}\label{prop:psi_stricly_monoton}
	Assume that the hypotheses of Proposition~\ref{prop:psi_convex} hold and suppose also that the kernel $P_{\rm out}$ is informative. Then $\Psi_{P_{\rm out}}$ is strictly increasing on $[0,\rho]$.
\end{proposition}

\begin{proof}
	Let us suppose that $\Psi_{P_{\rm out}}$ is not strictly increasing on $[0,\rho]$.
	There exists thus $q \in (0, \rho)$ such that $\Psi'_{P_{\rm out}}(q) = 0$. This means that $\langle w \rangle_{\rm sc} = 0$ almost surely and therefore that 
	$$
	\int_{\R} P_{\rm out}(Y^{(q)} \, | \, \sqrt{q}\, V + \sqrt{\rho - q} \,w) w e^{-w^2 / 2} dw = 0
	$$
	almost-surely. Let us write $\sigma = \sqrt{\rho-q}$. Consequently, 
	\begin{equation}\label{eq:pout_trivial}
		\int_{\R} P_{\rm out}(y \, | \, v + \sigma w) w e^{-w^2 / 2} dw = 0
	\end{equation}
	for almost all $y$ in $\R$ (if we are under assumption~\ref{item:psi_gauss}) or all $y \in \N$ (under assumption~\ref{item:psi_discrete}) and almost all $v \in \R$.
	We will now use the following lemma:
	\begin{lemma}\label{lem:constant_function}
		Let $Z \sim \cN(0,1)$ and let $f:\R \to \R$ be a bounded function. Suppose that for almost all $v \in \R$,
		$$
		\E[Z f(v+Z)] = 0 \,.
		$$
		Then, there exists a constant $C \in \R$ such that $f(v) = C$ for almost every $v$.
	\end{lemma}
	\begin{proof}
	Let us define the function
	$$
	h : t \mapsto \E [f(Z-t)] = \frac{1}{\sqrt{2\pi}}\int  f(x) e^{-(x+t)^2/2} dx \,.
	$$
	We have $h'(t) = \frac{-1}{\sqrt{2\pi}}\int  f(x) (x+t) e^{-(x+t)^2/2} dx = -\E [Zf(Z-t)] = 0$ and therefore $h$ is equal to some constant $C \in \R$. We are going to show that $f=C$ almost everywhere. Without loss of generality we can assume that $C = 0$, otherwise it suffices to consider the function $\tilde{f} = f-C$. Now we have for all $n \geq 0$, $t \in \R$
	$$
	0 = h^{(n)}(t) = \frac{1}{\sqrt{2\pi}}\int  f(x) \frac{\partial}{\partial t} e^{-(x+t)^2/2} dx 
	= \frac{1}{\sqrt{2\pi}}\int  f(x) (-1)^n H_n(x+t) e^{-(x+t)^2/2} dx  \,,
	$$
	where $H_n$ is $n^{\rm th}$ Hermite polynomial, defined as $H_n(x) = (-1)^n e^{x^2/2} \frac{d^n}{dx^n} e^{-x^2/2}$. Therefore, for all $n \geq 0$,
	$$
	\int  f(x) H_n(x) e^{-x^2/2} dx  = 0\,,
	$$
	which implies that $f=0$ almost everywhere since the Hermite functions form an orthonormal basis of $L^2(\R)$.
	\end{proof}
	We apply now Lemma~\ref{lem:constant_function} to \eqref{eq:pout_trivial} where the function $f$ is given by $f(x) = P_{\rm out}(y \, | \, \sigma x)$. We thus obtain that for almost every $y$, $P_{\rm out}( y \, | \, \cdot )$ is almost everywhere equal to a constant. 
	Under assumption~\ref{item:psi_discrete}, we get that for all $y \in \N$, $P_{\rm out}( y \, | \, \cdot )$ is almost everywhere equal to a constant: this contradicts the hypothesis that $P_{\rm out}$ is informative.

	If now assumption~\ref{item:psi_gauss} holds, then by \eqref{transition-kernel} the density function $P_{\rm out}(\cdot \, | \, x)$ is continuous on $\R$ for all $x \in \R$. 
	Let us fix $y \in \R$. We are going to show that $P_{\rm out}(y \, | \, \cdot)$ is almost everywhere equal to a constant $C_y$. Given what we just showed, we can construct a sequence
	$(y_n)_n \in \R^{\N}$ that converges to $y$ such that for all $n \geq 0$, there exists $E_n \subset \R$ with full Lebesgue's measure and $C_{n} \in \R$ such that for all $x \in E_n$,
	$$
	P_{\rm out}(y_n|x) = C_n \,.
	$$
	Let us define $E = \cap_{n \geq 0} E_n$. $E$ has therefore full Lebesgue's measure. Let now $x_1, x_2 \in E$. By continuity of $P_{\rm out}( \cdot | x_i )$, we get
	$$
	P_{\rm out}(y_n|x_i) \xrightarrow[n \to \infty]{} P_{\rm out}(y|x_i) , \qquad \text{for} \ i=1,2.
	$$
	Since we know that for all $n \geq0$ that $P_{\rm out}(y_n|x_1) = C_n = P_{\rm out}(y_n|x_2)$, we deduce that $P_{\rm out}(y|x_1) = P_{\rm out}(y|x_2)$. This proves that $P_{\rm out}(y \, | \, \cdot)$ is almost everywhere equal to a constant $C_y$ and contradicts the fact that $P_{\rm out}$ is informative.
\end{proof}

We turn now our attention to the study of the function:
\begin{equation}\label{eq:fun_gen}
	\mathcal{E}_{f}: 
	\left|
	\begin{array}{ccl}
		[0,\rho] & \to & \R_+ \\
		q & \mapsto &\E \big[(f(Y^{(q)}) - \E[f(Y^{(q)})|V])^2\big]
	\end{array}
	\right.
\end{equation}
where $f: \R \to \R$ is a continuous bounded function.
	We will prove that $\mathcal{E}_f$ is continuous (Proposition~\ref{prop:gen_continuous}) and strictly decreasing (Proposition~\ref{prop:gen_strict}) under the following hypotheses.
\begin{enumerate}[label=(\alph*),noitemsep]
	\item For all $x\in\R$, $P_{\rm out}(\cdot \, | \, x)$ is the law of
		 $\varphi(x,\bA) + \sqrt{\Delta}\, Z$
		 where $\varphi:\R \times \R^{k_A} \to \R$ is a measurable function and $(Z,\bA) \sim \cN(0,1) \otimes P_A$, for some probability distribution $P_A$ over $\R^{k_A}$. 
	 \item For almost all $a \in \R^{k_A}$ (w.r.t.\ $P_A$), $\varphi(\cdot, a)$ is continuous almost everywhere.
\end{enumerate}
	We suppose also that we are in one of the following cases:
\begin{enumerate}[label=(\roman*),noitemsep]
	\item \label{item:psi_gauss} $\Delta > 0$.
	\item \label{item:psi_discrete} $\Delta = 0$ and $\varphi$ takes values in $\N$.
\end{enumerate}

\begin{proposition}\label{prop:gen_continuous}
	Under the hypotheses presented above, $\mathcal{E}_f$ is continuous on $[0,\rho]$.
\end{proposition}
\begin{proof}
	Consider expression \eqref{eq:def_e_f_varphi}: The first term does not depend on $q$ and the second one is continuous by Lebesgue's convergence theorem.
\end{proof}

\begin{proposition}\label{prop:gen_strict}
	Assume that the hypotheses of Proposition~\ref{prop:gen_continuous} hold.
	Suppose that $x \mapsto \int f(y) P_{\rm out}(y \, | \, x) dy$ is not almost-everywhere equal to a constant.
	Then $\mathcal{E}_f$ is strictly decreasing on $[0,\rho]$.
\end{proposition}
\begin{proof}
	$\mathcal{E}_f(q)= \E [f(Y^{(q)})^2] - \E \big[ \E[f(Y^{(q)})|V]^2 \big]$. Since the first term does not depend on $q$, it suffices to show that $H :q \mapsto \E \big[ \E[f(Y^{(q)})|V]^2 \big]$ is strictly increasing on $[0,\rho]$.
We have for $q \in (0,\rho)$:
\begin{align*}
\E[f(Y^{(q)})|V] 
&= \int \int f(y) \frac{e^{-w^2/2}}{\sqrt{2 \pi}} P_{\rm out}(y|\sqrt{q}\, V + \sqrt{\rho - q}\,w) dy dw
= \int \int f(y) \frac{e^{-\frac{(x-\sqrt{q}V)^2}{2(\rho-q)}}}{\sqrt{2 \pi (\rho -q)}} P_{\rm out}(y|x) dy dx
\,.
\end{align*}
So we have, using \eqref{eq:diff_gauss}:
\begin{align*}
	\frac{\partial}{\partial q}
\E[f(Y^{(q)})|V] 
&=
\int \int \frac{f(y)}{2}\Big(\frac{1}{\rho -q} - \frac{(x-\sqrt{q}\,V)^2}{(\rho-q)^2} + \frac{V(x-\sqrt{q}\,V)}{\sqrt{q}(\rho-q)}\Big) \frac{e^{-\frac{(x-\sqrt{q}\,V)^2}{2(\rho-q)}}}{\sqrt{2 \pi (\rho -q)}} P_{\rm out}(y|x) dy dw
\\
&=
\frac{1}{2(\rho-q)}\E \left[ f(Y^{(q)}) \Big(1 - W^{* 2} + \frac{\sqrt{\rho-q}\,V W^*}{\sqrt{q}}\Big) \middle| V \right] \,.
\end{align*}
We obtain
\begin{equation}
H'(q) =
\frac{1}{\rho-q}\E \left[ 
	\E[f(Y^{(q)})|V] \,
	\E \left[ f(Y^{(q)}) \Big(1 - W^{* 2} + \frac{\sqrt{\rho-q}\,V W^*}{\sqrt{q}}\Big) \middle| V\right] 
\right] \,.
\label{eq:der_H}
\end{equation}
We compute by Gaussian integration by parts:
\begin{align}
	&\E \left[ 
	\E[f(Y^{(q)})|V] \,
	\E \left[ f(Y^{(q)}) V W^* \Big| V \right] 
\right]
=
\E \left[ V
	\E[f(Y^{(q)})|V] \,
	\E \left[ f(Y^{(q)}) W^* \Big| V \right] 
\right] \nonumber
\\
&=
\E \left[ 
	\frac{\partial}{\partial V} \E[f(Y^{(q)})|V] \,
	\E \left[ f(Y^{(q)}) W^* \Big| V \right] 
\right]
+
\E \left[ 
	 \E[f(Y^{(q)})|V] \,
	\frac{\partial}{\partial V}\E \left[ f(Y^{(q)}) W^* \Big| V \right] 
\right] \,. \label{eq:der_0}
\end{align}
We compute successively
\begin{align}
	\frac{\partial}{\partial V} \E[f(Y^{(q)})|V]
	&=\frac{\partial}{\partial V}
	\int \int f(y) \frac{e^{-\frac{(x-\sqrt{q}\,V)^2}{2(\rho-q)}}}{\sqrt{2 \pi (\rho -q)}} P_{\rm out}(y|x) dy dx \nonumber
	\\
	&=
	\int \int f(y) \frac{\sqrt{q}\,(x-\sqrt{q}\,V)}{\rho-q}
	\frac{e^{-\frac{(x-\sqrt{q}\,V)^2}{2(\rho-q)}}}{\sqrt{2 \pi (\rho -q)}} P_{\rm out}(y|x) dy dx
	= \frac{\sqrt{q}}{\sqrt{\rho-q}} \E \Big[ f(Y^{(q)}) W^* \Big|V \Big] \,.\label{eq:der_1}
	\\
	\frac{\partial}{\partial V} \E[f(Y^{(q)}) W^* |V]
	&=\frac{\partial}{\partial V}
	\int \int f(y) \frac{x - \sqrt{q}\,V}{\sqrt{\rho-q}} \frac{e^{-\frac{(x-\sqrt{q}\,V)^2}{2(\rho-q)}}}{\sqrt{2 \pi (\rho -q)}} P_{\rm out}(y|x) dy dx \nonumber
	\\
	&=
	\int \int 
	f(y) \Big(\frac{-\sqrt{q}}{\sqrt{\rho-q}} 
		+
		\frac{\sqrt{q}(x - \sqrt{q}\,V)^2}{(\rho-q)^{3/2}}
	\Big)
	\frac{e^{-\frac{(x-\sqrt{q}\,V)^2}{2(\rho-q)}}}{\sqrt{2 \pi (\rho -q)}} P_{\rm out}(y|x) dy dx \nonumber
	\\
	&=
	\frac{\sqrt{q}}{\sqrt{\rho-q}} \E \Big[f(Y^{(q)}) \big(-1 + W^{*2}\big)\Big|V\Big]\,. \label{eq:der_2}
\end{align}
By plugging \eqref{eq:der_0}-\eqref{eq:der_1}-\eqref{eq:der_2} back in \eqref{eq:der_H} we get:
$$
H'(q) = \frac{1}{\rho-q} \E\Big[\E \big[f(Y^{(q)})W^* \big| V\big]^2 \Big] \geq 0 \,.
$$
Let us suppose now that $H$ is not strictly increasing on $[0,\rho]$. This means that we can find $q \in (0,\rho)$ such that $H'(q)=0$ and therefore $\E[f(Y^{(q)})W^* | V] = 0$ almost-surely. This gives that for almost all $v \in \R$,
$$
\E \Big[ W \int f(y) P_{\rm out}(y|\sqrt{q}\, v  +\sqrt{\rho-q}\,W) dy \Big] = 0\,,
$$
where $\E$ is the expectation with respect to $W \sim \cN(0,1)$. Lemma~\ref{lem:constant_function} gives then that the function $x \mapsto \int f(y) P_{\rm out}(y|x) dy$ is almost everywhere equal to a constant: we obtain a contradiction. We conclude that $H$ is strictly increasing on $[0,\rho]$ and thus $\mathcal{E}_f$ is strictly decreasing on $[0,\rho]$.
\end{proof}

\begin{proposition}\label{prop:exists_gen_dec}
	Assume that the hypotheses of Proposition~\ref{prop:gen_continuous} hold.
	If the channel $P_{\rm out}$ is informative, then there exists a continuous bounded function $f: \R \to \R$ such that $x \mapsto \int f(y) P_{\rm out}(y|x)dy$ is not almost everywhere equal to a constant.
\end{proposition}
\begin{proof}
	Let us suppose that for all continuous bounded function $f: \R \to \R$ we have
	$$
	\int f(y) P_{\rm out}(y|x)dy = C_f
	$$
	for almost all $x \in \R$, for some constant $C_f \in \R$. 
	Let $X \sim \N(0,1)$ and $Y \sim P_{\rm out}(\cdot | X)$.
	We have then $\E[f(Y)|X] = C_f = \E[f(Y)]$ almost surely.
	Let $g: \R \to \R$ be another continuous bounded function and compute:
	$$
	\E[g(X)f(Y)] = \E\big[g(X)\E[f(Y)|X]\big] = \E[g(X)] \E[f(Y)] \,.
	$$
	It follows that $X$ and $Y$ are independent: The measures $P_{\rm out}(y|x) \frac{e^{-x^2/2}}{\sqrt{2\pi}} dy dx$ and $\E [P_{\rm out}(y|X)] \frac{e^{-x^2/2}}{\sqrt{2\pi}} dy dx$ are therefore equal. Consequently, for almost every $x,y$ we have
	$$
	P_{\rm out}(y|x) = \E[P_{\rm out}(y|X)] \,.
	$$
	This gives that for almost every $y$, $P_{\rm out}(y|\cdot)$ is almost everywhere equal to a constant. We conclude by the arguments presented at the end of the proof of Proposition \ref{prop:psi_stricly_monoton} that $P_{\rm out}$ is not informative, which is a contradiction.
\end{proof}
	\section{Approximation}\label{Appendix-approx}

Let us recall the various hypotheses considered in this paper, starting with the stronger set:
\begin{enumerate}[label=(H\arabic*),noitemsep]
	\item The prior distribution $P_0$ has a bounded support.
	\item $\varphi$ is a bounded $\cC^2$ function with bounded first and second derivatives w.r.t.\ its first argument.
	\item $(\Phi_{\mu i}) \iid \cN(0,1)$.
\end{enumerate}
The aim of this section is to relax them to the weaker ones:
\begin{enumerate}[label=(h\arabic*),noitemsep]
	\item The prior distribution $P_0$ admits a finite third moment and has at least two points in its support.
	\item There exists $\gamma>0$ such that the sequence $(\E [ \vert\varphi(\frac{1}{\sqrt n}[\boldsymbol{\Phi}\bX^*]_1, \bA_1)\vert^{2+\gamma}])_{n \geq 1}$ is bounded.
	\item The random variables $(\Phi_{\mu i})$ are independent with zero mean, unit variance and finite third moment that is bounded with $n$. 
	\item 
		For almost-all values of $\ba \in \R^{k_A}$ (w.r.t.\ $P_A$), the function $x \mapsto \varphi(x,\ba)$ is continuous almost everywhere. 
\end{enumerate}
The hypotheses on the precence or not of the Gaussian noise in \eqref{measurements} are:
\begin{enumerate}[label=(h5.\alph*),noitemsep]
	\item $\Delta > 0$.
	\item $\Delta = 0$ and $\varphi$ takes values in $\N$.
\end{enumerate}

In this section, we suppose that Theorem~\ref{th:RS_1layer} holds for channels of the form \eqref{measurements} (with $\Delta >0$) under the hypotheses~\ref{hyp:bounded},~\ref{hyp:c2} and~\ref{hyp:phi_gauss2}, as proven in Section~\ref{sec:interpolation}.

We show in this section that this imply that Theorem~\ref{th:RS_1layer} holds under the weaker hypotheses~\ref{hyp:third_moment}-\ref{hyp:um}-\ref{hyp:phi_general}-\ref{hyp:cont_pp}, and either~\ref{hyp:delta_pos} or~\ref{hyp:delta_0}. 
This section is organized as follows: We first prove Theorem~\ref{th:RS_1layer} under~\ref{hyp:third_moment}-\ref{hyp:um}-\ref{hyp:phi_general}-\ref{hyp:cont_pp} and~\ref{hyp:delta_pos} (i.e.\ $\Delta >0$). This is done by first relaxing the hypotheses on $P_0$ and $\bbf{\Phi}$ (Sec.~\ref{sec:relax_p0}) and then the hypotheses on $\varphi$ (Sec.~\ref{sec:relax_varphi}).
Finally, in Sec.~\ref{sec:relax_discrete}, we let $\Delta \to 0$ in order to prove Theorem~\ref{th:RS_1layer} under~\ref{hyp:third_moment}-\ref{hyp:um}-\ref{hyp:phi_general}-\ref{hyp:cont_pp} and~\ref{hyp:delta_0}.

Note that the statement of Theorem~\ref{th:RS_1layer} is equivalent to the statement of Corollary~\ref{cor:mi}, which simply express the result in terms of mutual information. This formulation will be slightly more convenient to relax the hypotheses. We will therefore prove in this section that \eqref{eq:lim_i} holds under the hypotheses~\ref{hyp:third_moment}-\ref{hyp:um}-\ref{hyp:phi_general}-\ref{hyp:cont_pp}, and either~\ref{hyp:delta_pos} or~\ref{hyp:delta_0}. The statement of Theorem~\ref{th:RS_1layer} can then be directly obtained by using the expressions of $I_{P_0}, \, \mathcal{I}_{P_{\rm out}}$ in terms of $\psi_{P_0}, \, \Psi_{P_{\rm out}}$, the relation \eqref{eq:i_f} and  Lemma~\ref{lem:sup_inf_2}.


\subsection{Relaxing the hypotheses on $P_0$ and $\bbf{\Phi}$}\label{sec:relax_p0}

As explained at the beginning of Sec.~\ref{sec:interpolation}, it suffices to consider the case $\Delta=1$. We start by relaxing the hypothesis \ref{hyp:bounded}.

\begin{lemma}[Relaxing $P_0$]\label{lemma:relax_bounded}
	Suppose that~\ref{hyp:third_moment}-\ref{hyp:c2}-\ref{hyp:phi_gauss2} and~\ref{hyp:delta_pos} hold. Then Theorem~\ref{th:RS_1layer} holds.
\end{lemma}
\begin{proof}
	The ideas are basically the same that in \cite{MiolaneXX} (Sec.~6.2.2). We omit the details here for the sake of brevity.
\end{proof}

We now relax the Gaussian assumption on the ``measurement matrix'' $\bbf{\Phi}$.

\begin{lemma}[Relaxing $\bbf{\Phi}$]\label{lemma:univ}
	Suppose that $\varphi: \R \times \R^{k_A} \to \R$ is $\cC^{\infty}$ with compact support and that~\ref{hyp:third_moment}-\ref{hyp:phi_general}-\ref{hyp:delta_pos} hold. Then Theorem~\ref{th:RS_1layer} holds.
\end{lemma}

\begin{proof}
	The proof is based on the Lindeberg generalization theorem (Theorem 2 from~\cite{korada2011lindeberg}) which is a variant of the generalized ``Lindeberg principle'' from \cite{chatterjee2006generalization}:
	\begin{thm}[Lindeberg generalization theorem] \label{th:lindeberg}
		Let $(U_i)_{1 \leq i \leq n}$ and $(V_i)_{1 \leq i \leq n}$ be two collections of random variables with independent components and $f: \R^n \to \R$ a $\mathcal{C}^3$ function. Denote $a_i = | \E U_i - \E V_i |$ and $b_i = | \E [U_i^2] - \E [V_i^2] |$. Then
		\vspace{-0.2cm}
		\begin{align*}
			|\E f(\bbf{U})& - \E f(\bbf{V}) |  \leq
			\sum_{i=1}^n \Big\{
				a_i \E | \partial_i f(U_{1:i-1},0,V_{i+1:n}) | + \frac{b_i}{2} \E | \partial^2_i f(U_{1:i-1}, 0,V_{i+1:n})| \\
				&+ \frac{1}{2} \E \int_{0}^{U_i} |\partial^3_i f(U_{1:i-1}, 0,V_{i+1:n})|(U_i - s)^2 ds 
				+ \frac{1}{2} \E \int_{0}^{V_i} |\partial^3_i f(U_{1:i-1}, 0,V_{i+1:n})|(V_i - s)^2 ds
			\Big\} \,.
		\end{align*}
	\end{thm}
	Let $(\Phi'_{\mu,i}) \iid \cN(0,1)$ and let $(\Phi_{\mu,i})$ be a family of independent random variables, with zero mean and unit variance.
	Let $f_n'$ be the free entropy \eqref{f} with design matrix $\bbf{\Phi'}$ and $f_n$ be the free entropy \eqref{f} with design matrix $\bbf{\Phi}$.

	We will apply Theorem~\ref{th:lindeberg} to the function
	$$
	F: \bbf{U} \in \R^{m \times n}
	\mapsto
	\frac{1}{n} \E\ln
	\int_{\bx,\ba} dP_A(\ba) dP_0(\bx) \, e^{
		-\frac{1}{2} \sum\limits_{\mu=1}^m 
		\left(\varphi\Big(\frac{1}{\sqrt{n}}[\bbf{U} \bX^*]_\mu , \bbf{A}_\mu \Big) - \varphi\Big(\frac{1}{\sqrt{n}}[\bbf{U} \bx]_\mu , \bbf{a}_\mu \Big) + Z_{\mu} \right)^2
	}
	$$
	where the expectation $\E$ is taken w.r.t.\ $\bX^*, \bA$ and $\bbf{Z}$.
	We have
	$$
	f_n = \E F(\bbf{\Phi}) \qquad {\rm and} \qquad f'_n = \E F(\bbf{\Phi'}) \,.
	$$
	It is not difficult to verify that $F$ is a $\cC^3$ function and that for all $1 \leq \mu \leq m$ and $1 \leq i \leq n$:
	$$
	\left\|\frac{\partial^3 F}{\partial U_{\mu,i}^3}\right\|_{\infty} \leq \frac{C}{n^{5/2}} \,,
	$$
	for some constant $C$ that only depends on $\varphi$ and the first three moments of $P_0$.
	Thus, an application of Theorem~\ref{th:lindeberg} gives $|f_n - f_n'| \leq \frac{C}{\sqrt{n}}$.
	By Proposition~\ref{lemma:relax_bounded}, we know that Theorem~\ref{th:RS_1layer} holds for $f_n'$, thus it holds for $f_n$.
\end{proof}

\subsection{Relaxing the hypotheses on $\varphi$} \label{sec:relax_varphi}

It remains to relax the hypotheses on $\varphi$. This section is dedicated to the proof of the following proposition, which is of course exactly the statement of Theorem~\ref{th:RS_1layer}.

\begin{proposition}[Relaxing $\varphi$]\label{prop:approx_phi}
	Suppose that~\ref{hyp:third_moment}-\ref{hyp:um}-\ref{hyp:phi_general}-\ref{hyp:cont_pp} and~\ref{hyp:delta_pos} hold.
	Then, Theorem~\ref{th:RS_1layer} holds for the output channel \eqref{measurements}.
\end{proposition}

To prove Proposition~\ref{prop:approx_phi} we will approximate the function $\varphi$ with a function $\widehat{\varphi}$ which is $\cC^{\infty}$ with compact support.
In the following, $G$ is a standard Gaussian random variable, independent of everything else.

\begin{proposition}\label{prop:density}
	Suppose that~\ref{hyp:third_moment}-\ref{hyp:um}-\ref{hyp:phi_general}-\ref{hyp:cont_pp} hold.
	Then, for all $\epsilon > 0$, there exist $\widehat{\varphi} \in \cC^{\infty}(\R \times \R^{k_A})$ with compact support, such that
	$$
	\E\left[
		(\varphi(\sqrt{\rho}G,\bA) - \widehat{\varphi}(\sqrt{\rho} G,\bA))^2
	\right] \leq \epsilon \,,
	$$
	and for $n$ large enough, we have
	$$
	\E\left[
		\left(\varphi\left(\frac{1}{\sqrt{n}}[\boldsymbol{\Phi} \bX^*]_1,\bA_1\right) - \widehat{\varphi}\left(\frac{1}{\sqrt{n}}[\boldsymbol{\Phi} \bX^*]_1,\bA_1\right)\right)^2
	\right] \leq \epsilon \,.
	$$
\end{proposition}

\begin{proof}
	By the Central Limit Theorem (using the fact that the third moments of $(X^*_i \Phi_{1,i})$ are bounded with $n$, because of hypotheses~\ref{hyp:third_moment} and~\ref{hyp:phi_general})
	\begin{equation}\label{eq:tcl1}
	\left(\frac{[\bbf{\Phi}\bX^*]_1}{\sqrt{n}},\bA_1\right)\xrightarrow[n \to \infty]{(d)}
	(\sqrt{\rho} G,\bA_1) \,.
\end{equation}
	This implies that 
	\begin{equation} \label{eq:conv_law}
		\varphi\left(\frac{[\bbf{\Phi}\bX^*]_1}{\sqrt{n}},\bA_1\right)
		\xrightarrow[n \to \infty]{(d)}
		\varphi(\sqrt{\rho} G,\bA_1) \,,
	\end{equation}
	because $\varphi(\cdot, \bbf{A}_1)$ is almost-surely continuous almost-everywhere, by assumption~\ref{hyp:cont_pp}.
	The following sequence $(\varphi(\frac{[\bbf{\Phi}\bX^*]_1}{\sqrt{n}},\bA_1))_n$ is by assumption~\ref{hyp:um} bounded in $L^2$, thus by \eqref{eq:conv_law} we have that $\E[\varphi(\sqrt{\rho}G,\bA_1)^2] < \infty$.
	Let $\epsilon >0$.
	We have just proved that $\varphi \in L^2(\R \times \R^{k_A})$ with the measure induced by $(\sqrt{\rho}G,\bA_1)$. There exists (see for instance Corollary~4.2.2 in \cite{bogachev2007measure}) a $\cC^{\infty}$ function with compact support $\widehat{\varphi}$ such that $ \E\left[ (\varphi(\sqrt{\rho}G,\bA) - \widehat{\varphi}(\sqrt{\rho} G,\bA))^2 \right] \leq \epsilon $.

	One deduce from \eqref{eq:tcl1} and \eqref{eq:conv_law} that
	$$
	\left(\varphi\left(\frac{1}{\sqrt{n}}[\boldsymbol{\Phi} \bX^*]_1,\bA_1\right) - \widehat{\varphi}\left(\frac{1}{\sqrt{n}}[\boldsymbol{\Phi} \bX^*]_1,\bA_1\right)\right)^2
	\xrightarrow[n \to \infty]{(d)}
	(\varphi(\sqrt{\rho}G,\bA) - \widehat{\varphi}(\sqrt{\rho} G,\bA))^2 \,.
	$$
	Now, hypothesis~\ref{hyp:um} gives that the sequence above is uniformly integrable. This gives that
	$$
	\E \left(\varphi\left(\frac{1}{\sqrt{n}}[\boldsymbol{\Phi} \bX^*]_1,\bA_1\right) - \widehat{\varphi}\left(\frac{1}{\sqrt{n}}[\boldsymbol{\Phi} \bX^*]_1,\bA_1\right)\right)^2
	\xrightarrow[n \to \infty]{}
	\E (\varphi(\sqrt{\rho}G,\bA) - \widehat{\varphi}(\sqrt{\rho} G,\bA))^2 \leq \epsilon \,.
	$$
	Consequently, the left-hand side is smaller that $2 \epsilon$ for $n$ large enough. This concludes the proof.
\end{proof}

In the remaining of this section, we prove Proposition~\ref{prop:approx_phi}.
Let $\epsilon>0$. Let $\varphi$ and $\widehat{\varphi}$ as in Proposition~\ref{prop:density}. 
Let us define $\bbf{Y} = \varphi(n^{-1/2} \bbf{\Phi}\bbf{X}^*, \bbf{A}) + \sqrt{\Delta} Z$ and $\bbf{\widehat{Y}} = \widehat{\varphi}(n^{-1/2} \bbf{\Phi}\bbf{X}^*, \bbf{A}) + \sqrt{\Delta} Z$.

\begin{lemma} \label{lem:approx_f}
	Suppose that~\ref{hyp:third_moment}-\ref{hyp:um}-\ref{hyp:phi_general}-\ref{hyp:cont_pp} and~\ref{hyp:delta_pos} hold.
	There exists a constant $C > 0$ such that for $n$ large enough
	$$
	\Big| \frac{1}{n} I(\bbf{X}^*; \bbf{Y}| \bbf{\Phi}) 
	- \frac{1}{n} I(\bbf{X}^*; \bbf{\widehat{Y}}| \bbf{\Phi}) \Big|
	\leq C \sqrt{\epsilon} \,.
	$$
\end{lemma}
\begin{proof}
	We have, for $n$ large enough
	$$
	\E \| \bbf{Y} - \bbf{\widehat{Y}} \|^2
	= m 
	\E\left[
		\left(\varphi\left(\frac{1}{\sqrt{n}}[\boldsymbol{\Phi} \bX^*]_1,\bA_1\right) - \widehat{\varphi}\left(\frac{1}{\sqrt{n}}[\boldsymbol{\Phi} \bX^*]_1,\bA_1\right)\right)^2
	\right] \leq m \epsilon \,.
	$$
	By Proposition~\ref{prop:free_wasserstein}, we obtain that there exists a constant $C>0$ (that depends only on $\Delta$ and $\varphi$) such that
	$$
	\big| I(\bbf{X}^*; \bbf{Y}| \bbf{\Phi}) 
	- I(\bbf{X}^*; \bbf{\widehat{Y}}| \bbf{\Phi}) \big|
	\leq C m \sqrt{\epsilon} \,,
	$$
	which gives the result.
\end{proof}

Let $P_{\rm out}$ denote the transition kernel associated to $\varphi$ and $\widehat{P}_{\rm out}$ the one associated to $\widehat{\varphi}$.
Analogously to the previous Lemma, one can show:
\begin{lemma}
	There exists a constant $C'>0$ such that for all $q \in [0, \rho]$,  $|\mathcal{I}_{P_{\rm out}}(q) - \mathcal{I}_{\widehat{P}_{\rm out}}(q)| \leq C' \sqrt{\epsilon}$.
\end{lemma}
From there we obtain that
\begin{equation} \label{eq:approx_RS}
	\Big|
	{\adjustlimits \inf_{q \in [0, \rho]} \sup_{r \geq 0}}\, i_{\rm RS}(q,r)
	-
	{\adjustlimits \inf_{q \in [0, \rho]} \sup_{r \geq 0}}\,  \widehat{i}_{\rm RS}(q,r)
	\Big| \leq C' \sqrt{\epsilon} \,.
\end{equation}
Applying Theorem~\ref{th:RS_1layer} for $\widehat{P}_{\rm out}$, we obtain that for $n$ large enough $|\frac{1}{n}I(\bbf{X}^*;\bbf{\widehat{Y}}|\bbf{\Phi}) - \inf_{q \in [0, \rho]} \sup_{r \geq 0} \widehat{i}_{\rm RS}(q,r)| \leq \sqrt{\epsilon}$. We now combine this with \eqref{eq:approx_RS} and Lemma~\ref{lem:approx_f} we obtain that for $n$ large enough
$$
\Big|
\frac{1}{n} I(\bbf{X}^*;\bbf{Y}|\bbf{\Phi})
-
\!\!{\adjustlimits \inf_{q \in [0, \rho]} \sup_{r \geq 0}}\, i_{\rm RS}(q,r)
\Big| 
\leq
\Big|
\frac{1}{n} I(\bbf{X}^*;\bbf{\widehat{Y}}|\bbf{\Phi})
-
\!\!{\adjustlimits \inf_{q \in [0, \rho]}\sup_{r \geq 0}}\,  \widehat{i}_{\rm RS}(q,r)
\Big| 
+(C + C') \sqrt{\epsilon}
\leq (C + C' + 1)  \sqrt{\epsilon} \,,
$$
which concludes the proof of Proposition~\ref{prop:approx_phi}, because of \eqref{eq:i_f} and the definition of the functions $I_{P_0}$ and $\mathcal{I}_{P_{\rm out}}$ in Corollary~\ref{cor:mi}.

\subsection{The case of discrete channels: Removing the Gaussian noise}\label{sec:relax_discrete}

Now that we proved (Proposition~\ref{prop:approx_phi}) that Theorem~\ref{th:RS_1layer} holds under hypotheses~\ref{hyp:third_moment}-\ref{hyp:um}-\ref{hyp:phi_general}-\ref{hyp:cont_pp} and~\ref{hyp:delta_pos}, we are going to show that it holds under~\ref{hyp:third_moment}-\ref{hyp:um}-\ref{hyp:phi_general}-\ref{hyp:cont_pp} and~\ref{hyp:delta_0} by letting $\Delta \to 0$.
We suppose in this section that $\varphi$ takes values in $\N$ and write $\bbf{Y} = \varphi\big(\bbf{\Phi}\bbf{X}^*/\sqrt{n}, \bbf{A}\big)$. By Proposition~\ref{prop:approx_phi} we know that for all $\Delta > 0$,
$$
\frac{1}{n} I(\bbf{X}^*;\bbf{Y} + \sqrt{\Delta} \bbf{Z}|\bbf{\Phi})
\xrightarrow[n \to \infty]{}
{\adjustlimits \inf_{q \in [0, \rho]}\sup_{r \geq 0} }\, \Big\{
	I_{P_0}(r) + \alpha I(W^*; \varphi(\sqrt{q}V + \sqrt{\rho-q}W^*,A) + \sqrt{\Delta}Z|V) - \frac{r}{2}(\rho-q)
\Big\} \,,
$$
where $\bbf{Z} \sim \cN(0,\bbf{I}_m)$ and $(V,W^*,Z,A) \sim \cN(0,1)^{\otimes 3} \otimes P_A$.
Since $\bbf{Y}$ takes values in $\N^m$ and $\varphi$ takes values in $\N$, we can apply Corollary~\ref{cor:cont_I_discrete} twice to obtain that for all $\Delta \in (0,1]$,
$$
\Big|
I(\bbf{X}^*;\bbf{Y} + \sqrt{\Delta} \bbf{Z}|\bbf{\Phi})
-
I(\bbf{X}^*;\bbf{Y}|\bbf{\Phi})
\Big| \leq 100 m e^{-1/(16\Delta)}
$$
and (recall that by definition $\mathcal{I}_{P_{\rm out}}(q)=I\big(W^*; \varphi(\sqrt{q}V + \sqrt{\rho-q}W^*,A)\big|V\big)$):
$$
|I\big(W^*; \varphi(\sqrt{q}V + \sqrt{\rho-q}W^*,A) + \sqrt{\Delta}Z\big|V\big) - \mathcal{I}_{P_{\rm out}}(q)| \leq 100 e^{-1/(16\Delta)}\,.
$$
Since our control over $\Delta$ is uniform in $n$, we can permute the $n \to \infty$ limit with the $\Delta \to 0$ limit to get:
$$
\frac{1}{n} I(\bbf{X}^*;\bbf{Y}|\bbf{\Phi})
\xrightarrow[n \to \infty]{}
{\adjustlimits \inf_{q \in [0, \rho]}\sup_{r \geq 0} }\, \Big\{
	I_{P_0}(r) + \alpha \mathcal{I}_{P_{\rm out}}(q) - \frac{r}{2}(\rho-q)
\Big\} \,.
$$

	\section{Some sup-inf formulas}\label{appendix_sup_inf}
This appendix gathers some useful lemmas for the manipulations of ``sup-inf'' formulas like \eqref{eq:rs_formula}.
\begin{lemma}\label{lem:sup_inf_dual}
	Let $f,g : \R_+ \to \R$ be two non-decreasing convex functions. We have
	$$
	{\adjustlimits \sup_{x \geq 0} \inf_{y \geq 0}} \big\{ f(x) + g(y) - x y\big\}
	=
	{\adjustlimits \sup_{y \geq 0} \inf_{x \geq 0}} \big\{ f(x) + g(y) - x y \big\}\,.
	$$
\end{lemma}
\begin{proof}
	Let us define the monotone conjugate (see the end of §12 of \cite{rockafellar2015convex}) of $f$ and $g$:
	$$
	f^*(y) = \sup_{x \geq 0} \big\{ xy - f(x) \big\}
	\quad \text{and} \quad
	g^*(y) = \sup_{x \geq 0} \big\{ xy - g(x) \big\} \,.
	$$
	These conjugates satisfy an analog of the Fenchel-Moreau Theorem: $f(x) = \sup_{y \geq 0} \{ xy - f^*(y) \}$ and $g(y) = \sup_{x \geq 0} \{ xy - g^*(x) \}$, see Theorem~12.4 from \cite{rockafellar2015convex}.
	We have then
	\begin{align*}
		{\adjustlimits \sup_{x \geq 0} \inf_{y \geq 0}} \big\{f(x) + g(y) - x y \big\}
	&= 
		\sup_{x \geq 0} \big\{f(x) - g^*(x)\big\}
	= 
	\sup_{x \geq 0} \sup_{y \geq 0} \big\{xy - f^*(y) - g^*(x) \big\}
	\\
	&= 
	\sup_{y \geq 0}\Big\{ - f^*(y) +  \sup_{x \geq 0} \big\{  xy  - g^*(x) \big\} \Big\}
	= 
	\sup_{y \geq 0} \big\{- f^*(y) +  g(y) \big\}
	\\
	&=
	{\adjustlimits \sup_{y \geq 0} \inf_{x \geq 0}} \big\{f(x) + g(y) - x y \big\}.
	\end{align*}
\end{proof}
The next Lemma on the Legendre transform will be useful.
\begin{lemma}\label{lem:legendre}
	Let $V \subset \R$ be a non-empty, closed interval and let $g: V \to \R$ be a continuous convex function. Define
	\begin{equation} \label{eq:legendre}
		g^*: x \in  \R \mapsto \sup_{y \in V} \big\{ xy - g(y) \big\} \in \R \cup \{+ \infty\} \,.
	\end{equation}
	Let ${\rm dom} \, g^* = \{ x \in \R \, | \, g^*(x) < \infty \}$.
	Then $g^*$ is a closed convex function and ${\rm dom} \, g^*$ is a non-empty interval.
	Moreover, for all $x \in {\rm dom} \, g^*$,
	\begin{equation}\label{eq:sub_gradient_legendre}
	\partial g^*(x) = {\rm arg\, max}_{y \in V} \big\{ xy - g(y) \big\}\,.
	\end{equation}
	In particular, if $g$ is strictly convex then $g^*$ is differentiable around every point in the interior of ${\rm dom} \, g^*$.
\end{lemma}
\begin{proof}
	We first extend the function $g$ on $\R$ by setting $g(x) = + \infty$ for all $x \notin V$. Notice that this does not change the definition of the function $g^*$.
	$g$ is then a proper, closed convex function (see for instance \cite{rockafellar2015convex} for the definitions of these properties). By Theorem~12.2 in \cite{rockafellar2015convex}, $g^*$ is also a proper closed convex function on $\R$, which gives that ${\rm dom} \, g^*$ is a non-empty interval. We now apply Corollary~23.5.1 from \cite{rockafellar2015convex} to obtain
	$$
	y \in \partial g^*(x) \iff x \in \partial g(y) \iff
		y \ {\rm maximizes}\ \big\{ xy - g(y) \big\}\,,
	$$
	for all $x \in {\rm dom}\, g^*$, which concludes the proof.
\end{proof}
\begin{corollary}\label{Cor:supinf_supinf}
	Let $f:\R_+ \to \R$ be a convex, Lipschitz, non-decreasing function. Define $\rho = \sup_{x \geq 0} f'(x^+)$. Let $g:[0,\rho] \to \R$ be a convex, Lipschitz, non-decreasing function.
	For $q_1\in \R_+$ and $q_2 \in [0,\rho]$ we define $\psi(q_1,q_2) = f(q_1) + g(q_2) - q_1 q_2$. Then
	$$
	{\adjustlimits\sup_{q_1 \geq 0} \inf_{q_2 \in [0,\rho]}} \psi(q_1,q_2)
	=
	{\adjustlimits\sup_{q_2 \in [0, \rho]} \inf_{q_1 \geq 0}} \psi(q_1,q_2) \,.
	$$
\end{corollary}
\begin{proof}
	In order to apply Lemma~\ref{lem:sup_inf_dual} we need to extend $g$ on $\R_+$. We thus define for $x \geq 0$
	$$
	g(x) = 
	\begin{cases}
		g(x) & \text{if} \ x \leq \rho \,, \\
		g(\rho) + (x-\rho)g'(\rho^-) & \text{if} \ x \geq \rho \,.
	\end{cases}
	$$
	Obviously $g$ is a convex, Lipschitz, non-decreasing function on $\R_+$. One can thus apply Lemma~\ref{lem:sup_inf_dual}:
	\begin{equation}\label{eq:sup_inf_R}
		{\adjustlimits \sup_{q_1 \geq 0} \inf_{q_2 \geq 0}} \big\{ f(q_1) + g(q_2) - q_1 q_2 \big\}
	=
	{\adjustlimits \sup_{q_2 \geq 0} \inf_{q_1 \geq 0}}  \big\{ f(q_1) + g(q_2) - q_1 q_2 \big\} \,.
	\end{equation}
	We will show now that $\sup_{q_1 \geq 0} \inf_{q_2 \geq 0} \psi(q_1,q_2)= \sup_{q_1 \geq 0} \inf_{q_2 \in [0,\rho]} \psi(q_1,q_2)$.
	Let us define for $q_1 \geq 0$
	$$
	g^*(q_1) = \sup\limits_{q_2 \in [0,\rho]} \{ q_1 q_2 - g(q_2) \}
	\quad \text{and} \quad
	h(q_1) = \inf_{q_2 \in [0,\rho]} \big\{f(q_1) + g(q_2) - q_1 q_2 \big\} = f(q_1) - g^*(q_1) \,.
	$$
	For $q_1 \geq g'(\rho^-)$ we have $g^*(q_1) = q_1 \rho - g(\rho)$. The function $h$ is therefore non-increasing on $[g'(\rho^-),+\infty)$, because $f$ is $\rho$-Lipschitz.
	We get that $h = f - g^*$ achieves its supremum on $[0,g'(\rho^-)]$. 
	Let $q_1^*$ be the smallest point at which this supremum is achieved. Let us show that $\inf_{q_2 \in [0,\rho]} \{ g(q_2) - q_1^* q_2 \}= \inf_{q_2 \geq 0}\{ g(q_2) -q_1^* q_2\}$. 
	\begin{itemize}
		\item If $q_1^* = 0$, then the minimum over $[0,\rho]$ is achieved at $q_2 = 0$, because $g$ is non-decreasing. By convexity, $q_2$ is also the minimizer over $\R_+$: both infimum are equal.
		\item If $q_1^* >0$, the optimality condition of $q_1^*$ gives $f'(q_1^{*-}) - (g^*)'(q_1^{*-}) \geq 0$. By \eqref{eq:sub_gradient_legendre} we obtain that there exists $q_2^* \in \argmin_{q_2 \in [0,\rho]} \{g(q_2) - q_1^* q_2 \}$ such that $f'(q_1^{*-}) \geq q_2^*$. If $q_2^* < \rho$ we conclude, as above, that both infimum are equal. Suppose now that $q^*_2 = \rho$ and define $q_1' = g'(\rho^-)$. By the optimality condition of $q_2^* = \rho$ we have $q_1' = g'(\rho^-) \leq q_1^*$. Compute
			$$
			h(q_1^*) - h(q_1') = f(q_1^*) - f(q_1') - \rho (q_1^*-q_1') \leq 0
			$$
			because $f$ is $\rho$-Lipschitz. Since $q_1' \leq q_1^*$ and $q_1^*$ is defined as the smallest maximizer of $h$, we get that $q_1^* = q_1'$. The left-hand derivative of $q_2 \mapsto g(q_2) - q_1^* q_2$ at $q=\rho$ is therefore equal to $0$: $\rho$ minimizes $q_2 \mapsto g(q_2) - q_1^* q_2$ over $\R_+$: both infimum are equal.
	\end{itemize}
	We have proved that $\inf_{q_2 \in [0,\rho]} \{ g(q_2) - q_1^* q_2 \}= \inf_{q_2 \geq 0} \{ g(q_2) -q_1^* q_2 \}$. Therefore
	$$
 {\adjustlimits \sup_{q_1 \geq 0} \inf_{q_2 \in [0,\rho]}} \psi(q_1,q_2)=
  \inf_{q_2 \in [0,\rho]} \psi(q_1^*,q_2)
  =
  \inf_{q_2 \geq 0} \psi(q_1^*,q_2)
  \leq
 {\adjustlimits \sup_{q_1 \geq 0} \inf_{q_2 \geq 0}} \psi(q_1,q_2) \,.
	$$
	We conclude that $\sup_{q_1 \geq 0} \inf_{q_2 \geq 0} \psi(q_1,q_2)= \sup_{q_1 \geq 0} \inf_{q_2 \in [0,\rho]} \psi(q_1,q_2)$ because the converse inequality is trivial.
	It remains to show now that	$\sup_{q_2 \geq 0} \inf_{q_1 \geq 0} \psi(q_1,q_2)= \sup_{q_2 \in [0,\rho]} \inf_{q_1 \geq 0} \psi(q_1,q_2)$ to prove the Lemma, because of \eqref{eq:sup_inf_R}. The inequality ``$\geq$'' is obvious and the inequality ``$\leq$'' follows from the fact that $\inf_{q_1 \geq 0} \psi(q_1,q_2) = - \infty$ if $q_2 > \rho$.
\end{proof}

\begin{lemma}\label{lem:sup_inf_2}
	Let $g$ be a strictly convex, differentiable, Lipschitz non-decreasing function on $\R_+$. Define $\rho = \sup_{x \geq 0} g'(x)$. Let $f$ be a convex, continuous, strictly increasing function on $[0,\rho]$, differentiable on $[0,\rho)$.
	For $(q_1, q_2) \in [0,\rho] \times \R_+$ we define $\psi(q_1,q_2) = f(q_1) + g(q_2) - q_1 q_2$.
	Then
	\begin{equation}\label{eq:sup_inf_no_lip}
	{\adjustlimits \sup_{q_1 \in [0, \rho]} \inf_{q_2 \geq 0}} \psi(q_1,q_2)
	=
	\sup_{(q_1,q_2)\in\Gamma}
	\psi(q_1,q_2) \,,
	\end{equation}
	where 
	$$\Gamma =
	\left\{
		(q_1,q_2) \in [0,\rho]\times (\R_+ \cup \{+\infty\}) \, \middle| \,
		\begin{array}{lll}
			q_1 &=& g'(q_2) \\
			q_2 &=& f'(q_1)
		\end{array}
	\right\} \,,
	$$
	where all the function are extended by there limits at the points at which they may not be defined (for instance $g'(+\infty) = \displaystyle\lim_{q \to \infty}g'(q)$, $f'(\rho) = \displaystyle\lim_{q \to \rho} f'(q)$).
	Moreover, the above extremas are achieved precisely on the same couples.
\end{lemma}
\begin{proof}
	Let $q_1^*$ be a maximizer of $f-g^*$ over $[0,\rho]$. $q_1^*$ is well defined because $f$ is continuous and $g^*$ is continuous over $[0,\rho)$ and is either continuous at $\rho$ or goes to $+\infty$ at $\rho$ (this comes from the fact that $g^*$ is a closed convex function, see Lemma \ref{lem:legendre}).

\textbf{Case 1:} $0 < q_1^* < \rho$. By strict convexity of $g$, $\psi(q_1, \cdot)$ admits a unique minimizer $q_2^*$ and $(g^*)'(q_1^*) = q_2^*$ by Lemma \ref{lem:legendre}. Thus, the optimality condition at $q_1^*$ gives
	$$
	0 = f'(q_1^*) - (g^*)'(q_1^*) = f'(q_1^*) - q_2^* \,.
	$$
	The optimality of $q_2^*$ gives then $q_1^* \leq g'(q_2^*)$. Suppose that $q_1^* < g'(q_2^*)$. This is only possible when $q_2^* = 0$. Define $q_1' = g'(q_2^*)=g'(0)$. Remark that $g^*(q_1') = -g(0) =g^*(q_1^*)$. We supposed that $q_1' > q_1^*$ thus, by strict monotonicity of $f$, $f(q_1') - g^*(q_1') > f(q_1^*) - g^*(q_1^*)$ which contradict the optimality of $q_1^*$. We obtain therefore that $q_1^* = g'(q_2^*)$.
	\\

	\textbf{Case 2:} $q_1^* = 0$. The optimality condition gives now
	\begin{equation}\label{eq:opt_case2}
	0 \leq f'(q_1^* = 0) \leq q_2^* \,,
\end{equation}
where $q_2^*$ is again the unique minimizer of $\psi(q_1^*= 0,\cdot)= f(0) + g$. $g$ is strictly increasing, so $q_2^* = 0$. Therefore $q_2^* = 0 = f'(q_1^*=0)$, by \eqref{eq:opt_case2}. As before we have necessarily, by optimality of $q_2^*$ that $q_1^* = g'(q_2^*)$. 
	\\

	\textbf{Case 3:} $q_1^* = \rho$. 
	In that case $\argmin_{q_2 \geq 0} \{ g(q_2) - q_1^* q_2 \} = \emptyset$ because $g$ is strictly convex and $\rho$-Lipschitz. Lemma \ref{lem:legendre} gives then that $\partial g^* (\rho) = \emptyset$ which implies (see Theorem~23.3 from \cite{rockafellar2015convex}) that $(g^*)'(\rho^-) = + \infty$. 
	Since $q_1^* = \rho$ maximizes $f-g^*$, we necessarily have then $f'(\rho^-) = + \infty$.
	
	Using the slight abuse of notation explained in the Lemma, we have $f'(q_1^*) = + \infty = q_2^*$, where $q_2^* = +\infty$ is the unique ``minimizer'' of $\psi(q_1^*,\cdot)$, by strict convexity of $g$. By definition of $\rho$ we have also $g'(q_2^*) = g'(+\infty) = \rho = q_1^*$.
	\\

	We conclude from the tree cases above that the ``sup-inf'' in \eqref{eq:sup_inf_no_lip} is achieved, and that all the couples $(q_1^*,q_2^*)$ that achieve this ``sup-inf'' belong to $\Gamma$. Thus
	$$
	{\adjustlimits \sup_{q_1 \in [0, \rho]} \inf_{q_2 \geq 0}} \psi(q_1,q_2)
	\leq
	\sup_{(q_1,q_2)\in\Gamma}
	\psi(q_1,q_2) \,.
	$$

	Let now be $(q_1,q_2) \in \Gamma$. By convexity of $g$ we see easily that
	$
	\psi(q_1,q_2) = \inf_{q_2'} \psi(q_1,q_2')
	$.
	Thus, $\psi(q_1,q_2) \leq \sup_{q_1'} \inf_{q_2'} \psi(q_1',q_2')$.
	Therefore 
	$$
	\sup_{(q_1,q_2)\in\Gamma}
	\psi(q_1,q_2)
	\leq
	{\adjustlimits \sup_{q_1 \in [0, \rho]} \inf_{q_2 \geq 0}} \psi(q_1,q_2)
	\,.
	$$
	This concludes the proof of \eqref{eq:sup_inf_no_lip}.
	It remains to see that a couple $(q_1^*,q_2^*) \in \Gamma$ that achieves the supremum in \eqref{eq:sup_inf_no_lip} also achieves the ``sup-inf''. This simply follows from the fact that $\psi(q_1^*,q_2^*) = \inf_{q_2} \psi(q_1^*,q_2)$ and \eqref{eq:sup_inf_no_lip}.
\end{proof}

		\section{Concentration of free entropy and overlaps}

		\subsection{Concentration of the free entropy}\label{appendix_concentration}
The goal of this appendix is to prove that the free entropy of the interpolating model studied in Sec.~\ref{interp-est-problem} concentrates around its expectation. 
To simplify the notations 
we use  $C(\varphi, S, \alpha)$ for a generic non-negative constant depending {\it only} on $\varphi$, $S$ and $\alpha$ ($S$ is the supremum over the signal values). We will also use the notation $K=1+\max(\rho,r_{\rm max})$ for a constant (depending only on $\rho\le S^2$ and $\varphi$) that upper bounds both $R_1$ and $R_2$ given by \eqref{R1R2}. It is also 
understood that $n$ and $m$ are large enough and $m/n \to \alpha$.

\begin{thm}[Free entropy concentration]\label{concentrationtheorem}
	Under assumptions~\ref{hyp:bounded},~\ref{hyp:c2} and \ref{hyp:phi_gauss2} there exists a non-negative constant $C(\varphi, S, \alpha)$ such that the partition function \eqref{Zt} concentrates as
\begin{align}\label{fluctuation}
	\Var\Big( \frac{1}{n}\ln {\cZ}_{t,\epsilon}\Big)=\mathbb{E}\Big[\Big( \frac{1}{n}\ln \mathcal{Z}_{t,\epsilon} - \frac{1}{n}\mathbb{E}\ln \mathcal{Z}_{t,\epsilon} \Big)^2\Big] 
\leq 
\frac{C(\varphi, S, \alpha)}{n}\,.
\end{align}
\end{thm}

The remaining of this appendix is dedicated to the proof of Theorem~\ref{concentrationtheorem}.
We first recall some set-up and notation for the convenience of the reader.
Recall that the interpolating Hamiltonian \eqref{interpolating-ham}--\eqref{stmu} is 
\begin{align}\label{interp-ham}
	\cH_{t,\epsilon}(\bx,\bw;\bY_t,\bY_t',\boldsymbol{\Phi},\bV)=- \sum_{\mu=1}^{m}
	\ln P_{\rm out} ( Y_{t,\mu} |s_{t, \mu}(\bx, w_\mu) ) + \frac{1}{2} \sum_{i=1}^{n}(Y'_{t,i}  - \sqrt{R_1(t)}\, x_i)^2
\end{align}
where 
\begin{align*}
s_{t, \mu}(\bx, w_\mu)\defeq \sqrt{\frac{1-t}{n}}[\boldsymbol{\Phi} \bx]_\mu  + k_1(t) V_{\mu} + k_2(t) w_{\mu}\,, 
\quad 
k_1(t) \defeq  \sqrt{R_2(t)}\,,
\quad
k_2(t) \defeq  \sqrt{\rho t - R_2(t)+2s_n}\,.
\end{align*}
We find it convenient to use the random function representation \eqref{measurements} for the interpolating model, namely 
\begin{align*}
\begin{cases}
Y_{t,\mu} =  \varphi\Big(\sqrt{\frac{1-t}{n}}[\boldsymbol{\Phi} \bX^*]_\mu+ k_1(t) V_{\mu} + k_2(t) W_{\mu}^*  , \bA_\mu \Big) + Z_\mu\,,\\
Y_{t, i}^\prime = \sqrt{R_1(t)}\, X_i^* + Z_i^\prime\,.
\end{cases}
\end{align*}
In this representation the random variables
$(\bA_\mu)_{1 \leq \mu \leq m} \iid P_A$ are arbitrary, and $(Z_\mu)_{1 \leq \mu \leq m} \iid {\cal N}(0,1)$, $(Z'_i)_{1 \leq i \leq n} \iid {\cal N}(0,1)$.
We have
\begin{align}
P_{\rm out}(Y_{t,\mu} | s_{t, \mu}(\bx, w_\mu))
&
= \int dP_A(\ba_\mu)\, \frac{1}{\sqrt{2\pi}}
\exp\Big\{ -\frac{1}{2}\Big(Y_{t,\mu} - \varphi(s_{t,\mu}(\bx,w_\mu), \ba_\mu) \Big)^2\Big\}
\nonumber \\ 
&
= \int dP_A(\ba_\mu)\, \frac{1}{\sqrt{2\pi}}\exp\Big\{-\frac{1}{2}\Big(\Gamma_{t,\mu}(\bx, w_\mu, \ba_\mu)+Z_{\mu}\Big)^2 \Big\}
\label{RFrep}
\end{align}
where, using the random function representation,  
\begin{align}
	\Gamma_{t,\mu}(\bx, &w_\mu, \ba_\mu) \label{3}
	\\
	&=  \varphi\Big(\sqrt{\frac{1-t}{n}}[\boldsymbol{\Phi} \bX^*]_\mu+ k_1(t) V_{\mu} + k_2(t) W_\mu^* , \bA_\mu \Big)  
- \varphi\Big(\sqrt{\frac{1-t}{n}}[\boldsymbol{\Phi} \bx]_\mu + k_1(t) V_{\mu} + k_2(t) w_\mu , \ba_\mu \Big)
\,.	 \nonumber
\end{align}
From \eqref{interp-ham}, \eqref{RFrep}, \eqref{3} we can express the free entropy of the interpolating  model as 
\begin{align}\label{entro-rep}
	\frac{1}{n}\ln \mathcal{Z}_{t,\epsilon} =\frac{1}{n} \ln \int  dP_0(\bx)dP_A(\ba)\mathcal{D}\bw \, e^{- \cH_{t,\epsilon}(\bx,\bw, \ba)} - \frac{m}{2n} \ln(2\pi)
\end{align}
where $\mathcal{D}\bw$ denote the standard $m$-dimensional Gaussian measure and where the Hamiltonian $\cH_{t,\epsilon}$ is re-expressed as
\begin{align}
\cH_{t,\epsilon}(\bx,\bw, \ba) =
	\frac{1}{2}\sum_{\mu=1}^{m} \big(\Gamma_{t,\mu}(\bx, w_\mu, \ba_\mu)+Z_\mu\big)^2
	+ 
	\frac{1}{2} \sum_{i=1}^{n}\big(\sqrt{R_1(t)}\, (X_i^*-x_i) + Z_i^\prime \big)^2
	\,.
\end{align}
The interpretation here is that $\bx, \bw, \ba$ are {\it annealed} variables and $\boldsymbol{\Phi}, \bV, \bA, \bY_t, \bY^\prime_t, \bX^*, \bW^*$, or equivalently $\boldsymbol{\Phi}, \bV, \bA, \bZ, \bZ^\prime, \bX^*, \bW^*$ are {\it quenched}. The inference problem is to recover $\bX^*, \bW^*$ given $\boldsymbol{\Phi}, \bV, \bY_t, \bY^\prime_t$. The free entropy can be further re-expressed as
\begin{align}\label{Z-Z}
\frac{1}{n}\ln \mathcal{Z}_{t,\epsilon} = \frac{1}{n}\ln\hat{\mathcal{Z}}_{t,\epsilon} - \frac{1}{2n}\sum_{\mu=1}^m Z_\mu^2 
-\frac{1}{2n}\sum_{i=1}^n Z_i^{\prime 2} - \frac{m}{2n} \ln(2\pi)
\end{align}
where
\begin{align}\label{newZ}
	\frac{1}{n}\ln\hat{\mathcal{Z}}_{t,\epsilon} &= \frac{1}{n}\ln \int  dP_0(\bx)dP_A(\ba)\mathcal{D}\bw \, e^{- \hat{\cH}_{t,\epsilon}(\bx,\bw, \ba)}\,,\\
	\hat{\cH}_{t,\epsilon}(\bx,\bw, \ba) &= 
\frac{1}{2}\sum_{\mu=1}^{m} \Big\{\Gamma_{t,\mu}(\bx, w_\mu, \ba_\mu)^2 
+ 2 Z_\mu \Gamma_{t,\mu}(\bx, w_\mu, \ba_\mu)\Big\} \nonumber
\\
&\qquad+ \frac{1}{2} \sum_{i=1}^{n}\Big\{R_1(t)( X_i^*- x_i)^2 + 2 Z_i^\prime \sqrt{R_1(t)}( X_i^*-  x_i)\Big\} \,.\label{explicit-ham}
\end{align}
In order to prove Theorem~\ref{concentrationtheorem} it remains to show that there exists a constant $C(\varphi, S, \alpha)>0$ such that $\Var( \ln \hat{\cZ}_{t,\epsilon}/n) \leq C(\varphi, S, \alpha)/n$.
This concentration property together with \eqref{Z-Z} implies \eqref{fluctuation}. 

%
We will first show concentration w.r.t.\ all Gaussian variables $\boldsymbol{\Phi}, \bV, \bZ, \bZ^\prime, \bW^*$ thanks to the classical Gaussian Poincaré inequality,
then the concentration w.r.t.\ $\bA$ and finally the one w.r.t.\ $\bX^*$ using classical bounded differences arguments. The order in which we prove the concentrations matters. We recall here these two variances bounds. The reader can refer to \cite{boucheron2004concentration} (Chapter 3) for detailed proofs of these statements.
\begin{proposition}[Gaussian Poincaré inequality]\label{poincare}
	Let $\bU = (U_1, \dots, U_N)$ be a vector of $N$ independent standard normal random variables. Let $g: \mathbb{R}^N \to \mathbb{R}$ be a continuously differentiable function. Then
 \begin{align}
	 \Var(g(\bU)) \leq \E \big[\| \nabla g (\bU) \|^2 \big] \,.
 \end{align}
\end{proposition}

\begin{proposition}[Bounded difference]\label{bounded_diff}
	Let $\mathcal{U} \subset \R$.
	Let $g: \mathcal{U}^N \to \mathbb{R}$ a function 
that satisfies the bounded difference property, i.e., there exists some constants $c_1, \dots, c_N \geq 0$ such that
$$
\sup_{\substack{u_1, \ldots, u_N\in \mathcal{U}^N \\ u_i' \in \mathcal{U}}}
\vert g(u_1, \dots, u_i, \ldots, u_N) - g(u_1, \dots, u_i', \ldots, u_N)\vert \leq c_i \quad \text{for all} \quad 1 \leq i \leq N \,.
$$ 
Let $\bbf{U}=(U_1, \dots, U_N)$ be a vector of $N$ independent random variables that take values in $\mathcal{U}$. Then
 \begin{align}
	 \Var(g(\bU)) \leq \frac{1}{4} \sum_{i=1}^N c_i^2 \,.
 \end{align}
\end{proposition}
\subsubsection{Concentration with respect to the Gaussian random variables $\bZ$, $\bZ'$, $\bV$, $\bW^*$, $\bbf\Phi$}
\begin{lemma} \label{lem:concentration_gauss}
Let $\mathbb{E}_{G}$ denotes the joint expectation w.r.t.\ $\bZ, \bZ', \bV, \bW^*, \bbf{\Phi}$ only.
	There exists a constant $C(\varphi, S, \alpha)>0$ such that
	\begin{equation}
		\E \Big[\Big(
				\frac{1}{n} \ln \hat{\cZ}_{t,\epsilon} - \frac{1}{n} \E_{G}\ln \hat{\cZ}_{t,\epsilon} 
		\Big)^2\Big] \leq \frac{C(\varphi, S, \alpha)}{n} \,.
	\end{equation}
\end{lemma}

Lemma~\ref{lem:concentration_gauss} follows directly from Lemmas~\ref{lem:concentration_gauss1} and~\ref{lem:concentration_gauss2} below.

\begin{lemma} \label{lem:concentration_gauss1}
Let $\mathbb{E}_{\bZ, \bZ^\prime}$ denotes the expectation w.r.t.\ $\bZ, \bZ^\prime$ only.
	There exists a constant $C(\varphi, S, \alpha)>0$ such that
	\begin{equation}
		\E \Big[\Big(
				\frac{1}{n} \ln \hat{\cZ}_{t,\epsilon} - \frac{1}{n} \E_{\bZ, \bZ'} \ln \hat{\cZ}_{t,\epsilon} 
		\Big)^2\Big] \leq \frac{C(\varphi, S, \alpha)}{n} \,.
	\end{equation}
\end{lemma}
\begin{proof}
We consider here $g = \ln\hat{\mathcal{Z}}_{t,\epsilon}/n$ only as a function of $\bZ$ and $\bZ^\prime$ and work conditionally on all other random variables. 
We have
\begin{align}\label{gra}
\Vert \nabla g\Vert^2 = 
\sum_{\mu=1}^m \Big|\frac{\partial g}{\partial Z_\mu}\Big|^2  
+ 
\sum_{i=1}^n \Big| \frac{\partial g}{\partial Z_i^\prime}\Big|^2.
\end{align} 
Each of these partial derivatives are of the form 
$|\partial_u g| = |n^{-1} \langle \partial_u \hat{\mathcal{H}}_{t,\epsilon}\rangle_{\hat{\mathcal{H}}_{t,\epsilon}}|$
where the Gibbs bracket $\langle - \rangle_{\hat{\mathcal{H}}_{t,\epsilon}}$ pertains to the effective Hamiltonian \eqref{explicit-ham}.
We find 
\begin{align*}
& 
\Big\vert\frac{\partial g}{\partial Z_\mu}\Big\vert 
= n^{-1} \big\vert\langle \Gamma_{t,\mu} \rangle_{\hat{\mathcal{H}}_{t,\epsilon}}\big\vert
\leq 
2 n^{-1} \sup\vert \varphi\vert\,,
\nonumber \\ &
\Big\vert\frac{\partial g}{\partial Z_i^\prime}\Big\vert 
 = n^{-1}\sqrt{R_1(t)} \big\vert X_i^*-\langle  x_i \rangle_{\hat{\mathcal{H}}_{t,\epsilon}}\big\vert
\leq 
2n^{-1}\sqrt{K}S\,,
\end{align*}
and replacing in \eqref{gra} we get
$ \Vert \nabla g\Vert^2 \leq 4 n^{-1}(\frac{m}{n}(\sup\vert\varphi\vert)^2 + KS^2)$.
Applying Proposition~\ref{poincare} we have 
\begin{align}\label{eq:var_part_Z}
	\E_{\bZ,\bZ'} \Big[\Big(
				\frac{1}{n} \ln \hat{\cZ}_{t,\epsilon} - \frac{1}{n}\E_{\bZ, \bZ'} \ln \hat{\cZ}_{t,\epsilon} 
		\Big)^2\Big] \leq \frac{C(\varphi, S, \alpha)}{n}\,.
\end{align}
Taking the expectation in \eqref{eq:var_part_Z} gives the lemma.
\end{proof}

\begin{lemma}\label{lem:concentration_gauss2}
	There exists a constant $C(\varphi, S, \alpha)>0$ such that
	\begin{equation}
		\E \Big[\Big(
				\frac{1}{n}\E_{\bZ, \bZ'} \ln \hat{\cZ}_{t,\epsilon} - \frac{1}{n}\E_{G} \ln \hat{\cZ}_{t,\epsilon} 
		\Big)^2\Big] \leq \frac{C(\varphi, S, \alpha)}{n} \,.
	\end{equation}
\end{lemma}
\begin{proof}
We consider here $g = \mathbb{E}_{\bZ, \bZ^\prime}\ln\hat{\mathcal{Z}}_{t,\epsilon}/n$ as a function of $\bV$, $\bW^*$, $\bbf{\Phi}$ and we work conditionally on the other random variables. Let $\partial_x\varphi$ be the derivative of $\varphi$ w.r.t. its first argument. We compute
\begin{align*}
\Big\vert\frac{\partial g}{\partial V_\mu}\Big\vert
& = n^{-1}\Big\vert \mathbb{E}_{\bZ, \bZ^\prime}\Big\langle (\Gamma_{t, \mu} +Z_\mu)\frac{\partial \Gamma_{t, \mu}}{\partial V_\mu} \Big\rangle_{\hat{\mathcal{H}}_{t,\epsilon}}\Big\vert
\nonumber \\ &
\leq 
n^{-1}\mathbb{E}_{\bZ, \bZ^\prime}\Big[(2\sup\vert \varphi\vert  +\vert Z_\mu\vert) \, 2\sqrt{K}\sup\vert \partial_x\varphi\vert\Big]
=
n^{-1} \Big(2\sup\vert \varphi\vert  + \sqrt{\frac{2}{\pi}}\Big) \,2\sqrt{K}\sup\vert \partial_x\varphi\vert\,.
\end{align*}
The same inequality holds for $\vert\frac{\partial g}{\partial W_\mu^*}\vert$. To compute 
the derivative w.r.t.\ $\Phi_{\mu i}$ we first remark 
\begin{align*}
\frac{\partial \Gamma_{t, \mu}}{\partial \Phi_{\mu i}}
= &
\sqrt{\frac{1-t}{n}} \Big\{X_i^*\,\partial_x\varphi\Big(\sqrt{\frac{1-t}{n}}[\boldsymbol{\Phi} \bX^*]_\mu+ k_1(t) V_{\mu} + k_2(t) W_\mu^* , \bA_\mu \Big)  
\nonumber \\ &
\qquad- x_i\,\partial_x\varphi\Big(\sqrt{\frac{1-t}{n}}[\boldsymbol{\Phi} \bx]_\mu+ k_1(t) V_{\mu} + k_2(t) w_\mu , \ba_\mu \Big)  \Big\}.
\end{align*}
Therefore, as $t\in[0,1]$,
\begin{align*}
\Big\vert\frac{\partial g}{\partial \Phi_{\mu i}}\Big\vert
& 
=
n^{-1}\Big\vert \mathbb{E}_{\bZ, \bZ^\prime}\Big\langle (\Gamma_{t, \mu} +Z_\mu)\frac{\partial \Gamma_{t, \mu}}{\partial \Phi_{\mu i}} \Big\rangle_{\hat{\mathcal{H}}_{t,\epsilon}}\Big\vert
\nonumber \\ &
\leq 
n^{-3/2}\mathbb{E}_{\bZ, \bZ^\prime}\Big[(2\sup\vert \varphi\vert +\vert Z_\mu\vert)
\,2S \sup\vert\partial_x\varphi\vert\Big]
=
n^{-3/2} \Big(2\sup\vert \varphi\vert +\sqrt{\frac{2}{\pi}}\Big)\,
2S \sup\vert\partial_x\varphi\vert\,.
\end{align*}
Putting these inequalities together we find
\begin{align*}
\Vert \nabla g\Vert^2 & = \sum_{\mu=1}^m \Big\vert\frac{\partial g}{\partial V_\mu}\Big\vert^2
+
\sum_{\mu=1}^m \Big\vert\frac{\partial g}{\partial W_\mu^*}\Big\vert^2
+
\sum_{\mu=1}^m\sum_{i=1}^n \Big\vert\frac{\partial g}{\partial \Phi_{\mu i}}\Big\vert^2
\nonumber \\ &
\leq 
2\frac{m}{n^2} \Big(2\sup\vert \varphi\vert  + \sqrt{\frac{2}{\pi}}\Big)^2 4K(\sup\vert \partial_x\varphi\vert)^2 + \frac{mn}{n^3}
\Big(2\sup\vert \varphi\vert +\sqrt{\frac{2}{\pi}}\Big)^2
\,4S^2 (\sup\vert\partial_x\varphi\vert)^2 \,.
\end{align*}
The lemma follows again from Proposition~\ref{poincare}.
\end{proof}

\subsubsection{Bounded difference with respect to $A_\mu$}

The next step is an application of the variance bound of Lemma~\ref{bounded_diff} to show that $\mathbb{E}_{G}\ln\hat{\mathcal{Z}}_{t,\epsilon}/n$ concentrates w.r.t.\ $\bA$ (we still keep $\bX^*$ fixed for the moment). 
\begin{lemma}\label{lem:concentration_A}
	Let $\E_{\bA}$ denotes the expectation w.r.t.\ $\bA$ only. 
There exists a constant $C(\varphi, \alpha)>0$ such that
	\begin{equation}
		\E \Big[\Big(
				\frac{1}{n}\E_{G} \ln \hat{\cZ}_{t,\epsilon}  - \frac{1}{n}\E_{G,\bA} \ln \hat{\cZ}_{t,\epsilon} 
		\Big)^2\Big] \leq \frac{C(\varphi, \alpha)}{n} \,.
	\end{equation}
\end{lemma}
\begin{proof}
	Let us consider $g=\mathbb{E}_{G}\ln\hat{\mathcal{Z}}_{t,\epsilon}/n$ as a function of $\bA$ only.
	Let $\nu \in \{1, \dots, m \}$.
	We must estimate variations $g(\bA) - g(\bA^{(\nu)})$ corresponding to two configurations $\bA$ and $\bA^{(\nu)}$
 with $A_{\mu}^{(\nu)} =
A_{\mu}$ for $\mu \neq \nu$ and $A_{\nu}^{(\nu)} = \tilde A_{\nu}$. 
We will use the notations $\hat{\mathcal{H}}_{t,\epsilon}^{(\nu)}$ and $\Gamma_{t,\mu}^{(\nu)}$ to
denote respectively the quantities $\hat{\mathcal{H}}_{t,\epsilon}$ and $\Gamma_{t,\mu}$ where $\bA$ is replaced by $\bA^{(\nu)}$.
By an application of 
Jensen's inequality one finds
\begin{align}\label{jensen1}
\frac{1}{n} \mathbb{E}_G \langle \hat{\mathcal{H}}_{t,\epsilon}^{(\nu)} - \hat{\mathcal{H}}_{t,\epsilon} \rangle_{\hat{\mathcal{H}}_{t,\epsilon}^{(\nu)}} 
\leq 
g(\bA) - g(\bA^{(\nu)}) 
\leq 
\frac{1}{n} \mathbb{E}_G\langle \hat{\mathcal{H}}_{t,\epsilon}^{(\nu)} - \hat{\mathcal{H}}_{t,\epsilon}\rangle_{\hat{\mathcal{H}}_{t,\epsilon}} 
\end{align}
where the Gibbs brackets pertain to the effective Hamiltonians \eqref{explicit-ham}.
From \eqref{explicit-ham} we obtain 
\begin{align*}
\hat{\mathcal{H}}_{t,\epsilon}^{(\nu)} - \hat{\mathcal{H}}_{t,\epsilon} 
= \frac{1}{2}\sum_{\mu=1}^m \left( \Gamma_{t,\mu}^{{(\nu)}2} - \Gamma_{t,\mu}^2 + 2 Z_\mu(\Gamma_{t,\mu}^{(\nu)} - \Gamma_{t,\mu})\right)
= \frac{1}{2} \left( \Gamma_{t,\nu}^{{(\nu)}2} - \Gamma_{t,\nu}^2 + 2 Z_\nu(\Gamma_{t,\nu}^{(\nu)} - \Gamma_{t,\nu}) \right) \,.
\end{align*}
Consequently
\begin{align}\label{entrop-diff-basic}
\frac{1}{2n} 
\mathbb{E}_G
\Big\langle
\Gamma_{t,\nu}^{(\nu)2}
-  \Gamma_{t,\nu}^2 + 2Z_\nu(\Gamma_{t,\nu}^{(\nu)} - \Gamma_{t,\nu}) \Big\rangle_{\hat{\mathcal{H}}_{t,\epsilon}^{(\nu)}} 
&\leq 
g(\bA) - g(\bA^{(\nu)})\nn
&\leq 
\frac{1}{2n}\mathbb{E}_G\Big\langle 
\Gamma_{t,\nu}^{{(\nu)}2}
-  \Gamma_{t,\nu}^2 + 2Z_\nu (\Gamma_{t,\nu}^{(\nu)} - \Gamma_{t,\nu})\Big\rangle_{\hat{\mathcal{H}}_{t,\epsilon}}.
\end{align}
Notice that $\big\vert \Gamma_{t,\nu}^{(\nu)2}
-  \Gamma_{t,\nu}^2 + 2Z_\nu (\Gamma_{t,\nu}^{(\nu)} - \Gamma_{t,\nu})\big\vert
\leq 
8 (\sup\vert\varphi\vert)^2 + 8 \vert Z_\nu\vert \sup\vert\varphi\vert$. Thus we conclude by \eqref{entrop-diff-basic} that $g$ satisfies a bounded difference property:
\begin{align}
	\vert g(\bA) - g(\bA^{(\nu)}) \vert \leq \frac{4}{n}\sup\vert\varphi\vert \Big( \sup\vert\varphi\vert + \sqrt\frac{2}{\pi}\Big)\,.
\end{align}
Lemma~\ref{lem:concentration_A} follows then by an application of Proposition~\ref{bounded_diff}.
\end{proof}

\subsubsection{Bounded difference with respect to $X_i^*$}

Let $\E_{\Theta} = \E_{\bA,G}$ denote the expectation w.r.t.\ all quenched variables except $\bX^*$. 
It remains to bound the variance of $\E_{\Theta} \ln \hat{\cZ}_{t,\epsilon}/n$ (which only depends on $\bX^*$).
\begin{lemma}\label{lem:concentration_X}
There exists a constant $C(\varphi, S, \alpha)>0$ such that
	\begin{equation}
		\E \Big[\Big(
				\frac{1}{n}\E_{\Theta} \ln \hat{\cZ}_{t,\epsilon}- \frac{1}{n}\E \ln \hat{\cZ}_{t,\epsilon} 
		\Big)^2\Big] \leq \frac{C(\varphi, S, \alpha)}{n} \,.
	\end{equation}
\end{lemma}

\begin{proof}
	The lemma is proved using again a bounded difference argument.
	Let $g=\mathbb{E}_{\Theta}\ln\hat{\mathcal{Z}}_{t,\epsilon}/n$ a function of $\bX^*$. Let $j \in \{1, \dots, n\}$. Let $\bX^*,\bX^{*(j)} \in [-S,S]^n$ be two input signals such that $X_i^{*(j)}= X_i^*$ for $i\neq j$.

	We are going to interpolate between $g(\bX^*)$ and $g(\bX^{*(j)})$. For $s \in [0,1]$ we define 
	$$
	\psi(s) = g(s \bX^* + (1-s) \bX^{*(j)})\,.
	$$
	Obviously $\psi(1) = g(\bX^*)$ and $\psi(0)=g(\bX^{*(j)})$. Using Gaussian integration by parts, it is not difficult to verify that for $s \in [0,1]$
	$$
	|\psi'(s)| \leq \frac{C(\varphi,S,\alpha)}{n} \,.
	$$
	This implies the bounded difference property $|g(\bX^*) - g(\bX^{*(j)}) | \leq C(\varphi,S,\alpha)/n$ and using Proposition~\ref{bounded_diff} we obtain the lemma.
\end{proof}

\subsubsection{Proof of Theorem~\ref{concentrationtheorem}}

From Lemmas~\ref{lem:concentration_gauss},~\ref{lem:concentration_A} and~\ref{lem:concentration_X} above, we obtain directly that 
$\Var( \ln \hat{\cZ}_{t,\epsilon}/n) \leq C(\varphi, S, \alpha)/n$
for some constant $C(\varphi, S, \alpha)>0$.
As mentioned before this implies, thanks to \eqref{Z-Z}, the Theorem~\ref{concentrationtheorem}.

		\subsection{Concentration of the overlap}\label{appendix-overlap}

In this appendix we provide the proof of Proposition~\ref{concentration}.
Recall the notation $\langle - \rangle_{n, t, \epsilon}$ for the Gibbs bracket associated to the Hamiltonian \eqref{interpolating-ham}.
It is crucial that it preserves the Nishimori identity of Appendix~\ref{appendix-nishimori}, i.e. it must come from an inference problem with known parameters.
Consider the corresponding average  
free entropy $f_{n,\epsilon}(t)$. In this section we think of it as a function of $R_1=R_1(t,\epsilon)$ and $R_2=R_2(t,\epsilon)$ given by \eqref{R1R2}, i.e. $(R_1,R_2)\mapsto f_{n, \epsilon}(t)$.  Similarly the free entropy for a realization of the quenched variables is also viewed here as a function $(R_1,R_2)\mapsto F_{n, \epsilon}(t) \defeq \ln \cZ_{t,\epsilon}(\bY_t,\bY_t',\boldsymbol{\Phi},\bV)/n$. For this section, we drop the indices in the Gibbs bracket $\langle -\rangle_{n,t,\epsilon}$ and simply write $\langle - \rangle$.

Let 
\begin{align*}
\mathcal{L} \defeq \frac{1}{n}\sum_{i=1}^n\Big(\frac{x_i^2}{2} - x_i X_i^* - \frac{x_i Z_i'}{2\sqrt{R_1}} \Big)\,.
\end{align*}
The fluctuations of the 
overlap $Q\defeq n^{-1}\sum_{i=1}^n X_i^* x_i$ and those of $\mathcal{L}$ are related through the remarkable identity
\begin{align}
\mathbb{E}\big\langle (\mathcal{L} - \mathbb{E}\langle \mathcal{L}\rangle)^2\big\rangle
= \frac{1}{4}\mathbb{E}\big\langle (Q - \mathbb{E}\langle Q \rangle)^2\big\rangle
&+ \frac{1}{2}\mathbb{E}[\langle Q^2\rangle -   \langle Q \rangle^2]+ \frac{1}{4n R_1} \mathbb{E}[(X_1^*)^2]\,.
\label{idd}
\end{align}
In particular
\begin{align}
\mathbb{E}\big\langle (\mathcal{L} - \mathbb{E}\langle \mathcal{L}\rangle)^2\big\rangle
\ge \frac{1}{4}\mathbb{E}\big\langle (Q - \mathbb{E}\langle Q \rangle)^2\big\rangle\,. \label{boundFLuctLQ}
\end{align}
A detailed derivation of \eqref{idd} involves only lengthy but straightforward algebra, using the Nishimori identity
and integrations by parts w.r.t.\ the Gaussian noise $Z_i'$, and can be found in Sec. 6 of \cite{barbier_stoInt}. 
Proposition~\ref{concentration} is then a direct consequence of the following:
\begin{proposition}[Concentration of $\mathcal{L}$ on $\mathbb{E}\langle \mathcal{L}\rangle$]\label{L-concentration}
	Let ${\cal B}_n\defeq[s_n,2s_n]^2$, where the sequence $(s_n)\in(0,1/2]^{\mathbb{N}}$. Assume that the interpolation functions $(r_\epsilon)$ and $(q_\epsilon)$ are regular (recall Definition \ref{def:reg}). Under assumptions~\ref{hyp:bounded},~\ref{hyp:c2} and \ref{hyp:phi_gauss2} there exists a constant $C(\varphi, S, \alpha)$ such that
\begin{align}
\int_{{\cal B}_n} d\epsilon\, 
\mathbb{E}\big\langle (\mathcal{L} - \mathbb{E}\langle \mathcal{L}\rangle_{n, t, \epsilon})^2\big\rangle_{n, t, \epsilon} \le \frac{C(\varphi, S, \alpha)}{n^{1/4}}\,.
\end{align}
\end{proposition}

The proof of this proposition is broken in two parts. Notice that 
\begin{align}
\mathbb{E}\big\langle (\mathcal{L} - \mathbb{E}\langle \mathcal{L}\rangle)^2\big\rangle
& = 
\mathbb{E}\big\langle (\mathcal{L} - \langle \mathcal{L}\rangle)^2\big\rangle
+ 
\mathbb{E}\big[(\langle \mathcal{L}\rangle - \mathbb{E}\langle \mathcal{L}\rangle)^2\big]\,.
\end{align}
Thus it suffices to prove the two following lemmas (see the proofs below). The first lemma expresses concentration w.r.t.\ the posterior distribution (or ``thermal fluctuations'') and is an elementary consequence of concavity properties of the free entropy and the Nishimori identity.
\begin{lemma}[Concentration of $\mathcal{L}$ on $\langle \mathcal{L}\rangle$]\label{thermal-fluctuations}
	Under the same hypotheses as in Proposition \ref{L-concentration} we have
\begin{align}\label{integral-form1}
 \int_{{\cal B}_n} d\epsilon\, 
  \mathbb{E} \big\langle (\mathcal{L} - \langle \mathcal{L}\rangle_{n,t,\epsilon})^2 \big\rangle_{n,t,\epsilon}  \le \frac{\rho(1 + \rho)}{n} \,.
\end{align}
\end{lemma}

The second lemma expresses the concentration of the average overlap w.r.t.\ the realizations of quenched disorder variables
and is a consequence of the concentration of the free entropy (more precisely Theorem \ref{concentrationtheorem} in Appendix~\ref{appendix_concentration}).
\begin{lemma}[Concentration of $\langle\mathcal{L}\rangle$ on $\mathbb{E}\langle \mathcal{L}\rangle$]\label{disorder-fluctuations}
Under the same hypotheses as in Proposition \ref{L-concentration} there exists a constant $C(\varphi, S, \alpha)$ such that
\begin{align}\label{integral-form2}
 \int_{{\cal B}_n} d\epsilon\, 
  \mathbb{E}\big[ (\langle \mathcal{L}\rangle_{n,t,\epsilon} - \mathbb{E}\langle \mathcal{L}\rangle_{n,t,\epsilon})^2 \big] \le \frac{C(\varphi, S, \alpha)}{n^{1/4}}\,.
\end{align}
\end{lemma}

We now turn to the proof of Lemmas \ref{thermal-fluctuations} and \ref{disorder-fluctuations}. The main ingredient is a set of formulas for the first two derivatives of the free entropy 
w.r.t. $R_1=R_1(t,\epsilon)$.  
For any given realisation of the quenched disorder,
\begin{align}
 \frac{dF_{n, \epsilon}(t)}{dR_1}  &= -\langle \mathcal{L} \rangle-\frac{1}{2n}\sum_{i=1}^n\Big((X_i^*)^2+\frac{1}{\sqrt{R_1}}X_i^*Z_i'\Big) \,,\label{first-derivative}\\
 \frac{1}{n}\frac{d^2F_{n, \epsilon}(t)}{dR_1^2}  &= \langle \mathcal{L}^2 \rangle - \langle \mathcal{L} \rangle^2 -
 \frac{1}{4 n^2R_1^{3/2}}\sum_{i=1}^n  \langle x_i\rangle Z_i'\,\label{second-derivative}.
\end{align}
Averaging \eqref{first-derivative} and \eqref{second-derivative}, using a Gaussian integration by parts w.r.t. $Z_i'$ and the (Nishimori) identity
$\mathbb{E}\langle x_i X_i^*\rangle  =  \mathbb{E}[\langle x_i\rangle^2]$ we find 
\begin{align}
 \frac{df_{n,\epsilon}(t)}{d R_1} &= -\mathbb{E}\langle \mathcal{L} \rangle-\frac{\rho}{2} =  \frac{1}{2n} \sum_{i=1}^n\mathbb{E}[\langle x_i\rangle^2]-\frac{\rho}{2}\,,\label{first-derivative-average}\\
 \frac{1}{n}\frac{d^2f_{n,\epsilon}(t)}{dR_1^2} &= \mathbb{E}[\langle \mathcal{L}^2 \rangle - \langle \mathcal{L} \rangle^2]
 -\frac{1}{4n^2R_1} \sum_{i=1}^n \mathbb{E}[\langle x_i^2\rangle - \langle x_i\rangle^2]\,\label{average-second-derivative}.
\end{align}

\subsubsection*{Proof of Lemma \ref{thermal-fluctuations}}
From \eqref{average-second-derivative} we have 
\begin{align}
\mathbb{E}\big\langle (\mathcal{L} - \langle \mathcal{L} \rangle)^2\big\rangle
& = 
\frac{1}{n}\frac{d^2f_{n, \epsilon}(t)}{dR_1^2}
+\frac{1}{4n^2R_1} \sum_{i=1}^n \mathbb{E}[\langle x_i^2\rangle - \langle x_i\rangle^2] 
\nonumber \\ &
\leq 
\frac{1}{n}\frac{d^2f_{n,\epsilon}(t)}{dR_1^2} +\frac{\rho}{4n\epsilon_1}\, ,\label{230}
\end{align}
where we used $\mathbb{E}\langle x_i^2\rangle = \mathbb{E}_{P_0}[(X^*)^2]=\rho$ by the Nishimori identity, and $R_1\ge \epsilon_1$. Recall ${\cal B}_n\defeq [s_n,2s_n]^2$. By assumption $q$ and $r$ are regular.
Therefore $R^t:(\epsilon_1,\epsilon_2)\mapsto (R_1(t,\epsilon),R_2(t,\epsilon))$ is a diffeomorphism whose Jacobian $J(R^t)$ verifies $J(R^t)(\epsilon)\ge 1$ for all $\epsilon\in{\cal B}_n$. Integrating over $\epsilon\in{\cal B}_n$ we obtain
\begin{align}
\int_{{\cal B}_n} d\epsilon\,\mathbb{E}\big\langle (\mathcal{L} - \langle \mathcal{L} \rangle)^2\big\rangle
&\le\frac{1}{n}\int_{R^t({\cal B}_n)} \frac{dR_1dR_2}{J(R^t)}\,\frac{d^2f_{n,\epsilon}(t)}{dR_1^2} +\frac{\rho s_n}{4n}\int_{s_n}^{2s_n} \frac{d\epsilon_1}{\epsilon_1}\nn
&\le \frac{1}{n}\int_{R^t({\cal B}_n)} dR_1dR_2\,\frac{d^2f_{n,\epsilon}(t)}{dR_1^2}+\frac{\rho s_n}{4n}\ln 2\,,
\end{align}
where in the integral above $J(R^t)$ is a function of $(R^{t})^{-1}(R_1,R_2)$. Note that from \eqref{R1R2} we have $R^t(\mathcal{B}_n)\subset [s_n, 2s_n + r_{\rm max}]\times [s_n, 2s_n + \rho]$ and therefore
\begin{align}
\int_{{\cal B}_n} d\epsilon\,\mathbb{E}\big\langle (\mathcal{L} - \langle \mathcal{L} \rangle)^2\big\rangle
&\leq \frac{1}{n}\int_{s_n}^{2s_n + \rho}dR_2\Big\{
\frac{df_{n, \epsilon}(t)}{dR_1}\Big\vert_{R_1=s_n}
- \frac{df_{n, \epsilon}(t)}{dR_1}\Big\vert_{R_1=2s_n + r_{\rm max}}\Big\} +\frac{\rho s_n}{4n}\ln2
\nn&\leq \frac{\rho(s_n + \rho)}{n}+\frac{\rho s_n}{4n}\ln 2 \label{93__}
\end{align}
using \eqref{first-derivative-average} combined with $\mathbb{E}\langle x_i^2\rangle =\rho$ to assert that the derivative of the free entropy is bounded in absolute value by $\rho/2$. This concludes the proof of Lemma \ref{thermal-fluctuations} using $s_n\le 1/2$ and $(\ln 2)/4< 1$.
\QEDA

\subsubsection*{Proof of Lemma \ref{disorder-fluctuations}}
Consider the two functions 
\begin{align}\label{new-free}
 \widetilde F(R_1) \defeq F_{n, \epsilon}(t) -\frac{\sqrt{R_1}}{n} 
 S\sum_{i=1}^n\vert Z_i'\vert\,,
 \qquad 
 \widetilde f(R_1) \defeq \E \widetilde F(R_1)= f_{n, \epsilon}(t) - \frac{\sqrt{R_1}}{n} S\sum_{i=1}^n \mathbb{E}\vert Z_i'\vert\,.
\end{align}
Because of 
\eqref{second-derivative} we see that the second derivative of $\widetilde F(R_1)$ is positive so that it is convex
(without this extra term $F_{n,\epsilon}(t)$ is not necessarily convex in $R_1$, although $f_{n,\epsilon}(t)$ is, which can be shown easily).
Note that $\widetilde f(R_1)$ is convex too.
Convexity allows us to use the following lemma (proved at the end of this section):
\begin{lemma}[A bound on differences of derivatives due to convexity]\label{lemmaConvexity}
Let $G(x)$ and $g(x)$ be convex functions. Let $\delta>0$ and define $C^{+}_\delta(x) \defeq g'(x+\delta) - g'(x) \geq 0$ and $C^{-}_\delta(x) \defeq g'(x) - g'(x-\delta) \geq 0$. Then
\begin{align*}
|G'(x) - g'(x)| \leq \delta^{-1} \sum_{u \in \{x-\delta, x, x+\delta\}} |G(u)-g(u)| + C^{+}_\delta(x) + C^{-}_\delta(x)\,.
\end{align*}
\label{thm:bound-by-convexity}
\end{lemma}
From \eqref{new-free}
\begin{align}\label{fdiff}
 \widetilde F(R_1) - \widetilde f(R_1) = F_{n, \epsilon}(t) - f_{n, \epsilon}(t) - \sqrt{R_1} S A \,, \quad \text{with} \quad A = \frac{1}{n}\sum_{i=1}^n \big(\vert Z_i'\vert -\mathbb{E}\vert Z_i'\vert\big)\,.
\end{align} 
and from \eqref{first-derivative}, \eqref{first-derivative-average} we obtain for the difference of derivatives (w.r.t. $R_1$)
\begin{align}\label{derdiff}
 \widetilde F'(R_1) - \widetilde f'(R_1) = 
\mathbb{E}\langle \mathcal{L} \rangle -\langle \mathcal{L} \rangle+\frac{\rho}{2}-\frac{1}{2n}\sum_{i=1}^n\Big((X_i^*)^2+\frac{1}{\sqrt{R_1}}X_i^*Z_i'\Big)  - \frac{SA}{2\sqrt{R_1}} \,.
\end{align}
From \eqref{fdiff}, \eqref{derdiff} it is easy to show that Lemma \ref{lemmaConvexity} implies
\begin{align}\label{usable-inequ}
\vert \langle \mathcal{L}\rangle - \mathbb{E}\langle \mathcal{L}\rangle\vert&\leq 
\delta^{-1} \sum_{u\in \{R_1 -\delta, R_1, R_1+\delta\}}
 \big(\vert F_{n, \epsilon}(t,R_1=u) - f_{n, \epsilon}(t,R_1=u) \vert + S\vert A \vert \sqrt{u} \big)\nn
 &\qquad\qquad
  + C_\delta^+(R_1) + C_\delta^-(R_1) + \frac{S\vert A\vert}{2\sqrt \epsilon_1} +\Big|\frac{\rho}{2}-\frac{1}{2n}\sum_{i=1}^n\Big((X_i^*)^2+\frac{1}{\sqrt{R_1}}X_i^*Z_i'\Big)\Big|
\end{align}
where $C_\delta^+(R_1)\defeq \widetilde f'(R_1+\delta)-\widetilde f'(R_1)\ge 0$ and $C_\delta^-(R_1)\defeq \widetilde f'(R_1)-\widetilde f'(R_1-\delta)\ge 0$. 
We used $R_1\ge \epsilon_1$ for the term $S\vert A\vert/(2\sqrt \epsilon_1)$. Note that $\delta$ will be chosen later on strictly smaller than $s_n$ (namely $\delta = s_n n^{-1/4}$) so that $R_1 -\delta \geq \epsilon_1 - \delta \geq s_n -\delta$ remains positive. Remark that by independence of the noise variables $\mathbb{E}[A^2]  \le  an^{-1}$ for some constant $a >0$; and that by independence between signal and noise, the last term in the absolute value in \eqref{usable-inequ}, call it $B$, is satisfies $\E[B^2]\le bn^{-1}$ for some constant $b>0$. We now square the identity \eqref{usable-inequ} and take its expectation. Then using $(\sum_{i=1}^pv_i)^2 \le p\sum_{i=1}^pv_i^2$ (by convexity), and that $R_1\le K$ ($K=1+\max(\rho,r_{\rm max})$ upper bounds both $R_1$ and $R_2$ given by \eqref{R1R2}), as well as the free entropy concentration Theorem \ref{concentrationtheorem},
\begin{align}\label{intermediate}
 \frac{1}{10}\mathbb{E}\big[\big(\langle \mathcal{L}\rangle - \mathbb{E}\langle \mathcal{L}\rangle\big)^2\big]
 \leq &\, 
 3\delta^{-2} \big(C +aS^2(K+\delta)\big) \frac{1}{n}
 + C_\delta^+(R_1)^2 + C_\delta^-(R_1)^2
 + \frac{S^2a}{4\epsilon_1n} + \frac{b}{n}\,.
\end{align}
where $C=C(\varphi, S, \alpha)$ is a positive constant depending only on $\varphi, S$ and $\alpha$ that comes from the use of Theorem \ref{concentrationtheorem}.
Recall $|C_\delta^\pm(R_1)|=|\widetilde f'(R_1\pm\delta)-\widetilde f'(R_1)|$. We have
\begin{align}
|\widetilde f'(R_1)|  \leq \frac12\Big(\rho  +\frac{S}{\sqrt R_1} \Big)\leq \frac12\Big(\rho  +\frac{S}{\sqrt \epsilon_1} \Big)\label{boudfprime}	
\end{align}
from \eqref{first-derivative-average}, \eqref{new-free} and $R_1\ge \epsilon_1$. This implies $|C_\delta^\pm(R_1)|\le \rho  +S/\sqrt \epsilon_1\le \rho  +S/\sqrt s_n$ as $\epsilon_1\ge s_n$. Recall also that ${\cal B}_n\defeq[s_n,2s_n]^2$. Then
\begin{align}
 \int_{{\cal B}_n} d\epsilon\, &\big\{C_\delta^+(R_1(t,\epsilon))^2 + C_\delta^-(R_1(t,\epsilon))^2\big\}\nn
 &\leq 
 \Big(\rho +\frac{S}{\sqrt s_n}\Big)
 \int_{{\cal B}_n} d\epsilon\, \big\{C_\delta^+(R_1(t,\epsilon)) + C_\delta^-(R_1(t,\epsilon))\big\}
 \nonumber \\ &
 =\Big(\rho+\frac{S}{\sqrt s_n}\Big)
 \int_{R^t({\cal B}_n)} \frac{dR_1dR_2}{J(R^t)}\, \big\{C_\delta^+(R_1) + C_\delta^-(R_1)\big\}
 \nonumber \\ &
  \le \Big(\rho +\frac{S}{\sqrt s_n}\Big)
 \int_{R^t({\cal B}_n)} dR_1dR_2\, \big\{C_\delta^+(R_1) + C_\delta^-(R_1)\big\}
 \nonumber \\ &
\leq  
\Big(\rho +\frac{S}{\sqrt s_n}\Big)\int_{s_n}^{2s_n + \rho} dR_2 \Big[\Big(\widetilde f(2s_n + r_{\rm max}+\delta) - \widetilde f(2s_n + r_{\rm max}-\delta)\Big) \nn
&\qquad\qquad\qquad\qquad\qquad\qquad\qquad\qquad+ \Big(\widetilde f(s_n-\delta) - \widetilde f(s_n+\delta)\Big)\Big]\,.
\end{align}
where we used that the Jacobian $J(R^t)$ of the ${\cal C}^1$ diffeomorphism $R^t:(\epsilon_1,\epsilon_2)\mapsto (R_1(t,\epsilon),R_2(t,\epsilon))$ is greater or equal to $1$ (by regularity of the interpolation functions $q$ and $r$)  and
$R^t(\mathcal{B}_n)\subset [s_n, 2s_n + r_{\rm max}]\times [s_n, 2s_n + \rho]$. The mean value theorem and \eqref{boudfprime} imply $|\widetilde f(R_1-\delta) - \widetilde f(R_1+\delta)|\le \delta(\rho  +S/\sqrt s_n)$ uniformly in $R_2$. Therefore
\begin{align}
 \int_{{\cal B}_n} d\epsilon\, &\big\{C_\delta^+(R_1(t,\epsilon))^2 + C_\delta^-(R_1(t,\epsilon))^2\big\}\leq 
 2\delta(s_n+\rho) \Big(\rho +\frac{S}{\sqrt s_n}\Big)^2\,.
\end{align}
Thus, integrating \eqref{intermediate} over $\epsilon\in {\cal B}_n$ yields (using ${\rm Vol}({\cal B}_n)=s_n^2$ and $s_n\le 1/2$)
\begin{align*}
 &\int_{{\cal B}_n} d\epsilon\, 
 \mathbb{E}\big[\big(\langle \mathcal{L}\rangle- \mathbb{E}[\langle \mathcal{L}\rangle\big)^2\big]\nn
 &\qquad\qquad\leq 30\big(C +aS^2\sqrt{K+\delta}\big)\delta^{-2} s_n^2 n^{-1} +20\delta(1/2+\rho) \Big(\rho +\frac{S}{\sqrt s_n}\Big)^2 + 5S^2a \frac{\ln 2}{2}\,\frac{s_n}{n} + \frac{bs_n^2}{n} \,.
\end{align*}
Finally we choose $\delta  =  s_n n^{-1/4}$ and obtain the desired result.$\QEDA$

\begin{proof}[Proof of Lemma \ref{lemmaConvexity}:]
Convexity implies that for any $\delta>0$ we have
\begin{align*}
G'(x) - g'(x)
 	& \leq \frac{G(x+\delta) - G(x)}{\delta} - g'(x) \nonumber \\
	& \leq \frac{G(x+\delta) - G(x)}{\delta} - g'(x) + g'(x+\delta) - \frac{g(x+\delta) - g(x)}{\delta} \nonumber \\
	& = \frac{G(x+\delta) - g(x+\delta)}{\delta} - \frac{G(x) - g(x)}{\delta} + C^+_\delta(x) \,,  \\
G'(x) - g'(x)
	& \geq \frac{G(x) - G(x-\delta)}{\delta} - g'(x) + g'(x-\delta) - \frac{g(x) - g(x-\delta)}{\delta} \nonumber \\
	& = \frac{G(x) - g(x)}{\delta} - \frac{G(x-\delta) - g(x-\delta)}{\delta} - C^-_{\delta}(x) \,.
\end{align*}
Combining these two inequalities ends the proof.
\end{proof}
    \section{Details on numerics}
Most of our experiments and codes are provided on the associated GitHub
repository \cite{githubrepo}, with codes in the Julia programming
language \cite{Julia} (with a Jupyter notebook interface) and in
matlab.  In this appendix, we shall give additional details on how the
plots have been obtained.
\subsection{General purpose algorithms}   
We have been using free available softwares in our experiments:
Standard machine learning tasks such as LASSO or logistic regression
were done using scikit-learn\cite{scikit-learn}. Keras
\cite{chollet2015}, with a tensorflow backend
\cite{tensorflow2015-whitepaper}, was used for neural networks. We
also used CVXPY, a python-embedded language for performing convex
optimization \cite {cvxpy}, as well as PhaseMax for phase retrieval
\cite{phasemax} experiments. 

Figure \ref{fig:classif} (in both the main text and SI) constrasts results of these general purpose algorithms with the optimal
generalization error in three classification problems. In the left
pannel of Fig. \ref{fig:classif} (for the binary perceptron), we used logistic regression with hand-tuned $\ell_2$ regularization
(basically, we have hand-selected the regularization parameter in
order to obtain the best results) with the function Logistic Regression
in the software scikit-learn \cite{scikit-learn}. In the center pannel
of Fig. \ref{fig:classif}
(this time for a sparse signal), we used the same software but this
time with a sparsity enhancing $\ell_1$ regularization, again
fine-tuned by hand.
In the right pannel of Fig. \ref{fig:classif}, we show, in the inset, how a neural
network with $2$ hidden layers was able to learn only approximately
the ``symmetric door'' rule. In this experiment, we used Keras with a
tensorflow backend. 

The data are created with a signal-vector of dimension $n=2500$ and are
sent into a network made with a first layer of dimension
$2500 \times 64$ followed by a rectified linear unit (ReLu), and a
dropout layer with fraction $0.2$ for regularization. This is followed
by a second layer of dimension $64 \times 64$, this time with a sigmoid
activation and again a dropout layer. Finally, we classify with a
final output layer with a softmax, using the categorical cross-entropy
as the loss function. The minimization is done using the RMSprop
optimizer for $1000$ epochs. The code is shown in the GitHub
repository \cite{githubrepo}.

We have tried many variations around this network. Interestingly, the
dropout layers have a strong effect on the regularization,  and help
significantly in improving the generalization error. 
Also interestingly, the number of epochs used for fitting was an
important parameter. Indeed, the quality of the fit improves
drastically as the number of epochs is increased: It seems that it
actually takes a lot of time to escape the initial point, where
prediction is just as bad as random. 
Finally, we also tried to increase the depth of the neural
nets. Interestingly this did not affect the performance and
the network was fitting the data and generalizing just as well with
deeper and deeper networks. We believe that it should be very
instructive to further study empirically this problem.

In Fig. \ref{fig:regression} in this SI, simular plots are
shown for three regression problems. In the left pannel, we used
the LASSO function in the software scikit-learn \cite{scikit-learn}
with its sparsity enhancing $\ell_1$ regularization, again fine-tuned
by hand. 

For the middle pannel, we had to turn to a different software. In this
case, the idea was to solve ${\bY}={\rm Relu}(\bbf{\Phi} {\bx})=\max(0,\bbf{\Phi} {\bx})$ (componentwise) subject
to a sparse penalty on $\bx$. Luckily, this can be turned into a
linear programming framework: Minimize the $\ell_1$ norm of $\bx$
subject to the constraint that ${\bY}={\rm Relu}(\bbf{\Phi} {\bx})$ which
can be implemented by enforcing $Y_{\mu}=\bbf{\Phi}_{\mu} \cdot \bx$ for
$\mu$'s such that $Y_{\mu}>0$, and $\bbf{\Phi}_{\mu} \cdot \bx<0$ for $\mu$'s such that 
$Y_{\mu}=0$. This linear program is solved with CVXPY \cite {cvxpy}
(CVX-2). We also show the results when only the indices $\mu$ associated with positive $Y_{\mu}>0$ are
used for comparaison (CVX-1). Finally, for the right figure, we used
PhaseMax \cite{phasemax} out-of-the-box to solve the problem.

\subsection{Evaluating the replica formula}   
In order to evaluate numerically the replica formula (\ref{frs}) we proceed as
it is common in the statistical physics literature since the early
papers on spin glasses. First, we found the critical points (\ref{fixed_points}) by
iterating the state evolution equations (\ref{GAMP_SE}) ---also called the
replica self-consistent equations--- starting from two different
initial conditions ($q^{t=0}=0$ and $q^{t=0}=\rho$). Next we computed the associated value of the
free entropy and then selected, if two different fixed points were found,
the correct one following the prescription given by Theorem \ref{th:RS_1layer}. We
also took special care in checking that we could not identify other
fixed points. An example of such a procedure is shown in the GitHub
repository \cite{githubrepo}, with codes in the Julia programming
language \cite{Julia}, for the perceptron problem.
\subsection{Breaking the symmetry in GAMP} 
A last notable point concerns the symmetry issue in GAMP. Indeed,
when $q=0$ is a fixed point of the state evolution (which is the case if
the prior has zero mean and the channel is symmetric), then GAMP should stay in this fixed point forever. This is the
case, for instance, for two problems considered in the present paper:
The symmetric door output function $\varphi(z)={\rm sgn}(|z|-K)$ with a
Rademacher prior $\pm 1$, and for the sign-less channel
$\varphi(z)=|z|$. In both cases, both $z$ and $-z$ are giving the same
output, and therefore so does both $\bX^*$ and $-\bX^*$.
Notice, however, that this is not a problem in the computation of the
free entropy. Here, one has to compute all the fixed points anyway. It
is also not a problem if the prior is breaking the symmetry (for
instance if one is working with a binary signal where $X^*_i=1$ with
probability $p_+=1/2+\epsilon$ and $X^*_i=-1$ with probability
$p_-=1/2-\epsilon$ with, say, $\epsilon=10^{-7}$. In this case the symmetry is broken, GAMP works, and
the state evolution predicts its behavior correctly.
Even though this problem is thus restricted to a very small class of
channels and priors, and even though perturbations solve
it, it is still an interesting mathematical challenge, especially from
the rigorous point of view. Indeed, this problem has attracted 
attention recently where initializations based on spectral algorithms
were analyzed \cite{mondelli2017fundamental}.

In the present paper, we adopted a pragmatic point of view. We did not
break the symmetry in the data generative model (as that would make
the problem slightly easier), instead we broke the symmetry
in the GAMP solver, thus making it slightly but un-noticeably suboptimal. 
An example of our code is given in the associated GitHub repository
\cite{githubrepo}. We created the data in the symmetric manner, but when we ran GAMP
to solve the problem, we broke the channel symmetry slightly. For instance, instead of solving with a door function that
returns $1$ only for $-0.674489<z<0. 674489$, we use a function that
instead returned $1$ for $-0.674489<z<0. 6745$. 
The same
strategy was used for the absolute value function, which can be
replaced by $\varphi(z)=z$ for $x>-\epsilon$ and $-z$ otherwise. Again, when
$\epsilon$ is small, this makes only an un-noticeable difference in the figures. 
This trick allowed GAMP to solve symmetric problems without
trouble in practice, and to reach perfect recovery even in the
symmetric problems as close to the theoretical threshold as
numerically desired. 

\end{appendices}
%
%
\end{document}